\documentclass[12pt,letterpaper]{article}
\usepackage[T1]{fontenc}
\usepackage[margin=1in]{geometry} 

\usepackage[titletoc]{appendix}
\usepackage{titling}
\usepackage[authoryear]{natbib}
\usepackage[dvipsnames]{xcolor}
\usepackage[colorlinks=true,
            linkcolor=Blue,
            citecolor=Blue,
            urlcolor=blue, 
            hyperfootnotes=false]{hyperref}
            
\usepackage[bottom]{footmisc}

\usepackage[onehalfspacing]{setspace} 

\usepackage[final,nopatch=footnote]{microtype}
\usepackage{lmodern}
\usepackage[subtle]{savetrees}
\usepackage[small,compact]{titlesec}
\titlespacing{\paragraph}{0pt}{2.25ex plus 1ex minus .2ex}{0.4em}

\usepackage{graphicx, csquotes, enumitem, booktabs, threeparttablex, longtable, pdflscape, array}
\usepackage{mathtools, amssymb, amsthm, bm, dsfont, etoolbox}
\allowdisplaybreaks

\theoremstyle{definition} 

\newcommand{\continuation}{??}

\newtheorem*{example*}{Example}
\newtheorem{remark}{Remark}
\newtheorem*{remark*}{Remark}

\makeatletter
\newcommand{\exampleqedsymbol}{$\triangle$}

\AtBeginEnvironment{example}{%
  \pushQED{\qed}%
  \let\example@oldsym\qedsymbol
  \renewcommand{\qedsymbol}{\exampleqedsymbol}%
}
\AtEndEnvironment{example}{%
  \popQED%
  \let\qedsymbol\example@oldsym
}

\AtBeginEnvironment{example*}{%
  \pushQED{\qed}%
  \let\example@oldsym\qedsymbol
  \renewcommand{\qedsymbol}{\exampleqedsymbol}%
}
\AtEndEnvironment{example*}{%
  \popQED%
  \let\qedsymbol\example@oldsym
}
\makeatother

\makeatletter
\newcommand{\remarkqedsymbol}{$\triangle$}

\AtBeginEnvironment{remark}{%
  \pushQED{\qed}%
  \let\remark@oldsym\qedsymbol
  \renewcommand{\qedsymbol}{\remarkqedsymbol}%
}
\AtEndEnvironment{remark}{%
  \popQED%
  \let\qedsymbol\remark@oldsym
}

\AtBeginEnvironment{remark*}{%
  \pushQED{\qed}%
  \let\remark@oldsym\qedsymbol
  \renewcommand{\qedsymbol}{\remarkqedsymbol}%
}
\AtEndEnvironment{remark*}{%
  \popQED%
  \let\qedsymbol\remark@oldsym
}
\makeatother

\usepackage{etoolbox}

\AtBeginEnvironment{procedure}{\vspace{\topsep}\par\kern8pt\hrule\kern2pt\hrule\relax}

\newtheorem{recipe}{Recipe}
\AtBeginEnvironment{recipe}{\vspace{\topsep}\par\kern8pt\hrule\kern2pt\hrule\relax}

\newenvironment{recipenumerate}{%
  \,
  \par\kern8pt\hrule\relax
  \begin{enumerate}%
}{%
  \end{enumerate}%
  \par\kern4pt\hrule\kern2pt\hrule\relax
  \kern8pt
}

\theoremstyle{plain} 
\newtheorem{assumption}{Assumption}

\newtheorem{proposition}{Proposition}

\newtheorem{lemma}{Lemma}

\newtheorem{assumptionU}{Assumption}

\newtheorem{propositionU}{Proposition}

\newtheorem{assumptionUK}{Assumption}

\newtheorem{propositionUK}{Proposition}

\newcommand{\R}{\mathbb{R}} 
\renewcommand{\P}[2][\theta]{\mathbb{P}_{#1}\left\{#2\right\}} 
\newcommand{\E}[2][\theta]{\mathbb{E}_{#1}\left[#2\right]} 
\newcommand{\Var}[2][\theta]{\textnormal{Var}_{#1}\left(#2\right)} 
\newcommand{\e}{\varepsilon} 
\newcommand{\1}{\bm{1}} 
\newcommand{\I}[1]{\mathds{1}\left\{#1\right\}} 
\newcommand{\ind}{\perp} 
\newcommand{\abs}[1]{\left|#1\right|} 
\newcommand{\norm}[1]{\left\|#1\right\|} 
\renewcommand{\to}[1][]{\overset{#1}{\rightarrow}} 
\newcommand{\dsqrt}[1]{\displaystyle\sqrt{#1}}

\DeclareMathOperator*{\argmin}{argmin} 
\DeclareMathOperator{\ATT}{ATT} 
\DeclareMathOperator{\ATE}{ATE} 
\DeclareMathOperator{\TWFE}{TWFE} 
\DeclareMathOperator*{\diam}{diam} 

\newcommand{\GLS}{\mathrm{GLS}} 
\newcommand{\SA}{\mathrm{SA}} 
\newcommand{\EW}{\mathrm{EW}} 

\newcommand{\curly}[1]{\left\{#1\right\}}
\renewcommand{\brack}[1]{\left[#1\right]}
\newcommand{\paren}[1]{\left(#1\right)}
\newcommand{\inner}[1]{\left\langle#1\right\rangle}
\newcommand{\sumk}{\sum_{k=1}^{K}}

\newcommand{\htheta}{\hat{\theta}}
\newcommand{\hB}{\widehat{B}}

\newcommand{\heta}{\hat{\eta}}
\newcommand{\teta}{\Tilde{\eta}}
\newcommand{\htau}{\hat{\tau}}
\newcommand{\htauw}{\hat{\tau}_{w}}
\newcommand{\htaul}{\hat{\tau}_{\lambda}}
\newcommand{\tauw}{\tau_{w}(\theta)}
\newcommand{\taul}{\tau_{\lambda}(\theta)}

\newcommand{\tsigma}{\Tilde{\sigma}}
\newcommand{\tSigma}{\Tilde{\Sigma}}

\newcommand{\hSigma}{\widehat{\Sigma}}

\renewcommand{\l}{\lambda}
\renewcommand{\L}{\Lambda}
\newcommand{\cE}{\mathcal{E}}
\newcommand{\W}{\mathcal{W}}

\newcommand{\bX}{\bm{X}}

\newcommand{\cv}[2][1-\alpha]{\mathrm{cv}_{#1}\left(#2\right)}

\newcommand{\hw}{\hat{w}}
\newcommand{\hL}{\hat{\Lambda}}
\newcommand{\hS}{\hat{S}}
\newcommand{\hbX}{\hat{\bm{X}}}

\newcommand{\tH}{\Tilde{H}}

\newcommand{\tQ}{\Tilde{Q}}

\newcommand{\tA}{\Tilde{A}}

\newcommand{\Pn}{P_{n}}
\newcommand{\cPn}{\mathcal{P}_{n}}
\newcommand{\Kn}{K_{n}}

\newcommand{\cU}{\mathcal{U}}
\newcommand{\cUn}{\mathcal{U}_{n}}

\newcommand{\Pns}[1][s]{P_{#1}}

\newcommand{\hthetan}{\hat{\theta}_{n}}
\newcommand{\thetan}{\theta_{n}}

\newcommand{\cthetan}{\Check{\theta}_{n}}
\newcommand{\hZn}{\hat{Z}_{n}}

\newcommand{\hSigman}{\widehat{\Sigma}_{n}}
\newcommand{\tSigman}{\Tilde{\Sigma}_{n}}
\newcommand{\Sigman}{\Sigma_{n}}

\newcommand{\eSigman}{\varepsilon_{\Sigma,n}}
\newcommand{\rhon}{\rho_{n}}

\newcommand{\hthetans}[1][s]{\hat{\theta}_{#1}}
\newcommand{\thetans}[1][s]{\theta_{#1}}

\newcommand{\hZns}[1][s]{\hat{Z}_{#1}}
\newcommand{\Zns}[1][s]{Z_{s}}
\newcommand{\hSigmans}[1][s]{\widehat{\Sigma}_{#1}}

\newcommand{\Sigmans}[1][s]{\Sigma_{#1}}

\newcommand{\limZ}{Z_{\infty}}
\newcommand{\limSigma}{\Sigma_{\infty}}
\newcommand{\limsigma}{\sigma_{\infty}}

\newcommand{\hwn}{\hat{w}_{n}}
\newcommand{\wn}{w_{n}}
\newcommand{\cwn}{\Check{w}_{n}}
\newcommand{\bbW}{\mathbb{W}}
\newcommand{\Wn}{\mathcal{W}_{n}}
\newcommand{\bbWn}{\mathbb{W}_{n}}

\newcommand{\hwns}[1][s]{\hat{w}_{#1}}
\newcommand{\wns}[1][s]{w_{#1}}

\newcommand{\limw}{w_{\infty}}

\newcommand{\htauwn}[1][\hat{w}_{n}]{\hat{\tau}_{#1,n}}
\newcommand{\tauwn}[1][w_{n}]{\tau_{#1,n}}

\newcommand{\tsigmawn}{\Tilde{\sigma}_{\hat{w}_{n},n}}

\newcommand{\hsigmawns}{\hat{\sigma}_{\hat{w}_{s},s}}
\newcommand{\tsigmawns}{\Tilde{\sigma}_{\hat{w}_{s},s}}

\renewcommand{\ln}{\lambda_{n}}
\newcommand{\hLn}{\hat{\Lambda}_{n}}
\newcommand{\Ln}{\Lambda_{n}}
\newcommand{\hbXn}{\hat{\bm{X}}_{n}}
\newcommand{\bXn}{\bm{X}_{n}}

\newcommand{\hSn}{\hat{S}_{n}}
\newcommand{\Sn}{S_{n}}
\newcommand{\cS}{\mathcal{S}}
\newcommand{\bbS}{\mathbb{S}}
\DeclareMathOperator{\dist}{dist}

\newcommand{\gn}{g_{n}}
\newcommand{\Lgn}{L_{g,n}}
\newcommand{\deltagn}{\delta_{g,n}}

\newcommand{\lns}[1][s]{\lambda_{#1}}
\newcommand{\hLns}[1][s]{\hat{\Lambda}_{#1}}
\newcommand{\Lns}[1][s]{\Lambda_{#1}}

\newcommand{\hSns}[1][s]{\hat{S}_{#1}}
\newcommand{\Sns}[1][s]{S_{#1}}

\newcommand{\limS}{S_{\infty}}
\newcommand{\limL}{\Lambda_{\infty}}

\newcommand{\rhwn}{\hat{w}^{*}_{n}\,}

\newcommand{\rwn}{w^{*}_{n}}

\newcommand{\rhtaun}{\hat{\tau}_{n}^{*}}

\newcommand{\rCIn}{CI^{*}_{\hat{w}_{n},n}}

\newcommand{\hBn}{\widehat{B}_{\hat{w}_{n},n}^{\beta}(\hat{\Lambda}_{n})}

\newcommand{\rCIns}[1][s]{CI^{*}_{\hat{w}_{#1},#1}}

\newcommand{\hBns}[1][s]{\widehat{B}_{\hat{w}_{#1},#1}^{\beta}(\hat{\Lambda}_{#1})}

\newcommand{\hHn}{\hat{H}_{n}}
\newcommand{\tHn}{\Tilde{H}_{n}}
\newcommand{\hQn}{\hat{Q}_{n}}

\newcommand{\hAn}{\hat{A}_{n}}

\newcommand{\Hn}{H_{n}}
\newcommand{\Qn}{Q_{n}}
\newcommand{\An}{A_{n}}
\newcommand{\hetan}{\hat{\eta}_{1-\beta,n}}
\newcommand{\tetan}{\Tilde{\eta}_{1-\beta,n}}

\newcommand{\hHns}[1][s]{\hat{H}_{#1}}

\newcommand{\hQns}[1][s]{\hat{Q}_{#1}}

\newcommand{\hAns}[1][s]{\hat{A}_{#1}}

\newcommand{\Hns}[1][s]{H_{#1}}
\newcommand{\Qns}[1][s]{Q_{#1}}
\newcommand{\Ans}[1][s]{A_{#1}}
\newcommand{\hetans}[1][s]{\hat{\eta}_{1-\beta,#1}}
\newcommand{\tetans}[1][s]{\Tilde{\eta}_{1-\beta,#1}}

\newcommand{\vns}[1][s]{v_{#1}}
\newcommand{\hvns}[1][s]{\hat{v}_{#1}}

\newcommand{\Dns}[1][s]{D_{#1}}
\newcommand{\hDns}[1][s]{\hat{D}_{#1}}

\newcommand{\Deltans}[1][s]{\Delta_{#1}}
\newcommand{\hDeltans}[1][s]{\hat{\Delta}_{#1}}

\newcommand{\limQ}{Q_{\infty}}
\newcommand{\limA}{A_{\infty}}
\newcommand{\limeta}{\eta_{\infty}}

\newcommand{\limu}{u_{\infty}}
\newcommand{\limv}{v_{\infty}}
\newcommand{\limD}{D_{\infty}}
\newcommand{\limDelta}{\Delta_{\infty}}

\newcommand{\limn}{\lim_{n \to \infty}}
\newcommand{\limsupn}{\limsup_{n \to \infty}}
\newcommand{\liminfn}{\liminf_{n \to \infty}}
\newcommand{\supPn}{\sup_{P_{n} \in \mathcal{P}_{n}}}
\newcommand{\infPn}{\inf_{P_{n} \in \mathcal{P}_{n}}}
\newcommand{\emin}{e_{\min}}
\newcommand{\emax}{e_{\max}}
\newcommand{\UB}[1]{\Bar{#1}} 
\newcommand{\LB}[1]{\underline{#1}}

\begin{document}

\title{Robust Inference for Weighted Estimands\thanks{I thank Isaiah Andrews, Anna Mikusheva, and Alberto Abadie for their guidance and support. I thank Kirill Borusyak, Patrick Kline, Ricky Li, Sarah Moon, Andreas Petrou-Zeniou, Jonathan Roth, Reca Sarfati, Liyang Sun, Whitney Zhang, and seminar participants from the MIT econometrics lunch for helpful comments and discussions. I gratefully acknowledge support from the Jerry A. Hausman Fellowship and the National Science Foundation Graduate Research Fellowship under Grant No. 1745302. I thank ChatGPT and Codex for research assistance. Refine.ink was used to check for consistency and clarity.}}
\author{Vod Vilfort\thanks{Department of Economics, Massachusetts Institute of Technology, vod@mit.edu.}}
\date{\today}
\maketitle

\begin{abstract}
\noindent Researchers often conduct inference on weighted estimands, defined as weighted averages of group-level effects. Example settings include event studies with cohort-level effects and experiments with site-level effects. Under heterogeneous effects, different weighting schemes yield estimands with distinct empirical and policy interpretations, leading to ambiguity and disagreement over the choice of weights. I establish bounds on differences between weighted estimands and confidence bounds on effect heterogeneity, which I use to construct estimators that minimize worst-case bias and confidence intervals that are uniformly valid over classes of weighted estimands. I apply these methods to an event study in Lakdawala, Nakasone, and Kho (2023), which studies the effects of school-based internet access on test scores. I find that results are robust to broad classes of weights. I then apply the methods to Tennessee's Project STAR experiment and find that results are sensitive to small departures from baseline weights.
\end{abstract}

\clearpage

\section{Introduction}\label{arxiv2:sec:introduction}
In empirical research, many estimands can be expressed as weighted averages of group-level parameters.\footnote{Under various assumptions, ordinary least squares (OLS) estimands can be expressed as weighted averages of conditional average treatment effects (CATEs), conditional mean derivatives, or other regression parameters that may vary across covariate cells and treatment intensities \citep{yitzhaki1996using, angrist1998estimating, aronow2016does, sloczynski2022interpreting, goldsmith2024contamination}; two-stage least squares (TSLS) estimands can be expressed as weighted averages of local average treatment effects (LATEs), which may vary across compliance types and covariate cells \citep{imbens1994identification, heckman2005structural, angrist2010extrapolate, kolesar2013ivheterogeneity, sloczynski2020should, huntington2020instruments, coussens2021improving, blandhol2022tsls, abadie2024instrumental}; and two-way fixed effects (TWFE) estimands can be expressed as weighted averages of average treatment effects on the treated (ATTs), which may vary across treatment cohorts and time periods \citep{de2020two, callaway2021difference, goodman2021difference, sun2021estimating, athey2022design, gardner2022two, borusyak2024revisiting, wing2024stacked}.} These \textit{weighted estimands} give concise summaries of the parameters of interest, making them frequent targets for statistical inference. However, conventional inference procedures can be sensitive to the choice of weights: conclusions reached under a researcher’s baseline weights may not hold under alternative weights entertained by readers. 

For example, in event studies with multiple treatment cohorts, many common estimands recover weighted averages of cohort-level causal effects, including conventional two-way fixed effects (TWFE) regression estimands and recently proposed alternatives.\footnote{For a review, see \citet{roth2023s}.} TWFE has a convenient and flexible structure, which can yield lower-variance estimators in practice \citep{de2022two, armstrong2025adapting}. However, TWFE can place negative weight on some of the cohort-level effects, which complicates the interpretation of the resulting estimands under heterogeneous effects \citep{de2020two, goodman2021difference, sun2021estimating, borusyak2024revisiting}. Recent alternatives allow a researcher to target weighted estimands that rule out negative weights by design, but this introduces an additional degree of freedom: the researcher must decide which positive weights to use.

This choice matters because different readers may prefer different weights, depending on the empirical or policy question that they have in mind \citep{callaway2021difference}. One reader may wish to upweight early-treated cohorts, another later-treated cohorts; one may care about short-run effects, another long-run effects; one may prefer weights that represent the characteristics of a policy-relevant target population, another may favor weights that trade off representativeness against estimation precision. These alternatives may all be reasonable, but they need not lead to the same conclusion under heterogeneous effects. Given the potential for disagreement, researchers may wish to assess the robustness of conclusions to the choice of weights. Current practice seems to reflect this desire: researchers often compare results for different weighting choices, such as default implementations of the recently proposed alternatives. In a selected survey of 60 recent American Economic Association (AEA) journal papers that cite the event-study literature on weighting and heterogeneity issues, 44 of the 47 papers that report a TWFE specification also report at least one heterogeneity-robust alternative.\footnote{For the survey results and additional details, see Section \ref{arxiv2:sec:implementation.event.studies} and Appendix \ref{arxiv2:app:event-study-survey}.} Such comparisons are informative, but they generally consider only a finite menu of weighting choices---which may not cover the range of plausible weights---and may not provide formal inferential guarantees for broader claims of robustness to the choice of weights.

In this paper, I develop robust inference procedures that directly account for ambiguity and disagreement over the choice of weights. I consider a general framework in which a researcher observes asymptotically normal estimates for a vector of parameters. The researcher reports a conventional point estimator and confidence interval (CI) for a baseline estimand defined by a vector of baseline weights. At the same time, the researcher faces a broad class of alternative weights, giving rise to alternative estimands for which the baseline estimator may be biased and the baseline CI may undercover. To address these concerns, I develop robust estimators that mitigate worst-case bias and robust CIs that ensure valid coverage.

To proceed, I first establish a sharp bound on the difference between the baseline weighted estimand and any given alternative weighted estimand. This bound is the product of (i) the \textit{heterogeneity in parameters}, defined as the square root of the residual sum of squares from a generalized least squares (GLS) regression of the parameters on a constant and (ii) the \textit{distance between weights}, defined as the standard deviation of the difference in estimators under the baseline and alternative weights. 
 
To develop my estimator, I study the problem of choosing weights to minimize the \textit{maximum distance} between weights over a given class of alternative weights. For a broad range of classes, this problem has a unique solution. I show that the corresponding weighted estimator is optimal for minimizing the maximum bias over the class of alternative estimands under any bound on the heterogeneity in parameters. I therefore refer to this estimator as the \textit{robust estimator}, and to the corresponding weights as the \textit{minimax-bias weights}. I propose that researchers report the robust estimator alongside the baseline estimator. Intuitively, the robust estimator accounts for the bias concerns of readers who may disagree with the baseline weights.

To develop my CI, I study the problem of constructing an upper confidence bound (UCB) for the heterogeneity in parameters. I show that a monotonic transformation of the heterogeneity in estimates produces an optimal UCB. This \textit{heterogeneity UCB}—together with the maximum distance between weights—yields a simple adjustment to the critical values of a baseline CI. The resulting \textit{robust CI} provides uniformly valid coverage over the class of alternative estimands. In particular, for each alternative estimand, the coverage probability is at least one minus the sum of (i) the significance level of the baseline CI and (ii) the significance level of the heterogeneity UCB. For example, suppose a 95\% baseline CI excludes zero, leading the researcher to reject the null of no average effect at the 5\% level. If zero remains excluded from a robust CI constructed using a 95\% heterogeneity UCB, then readers with alternative weights can robustly reject the null of no average effect at the 10\% level. To facilitate such inference procedures, I propose that researchers report the robust CI alongside the baseline CI.

The robust estimator and CI require the researcher to specify a class of alternative weights. My framework accommodates any compact and convex class of alternatives. I highlight three such classes that cover a broad range of empirical contexts. Below I give a high-level overview of these classes, deferring precise definitions and details to Section \ref{arxiv2:sec:setting} of the paper.

The first class is the \textit{bounded variance class}, which restricts attention to weights that yield estimators with variance no larger than a given bound, ensuring that the alternative estimands can be estimated with a reasonable degree of precision. For example, when assessing robustness of results, a researcher may wish to consider alternative estimands that can be estimated at least as precisely as the baseline. Under the bounded variance class, my proposed procedures admit a convenient reduction to the GLS regression of the group-level estimates on a constant, where (i) the GLS estimator coincides with the robust estimator, (ii) the GLS variance determines the maximum distance between weights for any choice of variance bound, and (iii) the GLS residual sum of squares determines the heterogeneity UCB. This reduction has two key implications. First, because the GLS estimator minimizes variance in the class of weighted estimators, the robust estimator is optimal for both bias and variance under the bounded variance class. Second, given the structure in (ii) and (iii), readers can construct robust CIs for their own choices of variance bounds, provided the researcher reports the GLS variance and heterogeneity UCB. These convenient properties make the bounded variance class a natural benchmark.

The simplex is the set of nonnegative weights, which yields weighted estimands that do not extrapolate beyond the range of the parameters. However, the unrestricted simplex allows for some groups to receive zero weight, which can be undesirable in practice. To ensure that groups are represented, one can impose a floor on the simplex. The \textit{truncated simplex class} is precisely the set of simplex weights that are bounded below by the given floor. Under positive baseline weights, the floor can be parameterized so that the truncated simplex yields the set of weights that are a given fraction between the baseline weights and the unrestricted simplex weights. In this form, the truncated simplex class allows one to investigate perturbations of any chosen magnitude from the baseline weights.

The simplex weights represent probability distributions over the set of groups, and can thus be viewed as different target populations that readers may be interested in. For instance, in settings where the groups index the sites of an experiment, a concern is that baseline results may not generalize to future sites of interest to policymakers \citep{allcott2015site}. If the researcher observes site-level covariates, however, the differences in covariate means under the baseline and alternative weights can inform the extent of external validity. Based on this idea, the \textit{covariate balance class} considers simplex weights whose covariate means differ from the baseline by at most a given constant. Under this class, one can assess whether results are robust to alternative populations whose covariate means are within a chosen distance of the baseline. 

I illustrate my framework in two empirical applications. The first application is an event study in \citet{lakdawala2023dynamic}, which uses the staggered rollout of internet access in Peruvian public schools to estimate the causal effects of school-based internet access on second grade test scores. The authors find a delayed achievement response: TWFE estimates are initially small but grow over time, indicating that schools require time to adapt to new internet access. The authors obtain similar patterns and magnitudes when using the \citet{sun2021estimating} estimator to account for negative weighting concerns with TWFE, concluding that results are not sensitive to the use of TWFE. My robust estimator and CI give a stronger conclusion: results are generally robust to \textit{classes} of nonnegative weights that yield comparable estimation precision to the \citet{sun2021estimating} weights.

The second application is Tennessee's Project STAR (Student/Teacher Achievement Ratio) experiment, which randomized students in seventy-nine Tennessee public elementary schools to classrooms of different sizes to estimate the causal effects of class size on test scores \citep{achilles2008tennessee, krueger1999experimental}. Project STAR has been used in many policy discussions, but the experiment's selection of schools raises concerns about whether results are representative of Tennessee or U.S. schools more broadly \citep{schanzenbach2006have}. To investigate these issues, I use an equal weights baseline to model the distribution of STAR schools and use the truncated simplex class to model departures from the STAR distribution. I show that, while the baseline CI implies medium-to-large positive effects, the robust CI includes small-to-zero effect sizes even under small departures from equal weighting. The robust CI shrinks when I intersect the truncated simplex with the covariate balance class based on school-level covariates, but results remain sensitive to small departures from the baseline weights.

While I focus on event studies and multisite experiments as the motivating examples, my framework applies broadly to settings where group-level parameters are averaged using weights that sum to one. My inference procedures are not designed for continuously indexed groups, but the bound on differences between weighted estimands has a natural analogue in that case.\footnote{For a discussion, see footnote \ref{arxiv2:footnote:bound.details} under Proposition \ref{arxiv2:prop:bound}.} Finally, the inferential guarantees of the robust estimator and CI do not depend on how the underlying parameters are defined. For example, in event studies one can define the parameters as cohort-level difference-in-differences (DiD) estimands rather than cohort-level average causal effects \citep{sun2021estimating}. This distinction affects the interpretation of the weighted estimands, but not the validity of the robust inference procedures.

A large literature highlights issues that arise when interpreting weighted estimands. In the context of my framework, these issues represent different sources of ambiguity and disagreement over the choice of weights. For example, regression-based estimands often recover weighted averages of treatment effects, but allow for negative weights \citep{de2020two, sloczynski2020should, goodman2021difference, mogstad2021causal, sun2021estimating, blandhol2022tsls, bhuller20242sls, goldsmith2024contamination}. Negative weights can flip the signs of treatment effects, complicating the interpretation of an estimand. In such contexts, there may be (i) ambiguity over how much negative weighting matters and (ii) disagreement over what alternative weighting schemes to consider.\footnote{\citet{abadie2025harvesting} and \citet{chiu2026causal} give different takeaways for issue (i) in the event study context.} Regarding the latter, a separate issue is that a positively weighted estimand can still lack empirical or policy relevance \citep{yitzhaki1996using, heckman2005structural, crump2006moving, angrist2010extrapolate, aronow2016does, li2018balancing, sloczynski2022interpreting, mogstad2024instrumental, poirier2024quantifying}. In this case, disagreement may arise from policymakers interested in the effects of treatment on their own target populations, while ambiguity may arise from researchers who must navigate such disagreement when summarizing results.

Common approaches to these issues are to report the extent of negative weighting in one's estimator or to examine the stability of results under alternative weighting schemes.\footnote{See \citet[Section 3.2.1]{roth2023s} for a review of such approaches in the event study context.} Intuitively, such approaches convey information about how far one's weights deviate from a benchmark, the degree of underlying heterogeneity, or some combination of both. My procedures build on this same intuition, but in a GLS-based geometry that yields a notion of heterogeneity amenable to inference with confidence bounds, drawing on statistical results on optimal quantile-unbiased estimation \citep{pfanzagl1994parametric}. This allows one to directly account for weighting issues by constructing robust CIs and estimators with explicit inferential guarantees, thereby facilitating robust inference for weighted estimands.

The remainder of this paper proceeds as follows. Section \ref{arxiv2:sec:setting} develops the model setting and notation. Section \ref{arxiv2:sec:bounding.differences} defines the heterogeneity and distance measures and establishes the bound on differences between weighted estimands. Section \ref{arxiv2:sec:robust.inference} develops the robust estimator and CI and establishes their properties. These results are developed under the assumption that group-level estimates are normally distributed with a known covariance matrix. Section \ref{arxiv2:sec:asymptotic.validity} shows that, with asymptotically normal estimates and a consistent covariance matrix estimator, my proposed procedures are uniformly asymptotically valid over a broad class of data generating processes. Section \ref{arxiv2:sec:implementation} discusses the practical implementation of these procedures. Section \ref{arxiv2:sec:applications} presents the empirical applications. Section \ref{arxiv2:sec:conclusion} concludes. The supplemental appendices contain proofs and additional results.

\section{Model Setting}\label{arxiv2:sec:setting}
Consider parameters $\theta_{k} \in \R$ indexed by groups $k \in \{1, \ldots, K\}$, where $K \geq 2$. Let $\theta = (\theta_{1}, \ldots, \theta_{K})'$ denote the vector of parameters. Researchers often conduct inference on estimands that can be expressed as weighted averages of the group-level parameters:
\begin{align*}
    \tauw = w'\theta = \sumk w_{k}\theta_{k}, \quad w \in \W, \quad \W = \curly{w \in \R^{K}: \1'w = 1},
\end{align*}
where $w \in \W$ is a vector of weights, $\W$ is the set of all weights, and $\1$ is the vector of ones. I refer to $\tauw$ as a \textit{weighted estimand}. I allow the weights $w$ to be negative unless stated otherwise. For the case of nonnegative (i.e., convex) weights, I define the simplex $\W_{+} = \{w \in \W: w \geq 0\}$.

\begin{example*}[Event Studies]
Given units $i$ and time periods $t \in \{0, 1, \ldots, T\}$, where $T \geq 2$, the researcher observes outcomes $Y_{it} \in \R$ and treatment indicators $D_{it} \in \{0,1\}$. A unit $i$'s cohort group $G_{i} = \min\curly{t: D_{it}=1}$ is the time period when first treated---if never treated, then $G_{i} = \infty$. There are no treated units in the base period (i.e., $G_{i} > 0$). Moreover, treatment is ``absorbing'' in the sense that a treated unit remains treated (i.e., $t \geq G_{i}$ implies $D_{it} = 1$).\footnote{For any treatment that is not absorbing, one can define an indicator for ever having received the treatment, which will be an absorbing treatment by construction; see \citet{sun2021estimating} for an example.} Let $Y_{it} = Y_{it}(G_{i})$, where $Y_{it}(g)$ denotes the potential outcome for unit $i$ when assigned to cohort $g$. Consider the group-time average treatment effect on the treated (ATT):
\begin{align}\label{arxiv2:eq:ATT.object}
    \ATT_{g,t} = E[Y_{it}(g) - Y_{it}(\infty)|G_{i} = g], \quad t \geq g.
\end{align}
This parameter gives the period $t$ causal effect of being treated in cohort $g$ versus being never-treated, averaged among the units $i$ in cohort $G_{i} = g$. Under assumptions of parallel trends and no anticipation, $\ATT_{g,t}$ is identified for different group-time pairs \citep{callaway2021difference, roth2023s}. Let $k$ index such pairs $(g_{k},t_{k})$. For parameters $\theta_{k} = \ATT_{k} = \ATT_{g_{k},t_{k}}$ and weights $w \in \W$, the corresponding weighted estimand is
\begin{align*}
    \tauw = \sumk w_{k}\ATT_{k}.
\end{align*}
Many common event-study estimands can be expressed like this, including conventional TWFE regression coefficients and recently proposed alternatives \citep{de2020two, callaway2021difference, goodman2021difference, sun2021estimating, gardner2022two, roth2023efficient, borusyak2024revisiting, wing2024stacked}. In this setting, each $w \in \W$ is a different way of summarizing treatment effect heterogeneity across cohort groups and time periods; see \citet[Table 1]{callaway2021difference} for examples.
\end{example*}

\begin{example*}[Multisite Experiments]
For units $i$, there are outcomes $Y_{i} \in \R$, treatment indicators $D_{i} \in \{0,1\}$, and covariates $X_{i} \in \R^{M}$. Let $(Y_{i}(1), Y_{i}(0))$ denote potential outcomes under treatment and control, so that $Y_{i} = D_{i}Y_{i}(1) + (1-D_{i})Y_{i}(0)$. The units $i$ come from different site populations $k$, represented by distributions $P_{k}$ over the unit-level random variables. Letting $E_{k} = E_{P_{k}}$ denote expectations under $P_{k}$, the site-level average treatment effects (ATEs) are
\begin{align*}
    \ATE_{k} = E_{k}[Y_{i}(1) - Y_{i}(0)], \quad k = 1, \ldots, K.
\end{align*}
Under assumptions of random treatment assignment within each site, $\ATE_{k}$ is identified for each site $k$ \citep{hotz2005predicting}. For parameters $\theta_{k} = \ATE_{k}$ and weights $w \in \W$, the corresponding weighted estimand is
\begin{align*}
    \tauw = \sumk w_{k}\ATE_{k}.
\end{align*}
For instance, the equal weights (EW) vector $w_{\EW} = \1/K$ represents an empirical distribution over the set of sites---this has been considered in, for example, \citet[Table VI]{allcott2015site}. More generally, each simplex vector $w \in \W_{+}$ represents a distribution over the set of sites.
\end{example*}

\subsection{Conventional Inference for Weighted Estimands}\label{arxiv2:sec:conventional.inferences}
The researcher observes a vector of estimates $\htheta$ for the parameters $\theta$. I model the relationship between $\htheta$ and $\theta$ as follows. 

\begin{assumption}\label{arxiv2:ass:normality}
$\htheta \sim N(\theta, \Sigma)$, where $\Sigma$ is a known positive definite covariance matrix.
\end{assumption}

Under the normal distribution $N(\theta, \Sigma)$, I denote probabilities $\P{\cdot}$, expectations $\E{\cdot}$, and variances $\Var{\cdot}$. Let $\Phi(z)$ denote the cumulative distribution function (CDF) of the standard normal distribution $N(0,1)$ and let $z_{\alpha}$ denote its $\alpha$-quantile for $\alpha \in (0,1)$.

Assumption \ref{arxiv2:ass:normality} is motivated by standard large-sample asymptotic results. For example, the central limit theorem (CLT) helps establish asymptotic normality of estimates, which motivates the normality condition $\htheta \sim N(\theta, \Sigma)$. Similarly, the law of large numbers (LLN) helps establish consistency of covariance matrix estimation, which motivates the condition that $\Sigma$ is known. I therefore develop my procedures under Assumption \ref{arxiv2:ass:normality}. In Section \ref{arxiv2:sec:asymptotic.validity}, I relax Assumption \ref{arxiv2:ass:normality} and establish the asymptotic validity of my procedures.

For each vector of weights $w \in \W$, Assumption \ref{arxiv2:ass:normality} implies $w'\htheta \sim N(w'\theta, \sigma_{w}^{2})$, where $\sigma_{w}^{2} = w'\Sigma w$. For a given target estimand $\tauw = w'\theta$, the conventional point estimator is defined as $\htauw = w'\htheta$ and the conventional CI at significance level $\alpha$ is defined as
\begin{align}\label{arxiv2:eq:conventional.CIs}
    CI_{w} = 
    \begin{cases}
    \hfil \brack{\htauw \pm z_{1-\alpha/2}\sigma_{w}}, & \text{two-sided}, \\[5pt]  
    \hfil \left(-\infty, \htauw + z_{1-\alpha}\sigma_{w}\right], & \text{one-sided (upper)}, \\[5pt]  
    \hfil \left[\htauw - z_{1-\alpha}\sigma_{w}, \infty\right), & \text{one-sided (lower)}.
    \end{cases}
\end{align}
The conventional estimator $\htauw$ is unbiased for $\tauw$:
\begin{align*}
    \E{\htauw} = \tauw, \quad \forall \theta.
\end{align*}
The conventional $CI_{w}$ has exact coverage for $\tauw$:
\begin{align*}
    \P{\tauw \in CI_{w}} = 1-\alpha, \quad \forall \theta.
\end{align*}
In this sense, the conventional statistics $(\htauw, CI_{w})$ provide valid inference for $\tauw$.

\begin{example*}[Event Studies, continued]
Let $\htheta_{k} = \widehat{\ATT}_{k}$ denote estimators of $\ATT_{k}$ for each $k$, such as DiD-based estimators. Under standard conditions, such estimators are asymptotically normal with a consistently estimable covariance matrix as the number of units goes to infinity \citep{callaway2021difference}. Often one must also estimate the weights $w$. For example, the \citet{sun2021estimating} estimand uses weights that depend on the population distribution of treatment $D_{it}$, so that in practice one must use the sample distribution of $D_{it}$ to estimate $w$. For ease of exposition, I develop my procedures assuming that $w$ is known and later establish asymptotic validity under estimated weights in Section \ref{arxiv2:sec:asymptotic.validity}.
\end{example*}

\begin{example*}[Multisite Experiments, continued]
Let $\htheta_{k} = \widehat{\ATE}_{k}$ denote estimators of $\ATE_{k}$ for each $k$, such as those based on sample differences in means between the treatment and control groups. Under site-level CLTs, such estimators are asymptotically normal with a consistently estimable covariance matrix as the number of units goes to infinity. In cases where the unit-level variables $(Y_{i}, D_{i}, X_{i})$ are unobserved or confidential, one may still have access to site-level ATE estimates $\htheta = (\widehat{\ATE}_{1}, \ldots, \widehat{\ATE}_{K})'$ and covariate means $\bX = (E_{1}[X_{i}], \ldots, E_{K}[X_{i}])'$, as in the case of the metadata in \citet{allcott2015site}. The application of my framework to multisite experiments requires only the site-level variables---I accommodate the case of estimated covariate means $\widehat{E}_{k}[X_{i}]$ in my asymptotic results.
\end{example*}

\subsection{Classes of Alternative Weights}\label{arxiv2:sec:alternative.weights}
The researcher reports $(\htauw, \sigma_{w})$ for some baseline choice of weights $w \in \W$. This report yields conventional statistics $(\htauw, CI_{w})$ that provide valid inference for the baseline estimand $\tauw$. However, the researcher may also consider a class $\L \subseteq \W$ of alternative weights $\l \in \L$, yielding alternative estimands $\taul$ for which $(\htauw, CI_{w})$ may be uninformative.

\begin{example*}[Event Studies, continued]
\citet{callaway2021difference} advocate choosing the weights $w$ to address well-posed empirical or policy questions.\footnote{This perspective is also prevalent in other empirical settings, such as instrumental variables estimation under heterogeneous treatment effects \citep{heckman2007econometric, mogstad2024instrumental}.} While this principle is ideal, two practical issues can arise, leading to consideration of a class of alternative weights $\L$.

\paragraph{Researcher Ambiguity.} First, even with a well-posed question, the researcher may find it difficult to articulate corresponding weights $w$. This costly introspection presumably underlies the various default weighting schemes considered in the event studies literature.\footnote{See \citet[Table 1]{callaway2021difference} and \citet[Table 2]{roth2023s} for examples.} A researcher may choose a default specification for $w$ and then assess robustness of results to other defaults---the survey in Section \ref{arxiv2:sec:implementation.event.studies} documents this pattern. But to the extent that these defaults fail to capture the scope of plausible weights, such robustness exercises may not adequately account for researcher ambiguity over the choice of $w$. To do better, it seems useful to consider a broader class of alternatives for $w$, such as the example classes $\L$ developed below. 

\paragraph{Reader Disagreement.} Second, even absent introspection costs, the researcher may have to communicate results to readers who are interested in different questions and hence disagree over the choice of weights. For example, colleagues with different priors about the empirical setting may wish to highlight different types of treatment effect heterogeneity: e.g., across cohorts versus across time \citep{callaway2021difference}. In such cases, $\L$ represents the range of questions that readers may be interested in, which can differ from the question addressed by $w$.
\end{example*}

\begin{example*}[Multisite Experiments, continued]
Given the site-level populations $P_{1}, \ldots, P_{K}$, a baseline vector of simplex weights $w \in \W_{+}$ represents the population $P_{w} = \sum_{k}w_{k}P_{k}$. The set of readers may include policymakers interested in the effect of treatment on their own target populations, leading to disagreement over the choice of weights. For concreteness, consider a policymaker who has a target population $P_{0}$ and a corresponding target parameter
\begin{align*}
    \ATE_{0} = E_{0}[Y_{i}(1) - Y_{i}(0)].
\end{align*}
This policymaker may find the baseline weights $w$ to be palatable when $P_{0}$ resembles $P_{w}$, such as when the mean $\mu_{0} = E_{0}[X_{i}]$ under $P_{0}$ is close to the mean $\mu_{w}(\bX) = w'\bX$ under $P_{w}$.\footnote{This mirrors the logic advanced in \citet[page 255]{aronow2016does} for interpreting the representativeness of weighted estimands in regression contexts. In their language, $\mu_{w}(\bX)$ is the covariate mean for the \textit{effective sample} of units spanned by the $w$-weighted sites.} In view of this, $\L \subseteq \W_{+}$ can be specified as a class of populations that a range of policymakers may find to be palatable.\footnote{If there were a single known policymaker, one could employ canonical approaches for extrapolating the site-level populations to the policymaker's target population $\ATE_{0}$, such as density reweighting under covariate shift assumptions \citep{hotz2005predicting}. However, such avenues are less tractable in the face of multiple and potentially unknown policymakers.}
\end{example*}

Whether the issue at hand is researcher ambiguity, reader disagreement, or some combination of both, my framework assumes that it can be represented by a class of alternative weights $\L$. To make this modeling assumption a practical one, I restrict attention to classes $\L$ that have the following structure.

\begin{assumption}\label{arxiv2:ass:alternative.weights}
$\L \subseteq \W$ is nonempty, compact, and convex.
\end{assumption}

This structure is useful for establishing the existence and uniqueness of solutions to upcoming optimization problems. I now provide examples of classes $\L$ that satisfy Assumption \ref{arxiv2:ass:alternative.weights} and map to various empirical contexts.

\begin{example*}[Bounded Variance] To analyze the robustness of baseline inferences to alternative weights, one may wish to account for differences in how well the corresponding estimands can be estimated. Formally, one can consider weights $\l$ for which the standard deviation of $\htaul$ is at most $r$ times that of the baseline $\htauw$, leading to the \textit{bounded variance class} 
\begin{align}\label{arxiv2:eq:class.bounded.variance}
    \L_{\sigma}(r) = \curly{\l \in \W: \sigma_{\l} \leq r\sigma_{w}}, \quad r \geq \frac{\sigma_{\min}}{\sigma_{w}}, \quad \sigma_{\min}^{2} = \min_{w \in \W} \sigma_{w}^{2} = \frac{1}{\1'\Sigma^{-1}\1}.
\end{align} 
This class reflects a preference for estimands that can be estimated with a reasonable degree of precision.\footnote{For example, \citet[pages 13-14]{mogstad2024instrumental} note that ``How interesting a target parameter is also cannot be divorced from the difficulty involved in estimating it.''} For instance, $r=1$ represents the class of estimands that can be estimated at least as precisely as the baseline estimand. When convenient, I leave the dependence of $\L_{\sigma}(r)$ on the \textit{standard deviation ratio bound} $r$ implicit, denoting $\L_{\sigma} = \L_{\sigma}(r)$.
\end{example*}

\begin{example*}[Truncated Simplex]
The simplex $\W_{+} = \{w \in \W: w \geq 0\}$ is the set of nonnegative weights, which yields estimands that do not extrapolate beyond the range of the parameters: $\min_{k}\theta_{k} \leq \tauw \leq \max_{k}\theta_{k}$. However, the simplex allows for some groups to receive zero weight, which can be undesirable in practice. To ensure that groups are represented, one can impose a floor on the simplex. In particular, given baseline simplex weights $w \in \W_{+}$, the floor $(1-\epsilon)w$ yields the \textit{truncated simplex class}
\begin{align}\label{arxiv2:eq:class.truncated.simplex}
    \L_{+}(\epsilon) = \curly{\l \in \W_{+}: \l  \geq (1-\epsilon)w} = \curly{(1-\epsilon)w + \epsilon \l: \l \in \W_{+}}, \quad \epsilon \in [0,1].
\end{align}
This class allows one to flexibly model departures from the baseline simplex weights $w$, where the \textit{truncation parameter} $\epsilon$ is the fraction of a given alternative $\l \in \L_{+}(\epsilon)$ that is allowed to deviate from $w$. Given $\epsilon < 1$, one can ensure that a group $k$ is represented in $\l \in \L_{+}(\epsilon)$ by using baseline weights with $w_{k} > 0$. Under baseline equal weights $w = w_{\EW} = \1/K$, one can interpret $\epsilon$ as the maximum possible discrepancy in weights $|\l_{k}-\l_{k'}|$ between two groups $k$ and $k'$. I give additional interpretations of $\epsilon$ when discussing practical implementations in Section \ref{arxiv2:sec:implementation}.
\end{example*}

\begin{example*}[Covariate Balance]
For $m \in \{1, \ldots, M\}$, consider group-level covariates $\bX_{m} \in \R^{K}$. A baseline simplex weight vector $w \in \W_{+}$ is a probability distribution over the set of groups, and can thus be viewed as a population with covariate means $w'\bX_{m}$. Likewise, alternative simplex weights $\l \in \W_{+}$ yield populations with covariate means $\l'\bX_{m}$. Letting $|(\l - w)'\bX_{m}|$ denote the covariate $m$ balance gap and $\text{sd}(\bX_{m})$ the standard deviation of $\bX_{m,k}$ across groups $k$, the set of populations with balance gaps no larger than $\Bar{c}$ standard deviations for all $m$ is given by the \textit{covariate balance class}
\begin{align}\label{arxiv2:eq:class.covariate.balance}
    \L_{X}(\Bar{c}) = \curly{\l \in \W_{+}: \max_{m \in \{1,\ldots,M\}}\frac{\abs{(\l - w)'\bX_{m}}}{\text{sd}(\bX_{m})} \leq \Bar{c}}, \quad \Bar{c} \geq 0, \quad \min_{m \in \{1,\ldots,M\}}\text{sd}(\bX_{m}) > 0.
\end{align}
This class allows one to consider a range of populations that have different covariate profiles from the baseline. For ease of reference, I abbreviate the balance gap measure as
\begin{align*}
    c_{\l}(\bX) = \max_{m \in \{1, \ldots, M\}}\frac{\abs{(\l - w)'\bX_{m}}}{\text{sd}(\bX_{m})}, \quad \bX = (\bX_{1}, \ldots, \bX_{M}).
\end{align*}
Thus, I can write $\L_{X}(\Bar{c}) = \curly{\l \in \W_{+}: c_{\l}(\bX) \leq \Bar{c}}$ for \textit{balance gap bound} $\Bar{c}$.
\end{example*}

\begin{remark}[Intersections]
One can also take intersections of the above classes. For example, intersecting the bounded variance class $\L_{\sigma}(r)$ with the unrestricted simplex $\L_{+}(1) = \W_{+}$ yields the \textit{bounded variance simplex class}
\begin{align}\label{arxiv2:eq:class.simplex.bounded.variance}
    \L_{\sigma}^{+}(r) = \L_{\sigma}(r) \cap \L_{+}(1) = \curly{\l \in \W_{+}: \sigma_{\l} \leq r\sigma_{w}}, \quad r \geq \frac{\min_{w \in \W_{+}} \sigma_{w}}{\sigma_{w}},
\end{align} 
which considers convex weights that yield estimators with reasonable estimation precision. As a second example, one can intersect the covariate balance class $\L_{X}(\Bar{c})$ with the truncated simplex $\L_{+}(\epsilon)$ to obtain the \textit{truncated covariate balance class}
\begin{align}\label{arxiv2:eq:class.truncated.covariate.balance}
    \L_{X}^{\epsilon}(\Bar{c}) = \L_{X}(\Bar{c}) \cap \L_{+}(\epsilon) = \curly{(1-\epsilon)w + \epsilon \l: \l \in \L_{X}(\Bar{c}/\epsilon)}, \quad \epsilon > 0,
\end{align} 
which considers $\epsilon$-departures from the baseline $w$ towards populations within a covariate mean radius $\Bar{c}/\epsilon$ from the baseline. I will consider these intersections in the empirical applications, but I focus on the benchmark classes $(\L_{\sigma}, \L_{+}, \L_{X})$ when interpreting my theoretical results.
\end{remark}

\subsection{Inference Distortions due to Alternative Weights}
I conclude Section \ref{arxiv2:sec:setting} by formally defining the inference distortions that arise when one uses the baseline statistics $(\htauw, CI_{w})$ to learn about alternative estimands $\taul$. In particular, I consider (i) the absolute bias of $\htauw$ for $\taul$, given by
\begin{align*}
    \abs{\E{\htauw} - \taul} = \abs{\taul- \tauw} \geq 0,
\end{align*}
and (ii) the noncoverage probability of $CI_{w}$ for $\taul$, given (in the two-sided case) by
\begin{align*}
    \P{\taul \notin CI_{w}} = \Phi\paren{z_{\alpha/2}-\frac{\taul- \tauw}{\sigma_{w}}} + \Phi\paren{z_{\alpha/2} + \frac{\taul- \tauw}{\sigma_{w}}} \geq \alpha.
\end{align*}
These inference distortions are governed by the (absolute) difference in estimands:
\begin{align*}
    \abs{\taul- \tauw}.
\end{align*}
The larger this difference, the higher the bias and noncoverage. In particular, when $\l \neq w$, there exist parameter values $\theta$ where $|\taul- \tauw|$ is arbitrarily large so that bias is arbitrarily large and noncoverage is arbitrarily close to one. In such cases, the baseline statistics $(\htauw, CI_{w})$ are completely uninformative for $\taul$. These observations suggest that, to obtain useful inferences for alternative estimands, one should account for potential differences in estimands across the alternative weights $\l \in \L$ and the potential parameter values $\theta \in \R^{K}$. To this end, the results in Section \ref{arxiv2:sec:bounding.differences} below provide a useful starting point.

\begin{remark}[Reporting Constraints]
Since $\L$ represents the broad range of beliefs and objectives that researchers and readers may have, it is impractical to report $(\htaul, CI_{\l})$ for all $\l \in \L$. Moreover, if there are confidentiality restrictions or communication costs, reporting the entire data $(\htheta, \Sigma)$ may not be a broadly applicable solution.\footnote{\cite{allcott2015site} is one example where confidentiality restrictions preclude the reporting of $(\htheta, \Sigma)$, but not of $(\htauw, \sigma_{w})$. Even absent such restrictions, researchers still often focus on reporting statistics of the form $(\htauw, \sigma_{w})$, as discussed in \citet{athey2023thirdnumber}. This reporting convention may be due to communication costs that induce researchers to report low-dimensional statistics for their readers.} With these issues in mind, I will develop methods that facilitate inferences across $\l \in \L$ based simply on (i) the baseline report $(\htauw, \sigma_{w})$ and (ii) a supplementary report of low-dimensional statistics, where the latter loosely refers to statistics whose dimension (e.g., the number of columns and rows occupied in a table) does not grow with the number of groups $K$, thus precluding $(\htheta, \Sigma)$. 
\end{remark}

\section{Bounding the Difference in Estimands}\label{arxiv2:sec:bounding.differences}
In this section, I bound the difference in estimands in terms of the heterogeneity in parameters and the distance between weights. I then show how to infer these quantities, which will be the basis for constructing the robust estimator and CI in Section \ref{arxiv2:sec:robust.inference}.

\paragraph{Heterogeneity in Parameters.} Consider the generalized least squares (GLS) regression of the parameters $\theta$ on a constant $\1$ under the weighting matrix $\Sigma^{-1}$. I define the \textit{heterogeneity in parameters} as the square root of the corresponding GLS residual sum of squares:
\begin{align*}
    H(\theta) = \sqrt{\min_{\gamma \in \R}(\theta - \1 \gamma)'\Sigma^{-1}(\theta - \1 \gamma)}.
\end{align*}
By construction, the heterogeneity in $\theta$ is zero if and only if $\theta_{k}$ is constant across $k$. The above regression is uniquely minimized at the GLS estimand
\begin{align}\label{arxiv2:eq:GLS.estimand}
    \arg\min_{\gamma \in \R}(\theta - \1 \gamma)'\Sigma^{-1}(\theta - \1 \gamma) = w_{\GLS}'\theta = \tau_{\GLS}(\theta), \quad w_{\GLS} = \frac{\Sigma^{-1}\1}{\1'\Sigma^{-1}\1}.
\end{align}
Evaluating $\gamma = \tau_{\GLS}(\theta)$ yields the formula
\begin{align}\label{arxiv2:eq:heterogeneity.quadratic}
    H(\theta) = \sqrt{\theta'Q\theta}, \quad Q = \Sigma^{-1/2}A\Sigma^{-1/2}, \quad A = I - \frac{\Sigma^{-1/2}\1\1'\Sigma^{-1/2}}{\1'\Sigma^{-1}\1},
\end{align}
where $A$ is the annihilator matrix for $\Sigma^{-1/2}\1$ and $I$ is the identity matrix. Thus, the squared heterogeneity can be represented as a particular quadratic form $\theta'Q\theta$ of the parameters, which will facilitate quantile-unbiased inference on $H(\theta)$ in Section \ref{arxiv2:sec:heterogeneity.UCB}.

\paragraph{Distance Between Weights.} The standard deviation $\norm{v}_{\Sigma} = \dsqrt{v'\Sigma v}$ of a given linear estimator $v'\htheta$ defines a norm on $v \in \R^{K}$. Given weights $\l$ and $w$, I define the \textit{distance between weights} as the corresponding norm of their difference:
\begin{align*}
    \norm{\l - w}_{\Sigma} = \sqrt{(\l-w)'\Sigma(\l-w)} = \sqrt{\Var{\htaul - \htauw}}.
\end{align*}
Intuitively, the distance between weights $\norm{\l - w}_{\Sigma}$ measures the disagreement between $\l$ and $w$ by taking the standard deviation of their corresponding difference in estimators. 

\begin{proposition}\label{arxiv2:prop:bound}
For any $w \in \W$ and $\l \in \W$, the difference in estimands is bounded as
\begin{align*}
    |\taul - \tauw| \leq H(\theta)\norm{\l - w}_{\Sigma}, \quad \forall \theta.
\end{align*}
This bound is sharp in the sense that, given $\l \neq w$ and any $\eta \geq 0$, there exists $\theta$ with $H(\theta) = \eta$ such that the bound holds with equality.
\end{proposition}

\begin{proof}
See Appendix \ref{arxiv2:app:proof:bound}.
\end{proof}

Proposition \ref{arxiv2:prop:bound} shows that the difference between a baseline estimand $\tauw$ and an alternative estimand $\taul$ is bounded by the product of (i) the heterogeneity in parameters $H(\theta)$, which is unknown due to $\theta$ being unobserved, and (ii) the distance between weights $\norm{\l - w}_{\Sigma}$, which is unknown due to ambiguity or disagreement over $\l \in \L$.\footnote{Proposition \ref{arxiv2:prop:bound} is based on the Cauchy-Schwarz inequality, similar to Scheffé-style arguments for bounding data-dependent linear combinations to obtain uniformly valid inference \citep{scheffe1953method, lehmann2024testing}. However, Proposition \ref{arxiv2:prop:bound} considers nonrandom linear combinations and obtains bounds on the inference distortions themselves. Analogous bounds hold if one replaces $\Sigma$ with another positive definite matrix when defining the heterogeneity and distance measures. It also has a natural analogue for continuously indexed groups, where one may define the weighted estimands as $\tauw = \int w(k)\theta(k)d\nu(k)$ and apply the same Cauchy-Schwarz argument in an $L^{2}(\nu)$ inner product. \label{arxiv2:footnote:bound.details}} In Section \ref{arxiv2:sec:maximum.distance}, I account for the latter by taking the maximum distance over the class of alternatives $\L$. In Section \ref{arxiv2:sec:heterogeneity.UCB}, I account for the former by constructing an upper confidence bound (UCB) on the heterogeneity $H(\theta)$.

\subsection{Maximum Distance}\label{arxiv2:sec:maximum.distance}
Given baseline weights $w$ and class of alternatives $\L$, the \textit{maximum distance} between weights is
\begin{align*}
    \max_{\l \in \L}\norm{\l - w}_{\Sigma} = \sqrt{\max_{\l \in \L}(\l-w)'\Sigma (\l-w)}.
\end{align*}  
Intuitively, the maximum distance measures the worst-case disagreement between the baseline and alternative weights across $\L$. Under the structure on $\L$ from Assumption \ref{arxiv2:ass:alternative.weights}, the maximum is attained at some $\l^{*} \in \L$, which can be viewed as the weights of a reader who disagrees with $w$ the most. Below I analyze the structure of the maximum distance under the example classes from Section \ref{arxiv2:sec:alternative.weights} and illustrate how this structure can facilitate communication between researchers and readers---I defer practical recommendations and implementation choices to Section \ref{arxiv2:sec:implementation}. 

\begin{example*}[Bounded Variance, continued]
One can show that 
\begin{align}\label{arxiv2:eq:distance.bounded.variance}
    \max_{\l \in \L_{\sigma}}\norm{\l - w}_{\Sigma} = \sqrt{r^{2}\sigma_{w}^{2} - \sigma_{\min}^{2}} + \sqrt{\sigma_{w}^{2} - \sigma_{\min}^{2}}, \quad \sigma_{\min}^{2} = \min_{w \in \W} \sigma_{w}^{2} = \frac{1}{\1'\Sigma^{-1}\1},
\end{align}
which is a function of $(r, \sigma_{w}, \sigma_{\min})$. If the researcher reports the minimum variance $\sigma_{\min}^{2}$ alongside the baseline variance $\sigma_{w}^{2}$, then a reader can compute the maximum distance for their own choice of standard deviation ratio bound $r$. Conveniently, $\sigma_{\min}^{2}$ coincides with the variance of the GLS estimator
\begin{align}\label{arxiv2:eq:GLS.estimator}
    \htau_{\GLS} = w_{\GLS}'\htheta = \argmin_{\gamma \in \R}(\htheta - \1 \gamma)'\Sigma^{-1}(\htheta - \1 \gamma), \quad \sigma_{\min}^{2} = \sigma_{\GLS}^{2} = w_{\GLS}'\Sigma w_{\GLS},
\end{align}
where $w_{\GLS}$ are the weights of the GLS estimand defined in \eqref{arxiv2:eq:GLS.estimand}. The above GLS regression is an input to my inference procedures for $H(\theta)$ in Section \ref{arxiv2:sec:heterogeneity.UCB}, so the GLS variance can be obtained at essentially no further cost.
\end{example*}

\begin{example*}[Truncated Simplex, continued]
Letting $v_{1}, \ldots, v_{K}$ denote the standard unit vectors, one can show that
\begin{align}\label{arxiv2:eq:distance.truncated.simplex}
    \max_{\l \in \L_{+}}\norm{\l - w}_{\Sigma} = \epsilon \max_{\l \in \W_{+}}\norm{\l - w}_{\Sigma} = \epsilon \max_{j \in \{1,\ldots,K\}}\norm{v_{j} - w}_{\Sigma},
\end{align}
which is a function of $(\epsilon, \max_{j}\norm{v_{j} - w}_{\Sigma})$. If the researcher reports $\max_{j}\norm{v_{j} - w}_{\Sigma}$, then a reader can compute the maximum distance for their own choice of truncation parameter $\epsilon$. In the case of an equal weights baseline $w = w_{\EW} = \1/K$ and a reader with weights $\l_{0} \in \W_{+}$ in the simplex, $\epsilon_{0} = 1-K\min_{k}\l_{0,k}$ is the smallest $\epsilon$ at which $\L_{+}(\epsilon)$ contains the reader's weights $\l_{0}$.
\end{example*}

\begin{example*}[Covariate Balance, continued]
Compared to the above examples, there does not appear to be a transparent expression for $\max_{\l \in \L_{X}}\norm{\l - w}_{\Sigma}$. However, the researcher can plot it for a range of balance gap bounds $\Bar{c}$, which allows a reader to examine the maximum distance at their preferred value of $\Bar{c}$ from the reported range. To support this exercise, the researcher can additionally report covariate information $(w'\bX_{m}, \text{sd}(\bX_{m}))_{m=1}^{M}$. This allows a policymaker with covariate means $\mu_{0} \in \R^{M}$ for a target population $P_{0}$ to compute the balance gap 
\begin{align*}
    c_{0}(\bX) = \max_{m \in \{1, \ldots, M\}}\frac{|\mu_{0,m} - w'\bX_{m}|}{\text{sd}(\bX_{m})}.
\end{align*}
Given the researcher's reported range of $\Bar{c}$, the policymaker can check the maximum distance at the value $\Bar{c}_{0}$ that is closest to $c_{0}(\bX)$---ideally while being larger than $c_{0}(\bX)$, so that $\L_{X}(\Bar{c}_{0})$ will contain any $\l \in \W_{+}$ that yields gap $c_{\l}(\bX) = c_{0}(\bX)$. Intuitively, $\L_{X}(\Bar{c}_{0})$ is a set of plausible candidates for $P_{0}$ in terms of possible balance-gap similarity and the maximum distance at $\Bar{c}_{0}$ gives a corresponding measure of ambiguity. 
\end{example*}

\subsection{Heterogeneity UCB}\label{arxiv2:sec:heterogeneity.UCB}
I now show how to infer the unknown heterogeneity in parameters $H(\theta)$. In particular, I derive an \textit{upper confidence bound} (UCB) for $H(\theta)$: given a significance level $\beta \in (0,1)$, I construct an estimator $\heta_{1-\beta}$ such that
\begin{align}\label{arxiv2:eq:inference.heterogeneity}
    \P{H(\theta) \leq \heta_{1-\beta}} \geq 1 - \beta, \quad \forall \theta.
\end{align}
In words, $\heta_{1-\beta}$ upper bounds the unknown heterogeneity $H(\theta)$ with probability at least $1-\beta$, for any value of $\theta$. I refer to $\heta_{1-\beta}$ as a \textit{heterogeneity UCB} at confidence level $1-\beta$. 

\paragraph{Derivation.} To derive $\heta_{1-\beta}$, I use the representation of $H(\theta) = \dsqrt{\theta'Q\theta}$ as a quadratic form in $\theta$. In particular, since $\htheta \sim N(\theta, \Sigma)$, the definition of $Q$ from \eqref{arxiv2:eq:heterogeneity.quadratic} implies that the \textit{heterogeneity in estimates} $H(\htheta) = \sqrt{\htheta'Q\htheta}$ satisfies 
\begin{align*}
    \htheta'Q\htheta = (\Sigma^{-1/2}\htheta)'A(\Sigma^{-1/2}\htheta) \sim \chi_{K-1}^{2}(H(\theta)), \quad \forall \theta,
\end{align*}
where $\chi_{K-1}^{2}(\eta)$ denotes the noncentral chi-squared distribution with $K-1$ degrees of freedom and noncentrality parameter $\eta^{2}$. Let $F_{\chi^{2}}(x; \eta)$ denote its CDF and consider the corresponding pivot function $\eta \mapsto F_{\chi^{2}}(\htheta'Q\htheta; \eta)$. The probability integral transform yields
\begin{align*}
    F_{\chi^{2}}(\htheta'Q\htheta; H(\theta)) \sim U(0,1), \quad \forall \theta.
\end{align*}
In particular, the probability of observing $F_{\chi^{2}}(\htheta'Q\htheta; H(\theta)) \geq \beta$ is equal to $1-\beta$. For each $x > 0$, the function $\eta \mapsto F_{\chi^{2}}(x;\eta)$ is strictly decreasing \citep{sun2010monotonicity}. Therefore, when $F_{\chi^{2}}(\htheta'Q\htheta; 0) > \beta$, I define $\heta_{1-\beta}$ as the unique solution to $F_{\chi^{2}}(\htheta'Q\htheta; \heta_{1-\beta}) = \beta$. Otherwise, when $F_{\chi^{2}}(\htheta'Q\htheta; 0) \leq \beta$, I define $\heta_{1-\beta} = 0$.\footnote{This construction is consistent with the general approach in \citet[Section 5.3]{pfanzagl1994parametric}, which  constructs confidence bounds for one-parameter distributions satisfying appropriate monotonicity conditions.} In summary,
\begin{align*}
    \heta_{1-\beta} = 
    \begin{cases}
    \hfil 0, & F_{\chi^{2}}(\htheta'Q\htheta; 0) \leq \beta, \\
    \hfil F_{\chi^{2}}^{-1}(\htheta'Q\htheta; \beta), & F_{\chi^{2}}(\htheta'Q\htheta; 0) > \beta.
    \end{cases}
\end{align*}
This construction depends only on $\htheta'Q\htheta$, which is conveniently obtained in \eqref{arxiv2:eq:GLS.estimator} as the residual sum of squares from the GLS regression of the estimates $\htheta$ on a constant $\1$ under the weighting matrix $\Sigma^{-1}$. The inversion for $\heta_{1-\beta}$ is computationally efficient, since it amounts to finding the root of a strictly monotone function.

The next result shows that the above $\heta_{1-\beta}$ is a valid UCB for $H(\theta)$. Moreover, it is \textit{quantile-unbiased} under heterogeneity in the sense that its (weak) overestimation probability for $H(\theta)$ is equal to $1-\beta$ when $H(\theta) > 0$. Note that when $H(\theta) = 0$, the overestimation probability is equal to one, as would be the case for any nonnegative estimator of $H(\theta)$.

\begin{proposition}\label{arxiv2:prop:eta.validity}
The above $\heta_{1-\beta}$ satisfies \eqref{arxiv2:eq:inference.heterogeneity}. Moreover, under heterogeneous $\theta$, the confidence level $1-\beta$ is attained:
\begin{align}\label{arxiv2:eq:exact.upperCI.heterogeneity}
    \P{H(\theta) \leq \heta_{1-\beta}} = 1 - \beta, \quad \forall \theta: H(\theta) > 0.
\end{align}
Under homogeneous $\theta$, the coverage probability is equal to one.
\end{proposition}

\begin{proof}
See Appendix \ref{arxiv2:app:proof:eta.validity}.
\end{proof}

The quantile-unbiasedness in Proposition \ref{arxiv2:prop:eta.validity} facilitates a variety of inference procedures for $H(\theta)$. First, $[0, \heta_{1-\beta}]$ is a one-sided upper confidence interval for $H(\theta)$, with exact coverage rate $1-\beta$ under heterogeneous $\theta$. Next, the one-sided lower interval $[\heta_{\beta}, \infty)$ satisfies 
\begin{align}\label{arxiv2:eq:exact.lowerCI.heterogeneity}
    \P{H(\theta) \geq \heta_{\beta}} = 1 - \beta, \quad \forall \theta: H(\theta) > 0.
\end{align}
Likewise, one can construct a two-sided interval $[\heta_{\beta/2}, \heta_{1-\beta/2}]$ that satisfies 
\begin{align}\label{arxiv2:eq:exact.twosidedCI.heterogeneity}
    \P{\heta_{\beta/2} \leq H(\theta) \leq \heta_{1-\beta/2}} = 1 - \beta, \quad \forall \theta: H(\theta) > 0.
\end{align}
Finally, one can construct a point estimator $\heta_{1/2}$ that is median-unbiased in the sense that
\begin{align}\label{arxiv2:eq:exact.median.heterogeneity}
    \P{H(\theta) \leq \heta_{1/2}} = \P{H(\theta) \geq \heta_{1/2}} = \frac{1}{2}, \quad \forall \theta: H(\theta) > 0.
\end{align}
That is, the median of $\heta_{1/2}$ is equal to $H(\theta)$ under heterogeneity.\footnote{For point estimation of heterogeneity, one can also consider mean-unbiased estimation of $\theta'Q\theta$, following a similar construction to, e.g., \citet{kline2020leave}. In particular, one can show that $\htheta'Q\htheta - (K-1)$ is a mean-unbiased estimator of $\theta'Q\theta$. However, the quantile-unbiased approach adopted here is more amenable to the construction of confidence bounds for my objects of interest.} Given my focus on upper bounding differences in weighted estimands, I focus on properties \eqref{arxiv2:eq:inference.heterogeneity} and \eqref{arxiv2:eq:exact.upperCI.heterogeneity}. Nevertheless, properties \eqref{arxiv2:eq:exact.lowerCI.heterogeneity}-\eqref{arxiv2:eq:exact.median.heterogeneity} illustrate the versatility of $\heta_{1-\beta}$ for inference on parameter heterogeneity. I conclude this section with comparative statics and an optimality result for $\heta_{1-\beta}$.

\paragraph{Comparative Statics.} For a given $\eta$, the CDF $F_{\chi^{2}}(x; \eta)$ is increasing in $x$ and decreasing in its degrees of freedom \citep{sun2010monotonicity}. Thus, $\heta_{1-\beta}$ is increasing in $\htheta'Q\htheta$ given a fixed $K$, and decreasing in $K$ given a fixed $\htheta'Q\htheta$. In this sense, $\heta_{1-\beta}$ is naturally summarized by the heterogeneity of the estimates relative to the number of groups: $H(\htheta)/\sqrt{K}$. Notice that, under homoskedasticity $\Sigma = \sigma^{2}I$, the expression $\htheta'Q\htheta/K$ reduces to the empirical variance of the estimates relative to the common sampling variance $\sigma^{2}$ of the estimates:
\begin{align*}
    \frac{\htheta'Q\htheta}{K} = \frac{\displaystyle \frac{1}{K}\sumk \htheta_{k}^{2} - \paren{\frac{1}{K}\sumk \htheta_{k}}^{2}}{\sigma^{2}}, \quad \Sigma = \sigma^{2}I.
\end{align*}
Thus, the proposed heterogeneity measure $\heta_{1-\beta}$ has intuitive comparative statics in terms of $\htheta'Q\htheta/K$. Finally, holding $\htheta'Q\htheta$ and $K$ fixed, a larger confidence level $1-\beta$ implies a larger $\heta_{1-\beta}$. This highlights a tradeoff when constructing $\heta_{1-\beta}$: to upper bound $H(\theta)$ with higher confidence, one must tolerate a larger UCB, corresponding to a lower significance level $\beta$.

\paragraph{Optimality.} Let $\Bar{\Theta}^{\eta} = \{\theta: H(\theta) = \eta\}$ denote the set of parameters for which the heterogeneity is equal to $\eta$, and consider the class of potentially randomized estimators $\Tilde{\eta}_{1-\beta}$ that are quantile-unbiased in the sense of \eqref{arxiv2:eq:exact.upperCI.heterogeneity}. The following result shows that $\hat{\eta}_{1-\beta}$ is minimax \textit{most accurate} in the sense that it minimizes the worst-case probability of overestimating $H(\theta)$, for any degree of underlying heterogeneity $\eta > 0$. 

\begin{proposition}\label{arxiv2:prop:optimality}
For any quantile-unbiased estimator $\Tilde{\eta}_{1-\beta}$ and any degree of heterogeneity $\eta > 0$, the heterogeneity measure $\heta_{1-\beta}$ yields a lower worst-case probability of overestimating $H(\theta)$:
\begin{align*}
    \sup_{\theta \in \Bar{\Theta}^{\eta}}\P{\heta_{1-\beta} \geq H(\theta) + \e} \leq \sup_{\theta \in \Bar{\Theta}^{\eta}}\P{\Tilde{\eta}_{1-\beta} \geq H(\theta) + \e}, \quad \forall (\eta,\e) > 0.
\end{align*}
\end{proposition}

\begin{proof}
See Appendix \ref{arxiv2:app:prop:optimality}.
\end{proof}

The above notion of optimality is a minimax analogue of the notion of a uniformly most accurate confidence bound \citep[Section 3.5]{lehmann2024testing}. In Appendix \ref{arxiv2:app:sec:pfanzagl.optimality}, I show that $\heta_{1-\beta}$ is uniformly most accurate in the class of quantile-unbiased estimators that depend on $\htheta$ through $\htheta'Q\htheta$. In fact, $\heta_{1-\beta}$ minimizes expected loss in this class under any quasiconvex loss function that attains its minimum at $H(\theta)$, for any heterogeneous $\theta$. The latter optimality statement is based on results from \citet{pfanzagl1994parametric} on optimal quantile-unbiased estimation in models with monotone likelihood ratios.

\section{Robust Inference Procedures}\label{arxiv2:sec:robust.inference}
Given the maximum distance between weights $\max_{\l \in \L}\norm{\l - w}_{\Sigma}$ and the heterogeneity UCB $\heta_{1-\beta}$, I define the bias UCB
\begin{align*}
    \hB_{w}^{\beta}(\L) = \heta_{1-\beta}\max_{\l \in \L}\norm{\l - w}_{\Sigma}.
\end{align*}
The following result shows that $\hB_{w}^{\beta}(\L)$ is indeed a valid bias UCB.

\begin{proposition}\label{arxiv2:prop:eta.to.UCB}
For any $w \in \W$, the above $\hB_{w}^{\beta}(\L)$ satisfies
\begin{align*}
    \P{\max_{\l \in \L}|\taul - \tauw| \leq \hB_{w}^{\beta}(\L)} \geq 1-\beta, \quad \forall \theta.
\end{align*}
Moreover, when $\L \neq \{w\}$, there exists $\theta$ such that the coverage probability is equal to $1-\beta$.
\end{proposition}

\begin{proof}
See Appendix \ref{arxiv2:app:proof:eta.to.UCB}.
\end{proof}

In words, $\hB_{w}^{\beta}(\L)$ upper bounds the maximum bias with probability at least $1-\beta$, for any value of $\theta$. This provides an avenue for addressing the bias and undercoverage of the baseline estimator $\htauw$ and $CI_{w}$ for inference on alternative estimands $\taul$ across $\l \in \L$. To this end, I develop the robust estimator in Section \ref{arxiv2:sec:robust.estimator} and the robust CI in Section \ref{arxiv2:sec:robust.CI}.

\subsection{Robust Estimator}\label{arxiv2:sec:robust.estimator}
Since $\hB_{w}^{\beta}(\L)$ is a valid bias UCB under any $w \in \W$, it provides a criterion for optimizing $w$. In the optimization, I allow for consideration sets $\W^{*} \subseteq \W$ that are nonempty, closed, and convex. I distinguish the weights in the optimization problem from the researcher's baseline weights, denoting the former by $\Bar{w} \in \W^{*}$ and the latter by $w \in \W$. I allow for classes $\L$ that depend on the baseline $w$, but not on the optimizer variable $\Bar{w}$. 

Given a consideration set $\W^{*}$, I define the \textit{robust estimator} $\htau^{*}$ as the weighted estimator induced by the \textit{minimax-bias weights} $w^{*} \in \W^{*}$, given by 
\begin{align*}
    \htau^{*} = \htau_{w^{*}}, \quad \quad w^{*} = \arg \adjustlimits \min_{\Bar{w} \in \W^{*}} \max_{\l \in \L} \norm{\l - \Bar{w}}_{\Sigma}.
\end{align*} 
Let $\hB_{\min}^{\beta}(\L) = \heta_{1-\beta}\max_{\l \in \L}\norm{\l - w^{*}}_{\Sigma}$ denote the corresponding minimax-bias UCB. 

\begin{proposition}\label{arxiv2:prop:optimal.weights}
The minimax-bias weights $w^{*}$ exist uniquely and satisfy
\begin{align*}
    \hB_{\min}^{\beta}(\L) \leq \hB_{\Bar{w}}^{\beta}(\L), \quad \forall \Bar{w} \in \W^{*}.
\end{align*}
Moreover, for any parameter space $\Theta^{\eta} = \{\theta: H(\theta) \leq \eta\}$ with a bound $\eta > 0$ on heterogeneity, the minimax-bias weights $w^{*}$ solve
\begin{align*}
    w^{*} = \arg \adjustlimits \min_{\Bar{w} \in \W^{*}} \max_{\l \in \L} \left.\max_{\theta \in \Theta^{\eta}}\right. \abs{\E{\htau_{\Bar{w}}} - \taul}, \quad  \forall \eta > 0.
\end{align*}
In this sense, the robust estimator $\htau^{*} = \htau_{w^{*}}$ is optimal for minimizing worst-case bias under any bound on heterogeneity.
\end{proposition}

\begin{proof}
See Appendix \ref{arxiv2:app:proof:optimal.weights}.
\end{proof}

Proposition \ref{arxiv2:prop:optimal.weights} establishes optimality of the minimax-bias weights $w^{*}$ from two perspectives: first, at the observed data, $w^{*}$ minimizes the bias $\hB_{\Bar{w}}^{\beta}(\L)$ that is inferred \textit{ex post}; and second, across hypothetical data realizations, $w^{*}$ minimizes the bias that is possible \textit{ex ante}. Conveniently, the maximum bias from either perspective is proportional to the maximum distance: $w^{*}$ does not depend on the estimated heterogeneity $\heta_{1-\beta}$ or on the heterogeneity bound $\eta$. In the above setup, the optimality of $w^{*}$ among weights $\Bar{w} \in \W^{*}$ is equivalent to optimality of the robust estimator $\htau^{*}$ among weighted estimators $\htau_{\Bar{w}}$. 

By default, I take the consideration set to be the class of alternative weights: $\W^{*} = \L$. In this case, $w^{*}$ can be interpreted as the weights that minimize worst-case disagreement across the class of alternative weights $\Bar{w} \in \L$. Moreover, since the minimax-bias weights depend only on the covariance geometry and the class of alternatives---rather than on estimated heterogeneity or a heterogeneity bound---they provide a natural default for researchers facing ambiguity over their initial choice of baseline weights. I now analyze the structure and interpretation of the minimax-bias weights $w^{*}$ for the benchmark classes $(\L_{\sigma}, \L_{+}, \L_{X})$ highlighted in Section \ref{arxiv2:sec:alternative.weights}.

\begin{example*}[Bounded Variance, continued]
Let $w^{*}(r)$ denote the minimax-bias weights under the bounded variance class $\L_{\sigma}(r)$. Following equation \eqref{arxiv2:eq:distance.bounded.variance}, one can show that
\begin{align}\label{arxiv2:eq:robust.estimator.GLS}
    \htau^{*}(r) = \htau_{\GLS}, \quad w^{*}(r) = w_{\GLS} = \frac{\Sigma^{-1}\1}{\1'\Sigma^{-1}\1}, \quad \min_{\Bar{w} \in \L_{\sigma}}\max_{\l \in \L_{\sigma}}\norm{\l - \Bar{w}}_{\Sigma} = \sqrt{r^{2}\sigma_{w}^{2} - \sigma_{\GLS}^{2}}.
\end{align}
The corresponding standard deviation is $\sigma^{*}(r) = \sigma_{\GLS}$. In summary, the minimax-bias weights coincide with the GLS weights so that the robust estimator and corresponding variance reduce to the GLS estimator and variance, for any choice of standard deviation ratio bound $r$.\footnote{GLS weights have appeared in discussions about how to aggregate independent CATEs. In this context, \citet{li2018balancing} call them overlap weights while \citet{goldsmith2024contamination} call them easiest-to-estimate weights. Established properties include variance-efficiency \citep{crump2006moving, li2018balancing, goldsmith2024contamination} and interpretation as the policy effect weights from a marginal increase in the log odds of treatment \citep{kennedy2019nonparametric, zhou2022marginal}. I establish a complementary property: GLS weights minimize worst-case bias when the only consensus over weighting schemes is that they should yield estimators with bounded variance.} Since the GLS weights minimize variance, $\sigma_{\GLS}^{2} = \sigma_{\min}^{2}$, it follows that the robust estimator is optimal for both bias \textit{and} variance under the bounded variance class.\footnote{This double-optimality of $\htau_{\GLS}$ parallels recent results in \citet{adusumilli2026you} and \citet{sarfati2026integrating}, which show that variance-efficient estimators are also bias-optimal under unrestricted---but bounded---forms of model misspecification. In the weighted estimand context, bounding heterogeneity in the GLS metric and the alternative weights in the standard deviation norm can be viewed as unrestricted bounds on the misspecification of a target estimand $\taul$. If one imposes further restrictions, such as intersecting $\L_{\sigma}$ with the simplex, then the robust estimator generally differs from GLS.} From Proposition \ref{arxiv2:prop:optimal.weights}, it follows that $\htau^{*} = \htau_{\GLS}$ is minimax-optimal for estimating $\taul$ over $(\l, \theta) \in \L \times \Theta^{\eta}$ under squared error in the class of weighted estimators.\footnote{Recent work studies how to optimally aggregate treatment effects from a decision-theoretic perspective \citep{armstrong2021finite, de2021trading, kwon2025estimating, lau2026aggregating}. In contrast to this work, my framework allows for ambiguity and disagreement over the target estimand.}

When $\Sigma$ is diagonal, the GLS weights are guaranteed to be strictly positive, weighting each group $k$ by the precision $1/\sigma_{k}^{2}$ of its estimate relative to the other groups:
\begin{align*}
    w_{\GLS, k} = \frac{1/\sigma_{k}^{2}}{\sum_{k=1}^{K}1/\sigma_{k}^{2}} > 0, \quad \Sigma =
    \begin{pmatrix}
    \sigma_{1}^{2} & \cdots & 0 \\
    \vdots & \ddots & \vdots \\
    0 & \cdots & \sigma_{K}^{2}
    \end{pmatrix}.
\end{align*}
For example, this occurs when $\htheta$ is a vector of estimates from independent sites $k$, or when $\htheta$ is a vector of CATE estimates computed on independent observations across covariate cells $k$. But when $\Sigma$ is non-diagonal, $w_{\GLS}$ may place negative weight on some groups $k$ in order to minimize variance. To prevent this negative weighting, one can restrict the bounded variance class to the simplex as in \eqref{arxiv2:eq:class.simplex.bounded.variance}. In this case, generally $w^{*} \neq w_{\GLS}$.
\end{example*}

\begin{example*}[Truncated Simplex, continued]
Let $w^{*}(1) = \arg\min_{\Bar{w} \in \W_{+}}\max_{\l \in \W_{+}}\norm{\l - \Bar{w}}_{\Sigma}$ denote the minimax-bias weights under $\L_{+}(1) = \W_{+}$ and $w^{*}(0) = w$ the minimax-bias weights under $\L_{+}(0) = \{w\}$. Following equation \eqref{arxiv2:eq:distance.truncated.simplex}, one can show that
\begin{align}\label{arxiv2:eq:robust.estimator.truncated.simplex}
    w^{*}(\epsilon) = (1-\epsilon)w^{*}(0) + \epsilon w^{*}(1), \quad \min_{\Bar{w} \in \L_{+}}\max_{\l \in \L_{+}}\norm{\l - \Bar{w}}_{\Sigma} = \epsilon\min_{\Bar{w} \in \W_{+}}\max_{\l \in \W_{+}}\norm{\l - \Bar{w}}_{\Sigma}.
\end{align}
Thus, the minimax-bias weights $w^{*}(\epsilon)$ are a fraction $1-\epsilon$ consistent with the baseline weights and a fraction $\epsilon$ consistent with the weights that minimize disagreement over the unrestricted simplex. The robust estimator is therefore $\htau^{*}(\epsilon) = (1-\epsilon)\htau^{*}(0) + \epsilon\htau^{*}(1)$, which is between the baseline estimator $\htau^{*}(0) = \htauw$ and the simplex-robust estimator $\htau^{*}(1)$.
\end{example*}

\begin{example*}[Covariate Balance, continued]
There does not appear to be a transparent formula for $w^{*}(\Bar{c})$ under $\L_{X}(\Bar{c})$. However, in the case of $\Bar{c} = 0$, one can interpret $\L_{X}(0)$ in \eqref{arxiv2:eq:class.covariate.balance} as the set of solutions to a synthetic control problem \citep{abadie2003economic, abadie2010synthetic}. In particular, if the baseline $w$ represents synthetic control weights derived under the max-norm and standardized covariates, the corresponding predictor means $w'\bX$ produce the same synthetic control fit as each alternative $\l \in \L_{X}(0)$, meaning that $\L_{X}(0)$ is a set of synthetic control weights.\footnote{See \citet{liu2025synthetic} for a recent identification framework predicated on such ambiguity sets.} The minimax-bias weights $w^{*}(0)$ under $\L_{X}(0)$ can then be viewed as a point of centrality among a set of synthetic control weights.
\end{example*}

\subsection{Robust CI}\label{arxiv2:sec:robust.CI}
I define the robust CI centered at $w \in \W$ as
\begin{align*}
    CI_{w}^{*} = 
    \begin{cases}
    \hfil \displaystyle \brack{\htauw \pm \cv{\hB_{w}^{\beta}(\L)/\sigma_{w}}\sigma_{w}}, & \text{two-sided}, \\[5pt]  
    \hfil \left(-\infty, \htauw + z_{1-\alpha}\sigma_{w} + \hB_{w}^{\beta}(\L)\right], & \text{one-sided (upper)}, \\[5pt]  
    \hfil \left[\htauw - z_{1-\alpha}\sigma_{w} - \hB_{w}^{\beta}(\L), \infty\right), & \text{one-sided (lower)},
    \end{cases}
\end{align*}
where the critical value function $\cv{b}$ gives the $(1-\alpha)$-quantile of the folded normal distribution $|N(b,1)|$. Thus, the robust $CI_{w}^{*}$ parallels the baseline $CI_{w}$, but uses critical values that adjust for the inferred bias $\hB_{w}^{\beta}(\L)$. This adjustment widens the baseline endpoints: $CI_{w} \subseteq CI_{w}^{*}$, with equality if and only if $\hB_{w}^{\beta}(\L) = 0$. Provided that $\L \neq \{w\}$, this equality holds if and only if the heterogeneity in estimates $H(\htheta)$ is negligible at level $\beta$ in the sense that $\heta_{1-\beta} = 0$.

\begin{proposition}\label{arxiv2:prop:robust.coverage}
$CI_{w}^{*}$ provides uniformly valid coverage at confidence level $1-(\alpha+\beta)$:
\begin{align}\label{arxiv2:eq:robust.coverage}
    \P{\taul \in CI_{w}^{*}} \geq 1-(\alpha+\beta), \quad \forall \l \in \L, \quad \forall \theta.
\end{align}
\end{proposition}
\begin{proof}
See Appendix \ref{arxiv2:app:proof:robust.coverage}.
\end{proof}

Proposition \ref{arxiv2:prop:robust.coverage} shows that $CI_{w}^{*}$ covers $\taul$ at confidence level $1-(\alpha+\beta)$, uniformly over $\l \in \L$.\footnote{$CI_{w} \subseteq CI_{w}^{*}$ implies $\P{\tauw \in CI_{w}^{*}} \geq \P{\tauw \in CI_{w}} = 1-\alpha$ for all $\theta$. Thus, the robust $CI_{w}^{*}$ still covers the baseline estimand $\tauw$ at conventional confidence levels. The additional error level $\beta$ when considering alternative estimands $\taul$ is a Bonferroni correction to account for the first-step inference on heterogeneity; such Bonferroni corrections have been used for two-step construction of test statistics and CIs in other contexts, such as \citet{romano2014practical} and \citet{mccloskey2017bonferroni}.} This uniform coverage differs from the conventional coverage in \eqref{arxiv2:eq:conventional.CIs}. The latter provides inference guarantees for one given $w$, while the former extends such guarantees to every $\l \in \L$. Relative to conventional coverage, the price of uniform coverage is (i) a higher error level from using $\heta_{1-\beta}$ to infer heterogeneity and (ii) longer intervals to obtain coverage that is robust to alternative weights $\l \in \L$. 

\paragraph{Interpretation.} Intuitively, if $\l \in \L$ indexes the preferred weights of different readers, then the robust $CI_{w}^{*}$ provides valid coverage for each reader's estimand $\taul$. For example, suppose the baseline $CI_{w}$ excludes zero, leading the researcher to reject the $w$-null of no average effect $H_{0,w}: \tauw = 0$ at significance level $\alpha$. If zero remains excluded from the robust $CI_{w}^{*}$, then a reader with alternative weights $\l \in \L$ can reject the $\l$-null of no average effect $H_{0,\l}: \taul = 0$ at significance level $\alpha + \beta$. In this sense, $CI_{w}^{*}$ facilitates robust inference.

\paragraph{Additional Coverage Properties.} In Appendix \ref{arxiv2:app:sec:simultaneous.coverage}, I show that $CI_{w}^{*}$ also provides a form of simultaneous coverage---but at a lower confidence level in the two-sided case. I compare this to the uniform coverage in \eqref{arxiv2:eq:robust.coverage}. In comparing the two coverage notions, I draw connections to the literature on inference for partially identified parameters \citep{molinari2020microeconometrics}. In Appendix \ref{arxiv2:app:sec:coverage.upper.bound}, I derive an upper bound on the coverage rate of $CI_{w}^{*}$. The bound is strictly less than one for any $\l \in \L$ under any degree of heterogeneity. This implies that $CI_{w}^{*}$ has nontrivial coverage for $\taul$ across all values of $\l \in \L$ and $\theta$.

\paragraph{Choice of Centering Weights.} The robust $CI_{w}^{*}$ can be centered at any $w \in \W$ while maintaining uniform coverage over $\l \in \L$. By default, I center the robust CIs at the minimax-bias weights $w^{*}$ and denote $CI^{*} = CI_{w^{*}}^{*}$. In particular, given the minimax-bias UCB $\hB_{\min}^{\beta}(\L)$, robust estimator $\htau^{*}$, and standard deviation $\sigma^{*} = \sqrt{(w^{*})'\Sigma w^{*}}$, the robust CIs become
\begin{align*}
    CI^{*} = 
    \begin{cases}
    \hfil \displaystyle \brack{\htau^{*} \pm \cv{\hB_{\min}^{\beta}(\L)/\sigma^{*}}\sigma^{*}}, & \text{two-sided}, \\[5pt]  
    \hfil \left(-\infty, \htau^{*} + z_{1-\alpha}\sigma^{*} + \hB_{\min}^{\beta}(\L)\right], & \text{one-sided (upper)}, \\[5pt]  
    \hfil \left[\htau^{*} - z_{1-\alpha}\sigma^{*} - \hB_{\min}^{\beta}(\L), \infty\right), & \text{one-sided (lower)}.
    \end{cases}
\end{align*}
This choice of center yields the smallest bias UCB $\hB_{\min}^{\beta}(\L)$ across the consideration set weights. In this sense, $w^{*}$ prioritizes the component of robust CI length that comes from the ambiguity or disagreement over $\L$. However, the resulting standard deviation $\sigma^{*}$ may be larger than the baseline $\sigma_{w}$, so the overall effect on length is ambiguous ex ante.\footnote{Nevertheless, when the maximum bias is fixed and positive, the bias UCB is asymptotically non-negligible while the variance shrinks to zero with the sample size, which supports choosing $w^{*}$ as the default center.} When $\heta_{1-\beta}$ is small enough, it is even possible for $CI^{*}$ to be shorter than $CI_{w}$. This is easiest to see with the bounded variance class. Indeed, since the bounded variance class yields $w^{*}=w_{\GLS}$ and the GLS weights minimize variance, then $CI^{*}$ must be weakly shorter than $CI_{w}$ in realizations where $\heta_{1-\beta} = 0$.

\paragraph{Relationship to Bias-Aware CIs.} My robust CIs have a similar structure to the bias-aware CIs advanced in \citet{armstrong2018optimal,armstrong2020simple,armstrong2021finite,armstrong2021sensitivity}, but there are important differences in model setup and inferential objective. The bias-aware CIs are designed to provide valid inference for a fixed target estimand $\l'\theta$ subject to bounds on the parameter space for $\theta$, such as bounds on heterogeneity---see \citet{kwon2025estimating} for results of this flavor. By contrast, my robust CIs are designed to provide valid inference for a class of target estimands $\{\l'\theta: \l \in \L\}$ and impose no bound on heterogeneity, opting instead to infer the heterogeneity using $\heta_{1-\beta}$. A downside of bounding heterogeneity is the potential for undercoverage when the bound is incorrect, while a downside of inferring heterogeneity is the additional error level $\beta$ in the uniform coverage statement \eqref{arxiv2:eq:robust.coverage}.

\section{Uniform Asymptotic Validity}\label{arxiv2:sec:asymptotic.validity}
The foregoing results are developed under normally distributed estimates $\htheta \sim N(\theta, \Sigma)$, a known covariance matrix $\Sigma$, known baseline weights $w$, and a known class of alternative weights $\L$. In this section, I establish the uniform asymptotic validity of my procedures under asymptotically normal estimates, consistent covariance matrix estimators, consistent estimators for the baseline weights, and consistent estimators for the class of alternative weights. Sections \ref{arxiv2:sec:asymptotic.environment} and \ref{arxiv2:sec:asymptotics.baseline.alternative.weights} formalize the asymptotic setup and assumptions. Sections \ref{arxiv2:sec:consistent.classes}, \ref{arxiv2:sec:asymptotics.distance.measures}, and \ref{arxiv2:sec:asymptotics.robust.estimator} establish asymptotic results for the class of alternative weights, maximum distance, and robust estimator. Sections \ref{arxiv2:sec:measures.heterogeneity.asymptotic} and \ref{arxiv2:sec:asymptotics.robust.CI} establish asymptotic results for the heterogeneity UCB and robust CI. For self-contained practical implementations of my inference procedures, see Section \ref{arxiv2:sec:implementation}.

\subsection{Environment}\label{arxiv2:sec:asymptotic.environment} 
There is a sample of size $n$ drawn from some unknown distribution $\Pn \in \cPn$, where $\cPn$ is a class of distributions for samples of size $n$. Based on this data, the researcher constructs a vector of estimates $\hthetan \in \R^{K}$ and a positive definite covariance matrix estimator $\tSigman \in \R^{K \times K}$. Let $\hSigman = n\tSigman$ denote the normalized covariance matrix estimator. Under distribution $\Pn$, the sample objects $(\hthetan, \hSigman)$ have population analogues $(\theta(\Pn), \Sigma(\Pn))$, where $\Sigma(\Pn)$ is positive definite. I focus on asymptotics where the number of groups $K \geq 2$ is fixed while the sample size $n \to \infty$ grows. For extensions to growing-$K$ asymptotics, see Appendix \ref{arxiv2:app:sec:UK}.

First, I assume that the vector of normalized estimates $\hZn(\Pn) = \sqrt{n}(\hthetan - \theta(\Pn))$ is uniformly asymptotically normal in the sense that its $\Sigma(\Pn)$-norm-standardized unit-vector linear combinations converge to the standard normal distribution in the Kolmogorov metric, uniformly over $\Pn \in \cPn$. Moreover, I assume that the vector of parameters $\theta(\Pn)$ is uniformly bounded (over $\Pn \in \cPn$ and $n \geq 1$).
\begin{assumptionU}\label{arxiv2:ass:linear.CLT}
For the set of unit vectors $\cU = \{u \in \R^{K}: \norm{u} = 1\}$,
\begin{align*}
    \limn \supPn \sup_{u \in \cU} \sup_{x \in \R}
    \abs{\P[\Pn]{\frac{u'\hZn(\Pn)}{\sqrt{u'\Sigma(\Pn) u}} \leq x} - \Phi(x)} = 0, \quad \hZn(\Pn) = \sqrt{n}(\hthetan - \theta(\Pn)).
\end{align*}
Moreover, there exists a constant $\UB{C}_{\theta} > 0$ such that $\supPn \norm{\theta(\Pn)} \leq \UB{C}_{\theta}$ for all $n$.
\end{assumptionU}

Assumption \ref{arxiv2:ass:linear.CLT} is a natural way to formulate the uniform convergence in distribution of the normalized estimates $\hZn(\Pn)$ to the normal distribution $N(0, \Sigma(\Pn))$.\footnote{This formulation of uniform asymptotic normality is motivated by the Cramér-Wold device for establishing convergence in distribution to multivariate normal distributions \citep[Example 2.18]{van2000asymptotic}. One can also formulate uniform convergence in terms of bounded Lipschitz metric \citep[Section 1.12]{van1996weak}, but the current formulation more easily extends to the large-$K$ asymptotic results in Appendix \ref{arxiv2:app:sec:UK}.} If the components of $\hthetan$ are regression coefficients or sample averages computed over unit-level observations, then Assumption \ref{arxiv2:ass:linear.CLT} follows from bounds on the moments of the observations and bounds on the dependence across observations.

Next, I assume that the covariance matrix estimator $\hSigman = n\tSigman$ is uniformly consistent for the population covariance matrix $\Sigma(\Pn)$, and that the eigenvalues of $\Sigma(\Pn)$ are uniformly bounded above and away from zero. 

\begin{assumptionU}\label{arxiv2:ass:consistent.covariance} 
For each $\e > 0$, 
\begin{align*}
    \lim_{n \to \infty} \sup_{P_{n} \in \mathcal{P}_{n}} \P[\Pn]{\norm{\hSigman - \Sigma(\Pn)} > \e} = 0,
\end{align*}
where $\norm{\Sigma}$ denotes the matrix operator norm. Moreover, there exist constants $\UB{e}, \LB{e} > 0$ such that
\begin{align*}
    \LB{e} \leq \inf_{n}\infPn \emin(\Sigma(\Pn)) \leq \sup_{n}\supPn \emax(\Sigma(\Pn)) \leq \UB{e},
\end{align*}
where $\emin(\Sigma)$ and $\emax(\Sigma)$ denote the minimum and maximum eigenvalues of a matrix $\Sigma$.
\end{assumptionU}

If $\hSigman$ is appropriate for the setting at hand (e.g., sample covariance for iid data, long-run covariance for time series data, etc.), then uniform consistency follows from the same sorts of sufficient conditions as for Assumption \ref{arxiv2:ass:linear.CLT}. The uniform upper and lower eigenvalue bounds for $\Sigma(\Pn)$ ensure that the normal distributions $N(0, \Sigma(\Pn))$ are uniformly tight and nondegenerate.

\begin{remark}[Analogy to the Normal Model]\label{arxiv2:remark:analogy.environment.normal}
The normal model maintains exact normality $\htheta \sim N(\theta, \Sigma)$ and known $\Sigma$. The asymptotic environment instead has a normal approximation $\hthetan \overset{a}{\sim} N(\theta(\Pn), \Sigma(\Pn)/n)$ and estimated $\tSigman = \hSigman/n \overset{a}{\approx} \Sigma(\Pn)/n$. Thus, the objects $(\htheta, \Sigma, \theta)$ in the normal model are analogous to the objects $(\hthetan, \tSigman, \theta(\Pn))$ in the asymptotic environment. 
\end{remark}

\subsection{Baseline and Alternative Weights}\label{arxiv2:sec:asymptotics.baseline.alternative.weights}
The researcher constructs a vector of baseline weights $\hwn \in \W$ and a nonempty, compact, and convex class of alternative weights $\hLn \subseteq \W$. The population analogues are $(w(\Pn), \L(\Pn))$. The baseline and alternative estimators are $\htauwn = \hwn'\hthetan$ and $\htauwn[\l] = \l'\hthetan$, where $\l \in \hLn$. The baseline and alternative estimands are $\tau_{\wn}(\Pn) = w(\Pn)'\theta(\Pn)$ and $\tau_{\l}(\Pn) = \l'\theta(\Pn)$, where $\l \in \L(\Pn)$. If the baseline and alternative weights are known, then I define $\hwn = w(\Pn)$ and $\hLn = \L(\Pn)$.

I assume that the sample baseline weights $\hwn$ are uniformly consistent for the population baseline weights $w(\Pn)$. Moreover, I assume that $w(\Pn)$ is uniformly bounded---i.e., the baseline weighting scheme cannot place arbitrarily large negative weight on any group.

\begin{assumptionU}\label{arxiv2:ass:consistent.weights} 
For each $\e > 0$,
\begin{align*}
    \limn \supPn \P[\Pn]{\norm{\hwn - w(\Pn)} > \e} = 0.
\end{align*}
Moreover, there exists a constant $\UB{C}_{w} > 0$ such that $\supPn \norm{w(\Pn)} \leq \UB{C}_{w}$ for all $n$.
\end{assumptionU}

For example, under Assumption \ref{arxiv2:ass:consistent.covariance}, it follows that Assumption \ref{arxiv2:ass:consistent.weights} holds for GLS weights
\begin{align}\label{arxiv2:eq:GLS.weights.Sigman}
    w_{\GLS}(\Pn) = \arg\min_{w \in \W} \left. w'\Sigma(\Pn)w \right. = \frac{\Sigma(\Pn)^{-1}\1}{\1'\Sigma(\Pn)^{-1}\1}, \quad \norm{w_{\GLS}(\Pn)} \leq \frac{\UB{e}/\LB{e}}{\sqrt{K}}.
\end{align}
Assumption \ref{arxiv2:ass:consistent.weights} also holds for weights considered in event studies, such as those in \citet{callaway2021difference} under uniform versions of their asymptotic linearity assumptions, and those in \citet{sun2021estimating} under uniform versions of their moment conditions.

I suppose that the class of alternative weights $\hLn$ depends on the data through statistics $\hSn \in \cS$ taking values in a metric space $(\cS,d_{\cS})$. In particular, $\hLn = \L(\hSn)$ and $\L(\Pn) = \L(S(\Pn))$, where $S(\Pn)$ is the population analogue of $\hSn$. I assume the sample statistics $\hSn$ are uniformly consistent for the population statistics $S(\Pn)$. Moreover, $S(\Pn)$ is contained in a compact set $\bbS \subseteq \cS$ and $\hSn$ is contained in $\bbS$ with probability uniformly approaching one.

\begin{assumptionU}\label{arxiv2:ass:consistent.S} 
For each $\e > 0$,
\begin{align*}
    \limn \supPn \P[\Pn]{d_{\cS}(\hSn, S(\Pn)) > \e} = 0.
\end{align*}
Moreover, $S(\Pn) \in \bbS$ for all $\Pn \in \cPn$ and $n$, and
\begin{align*}
    \limn \supPn \P[\Pn]{\hSn \notin \bbS} = 0.
\end{align*}
\end{assumptionU}

For a bounded variance class with $\hSn = (\hSigman, \hwn)$, Assumption \ref{arxiv2:ass:consistent.S} follows from Assumptions \ref{arxiv2:ass:consistent.covariance} and \ref{arxiv2:ass:consistent.weights}, where $\bbS$ can be defined with constants $\LB{e}/2$, $2\UB{e}$, and $2\UB{C}_{w}$; see the example below for details. For a truncated simplex class with $\hSn = \hwn \in \W_{+}$, Assumption \ref{arxiv2:ass:consistent.S} follows from Assumption \ref{arxiv2:ass:consistent.weights} and $\wn \in \bbS = \W_{+}$. For a covariate balance class with $\hSn = (\hwn, \hbXn) \in \W_{+} \times \R^{K \times M}$, where $\hbXn$ is an estimated covariate mean matrix, Assumption \ref{arxiv2:ass:consistent.S} follows from Assumption \ref{arxiv2:ass:consistent.weights} with $\wn \in \W_{+}$ and uniform consistency of $\hbXn$ to a population matrix $\bX(\Pn)$ with $\norm{\bX(\Pn)} \leq \UB{C}_{X}$ uniformly for some constant $\UB{C}_{X} > 0$ and $\min_{m}\text{sd}(\bX_{m}(\Pn)) \geq \varsigma$ uniformly for some constant $\varsigma > 0$, where the corresponding component of $\bbS$ can be defined with constants $2\UB{C}_{X}$ and $\varsigma/2$; see the example below for details.

Let $\diam(A) = \sup_{a,b \in A}\norm{a - b}$ denote the diameter of a given set $A \subseteq \R^{K}$. I assume that the class of alternative weights has the following structure. 
\begin{assumptionU}\label{arxiv2:ass:Lambda.representation} 
For a compact and convex set $\bbW \subseteq \W$ with positive diameter $\diam(\bbW) = \max_{\l,w \in \bbW}\norm{\l - w} > 0$, there exists a function $g: \bbW \times \bbS \to \R$ with constants $L_{g}, \delta_{g} > 0$ such that
\begin{enumerate}[label=(\roman*)]
    \item for each $S \in \bbS$, the class of alternative weights can be represented as
    \begin{align*}
        \L(S) = \curly{\l \in \bbW: g(\l, S) \leq 0},
    \end{align*}
    where the functions $\{\l \mapsto g(\l,S): S \in \bbS\}$ are continuous and convex on $\bbW$;
    \item the functions $\{S \mapsto g(\l,S): \l \in \bbW\}$ are uniformly $L_{g}$-Lipschitz on $\bbS$ in the sense that
    \begin{align*}
        \max_{\l \in \bbW}\abs{g(\l, S_{1}) - g(\l, S_{2})} \leq L_{g}d_{\cS}(S_{1},S_{2}), \quad \forall S_{1},S_{2} \in \bbS;
    \end{align*}
    \item for each $S \in \bbS$, there exists $\l^{\circ}(S) \in \bbW$ satisfying the Slater condition
    \begin{align*}
        g(\l^{\circ}(S), S) \leq -\delta_{g}.
    \end{align*}
\end{enumerate}
\end{assumptionU}
Condition \ref{arxiv2:ass:Lambda.representation}(i) requires that, on the set $\bbS$ containing the population statistics $S(\Pn)$, the classes of interest can be represented as inequality constraints of continuous and convex functions $\l \mapsto g(\l,S)$ on a compact $\bbW$, which implies that the population classes $\L(\Pn) \subseteq \bbW$ are compact and convex. Condition \ref{arxiv2:ass:Lambda.representation}(ii) requires that the functions $S \mapsto g(\l,S)$ in the representation are uniformly Lipschitz on $\bbS$, which helps in establishing continuity of $S \mapsto \L(S)$ in the Hausdorff metric defined below in Section \ref{arxiv2:sec:consistent.classes}. Condition \ref{arxiv2:ass:Lambda.representation}(iii) requires that the inequality constraints are uniformly strictly feasible, which helps in establishing continuity of $S \mapsto \L(S)$ and implies that the population classes $\L(\Pn)$ are nonempty. Moreover, it implies that the classes are nondegenerate in the sense that the maximum distance between weights is uniformly bounded away from zero. In particular, the maintained assumptions combined with Lemma \ref{arxiv2:app:lemma:nondegenerate.disagreement} yield
\begin{align*}
    \inf_{n} \inf_{\Pn \in \cPn} \inf_{w \in \W}\max_{\l \in \L(\Pn)} \norm{\l - w}_{\Sigma(\Pn)} \geq  \frac{\sqrt{\LB{e}}}{2}\frac{\delta_{g}}{\delta_{g} + \UB{C}_{g}}\diam(\bbW) > 0,
\end{align*}
where $\UB{C}_{g} = \max_{S \in \bbS} \max_{\l \in \bbW} \max\curly{g(\l,S), 0} < \infty$. Thus, condition \ref{arxiv2:ass:Lambda.representation}(iii) rules out asymptotic regimes in which the class of alternatives collapses to a singleton, which helps in formulating general arguments for the asymptotic validity of the robust CIs. I now map the conditions of Assumption \ref{arxiv2:ass:Lambda.representation} to the example classes of interest.

\begin{example*}[Bounded Variance, continued] Let $S = (\Sigma,w)$, where $\Sigma$ is a positive definite matrix and $w \in \W$. Given $r \geq 1$, define the bounded variance class
\begin{align*}
    \L_{\sigma}(\Sigma,w) = \curly{\l \in \W: g(\l,\Sigma,w) \leq 0}, \quad g(\l,\Sigma,w) = \sqrt{\l'\Sigma\l} -r\sqrt{w'\Sigma w}.
\end{align*}
Take the eigenvalue and norm bounds in Assumptions \ref{arxiv2:ass:consistent.covariance} and \ref{arxiv2:ass:consistent.weights} as given and define
\begin{align*}
    \bbS = \curly{(\Sigma,w): \LB{e}/2 \leq \emin(\Sigma) \leq \emax(\Sigma) \leq 2\UB{e}, w \in \W, \norm{w} \leq 2\UB{C}_{w}}.
\end{align*}
Choose a finite constant $\UB{C}_{\L} > 4r(\UB{e}/\LB{e})^{1/2}\UB{C}_{w}$ and define
\begin{align*}
    \bbW = \curly{\l \in \W: \norm{\l} \leq \UB{C}_{\L}}.
\end{align*}
If $\l \in \L_{\sigma}(\Sigma,w)$ and $(\Sigma,w) \in \bbS$, then
\begin{align*}
    \norm{\l}^{2}\LB{e}/2 \leq \l'\Sigma\l \leq r^{2}w'\Sigma w \leq r^{2}(2\UB{e})(2\UB{C}_{w})^{2} = 8r^{2}\UB{e}\UB{C}_{w}^{2},
\end{align*}
which implies $\norm{\l} \leq 4r(\UB{e}/\LB{e})^{1/2}\UB{C}_{w} \leq \UB{C}_{\L}$. This means $\L_{\sigma}(\Sigma,w) \subseteq \bbW$ for each $(\Sigma,w) \in \bbS$. One can thus consider $g(\l,\Sigma,w)$ over $(\l, \Sigma, w) \in \bbW \times \bbS$, which is continuous and convex in $\l$, and Lipschitz in $(\Sigma, w)$. For the Slater condition, consider
\begin{align*}
    w_{\GLS}(\Sigma) = \frac{\Sigma^{-1}\1}{\1'\Sigma^{-1}\1}, \quad \sigma_{\min}(\Sigma) = \min_{w \in \W}\sqrt{w'\Sigma w} = \sqrt{w_{\GLS}(\Sigma)'\Sigma w_{\GLS}(\Sigma)} = \frac{1}{\sqrt{\1'\Sigma^{-1}\1}}.
\end{align*}
If $r > 1$, then $\l^{\circ}(\Sigma,w) = w_{\GLS}(\Sigma)$ is a Slater point:
\begin{align*}
    g(\l^{\circ}(\Sigma,w),\Sigma,w) \leq (1-r)\sigma_{\min}(\Sigma) \leq -\delta_{g}, \quad \delta_{g} = (r-1)\sqrt{\frac{\LB{e}}{2K}} > 0.
\end{align*}
For $r=1$, it suffices to assume that $w$ is bounded away from $\l^{\circ}(\Sigma,w) = w_{\GLS}(\Sigma)$. But since this may be difficult to verify in practice, a simple alternative is to use $r = 1+\delta$ for some small fixed $\delta > 0$.\footnote{If one wishes to allow for data-dependent $\delta$ and $r$, it suffices to embed these parameters into $\bbS$ and then investigate the conditions of Assumption \ref{arxiv2:ass:Lambda.representation} (and Assumption \ref{arxiv2:ass:consistent.S}) relative to that new class.} The same logic applies to the boundary cases in the examples below.
\end{example*}

\begin{example*}[Truncated Simplex, continued]
Let $S = w \in \W_{+}$ and $\bbS = \bbW = \W_{+}$. Given $\epsilon \in [0,1]$, define the truncated simplex class
\begin{align*}
    \L_{+}(w) = \curly{\l \in \bbW: g(\l,w) \leq 0}, \quad g(\l,w) = \max_{k \in \{1,\ldots,K\}}\paren{(1-\epsilon)w_{k}-\l_{k}}, \quad w \in \bbS.
\end{align*}
The function $g(\l,w)$ is continuous and convex in $\l$ and Lipschitz in $w$. For the Slater condition, consider the mixture $\l^{\circ}(w) = (1-\epsilon)w + \epsilon \1/K$. If $\epsilon > 0$, then
\begin{align*}
    g(\l^{\circ}(w),w) = -\delta_{g}, \quad \delta_{g} = \epsilon/K > 0.
\end{align*}
\end{example*}

\begin{example*}[Covariate Balance, continued]
Let $S = (w,\bX)$, where $\bX = (\bX_{1},\ldots,\bX_{M}) \in \R^{K \times M}$ and $w \in \W_{+}$. Given $\Bar{c} \geq 0$, define the balance gap measure
\begin{align*}
    c_{\l}(w,\bX) = \max_{m \in \{1,\ldots,M\}} \frac{\abs{(\l-w)'\bX_{m}}}{\text{sd}(\bX_{m})}.
\end{align*}
Take $\bbW = \W_{+}$ and define the covariate balance class
\begin{align*}
    \L_{X}(w,\bX) = \curly{\l \in \bbW: g(\l,w,\bX) \leq 0}, \quad g(\l,w,\bX) = c_{\l}(w,\bX) - \Bar{c}.
\end{align*}
Define the set
\begin{align*}
    \bbS = \curly{(w,\bX): w \in \W_{+}, \norm{\bX} \leq 2\UB{C}_{X}, \min_{m \in \{1,\ldots,M\}}\text{sd}(\bX_{m}) \geq \varsigma/2}.
\end{align*}
Over $(w,\bX) \in \bbS$, the standard deviation functions $\bX_{m} \mapsto \text{sd}(\bX_{m})$ are Lipschitz and bounded away from zero and the functions $\{(w,\bX_{m}) \mapsto \abs{(\l-w)'\bX_{m}}: \l \in \bbW\}$ are (uniformly) Lipschitz. Thus, the function $g(\l, w, \bX)$ is Lipschitz in $(w, \bX)$. It is also continuous and convex in $\l$. For the Slater condition, consider $\l^{\circ}(w,\bX) = w$. If $\Bar{c} > 0$, then
\begin{align*}
    g(\l^{\circ}(w, \bX),w,\bX) = -\delta_{g}, \quad \delta_{g} = \Bar{c} > 0.
\end{align*}
\end{example*}

\begin{remark}[Mapping to Intersections]
Finite intersections of the above classes can be handled by considering the maximum of their corresponding $g$-functions, yielding
\begin{align*}
    \L(S) = \bigcap_{j=1}^{J}\curly{\l \in \bbW: g_{j}(\l,S) \leq 0} = \curly{\l \in \bbW: g(\l,S) \leq 0}, \quad g(\l,S)=\max_{j \in \{1,\ldots,J\}} g_{j}(\l,S).
\end{align*}
If each function $g_{j}(\l,S)$ is continuous and convex in $\l$ and Lipschitz in $S$, then so is $g(\l,S)$. The Slater condition holds when there is a common point $\l^{\circ}(S) \in \bbW$ satisfying
\begin{align*}
    g_{j}(\l^{\circ}(S),S) \leq -\delta_{g}, \quad \forall j \in \{1,\ldots,J\}.
\end{align*}
The above observations can be used to map the conditions of Assumption \ref{arxiv2:ass:Lambda.representation} to the bounded variance simplex class and the truncated covariate balance class. For example, the bounded variance simplex $\L_{\sigma}^{+}(\Sigma,w) = \{\l \in \W_{+}: \sqrt{\l'\Sigma \l} - r\sqrt{w'\Sigma w} \leq 0\}$ can be mapped to $\bbW = \W_{+}$, $\bbS = \{(\Sigma, w): \LB{e}/2 \leq \emin(\Sigma) \leq \emax(\Sigma) \leq 2\UB{e}, w \in \W_{+}\}$, $g(\l, \Sigma, w) = \sqrt{\l'\Sigma \l} - r\sqrt{w'\Sigma w}$, and Slater point $\l^{\circ}(\Sigma,w) = \argmin_{\l \in \W_{+}}\sqrt{\l'\Sigma \l}$ for $r > 1$ given by the simplex-constrained variance minimizing weights. For $r=1$, the aforementioned caveats apply.
\end{remark}

\paragraph{Notation for Population Objects.} When convenient, I will use the shorthands
\begin{align*}
    (\thetan, \Sigman, \wn, \Ln, \Sn) = (\theta(\Pn), \Sigma(\Pn), w(\Pn), \L(\Pn), S(\Pn)),
\end{align*}
and likewise for other objects that depend on $\Pn$.

\subsection{Consistency of the Class of Alternatives}\label{arxiv2:sec:consistent.classes}
To formulate uniform consistency of the estimated class $\hLn$ for the population class $\Ln$, I use a standard notion of distance between sets. Given two nonempty sets $A \subseteq \R^{K}$ and $B \subseteq \R^{K}$, let $\dist(a,B) = \inf_{b \in B}\norm{a-b}$ and $\dist(b,A) = \inf_{a \in A}\norm{b-a}$ denote the distance from $a \in A$ to $B$ and from $b \in B$ to $A$, respectively. The Hausdorff distance between $A$ and $B$ is defined as
\begin{align*}
    d_{H}(A,B) = \max\curly{\sup_{a \in A}\left.\dist(a,B),\right. \sup_{b \in B}\left.\dist(b,A)\right.}.
\end{align*}
Under Assumption \ref{arxiv2:ass:Lambda.representation}, Lemma \ref{arxiv2:app:lemma:hausdorff.continuity} shows that $S \mapsto \L(S)$ is Lipschitz continuous in the Hausdorff metric $d_{H}$:
\begin{align*}
    d_{H}(\L(S_{1}),\L(S_{2})) \leq \frac{\diam(\bbW)}{\delta_{g}}L_{g}d_{\cS}(S_{1},S_{2}), \quad \forall S_{1},S_{2} \in \bbS.
\end{align*}
Based on this continuity, the following result shows that $\hLn = \L(\hSn)$ is a uniformly consistent estimator of $\Ln = \L(\Sn)$ in the Hausdorff metric.

\begin{propositionU}\label{arxiv2:prop:uniform.hausdorff.convergence}
Under Assumptions \ref{arxiv2:ass:consistent.S} and \ref{arxiv2:ass:Lambda.representation}, for each $\e > 0$,
\begin{align*}
    \limn \supPn \P[\Pn]{d_{H}(\hLn, \Ln) > \e} = 0.
\end{align*}
\end{propositionU}

\begin{proof}
See Appendix \ref{arxiv2:app:proof:uniform.hausdorff.convergence}.        
\end{proof}

\subsection{Consistency of the Maximum Distance}\label{arxiv2:sec:asymptotics.distance.measures}
The estimated and population maximum distances between weights are
\begin{align*}
    \max_{\l \in \hLn}\norm{\l - \hwn}_{\hSigman} = \max_{\l \in \hLn}\sqrt{(\l - \hwn)'\hSigman(\l - \hwn)}, \quad \max_{\l \in \Ln}\norm{\l - \wn}_{\Sigman} = \max_{\l \in \Ln}\sqrt{(\l - \wn)'\Sigman(\l - \wn)}.
\end{align*}
Intuitively, $\norm{\l - \wn}_{\Sigman}/\sqrt{n}$ is the standard deviation for $(\l - \wn)'\hthetan$ under the normal approximation $\hthetan \overset{a}{\sim} N(\thetan, \Sigman/n)$, while $\norm{\l - \hwn}_{\hSigman}/\sqrt{n} = \norm{\l - \hwn}_{\tSigman}$ is the plug-in standard error. The following result shows that the estimated maximum distance is uniformly consistent for the population maximum distance.

\begin{propositionU}\label{arxiv2:prop:consistent.distance}
Under Assumptions \ref{arxiv2:ass:consistent.covariance}-\ref{arxiv2:ass:Lambda.representation}, for each $\e > 0$,
\begin{align*}
    \limn \supPn \P[\Pn]{\abs{\max_{\l \in \hLn}\norm{\l - \hwn}_{\hSigman} - \max_{\l \in \Ln}\norm{\l - \wn}_{\Sigman}} > \e} = 0.
\end{align*}
\end{propositionU}

\begin{proof}
See Appendix \ref{arxiv2:app:proof:consistent.distance}.
\end{proof}

\subsection{Optimality of the Robust Estimator}\label{arxiv2:sec:asymptotics.robust.estimator}
The estimated and population minimax-bias weights are
\begin{align*}
    \rhwn = \arg \adjustlimits \min_{\Bar{w} \in \hLn} \max_{\l \in \hLn} \norm{\l - \Bar{w}}_{\hSigman}, \quad \rwn = \arg \adjustlimits \min_{\Bar{w} \in \Ln} \max_{\l \in \Ln} \norm{\l - \Bar{w}}_{\Sigman}.
\end{align*}
As in the normal model, I allow for classes $(\hLn, \Ln)$ that depend on the baseline weights $(\hwn, \wn)$, but not on the optimizer variable $\Bar{w}$. The following result shows that the minimax-bias weights satisfy Assumption \ref{arxiv2:ass:consistent.weights}.

\begin{propositionU}\label{arxiv2:prop:consistent.minimax.weights}
Under Assumptions \ref{arxiv2:ass:consistent.covariance}, \ref{arxiv2:ass:consistent.S}, and \ref{arxiv2:ass:Lambda.representation}, for each $\e > 0$,
\begin{align*}
    \limn \supPn \P[\Pn]{\norm{\rhwn - \rwn} > \e} = 0.
\end{align*}
Moreover, there exists a constant $\UB{C}_{w} > 0$ such that $\supPn \norm{\rwn} \leq \UB{C}_{w}$ for all $n$.
\end{propositionU}

\begin{proof}
See Appendix \ref{arxiv2:app:proof:consistent.minimax.weights}.
\end{proof}

To establish asymptotic optimality of the robust estimator
\begin{align*}
    \rhtaun = \htauwn[\rhwn] = \rhwn'\hthetan,
\end{align*}
I first characterize the asymptotic maximum bias of weighted estimators $\htauwn = \hwn'\hthetan$ under the following asymptotic uniform integrability (UI) condition:
\begin{align}\label{arxiv2:eq:UI.condition}
    \lim_{C \to \infty} \limsupn \supPn
    \begin{pmatrix}
    \E[\Pn]{\|\hthetan - \thetan\|^{2} \I{\|\hthetan - \thetan\|^{2} > C}} \\
    \E[\Pn]{\|\hwn - \wn\|^{2} \I{\|\hwn - \wn\|^{2} > C}}
    \end{pmatrix}
    = 0.
\end{align}
UI condition \eqref{arxiv2:eq:UI.condition} controls the tails of the squared estimation errors $\|\hthetan - \thetan\|^{2}$ and $\|\hwn - \wn\|^{2}$, ensuring that the corresponding mean squared errors uniformly converge to zero under Assumptions \ref{arxiv2:ass:linear.CLT} and \ref{arxiv2:ass:consistent.weights}. Similar to the discussions for Assumptions \ref{arxiv2:ass:linear.CLT} and \ref{arxiv2:ass:consistent.weights}, UI condition \eqref{arxiv2:eq:UI.condition} follows from moment and dependence bounds on the underlying observations used to construct $(\hthetan, \hwn)$.\footnote{By Markov's inequality, UI condition \eqref{arxiv2:eq:UI.condition} follows from bounds strong enough to yield uniformly bounded $L^{2+\delta}(\Pn)$ norms on the estimation errors for some $\delta > 0$; e.g., see \citet[Example 2.21]{van2000asymptotic}.}

For the population heterogeneity in parameters $\Hn(\thetan) = \min_{\gamma \in \R}\norm{\thetan - \1\gamma}_{\Sigman^{-1}}$, denote the largest heterogeneity over $\cPn$ as $\UB{\eta}(\cPn) = \supPn\Hn(\thetan)$, which is uniformly bounded over $n$ under the lower eigenvalue bound on $\Sigman$ and the norm bound on $\thetan$. For each $\Pn \in \cPn$ and $n$, let $P_{n}^{\dagger}$ denote a distribution where $\Sigma(P_{n}^{\dagger}) = \Sigman$, $S(P_{n}^{\dagger}) = \Sn$, $w(P_{n}^{\dagger}) = \wn$, and
\begin{align*}
    \theta(P_{n}^{\dagger}) = \I{\Ln \neq \{\wn\}}\UB{\eta}(\cPn)\Sigman(\ln^{*} - \wn)/\norm{\ln^{*} - \wn}_{\Sigman}, \quad \ln^{*} \in \arg\max_{\l \in \Ln}\norm{\l - \wn}_{\Sigman},
\end{align*}
which satisfies $\Hn(\theta(P_{n}^{\dagger})) = \UB{\eta}(\cPn)$ whenever $\Ln \neq \{\wn\}$. If $P_{n}^{\dagger} \in \cPn$ for all $\Pn \in \cPn$ and $n$, then the class of distributions will be rich enough for the bias bound below to be sharp, in analogy to the sharpness of the bias bound from Proposition \ref{arxiv2:prop:optimal.weights}. Let $\tauwn[\l] = \l'\thetan$.

\begin{propositionU}\label{arxiv2:prop:asymptotic.bias}
Under Assumptions \ref{arxiv2:ass:linear.CLT}-\ref{arxiv2:ass:Lambda.representation} and UI condition \eqref{arxiv2:eq:UI.condition},
\begin{align*}
    \limsupn \paren{\supPn \max_{\l \in \Ln} \abs{\E[\Pn]{\htauwn} - \tauwn[\l]}}
    \leq \limsupn \paren{\UB{\eta}(\cPn) \supPn \max_{\l \in \Ln} \norm{\l - \wn}_{\Sigman}},
\end{align*}
where the right-hand side is finite. Moreover, equality holds if $P_{n}^{\dagger} \in \cPn$ for all $\Pn \in \cPn$ and $n$.
\end{propositionU}

\begin{proof}
See Appendix \ref{arxiv2:app:proof:asymptotic.bias}.
\end{proof}

The following result establishes asymptotic optimality of the robust estimator, serving as the asymptotic analogue of Proposition \ref{arxiv2:prop:optimal.weights} from the normal model.

\begin{propositionU}\label{arxiv2:prop:asymptotic.optimality}
Let Assumptions \ref{arxiv2:ass:linear.CLT}, \ref{arxiv2:ass:consistent.covariance}, \ref{arxiv2:ass:consistent.S}, and \ref{arxiv2:ass:Lambda.representation} be satisfied and suppose that $(\hthetan, \rhwn)$ satisfies UI condition \eqref{arxiv2:eq:UI.condition}. Then for any estimator $\htauwn$ with weights $\hwn$ satisfying Assumption \ref{arxiv2:ass:consistent.weights}, the weights component of UI condition \eqref{arxiv2:eq:UI.condition}, $P_{n}^{\dagger} \in \cPn$, and $\wn \in \Ln$ for all $\Pn \in \cPn$ and $n$, the robust estimator $\rhtaun = \rhwn'\hthetan$ has a lower asymptotic maximum bias in the sense that
\begin{align*}
    \limsupn \paren{\supPn \max_{\l \in \Ln} \abs{\E[\Pn]{\rhtaun} - \tauwn[\l]}} \leq \limsupn \paren{\supPn \max_{\l \in \Ln} \abs{\E[\Pn]{\htauwn} - \tauwn[\l]}},
\end{align*}
where the right-hand side is finite.
\end{propositionU}

\begin{proof}
See Appendix \ref{arxiv2:app:proof:asymptotic.optimality}.    
\end{proof}

This optimality result presumes that $\rhwn$ satisfies the weights component of UI condition \eqref{arxiv2:eq:UI.condition}. This holds, for example, when the class of alternatives $\L(S)$ is uniformly bounded over $S \in \cS$---not just $S \in \bbS$---as with any subset of the simplex. In other cases, it may be possible to argue the UI condition directly. For example, $\rhwn$ reduces to the GLS weights in formula \eqref{arxiv2:eq:GLS.weights.Sigman} under the bounded variance class, from which it suffices to impose appropriate UI conditions on $\hSigman$. Note that even without UI conditions, $\rhwn$ is still uniformly consistent for the weights $\rwn$ that minimize the population maximum distance, as established in Proposition \ref{arxiv2:prop:consistent.minimax.weights}.

\subsection{Validity of the Heterogeneity UCB}\label{arxiv2:sec:measures.heterogeneity.asymptotic}
The \textit{sample heterogeneity in parameters} is the square root of the residual sum of squares from the GLS regression of the parameters $\thetan$ on a constant $\1$ weighted by $\hSigman^{-1}$:
\begin{align*}
    \hHn(\thetan) = \min_{\gamma \in \R}\norm{\thetan - \1\gamma}_{\hSigman^{-1}} =  \sqrt{\thetan'\hQn\thetan}, \quad \hQn = \hSigman^{-1/2}\hAn \hSigman^{-1/2}, \quad \hAn = I - \frac{\hSigman^{-1/2}\1\1'\hSigman^{-1/2}}{\1'\hSigman^{-1}\1},
\end{align*}
where $\hAn$ is the annihilator matrix for $\hSigman^{-1/2}\1$. The function $\hHn$ used to measure heterogeneity is random due to $\hSigman$. A different way to measure heterogeneity is the \textit{population heterogeneity in parameters}, which replaces $\hSigman$ with $\Sigman$:
\begin{align*}
    \Hn(\thetan) = \min_{\gamma \in \R}\norm{\thetan - \1\gamma}_{\Sigman^{-1}} = \sqrt{\thetan'\Qn\thetan}, \quad \Qn = \Sigman^{-1/2}\An \Sigman^{-1/2}, \quad \An = I - \frac{\Sigman^{-1/2}\1\1'\Sigman^{-1/2}}{\1'\Sigman^{-1}\1}.
\end{align*}
Note that $\hHn(\thetan) = \norm{\hAn \hSigman^{-1/2} \thetan}$ and $\Hn(\thetan) = \norm{\An \Sigman^{-1/2} \thetan }$. The two heterogeneity measures are close to each other in the following sense. 

\begin{propositionU}\label{arxiv2:prop:consistent.heterogeneity}
Under Assumption \ref{arxiv2:ass:consistent.covariance}, for each $\e > 0$,
\begin{align*}
    \limn \supPn \P[\Pn]{\norm{\hAn \hSigman^{-1/2} \thetan - \An \Sigman^{-1/2} \thetan} > \e\norm{\An \Sigman^{-1/2} \thetan}} = 0.
\end{align*}
Consequently, for each $\e > 0$,
\begin{align*}
    \limn \supPn \P[\Pn]{\abs{\hHn(\thetan) - \Hn(\thetan)} > \e\Hn(\thetan)} = 0.
\end{align*}
\end{propositionU}

\begin{proof}
See Appendix \ref{arxiv2:app:proof:consistent.heterogeneity}.    
\end{proof}

The \textit{heterogeneity in estimates} is $\hHn(\hthetan)$. The \textit{heterogeneity UCB} at significance level $\beta$ is
\begin{align*}
    \hetan = \frac{\tetan}{\sqrt{n}}, \quad \tetan =  
    \begin{cases}
    \hfil 0, & F_{\chi^{2}}(n\hHn^{2}(\hthetan); 0) \leq \beta, \\
    \hfil F_{\chi^{2}}^{-1}(n\hHn^{2}(\hthetan); \beta), & F_{\chi^{2}}(n\hHn^{2}(\hthetan); 0) > \beta.
    \end{cases}
\end{align*}
The following result establishes uniform asymptotic validity of $\hetan$ for inference on the sample heterogeneity in parameters $\hHn(\thetan)$, which will turn out to be sufficient for establishing uniform asymptotic validity of the robust CI---see Remark \ref{arxiv2:remark:random.heterogeneity.measure} below for why I consider inference on the sample measure $\hHn(\thetan)$ instead of the population measure $\Hn(\thetan)$.

\begin{propositionU}\label{arxiv2:prop:asymptotic.UCB}
Under Assumptions \ref{arxiv2:ass:linear.CLT} and \ref{arxiv2:ass:consistent.covariance},
\begin{align}\label{arxiv2:eq:UCB.criteria}
    \limn \supPn \abs{\P[\Pn]{\hHn(\thetan) \leq \hetan} - (1-\beta)}\I{\Hn(\thetan) > 0} = 0.
\end{align}
\end{propositionU}

\begin{proof}
See Appendix \ref{arxiv2:app:proof:asymptotic.UCB}.    
\end{proof}

Note that when $\Hn(\thetan) = 0$, the UCB $\hetan \geq 0$ always covers $\hHn(\thetan)$, so multiplication by the indicator $\I{\Hn(\thetan) > 0}$ restricts attention to sequences where heterogeneity is nontrivial, in analogy to quantile-unbiasedness criterion \eqref{arxiv2:eq:exact.upperCI.heterogeneity} from the normal model. Thus, \eqref{arxiv2:eq:UCB.criteria} implies
\begin{align*}
    \liminf_{n \to \infty} \infPn \P[\Pn]{\hHn(\thetan) \leq \hetan} \geq 1-\beta.
\end{align*}
In this sense, $\hetan$ is a uniformly asymptotically valid UCB, in analogy to criterion \eqref{arxiv2:eq:inference.heterogeneity}.

\begin{remark}[Sample versus Population Heterogeneity]\label{arxiv2:remark:random.heterogeneity.measure}
Consider a sequence where the normalized population heterogeneity diverges: $\sqrt{n}\Hn(\thetan) \to \infty$. Given Proposition \ref{arxiv2:prop:consistent.heterogeneity}, this implies $\sqrt{n}\hHn(\thetan) \to[p] \infty$. To see what this means for validity of $\hetan$, consider noncentral chi-squared asymptotics for the heterogeneity in estimates $\hHn(\hthetan)$ relative to the sample heterogeneity in parameters $\hHn(\thetan)$:
\begin{align}\label{arxiv2:eq:sample.noncentral.asymptotics}
    \frac{\sqrt{n}(\hHn^{2}(\hthetan) -  \hHn^{2}(\thetan))}{2\hHn(\thetan)} = \frac{\hZn' \hQn \hZn}{2\sqrt{n}\hHn(\thetan)} + \frac{(\hAn \hSigman^{-1/2}\thetan)'}{\norm{\hAn \hSigman^{-1/2}\thetan}} \hSigman^{-1/2}\hZn, \quad \hZn = \sqrt{n}(\hthetan - \thetan).
\end{align}
In this asymptotic regime, the first term is $o_{p}(1)$ while the second term converges in distribution to a standard normal. The asymptotic validity of $\hetan$ follows from expression \eqref{arxiv2:eq:sample.noncentral.asymptotics} combined with bounds from \citet{seri2015tight} that imply the uniform convergence of a noncentral chi-squared CDF with diverging noncentrality parameter to a normal CDF; see Appendix \ref{arxiv2:app:proof:asymptotic.UCB} for details. By comparison, the population heterogeneity analogue of \eqref{arxiv2:eq:sample.noncentral.asymptotics} yields
\begin{align*}
    \frac{\sqrt{n}(\hHn^{2}(\hthetan) -  \Hn^{2}(\thetan))}{2\Hn(\thetan)} = \frac{\hZn' \hQn \hZn}{2\sqrt{n}\Hn(\thetan)} + \frac{(\hAn \hSigman^{-1/2}\thetan)'}{\norm{\An \Sigman^{-1/2}\thetan}} \hSigman^{-1/2}\hZn + \frac{\sqrt{n}(\hHn^{2}(\thetan) - \Hn^{2}(\thetan))}{2\Hn(\thetan)},
\end{align*}
where the first two terms behave like before, but now there is a third term where
\begin{align*}
    \abs{\frac{\sqrt{n}(\hHn^{2}(\thetan) - \Hn^{2}(\thetan))}{2\Hn(\thetan)}} \lesssim \sqrt{n}\abs{\frac{\hHn(\thetan)}{\Hn(\thetan)} - 1},
\end{align*}
which need not be $o_{p}(1)$. Proposition \ref{arxiv2:prop:asymptotic.UCB} avoids this term altogether.
\end{remark}

\begin{remark}[Analogy to Heterogeneity in the Normal Model]
As noted in Remark \ref{arxiv2:remark:analogy.environment.normal}, objects $(\htheta, \Sigma, \theta)$ in the normal model are analogous to objects $(\hthetan, \tSigman, \thetan)$ in the asymptotic environment, where $\tSigman = \hSigman/n$. The corresponding GLS weighting matrices are $\Sigma^{-1}$ and $\tSigman^{-1} = n\hSigman^{-1}$. Thus, using $\heta_{1-\beta}$ for inference on $H(\theta)$ in the normal model is analogous to using $\tetan = \sqrt{n}\hetan$ for inference on $\tHn(\thetan) = \sqrt{n}\hHn(\thetan)$ in the asymptotic environment.
\end{remark}

\subsection{Validity of the Robust CI}\label{arxiv2:sec:asymptotics.robust.CI}
I now establish the asymptotic validity of robust CIs centered at weight vectors $\hwn$ that satisfy Assumption \ref{arxiv2:ass:consistent.weights}, which include the minimax-bias weights $\rhwn$ in view of Proposition \ref{arxiv2:prop:consistent.minimax.weights}. 

For centering weights $\hwn$, the bias UCB is 
\begin{align*}
    \hBn = \hetan \max_{\l \in \hLn}\norm{\l - \hwn}_{\hSigman} = \tetan \max_{\l \in \hLn}\norm{\l - \hwn}_{\tSigman}, \quad \tSigman = \hSigman/n.
\end{align*}
For standard error $\tsigmawn = \sqrt{\hwn'\tSigman \hwn} = \sqrt{\hwn'\hSigman \hwn}/\sqrt{n} = \hat{\sigma}_{\hwn,n}/\sqrt{n}$, the robust CI is
\begin{align*}
    \rCIn = 
    \begin{cases}
    \hfil \displaystyle \brack{\htauwn \pm \cv{\hBn/\tsigmawn}\tsigmawn}, & \text{two-sided}, \\[5pt]  
    \hfil \left(-\infty, \htauwn + z_{1-\alpha}\tsigmawn + \hBn\right], & \text{one-sided (upper)}, \\[5pt]  
    \hfil \left[\htauwn - z_{1-\alpha}\tsigmawn - \hBn, \infty\right), & \text{one-sided (lower)}.
    \end{cases}
\end{align*}
The validity of $\rCIn$ depends on its coverage for weights in the population class $\Ln \subseteq \bbW$, where $\bbW \subseteq \W$ is the compact convex set with positive diameter from Assumption \ref{arxiv2:ass:Lambda.representation}. The uniform consistency in Proposition \ref{arxiv2:prop:uniform.hausdorff.convergence} shows that the estimated class $\hLn$ is close to the population class $\Ln$ with high probability as $n \to \infty$. However, this is not the same as $\hLn$ containing the alternatives from $\Ln$ with high probability. To see this, recall that Assumption \ref{arxiv2:ass:Lambda.representation} yields the inequality constraint representations 
\begin{align*}
    \hLn = \L(\hSn) = \curly{\l \in \bbW: g(\l, \hSn) \leq 0}, \quad \Ln = \L(\Sn) = \curly{\l \in \bbW: g(\l, \Sn) \leq 0},
\end{align*}
where the former holds with probability uniformly approaching one under Assumption \ref{arxiv2:ass:consistent.S}. The issue is at the boundary of the inequality constraint, where small estimation errors can lead to exclusion, $g(\l, \hSn) > 0$, of population alternatives where $g(\l, \Sn) = 0$. Since membership is discontinuous at the boundary, uniform consistency of $\hSn$ to $\Sn$ generally does not ensure that $\hLn$ contains the alternatives in $\Ln$. However, as I show in examples further below, there will typically exist a subclass $\L_{0}(S) \subseteq \L(S)$ such that, for $\L_{0,n} = \L_{0}(\Sn)$,
\begin{align}\label{arxiv2:eq:target.containment}
    \limn \infPn \inf_{\l \in \L_{0,n}} \P[\Pn]{\l \in \hLn} = 1.
\end{align}
For subclasses $\L_{0,n}$ satisfying containment condition \eqref{arxiv2:eq:target.containment}, the robust CI constructed with $\hLn$ provides uniformly asymptotically valid coverage for the target alternatives $\l \in \L_{0,n}$.

\begin{propositionU}\label{arxiv2:prop:asymptotic.robust.CI}
Let Assumptions \ref{arxiv2:ass:linear.CLT}-\ref{arxiv2:ass:Lambda.representation} be satisfied and suppose that $\L_{0,n} \subseteq \Ln$ satisfies containment condition \eqref{arxiv2:eq:target.containment}. Then
\begin{align}\label{arxiv2:eq:asymptotic.robust.coverage}
    \liminfn \infPn \inf_{\l \in \L_{0,n}}\P[\Pn]{\tauwn[\l] \in \rCIn} \geq 1-(\alpha+\beta).
\end{align}
\end{propositionU} 

\begin{proof}
See Appendix \ref{arxiv2:app:proof:asymptotic.robust.CI}.        
\end{proof}

The asymptotic coverage in \eqref{arxiv2:eq:asymptotic.robust.coverage} is analogous to the coverage in \eqref{arxiv2:eq:robust.coverage} from the normal model, but stated for the subclass $\L_{0,n}$ rather than the population class $\Ln$. This suggests viewing the class $\L_{0,n}$ in Proposition \ref{arxiv2:prop:asymptotic.robust.CI} as a target class, $\Ln \supseteq \L_{0,n}$ as an enlargement of the target class, and $\hLn$ as an estimator of the enlargement $\Ln$. In other words, to ensure asymptotic coverage for target alternatives $\L_{0,n}$, it suffices to use an enlarged estimator $\hLn \supseteq \hL_{0,n}$ to construct the robust CI rather than using the sample analogue $\hL_{0,n}$ of the target. For example, if the class $\hL_{0,n} = \L(\hSn; \kappa_{0})$ increases in a scalar parameter $\kappa$, this suggests using $\hLn = \L(\hSn;\kappa)$ with a larger $\kappa > \kappa_{0}$. Importantly, unlike the enlarged $\Ln$, I do not require the target $\L_{0,n}$ to satisfy the conditions of Assumption \ref{arxiv2:ass:Lambda.representation}, so long as it satisfies containment condition \eqref{arxiv2:eq:target.containment} relative to $\Ln$. For example, a target bounded variance class with $r_{0}=1$ may not satisfy the Slater condition, but nevertheless satisfy containment condition \eqref{arxiv2:eq:target.containment} relative to an enlargement with $r > 1$ that does satisfy the Slater condition; see the examples below.

\begin{remark}[Caveats Regarding Class Estimation Error]
If there is no class estimation error (i.e., $\hLn = \Ln$), then containment condition \eqref{arxiv2:eq:target.containment} holds for $\L_{0,n} = \Ln$ and Proposition \ref{arxiv2:prop:asymptotic.robust.CI} applies with $\Ln = \L_{0,n}$. However, Proposition \ref{arxiv2:prop:asymptotic.robust.CI} still assumes that $\Ln$ satisfies Assumption \ref{arxiv2:ass:Lambda.representation}. Thus, even when there is no class estimation error, it can still be useful to apply Proposition \ref{arxiv2:prop:asymptotic.robust.CI} to $\Ln \supseteq \L_{0,n}$ when the target class $\L_{0,n}$ may violate the conditions of Assumption \ref{arxiv2:ass:Lambda.representation}, as with the aforementioned bounded variance class with $r=1$ and the Slater condition.
\end{remark}

To establish containment condition \eqref{arxiv2:eq:target.containment}, a practical approach is to show that the inequality constraint from $\Ln$ is uniformly slack when evaluated over $\L_{0,n}$. In particular, it suffices to show the existence of a constant $\nu > 0$ such that for all $\Pn \in \cPn$ and $n$,
\begin{align}\label{arxiv2:eq:slack.condition}
    \sup_{\l \in \L_{0,n}} g(\l, \Sn) \leq -\nu.
\end{align}
By Lemma \ref{arxiv2:app:lemma:sufficient.slack}, slack condition \eqref{arxiv2:eq:slack.condition} implies containment condition \eqref{arxiv2:eq:target.containment} under Assumptions \ref{arxiv2:ass:consistent.S} and \ref{arxiv2:ass:Lambda.representation}. I now show how to derive $\L_{0,n}$ and $\nu$ for the benchmark classes of interest.

\begin{example*}[Bounded Variance, continued] 
For $\Sn = (\Sigman, \wn)$ and $r_{0} \geq 1$, the target class is
\begin{align*}
    \L_{0,n} = \curly{\l \in \W: g(\l, \Sigman, \wn; r_{0}) \leq 0} = \curly{\l \in \W: \sqrt{\l'\Sigman \l} \leq r_{0}\sqrt{\wn'\Sigman \wn}}.
\end{align*}
For $r > r_{0}$, the enlarged class is
\begin{align*}
    \Ln = \curly{\l \in \W: g(\l, \Sigman, \wn; r) \leq 0} = \curly{\l \in \W: \sqrt{\l'\Sigman \l} \leq r\sqrt{\wn'\Sigman \wn}}.
\end{align*}
Observe that
\begin{align*}
    \sup_{\l \in \L_{0,n}}g(\l, \Sigman, \wn; r) = \sup_{\l \in \L_{0,n}}\sqrt{\l'\Sigman \l} - r\sqrt{\wn'\Sigman \wn} \leq -(r - r_{0})\sqrt{\wn'\Sigman \wn} \leq 0.
\end{align*}
Since $\emin(\Sigman) \geq \LB{e}$ and $\norm{\wn} \geq 1/\sqrt{K}$, slack condition \eqref{arxiv2:eq:slack.condition} is satisfied:
\begin{align*}
    \sup_{\l \in \L_{0,n}}g(\l, \Sigman, \wn; r) \leq -\nu, \quad \nu = (r - r_{0})\sqrt{\LB{e}/K} > 0.
\end{align*}
By Proposition \ref{arxiv2:prop:asymptotic.robust.CI}, a robust CI constructed with estimated class
\begin{align*}
    \hLn = \curly{\l \in \W: g(\l, \hSigman, \hwn; r) \leq 0} = \curly{\l \in \W: \sqrt{\l'\hSigman \l} \leq r\sqrt{\hwn'\hSigman \hwn}}
\end{align*}
provides uniformly asymptotically valid coverage for $\l \in \L_{0,n}$. These derivations also apply to the bounded variance simplex by replacing $\W$ with $\W_{+}$ and restricting to $\hwn, \wn \in \W_{+}$.
\end{example*}

\begin{example*}[Truncated Simplex, continued] 
For $\Sn = \wn \in \W_{+}$ and $\epsilon_{0} \in [0,1)$, the target class is
\begin{align*}
    \L_{0,n} = \curly{\l \in \W_{+}: g(\l, \wn; \epsilon_{0}) \leq 0} = \curly{\l \in \W_{+}: \l \geq (1-\epsilon_{0})\wn}.
\end{align*}
For $\epsilon \in (\epsilon_{0}, 1]$, the enlarged class is
\begin{align*}
    \Ln = \curly{\l \in \W_{+}: g(\l, \wn; \epsilon) \leq 0} = \curly{\l \in \W_{+}: \l \geq (1-\epsilon)\wn}.
\end{align*}
Observe that
\begin{align*}
    \sup_{\l \in \L_{0,n}}g(\l, \wn; \epsilon) = \max_{k \in \{1,\ldots,K\}}\paren{(1-\epsilon)w_{n,k}-\inf_{\l \in \L_{0,n}}\l_{k}} \leq -(\epsilon - \epsilon_{0})\min_{k \in \{1,\ldots,K\}}w_{n,k}.
\end{align*}
Suppose that $w_{n,k}$ is uniformly bounded away from zero:
\begin{align*}
    \inf_{n}\infPn w_{n,k} \geq w_{\min} > 0, \quad \forall k \in \{1,\ldots,K\}.
\end{align*}
This holds, for instance, for the equal weights vector $\wn = \1/K$. Given this uniform lower bound on $w_{n,k}$, slack condition \eqref{arxiv2:eq:slack.condition} is satisfied:
\begin{align*}
    \sup_{\l \in \L_{0,n}}g(\l, \wn; \epsilon) \leq -\nu, \quad \nu = (\epsilon - \epsilon_{0})w_{\min} > 0.
\end{align*}
By Proposition \ref{arxiv2:prop:asymptotic.robust.CI}, a robust CI constructed with estimated class
\begin{align*}
    \hLn = \curly{\l \in \W_{+}: g(\l, \hwn; \epsilon) \leq 0} = \curly{\l \in \W_{+}: \l \geq (1-\epsilon)\hwn}
\end{align*}
provides uniformly asymptotically valid coverage for $\l \in \L_{0,n}$.
\end{example*}

\begin{example*}[Covariate Balance, continued] 
For $\Sn = (\wn, \bXn) \in \W_{+} \times \R^{K \times M}$ and $\Bar{c}_{0} \geq 0$, the target class is
\begin{align*}
    \L_{0,n} = \curly{\l \in \W_{+}: g(\l, \wn, \bXn; \Bar{c}_{0}) \leq 0} = \curly{\l \in \W_{+}: c_{\l}(\wn,\bXn) \leq \Bar{c}_{0}}.
\end{align*}
For $\Bar{c} > \Bar{c}_{0}$, the enlarged class is
\begin{align*}
    \Ln = \curly{\l \in \W_{+}: g(\l, \wn, \bXn; \Bar{c}) \leq 0} = \curly{\l \in \W_{+}: c_{\l}(\wn,\bXn) \leq \Bar{c}}.
\end{align*}
Slack condition \eqref{arxiv2:eq:slack.condition} is satisfied:
\begin{align*}
    \sup_{\l \in \L_{0,n}}g(\l, \wn, \bXn; \Bar{c}) \leq -\nu, \quad \nu = \Bar{c} - \Bar{c}_{0} > 0.
\end{align*}
By Proposition \ref{arxiv2:prop:asymptotic.robust.CI}, a robust CI constructed with estimated class
\begin{align*}
    \hLn = \curly{\l \in \W_{+}: g(\l, \hwn, \hbXn; \Bar{c}) \leq 0} = \curly{\l \in \W_{+}: c_{\l}(\hwn,\hbXn) \leq \Bar{c}}
\end{align*}
provides uniformly asymptotically valid coverage for $\l \in \L_{0,n}$.
\end{example*}

\begin{remark}[Reporting Convention for Boundary Alternatives]\label{arxiv2:remark:enlarged.convention}
The takeaway from the above examples and discussions is that asymptotic coverage for boundary alternatives requires some care. For instance, a target bounded variance simplex class
\begin{align*}
    \L_{0,n} = \curly{\l \in \W_{+}: \sqrt{\l'\Sigman \l} \leq r_{0}\sqrt{\wn'\Sigman \wn}}, \quad r_{0} \geq 1,
\end{align*}
may contain alternative weights that lie on the boundary of the variance constraint: $\dsqrt{\l'\Sigman \l} = r_{0}\sqrt{\wn'\Sigman \wn}$. Proposition \ref{arxiv2:prop:asymptotic.robust.CI} says that a sufficient avenue for asymptotic coverage of the entire $\L_{0,n}$ is to use a robust CI constructed with an enlarged estimated class
\begin{align*}
    \hLn = \curly{\l \in \W_{+} : \sqrt{\l'\hSigman \l} \leq r\sqrt{\hwn'\hSigman \hwn}}, \quad r = r_{0}+\delta, \quad \delta > 0.
\end{align*}
Proposition \ref{arxiv2:prop:asymptotic.robust.CI} guarantees asymptotic coverage under any fixed $\delta > 0$. Nevertheless, for ease of exposition in the implementation and empirical applications, I simply report robust CIs under $\delta = 0$. Using $\delta = 0.0001$ leaves the reported CIs unchanged at the displayed precision.
\end{remark}

\section{Practical Implementation}\label{arxiv2:sec:implementation}
Motivated by the asymptotic results in Section \ref{arxiv2:sec:asymptotic.validity}, but suppressing the sample size index $n$ for conciseness, I consider asymptotically normal estimates $\htheta \overset{a}{\sim} N(\theta, \Sigma/n)$, consistent covariance matrix estimator $\tSigma = \hSigma/n \overset{a}{\approx} \Sigma/n$, consistent baseline weights estimator $\hw \overset{a}{\approx} w$, and classes of alternative weights $\hL = \L(\hS)$ characterized by consistently estimated statistics $\hS \overset{a}{\approx} S$: e.g., $\hS = (\hw, \hSigma)$ for bounded variance class, $\hS = \hw$ for truncated simplex class, and $\hS = (\hw, \hbX)$ for covariate balance class with estimated covariate mean matrix $\hbX$. 

Section \ref{arxiv2:sec:feasible.procedures} summarizes my robust inference procedures in terms of the above objects and discusses specification choices for $\hL$. Sections \ref{arxiv2:sec:implementation.event.studies} and \ref{arxiv2:sec:implementation.multisite.experiments} discuss practical implementations of these procedures in the context of event studies and multisite experiments, respectively. Section \ref{arxiv2:sec:implementation.event.studies} also gives a survey of current practice for the event study context.

\subsection{Feasible Inference Procedures}\label{arxiv2:sec:feasible.procedures}
The GLS regression of $\htheta$ on a constant $\1$ weighted by $\tSigma^{-1}$ yields 
\begin{align}\label{arxiv2:eq:feasible.GLS.output}
    \htau_{\GLS} = \frac{\1'\tSigma^{-1}\htheta}{\1'\tSigma^{-1}\1}, \quad \tsigma_{\GLS} = \frac{1}{\dsqrt{\1'\tSigma^{-1}\1}}, \quad \tH^{2}(\htheta) = (\htheta - \1\htau_{\GLS})'\tSigma^{-1}(\htheta - \1\htau_{\GLS}),
\end{align}
where $\htau_{\GLS}$ is the GLS estimator, $\tsigma_{\GLS}$ is the GLS standard error, $\tH^{2}(\htheta)$ is the GLS residual sum of squares, which yields heterogeneity in estimates $\tH(\htheta)$ with explicit formula
\begin{align}\label{arxiv2:eq:feasible.heterogeneity.estimates}
    \tH(\htheta) = \sqrt{\htheta'\tQ\htheta}, \quad \tQ = \tSigma^{-1/2}\tA\tSigma^{-1/2}, \quad \tA = I - \frac{\tSigma^{-1/2}\1\1'\tSigma^{-1/2}}{\1'\tSigma^{-1}\1}.
\end{align}
Following the steps of Section \ref{arxiv2:sec:heterogeneity.UCB}, the heterogeneity UCB is
\begin{align}\label{arxiv2:eq:feasible.heterogeneity.UCB}
    \teta_{1-\beta} = 
    \begin{cases}
    \hfil 0, & F_{\chi^{2}}(\tH^{2}(\htheta); 0) \leq \beta, \\
    \hfil F_{\chi^{2}}^{-1}(\tH^{2}(\htheta); \beta), & F_{\chi^{2}}(\tH^{2}(\htheta); 0) > \beta.
    \end{cases}
\end{align}
Following the steps of Section \ref{arxiv2:sec:robust.inference}, the minimax-bias weights and bias UCB are 
\begin{align}\label{arxiv2:eq:feasible.minimax.bias.objects}
    \hw^{*} = \arg\min_{\Bar{w} \in \hL}\max_{\l \in \hL}\norm{\l - \Bar{w}}_{\tSigma}, \quad \quad \hB_{\min}^{\beta}(\hL) = \teta_{1-\beta}\max_{\l \in \hL}\norm{\l - \hw^{*}}_{\tSigma}.
\end{align}
For standard error $\tsigma^{*} = \sqrt{(\hw^{*})'\tSigma \hw^{*}}$, the robust estimator and corresponding robust CI are 
\begin{align}\label{arxiv2:eq:feasible.robust.estimator.CI}
    \htau^{*} = (\hw^{*})'\htheta, \quad \quad 
    CI^{*} = 
    \begin{cases}
    \hfil \displaystyle \brack{\htau^{*} \pm \cv{\hB_{\min}^{\beta}(\hL)/\tsigma^{*}}\tsigma^{*}}, & \text{two-sided}, \\[5pt]  
    \hfil \left(-\infty, \htau^{*} + z_{1-\alpha}\tsigma^{*} + \hB_{\min}^{\beta}(\hL)\right], & \text{one-sided (upper)}, \\[5pt]  
    \hfil \left[\htau^{*} - z_{1-\alpha}\tsigma^{*} - \hB_{\min}^{\beta}(\hL), \infty\right), & \text{one-sided (lower)},
    \end{cases}
\end{align}
where $\cv{b}$ can be computed as the square root of the $(1-\alpha)$-quantile of the noncentral chi-squared distribution with one degree of freedom and noncentrality parameter $b^{2}$.

\begin{recipe}\label{arxiv2:recipe:robust.procedures} Implementation of Inference Procedures
\begin{recipenumerate}
  \item Construct $\teta_{1-\beta}$ as in \eqref{arxiv2:eq:feasible.heterogeneity.UCB}, based on \eqref{arxiv2:eq:feasible.heterogeneity.estimates} or \eqref{arxiv2:eq:feasible.GLS.output}.
  \item Construct $\hw^{*}$ and $\hB_{\min}^{\beta}(\hL)$ as in \eqref{arxiv2:eq:feasible.minimax.bias.objects}.
  \item Construct $\htau^{*}$ and $CI^{*}$ as in \eqref{arxiv2:eq:feasible.robust.estimator.CI}.
  \item Report $(\htau^{*}, CI^{*})$.
\end{recipenumerate}
\end{recipe}

Recipe \ref{arxiv2:recipe:robust.procedures} summarizes the implementation of my inference procedures for inputs given by (i) the estimates and covariance matrix estimator $(\htheta, \tSigma)$, for which I describe typical constructions in Section \ref{arxiv2:sec:implementation.event.studies} for event studies and Section \ref{arxiv2:sec:implementation.multisite.experiments} for multisite experiments; (ii) the significance levels $\alpha$ and $\beta$, which by default can be set to $\alpha = \beta = 0.05$; and (iii) the class of alternatives $\hL$. 

In some cases there may be ambiguity over which class of alternatives to use. In particular, the class $\hL=\hL(\kappa)$ may be indexed by a scalar parameter $\kappa\in\R$, where $\hL(\kappa)$ is increasing in $\kappa$ so that larger $\kappa$ permits weakly more disagreement over the choice of weights. One can execute Recipe \ref{arxiv2:recipe:robust.procedures} for a range of plausible $\kappa$, or report the \textit{breakdown value} $\kappa^{*}=\inf\{\kappa:\tau_{0}\in CI^{*}(\kappa)\}$ for a threshold value $\tau_{0}$, such as $\tau_{0}=0$. A large $\kappa^{*}$ means that substantial disagreement over weights must be allowed before the robust CI first includes $\tau_{0}$.\footnote{When inclusion of $\tau_{0}$ in $CI^{*}(\kappa)$ is monotone in $\kappa$, $\kappa^{*}$ is also the cutoff above which $\tau_{0}$ remains included.}

\begin{example*}[Bounded Variance, continued]
From the derivations in \eqref{arxiv2:eq:robust.estimator.GLS}, it follows that Recipe \ref{arxiv2:recipe:robust.procedures} is determined by the GLS output \eqref{arxiv2:eq:feasible.GLS.output}. In particular, for a lower one-sided $CI^{*}$ and baseline weights $\hw$, Recipe \ref{arxiv2:recipe:robust.procedures} under $\hL_{\sigma}(r) = \{\l \in \W: \tsigma_{\l} \leq r\tsigma_{\hw}\}$ for standard error ratio bound $r$ and standard error $\tsigma_{\hw} = \dsqrt{\hw'\tSigma\hw}$ yields
\begin{align*}
    \htau^{*} = \htau_{\GLS}, \quad \quad CI^{*}(r) = \left[\htau_{\GLS} - z_{1-\alpha}\tsigma_{\GLS} - \teta_{1-\beta}\sqrt{r^{2}\tsigma_{\hw}^{2} - \tsigma_{\GLS}^{2}}, \infty\right).
\end{align*}
Assuming $\teta_{1-\beta} > 0$, the breakdown value for a given threshold $\tau_{0}$ is
\begin{align*}
    r^{*} = \frac{\tsigma_{\GLS}}{\tsigma_{\hw}}\sqrt{\paren{\frac{\max\curly{T_{\GLS}(\tau_{0}) - z_{1-\alpha}, 0}}{\teta_{1-\beta}}}^{2} + 1}, \quad T_{\GLS}(\tau_{0}) = \frac{\htau_{\GLS} - \tau_{0}}{\tsigma_{\GLS}}.
\end{align*}
The breakdown value $r^{*}$ is large when the GLS $t$-statistic $T_{\GLS}(\tau_{0})$ is large relative to inferred heterogeneity $\teta_{1-\beta}$. For a given $r$, the above $CI^{*}(r)$ includes $\tau_{0}$ when $r \geq r^{*}$.
\end{example*}

\begin{example*}[Truncated Simplex, continued]
Given baseline weights $\hw$, one can execute Recipe \ref{arxiv2:recipe:robust.procedures} with $\hL_{+}(\epsilon) = \{\l \in \W_{+}: \l \geq (1-\epsilon)\hw\} = \{(1-\epsilon)\hw + \epsilon \l: \l \in \W_{+}\}$ for some range of truncation parameters $\epsilon$, or compute breakdown values $\epsilon^{*}$. To interpret magnitudes of $\epsilon$, consider the equal weights baseline $\hw = w_{\EW} = \1/K$. In this case, $\epsilon$ is the maximum possible discrepancy in weights $|\l_{k}-\l_{k'}|$ between two groups $k$ and $k'$. Moreover, note that the maximum amount of weight that can be reallocated across groups when moving from $w_{\EW}$ to $\l \in \L_{+}(\epsilon)$ is
\begin{align*}
    \paren{(1-\epsilon)\frac{1}{K} + \epsilon} - \frac{1}{K} = \frac{K-1}{K}\epsilon.
\end{align*}
On the other hand, the amount of weight that is reallocated when moving \textit{all} the weight from $J \in \{1,\ldots,K-1\}$ of the groups to the other $K-J$ groups (i.e., a leave-$J$-out procedure) is equal to $J/K$. Equating these two reallocation quantities yields the $\epsilon$ calibration
\begin{align}\label{arxiv2:eq:leave.out.calibration}
    \frac{K-1}{K}\epsilon = \frac{J}{K} \iff \epsilon = \frac{J}{K-1}.
\end{align}
Thus, in the EW baseline, choosing $\epsilon$ according to \eqref{arxiv2:eq:leave.out.calibration} allows for the same maximal reallocation of weight as simply dropping $J$ of the groups. In this sense, $\epsilon$ is a rescaled leave-$J$-out magnitude: $(K-1)\epsilon$ is the corresponding number of groups dropped, and $(K-1)\epsilon/K$ is the corresponding fraction of groups dropped. This calibration allows one to maintain the mathematical structure of $\L_{+}(\epsilon)$ while casting the logic of $\epsilon$-departures in terms of the potentially more familiar logic of leave-$J$-out procedures.
\end{example*}

\begin{example*}[Covariate Balance, continued]
Given covariates $\hbX$, one can execute Recipe \ref{arxiv2:recipe:robust.procedures} with $\hL_{X}(\Bar{c}) = \curly{\l \in \W_{+}: c_{\l}(\hbX) \leq \Bar{c}}$ for some range of balance gap bounds $\Bar{c}$, or compute breakdown values $\Bar{c}^{*}$. To interpret magnitudes of $\Bar{c}$, consider the group-level balance gaps
\begin{align}\label{arxiv2:eq:group.balance.gaps}
    c_{k}(\hbX) = \max_{m \in \{1,\ldots,M\}}\frac{\abs{\hbX_{m,k} - \hw'\hbX_{m}}}{\text{sd}(\hbX_{m})}.
\end{align}
Intuitively, $c_{k}(\hbX)$ is the balance gap that arises from a population with covariate means equal to those of group $k$. If one assumes that readers are interested in balance gaps within the range of $\{c_{k}(\hbX)\}_{k=1}^{K}$, then the deciles $(\Bar{c}_{d})_{d=1}^{9}$ of the group-level balance gaps provide one simple way to determine relevant magnitudes of $\Bar{c}$. For example, $\Bar{c} = \Bar{c}_{5}$ allows one to assess robustness to populations with covariate profiles that yield balance gaps no higher than the median group's gap. An analogous calibration can be formulated for truncated covariate balance class \eqref{arxiv2:eq:class.truncated.covariate.balance} using the scaled gap deciles $(\epsilon\Bar{c}_{d})_{d=1}^{9}$:
\begin{align}\label{arxiv2:eq:class.truncated.balance.scaled.deciles}
    \hL_{X}^{\epsilon}(\epsilon\Bar{c}_{d}) = \hL_{X}(\epsilon\Bar{c}_{d}) \cap \hL_{+}(\epsilon) = \curly{(1-\epsilon)\hw + \epsilon \l: \l \in \hL_{X}(\Bar{c}_{d})}, \quad \epsilon > 0.
\end{align}
For each $d$, this class considers $\epsilon$-departures from the baseline $\hw$ towards populations that have balance gaps no higher than decile $\Bar{c}_{d}$.
\end{example*}

\subsection{Current Practice and Implementation in Event Studies}\label{arxiv2:sec:implementation.event.studies}

\subsubsection{Current Practice in Event Studies}
To illustrate the empirical practice that motivates my proposed procedures, I survey a sample of recent event study papers. The sample focuses on 2024-2025 articles in the \textit{American Economic Review} (AER), \textit{American Economic Journal: Applied Economics} (AEJ: Applied), and \textit{American Economic Journal: Economic Policy} (AEJ: Policy) that cite at least one of the following canonical contributions to the event study literature on weighting and heterogeneity issues: \citet{de2020two}, \citet{callaway2021difference}, \citet{goodman2021difference}, \citet{sun2021estimating}, \citet{borusyak2024revisiting}. In particular, I use the economics literature search tool at \url{https://paulgp.com/econlit-pipeline/search.html} to locate papers whose textual references contain the titles of the five aforementioned papers. This title-citation screen is \textit{not} a representative sample of all event studies. It instead identifies papers whose authors are likely familiar with---and wish to account for---event-study weighting and heterogeneity issues. My survey examines how researchers engage with these issues.

The title-citation screen yields 66 papers. I manually classify 64 as containing an event-study or dynamic DiD design.\footnote{I coded initial classifications based on manual reviews of the published main-text and online appendix PDFs of the 66 papers. I then used Codex as a second-pass consistency check against the coding spreadsheet and the published papers' main-text PDFs. I made all final coding decisions.} Of those, I classify 60 as using a generalized staggered-timing design.\footnote{Here ``generalized'' refers to designs with staggered or otherwise varying treatment timing or exposure across groups, including designs in which the relevant variation is not temporal. For example, \citet[Online Appendix, page 23]{auer2025communication} consider a setting with spatial variation: ``While the treatments discussed in the literature so far have not covered our case of spacial variation (instead of variation over time) and switching of treatment, it is obvious from Figure B.1 that we are aggregating across different 2x2 DiDs. Thus similar problems of potentially negative implicit weights might arise.''} These 60 papers are the main survey denominator.\footnote{It is useful to benchmark this number to the number of articles published in AER, AEJ: Applied, and AEJ: Policy during 2024-2025. There are 110 articles for AER in 2024 \citep{10.1257/pandp.115.743}, 115 for AER in 2025 \citep{10.1257/pandp.116.742}, 64 for AEJ: Applied in 2024 \citep{10.1257/pandp.115.779}, 44 for AEJ: Applied in 2025 \citep{10.1257/pandp.116.780}, 64 for AEJ: Policy in 2024 \citep{10.1257/pandp.115.789}, and 61 for AEJ: Policy in 2025 \citep{10.1257/pandp.116.790}, for 458 articles in total. The 66 title-citation hits and 60 included papers therefore represent 14.4 and 13.1 percent of this article universe.} For each paper, I record whether it reports a conventional TWFE event-study specification; any of five core heterogeneity-robust (HR) methods (Sun-Abraham, Callaway-Sant'Anna, stacked DiD, Borusyak-Jaravel-Spiess imputation, de Chaisemartin-D'Haultfoeuille); any non-core alternative intended to address weighting or heterogeneity concerns; and any diagnostic method for assessing such concerns.\footnote{An example non-core alternative is the two-stage DiD method of \citet{gardner2022two}. An example diagnostic is the \citet{goodman2021difference} decomposition.} I count a method as reported when its estimates or results are given in either the main text or online appendix. Table \ref{arxiv2:tab:event-study-survey-summary} summarizes the survey. For additional details, including the list of 60 papers and their method codes, see Appendix \ref{arxiv2:app:event-study-survey}.

\begin{table}[!htbp]
\centering
\small
\caption{Reporting of TWFE and HR event-study methods}
\label{arxiv2:tab:event-study-survey-summary}
\vspace{0.5em}
\begin{threeparttable}
\begin{tabular}{p{0.48\textwidth}cc}
\toprule
Reporting outcome & All included papers & Among TWFE papers \\
 & $(N=60)$ & $(N=47)$ \\
\midrule
Reports TWFE event-study specification & 47 (78.3\%) & -- \\
Reports core HR method & 49 (81.7\%) & 37 (78.7\%) \\
Reports exactly one core HR method & 36 (60.0\%) & 27 (57.4\%) \\
Reports core HR method or other alternative & 57 (95.0\%) & 44 (93.6\%) \\
Reports weighting or heterogeneity diagnostic & 12 (20.0\%) & 12 (25.5\%) \\
\bottomrule
\end{tabular}
\vspace{0.35em}
\footnotesize
\textit{Notes}: The included-paper sample consists of the 60 manually classified papers described in the text. The core HR methods are Sun-Abraham, Callaway-Sant'Anna, stacked DiD, Borusyak-Jaravel-Spiess imputation, and de Chaisemartin-D'Haultfoeuille. ``Other alternative'' refers to a non-core alternative intended to address weighting or heterogeneity concerns.
\end{threeparttable}
\end{table}

TWFE is common: 47 of the 60 included papers report a TWFE event-study specification. Among the 47 TWFE papers, 37 report a core HR method and 44 report either a core HR method or another alternative. Across all 60 papers, 49 report at least one core HR method---with 13 reporting two or more---and 12 report a weighting or heterogeneity diagnostic. I view these patterns as evidence for the motivation given in Section \ref{arxiv2:sec:introduction}: researchers often assess robustness by comparing results across weighting choices. As I noted previously, such comparisons may have limited scope for establishing robustness to the choice of weights.

\subsubsection{Implementation in Event Studies}

\paragraph{Context.} Following the exposition in Section \ref{arxiv2:sec:setting}, an event study has units $i$ treated at different time periods $t$, allowing one to estimate the dynamic effects of treatment on outcomes $Y_{it}$. In this setting, two key sources of heterogeneity are (i) across treatment cohorts $G_{i}$, which denote the periods of initial treatment for units $i$; and (ii) across \textit{event times} $\ell \geq 0$, which denote the number of periods since initial treatment. Given the $\ATT_{g,t}$ object in \eqref{arxiv2:eq:ATT.object}, the average treatment effect $\ell$ periods after initial treatment for cohort $G_{i} = g$ is
\begin{align*}
    \ATT_{g,g+\ell} = E[Y_{i,g+\ell}(g) - Y_{i,g+\ell}(\infty)|G_{i} = g].
\end{align*}
Under event-study assumptions, \citet{sun2021estimating} show that event-time $\ell$ coefficients from conventional TWFE regression specifications recover estimands of the form
\begin{align*}
    \TWFE_{\ell}(\Tilde{w}_{\ell}) = \sum_{g, \ell'}\Tilde{w}_{\ell,(g,\ell')}\ATT_{g,g+\ell'},
\end{align*}
which allow treatment effects from different event times $\ell' \neq \ell$ and negative weights in $\Tilde{w}_{\ell}$. As surveyed in \citet{roth2023s}, the recent event studies literature proposes heterogeneity-robust (HR) methods that, for weights $w_{\ell} > 0$, recover event-time estimands of the form
\begin{align*}
    \ATT_{\ell}(w_{\ell}) = \sum_{g}w_{\ell,g}\ATT_{g,g+\ell}, 
\end{align*}
preventing negative weighting and excluding treatment effects from different event times. 

For each target event time $\ell$, a researcher may report conventional estimators and CIs for both a TWFE specification and one of the proposed HR approaches and examine the stability of results. There are two ways to map this setup to the notation of Recipe \ref{arxiv2:recipe:robust.procedures}.
\begin{itemize}
    \item Consider TWFE as the baseline and estimate $(\Tilde{w}_{\ell, (g,\ell')}, \ATT_{g,g+\ell'})$ across $(g,\ell')$.
    \item Consider an HR method as the baseline and estimate $(w_{\ell,g}, \ATT_{g,g+\ell})$ across $g$.
\end{itemize}
While it is possible to do the former (see \citet[equations (13) and (27)]{sun2021estimating}), the latter has the advantage of automatically shutting off the $\ell' \neq \ell$ channels and reducing the number of objects to be estimated. In fact, the ingredients required for implementing Recipe \ref{arxiv2:recipe:robust.procedures} are readily available in many applications of HR methods. For example, the \citet{sun2021estimating} procedure already requires estimation of $(w_{\ell,g}, \ATT_{g,g+\ell})$, and so Recipe \ref{arxiv2:recipe:robust.procedures} can be readily executed in applications that were already using the \citet{sun2021estimating} procedure, or any other procedure that first estimates ATT objects and their corresponding joint covariance matrix, such as the doubly robust procedure of \citet{callaway2021difference}, or stacked DiD procedures as in \citet[page 24]{wing2024stacked}.

\paragraph{Implementation.} I now use $k$ to index the groups $g$. For each target event time $\ell$, denote the number of available cohorts by $K_{\ell}$. Take an existing HR method for which one can estimate baseline weights $\hw_{\ell} \overset{a}{\approx} w_{\ell}$ and ATT objects
\begin{align*}
    \htheta_{\ell} = \paren{\widehat{\ATT}_{k,k+\ell}}_{k=1}^{K_{\ell}} \overset{a}{\sim} N\paren{\theta_{\ell}, \tSigma_{\ell}},
\end{align*}
where $\tSigma_{\ell}$ is a covariance matrix estimator. Let $\tsigma_{\hw,\ell} = \dsqrt{\hw_{\ell}'\tSigma_{\ell}\hw_{\ell}}$ denote the plug-in standard error for $\htau_{\hw,\ell} = \hw_{\ell}'\htheta_{\ell}$. Following equation \eqref{arxiv2:eq:class.simplex.bounded.variance}, consider the bounded variance simplex class
\begin{align*}
    \hL_{\sigma, \ell}^{+}(r) = \curly{\l \in \W_{+,\ell}: \tsigma_{\l,\ell} \leq r\tsigma_{\hw,\ell}}, \quad \W_{+,\ell} = \curly{w \in \R^{K_{\ell}}: \1'w = 1, w \geq 0}.
\end{align*} 
Executing Recipe \ref{arxiv2:recipe:robust.procedures} with $\hL_{\sigma, \ell}^{+}(r)$ produces $(\htau_{\ell}^{*}(r), CI_{\ell}^{*}(r))$. For the choice of standard error ratio bound $r$, I propose the following implementation.
\begin{enumerate}
    \item Choose $r=1$, allowing one to assess robustness to the class of nonnegative weights that yield standard errors no larger than the HR baseline $\tsigma_{\hw,\ell}$. This is the minimal $r$ for which the baseline weights are included in $\hL_{\sigma, \ell}^{+}(r)$.
    \item Report breakdown value $r_{\ell}^{*}$, which gives the smallest value of $r$ at which the robust $CI_{\ell}^{*}(r)$ includes a threshold value $\tau_{0}$ (e.g., $\tau_{0} = 0$). If $r_{\ell}^{*}$ is large, then the robust CI includes $\tau_{0}$ only after the class is expanded to allow standard errors that are large relative to $\tsigma_{\hw,\ell}$.
\end{enumerate}
I apply this implementation to \citet{lakdawala2023dynamic} in Section \ref{arxiv2:sec:application.lakdawala}.

\subsection{Implementation in Multisite Experiments}\label{arxiv2:sec:implementation.multisite.experiments}
Following the exposition in Section \ref{arxiv2:sec:setting}, a multisite experiment has sites $k$ containing units $i$ over which ATEs are estimated:
\begin{align*}
    \htheta = \paren{\widehat{\ATE}_{k}}_{k=1}^{K} \overset{a}{\sim} N\paren{\theta, \tSigma},
\end{align*}
where $\tSigma$ is a covariance matrix estimator. Given independent sites, it suffices to obtain site-level objects $(\htheta_{k}, \tsigma_{k}, \widehat{E}_{k}[X_{i}])$, where estimated variances $\tsigma_{k}^{2}$ can be collected into a diagonal matrix $\tSigma$ and estimated covariate means $\widehat{E}_{k}[X_{i}]$ can be collected into a matrix $\hbX = (\widehat{E}_{1}[X_{i}], \ldots, \widehat{E}_{K}[X_{i}])'$. Given microdata on units $i$ from the sites, one can directly estimate $\widehat{E}_{k}[X_{i}]$ with sample means and $(\htheta_{k}, \tsigma_{k})$ via regression. For example, letting $\I{i \in k}$ be an indicator for membership in site $k$, one can estimate the (no-intercept) saturated model
\begin{align*}
    Y_{i} = \sum_{k=1}^{K}\gamma_{0k}\I{i \in k} + \sum_{k=1}^{K}\theta_{k}D_{i}\I{i \in k} + u_{i}
\end{align*}
with OLS to obtain ATE estimates $\htheta_{k}$ and heteroskedasticity-robust standard errors $\tsigma_{k}$.

The baseline weights $\hw \in \W_{+}$ represent a particular population/distribution over the set of sites. To model departures from the baseline population, one can use the truncated simplex class $\hL_{+}(\epsilon)$ defined in \eqref{arxiv2:eq:class.truncated.simplex} or the covariate balance class $\hL_{X}(\Bar{c})$ defined in \eqref{arxiv2:eq:class.covariate.balance}. 
\begin{itemize}
    \item Departures in $\hL_{+}(\epsilon)$ ensure that groups $k$ represented in the baseline population remain represented in the alternative population: $\hw_{k} > 0 \implies \l_{k} > 0$ for each $\l \in \hL_{+}(\epsilon)$ when $\epsilon < 1$.
    \item Departures in $\hL_{X}(\Bar{c})$ incorporate information about group-level covariates $\hbX$.
\end{itemize}
To obtain benefits from both approaches, I propose the following implementation based on the truncated covariate balance class $\hL_{X}^{\epsilon}(\epsilon \Bar{c}) = \hL_{X}(\epsilon\Bar{c}) \cap \hL_{+}(\epsilon)$ with scaled gap $\epsilon\Bar{c}$ in view of \eqref{arxiv2:eq:class.truncated.balance.scaled.deciles}.
\begin{enumerate}
    \item Choose the equal weights baseline $\hw = w_{\EW} = \1/K$, which represents an empirical distribution over the sites, under which all groups $k$ will be represented in the alternative populations within $\hL_{X}^{\epsilon}(\epsilon\Bar{c}) = \hL_{X}(\epsilon\Bar{c}) \cap \L_{+}(\epsilon)$ given $\epsilon < 1$.
    \item Formulate a choice of truncation parameter $\epsilon$. Following \eqref{arxiv2:eq:leave.out.calibration}, one can interpret $\epsilon$ in terms of the leave-$J$-out calibration.
    \item Execute Recipe \ref{arxiv2:recipe:robust.procedures} with $\hL_{X}^{\epsilon}(\epsilon\Bar{c})$ for a range of balance gap bounds $\Bar{c}$. Following \eqref{arxiv2:eq:group.balance.gaps} and \eqref{arxiv2:eq:class.truncated.balance.scaled.deciles}, one can use the site-level balance gap deciles: $\Bar{c} \in \{\Bar{c}_{d}\}_{d=1}^{9}$.
\end{enumerate}
I apply this implementation to Project STAR in Section \ref{arxiv2:sec:application.STAR}.

\section{Empirical Applications}\label{arxiv2:sec:applications}
I now apply my procedures to the event study in \citet{lakdawala2023dynamic} and the multisite experiment in Project STAR, following the implementations described in Sections \ref{arxiv2:sec:implementation.event.studies} and \ref{arxiv2:sec:implementation.multisite.experiments}. Below I use significance levels $\alpha = \beta = 0.05$ and hence confidence levels $1-\alpha = 95\%$ for conventional CIs, $1-\beta = 95\%$ for heterogeneity UCBs, and $1-(\alpha+\beta)=90\%$ for robust CIs.

\subsection{Application to Lakdawala, Nakasone and Kho (2023)}\label{arxiv2:sec:application.lakdawala}
\citet{lakdawala2023dynamic}---henceforth LNK---study the rollout of internet access in Peruvian public primary schools to estimate the causal effects of school-based internet access on second grade test scores. The estimating sample contains schools that either gained internet access between 2007 and 2020 or remained unconnected by 2020, with second grade math and reading test scores observed from 2007--2016. LNK use a TWFE event-study specification: after absorbing school fixed effects, calendar-time controls, and time-varying school and student controls, the event-time coefficients exploit variation across schools with different installation timing. LNK's event-study figures report school-level dynamic effects for successive cohorts of second graders, with scores standardized within calendar year and effects interpreted relative to the year before internet installation.

LNK's Figure 2 shows a delayed achievement response. The TWFE event-study estimates, replicated in Figure \ref{arxiv2:fig:lakdawala-simplex-robust}, are small immediately after internet installation and grow over time. In LNK's page 234 description of the medium-run effect, ``By year 5, scores are 0.110 standard deviations higher for math'' and ``0.063 standard deviations higher for reading'' relative to the year before installation. These dynamics are central to the interpretation of the intervention because LNK argue that schools require time to adapt to new internet access. The empirical target is therefore a dynamic path of treatment effects.

\paragraph{Weighting Issues.} Following Section \ref{arxiv2:sec:implementation.event.studies}, let cohorts $k$ index the first year in which a school gains internet access and event times $\ell$ denote the time relative to first internet access. LNK's pages 241-42 discuss weighting issues: ``Recent work \ldots raises some important issues with using two-way fixed effects estimation when there is staggered timing of treatment. We show that these issues do not explain our main estimates by illustrating robustness to the estimator proposed in Sun and Abraham (2021) \ldots The patterns and magnitudes of these estimates are very similar to our main estimates.'' The \citet{sun2021estimating}---henceforth SA---estimator corresponds to weights $\hw_{\SA,\ell}$. LNK's Figure A.12, replicated in Figure \ref{arxiv2:fig:lakdawala-simplex-robust}, is therefore a robustness check that replaces the TWFE event-study aggregation with the SA aggregation in order to assess whether results are robust to alternative weights.\footnote{LNK use a bootstrap procedure to construct the SA CIs. Their replication files estimate cohort-specific event-study coefficients using the same controls and fixed effects as the main specification, drop cohorts not identified relative to the omitted period, aggregate cohort-event coefficients using the SA event-time weights, and bootstrap the aggregated coefficients using 500 replications stratified by first access year and clustered by school. The CIs reported here instead use the plug-in standard error computed from the estimated covariance matrix $\tSigma_{\ell}$ for the cohort-event coefficient vector $\htheta_{\ell}$, treating the SA weights as fixed. Thus, while I replicate the same SA estimates, my reported SA CIs are not identical to LNK's bootstrap SA CIs.} This suggests a sharper robustness question: are the substantive conclusions robust only to the SA weights, or to a broader class of plausible event-time weights? To answer this question, I compute robust estimators and CIs for classes of alternative weights that nest the SA weights.

\paragraph{Measuring Heterogeneity.} Table \ref{arxiv2:tab:lakdawala-heterogeneity-ucb} reports the estimated event-time heterogeneity UCBs $\teta_{\ell}$ used to construct the robust CIs. For event times with $\teta_{\ell}=0$, the heterogeneity adjustment in the respective robust CI drops out. The last row reports $K_{\ell}$, the total number of cohorts $k$ at each event time $\ell$. Across all but one event time, the heterogeneity UCB is lower for the math cohort-level estimates than for the reading cohort-level estimates, previewing the greater degree of robustness for the math results. 

\begin{table}[!htbp]
\centering

\caption{Heterogeneity UCB $\teta_{\ell}$ and number of event-time cohorts $K_{\ell}$}
\label{arxiv2:tab:lakdawala-heterogeneity-ucb}

\vspace{0.5em}
\begin{threeparttable}
\begin{tabular}{lccccccc}
\toprule
Quantity & $\ell=0$ & $\ell=1$ & $\ell=2$ & $\ell=3$ & $\ell=4$ & $\ell=5$ & $\ell=6$ \\
\midrule
$\teta_{\ell}$ for Math & 0.00 & 2.77 & 2.41 & 2.93 & 3.29 & 2.70 & 0.00 \\
$\teta_{\ell}$ for Reading & 2.58 & 2.69 & 4.42 & 3.54 & 5.59 & 4.53 & 2.79 \\
Number of cohorts $K_{\ell}$ & 9 & 8 & 7 & 6 & 5 & 4 & 3 \\
\bottomrule
\end{tabular}
\end{threeparttable}
\end{table}

\paragraph{Bounded Variance Simplex Class.} Consider the bounded variance simplex
\begin{align*}
    \hL^{+}_{\sigma,\ell}(r) = \curly{\l \in \W_{+,\ell}: \tsigma_{\l,\ell} \leq r\tsigma_{\SA,\ell}}, \quad \tsigma_{\SA,\ell} = \sqrt{\hw_{\SA,\ell}'\tSigma_{\ell}\hw_{\SA,\ell}},
\end{align*}
with $r=1$, which corresponds to the class of simplex-weighted estimands that can be estimated at least as precisely as the SA estimand. For each event time $\ell$, I compute the robust $CI^{*}_{\ell}(1)$ centered at the robust estimator $\htau_{\ell}^{*}(1)$. Figure \ref{arxiv2:fig:lakdawala-simplex-robust} presents three series: the TWFE CI centered at the TWFE estimate, the SA CI centered at the SA estimate, and the robust CI centered at the robust estimate. For math, the robust CIs are positive for event times 1 through 6 and only include zero at event time 0. At event time 5, the robust estimate is $0.099$ with CI $[0.055,0.143]$, close to both the TWFE and SA estimates. Thus, LNK's medium-run math conclusion is robust to alternative simplex weights that yield estimators at least as precise as the SA estimator. For reading, the robust estimates remain close to the SA path, but the conclusion that results are significantly positive is somewhat more sensitive. The robust CIs are positive for event times 1 through 4 and 6, while event time 5 barely includes zero: the robust estimate is $0.042$ with CI $[-0.0002,0.084]$. 

\begin{figure}[!htbp]
\centering

\caption{Impact of Internet Access on Test Scores}
\label{arxiv2:fig:lakdawala-simplex-robust}

\vspace{1em}

\begin{minipage}{0.82\textwidth}
\centering
\textbf{Panel A. Standardized Math Scores}\\[0.5em]

\includegraphics[width=\textwidth]{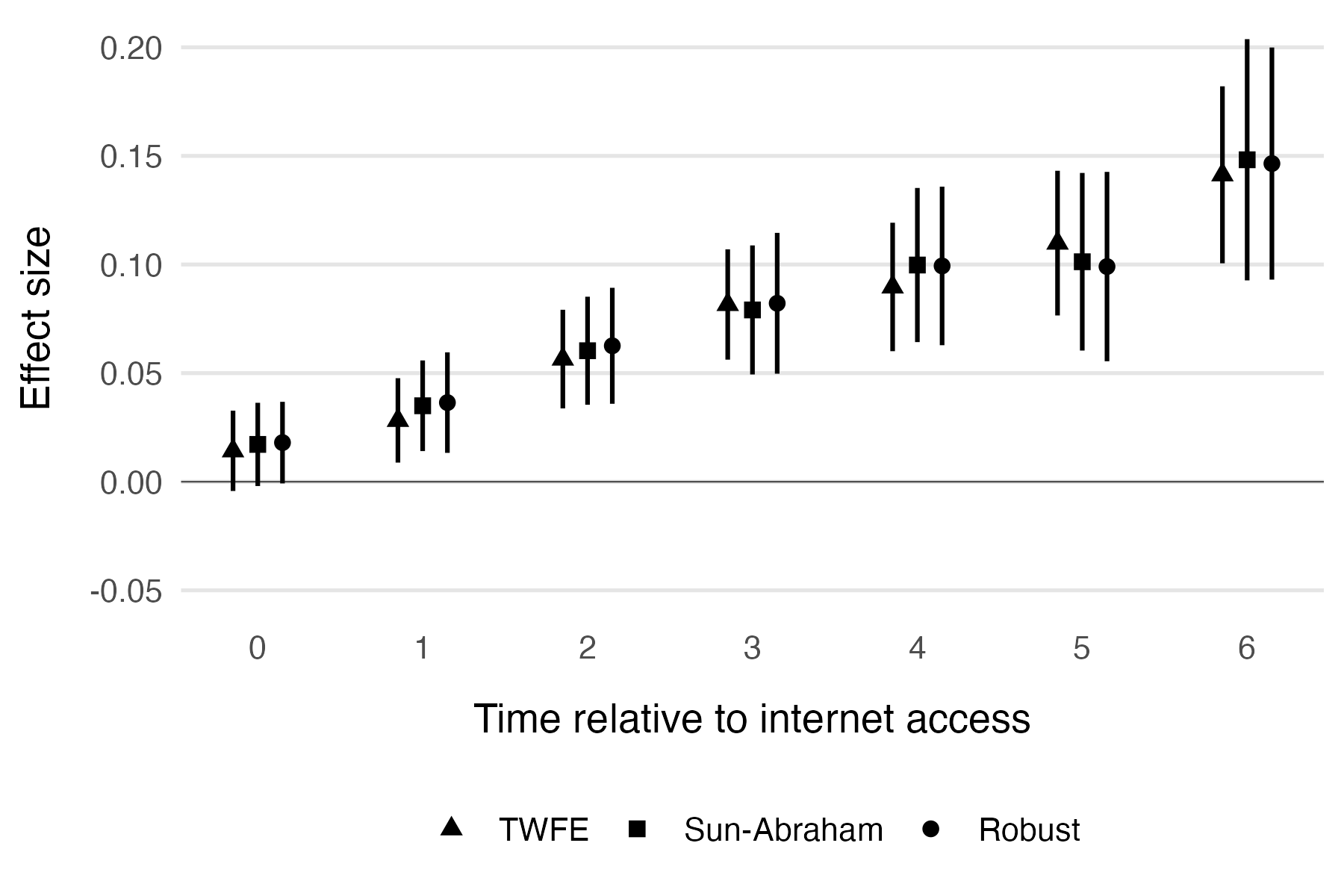}
\end{minipage}

\vspace{1.5em}

\begin{minipage}{0.82\textwidth}
\centering
\textbf{Panel B. Standardized Reading Scores}\\[0.5em]

\includegraphics[width=\textwidth]{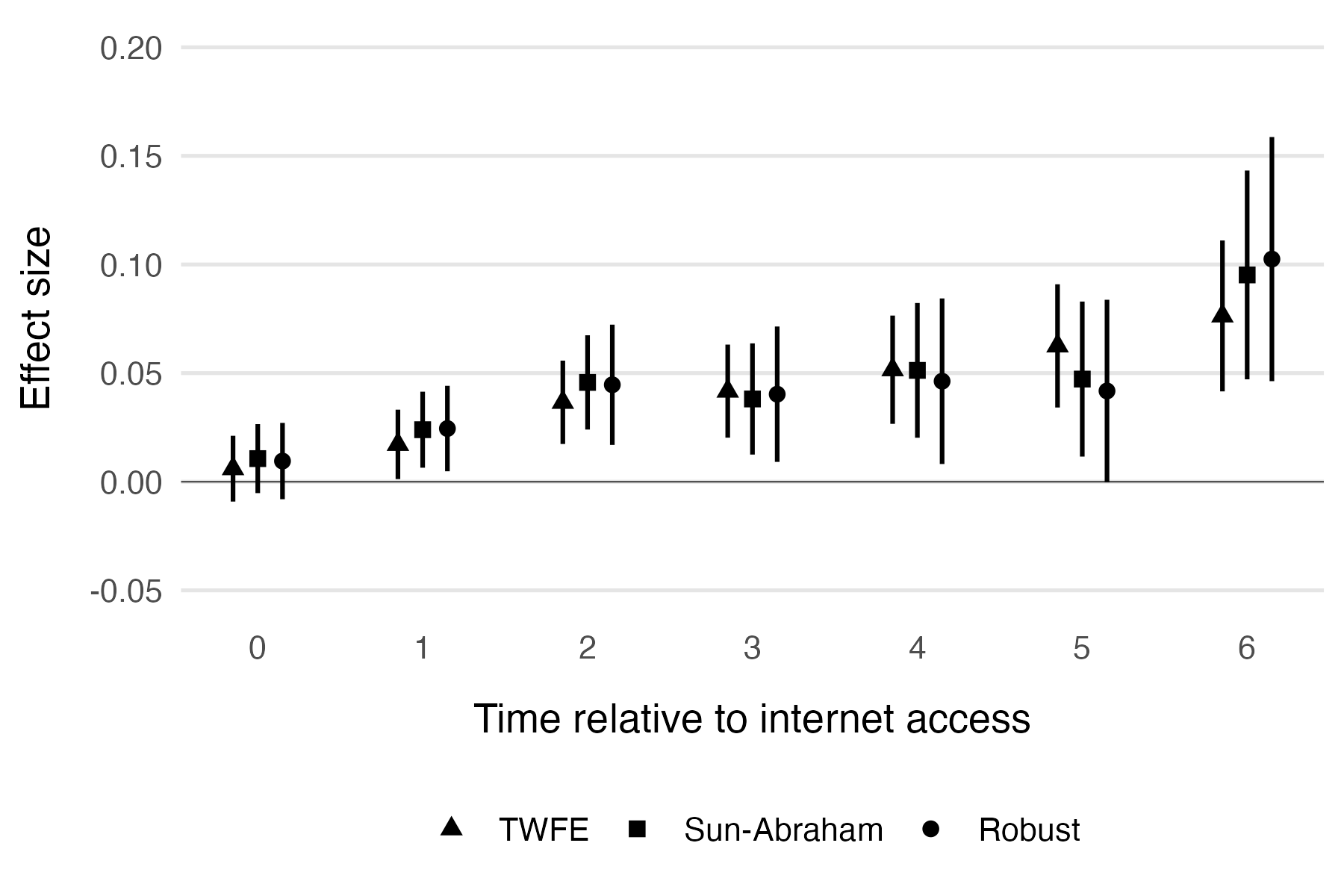}
\end{minipage}


\end{figure}

Table \ref{arxiv2:tab:lakdawala-case-i-simplex-breakdown} presents
breakdown values $r^{*}_{\ell}$, which for each $\ell$ gives the smallest value of $r$ at which
the robust $CI^{*}_{\ell}(r)$ includes zero for the class $\hL_{\sigma,\ell}^{+}(r)$. A value of $1.55$, for example, means that the robust CI begins to include zero once the allowable $\tsigma$ bound is expanded to about $1.55$ times the SA $\tsigma$. A dash means that zero is not reached even when considering the simplex without the variance constraint.\footnote{If $\teta_{\ell}=0$, I still compute breakdown values using the same definition: report the smallest feasible value of the robustness parameter at which the robust CI contains zero, and report ``--'' if zero is never reached over the full class. This convention implies, for example, that math at $\ell=0$ has $r^{*}_{0}<1$ because zero is already included at the smallest feasible bounded-variance simplex class, while math at $\ell=6$ has no finite breakdown value because zero is not reached even when the variance constraint is no longer binding over the simplex.} The breakdown values reinforce the visual evidence. For math, event times 2 through 5 tolerate substantial increases in $r$ before the robust CI includes zero, with $r^{*}_{5}=1.55$ at the headline year-5 event time. Event time 6 is even stronger in this sense: zero is not reached over the full simplex. For reading, by contrast, the breakdown values are close to one for event times 1 through 5, and event time 5 already includes zero at $r=1$. Overall, these results support the baseline math conclusions strongly, but not as strongly for the reading conclusions.

\begin{table}[!htbp]
\centering

\caption{Breakdown $r_{\ell}^{*}$ for $0 \in CI_{\ell}^{*}(r)$ under $\hL_{\sigma,\ell}^{+}(r)$}
\label{arxiv2:tab:lakdawala-case-i-simplex-breakdown}

\vspace{0.5em}
\begin{threeparttable}
\begin{tabular}{lccccccc}
\toprule
Outcome & $\ell=0$ & $\ell=1$ & $\ell=2$ & $\ell=3$ & $\ell=4$ & $\ell=5$ & $\ell=6$ \\
\midrule
$r_{\ell}^{*}$ for Math & 0.98 & 1.18 & 1.69 & 1.63 & 1.54 & 1.55 & -- \\
$r_{\ell}^{*}$ for Reading & 0.98 & 1.06 & 1.12 & 1.06 & 1.02 & 1.00 & 1.34 \\
\bottomrule
\end{tabular}
\end{threeparttable}
\end{table}

\subsection{Application to Project STAR}\label{arxiv2:sec:application.STAR}
The Tennessee Project STAR (Student/Teacher Achievement Ratio) experiment randomized students in seventy-nine Tennessee public elementary schools to classes with different numbers of students to estimate the causal effects of class size on test scores \citep{achilles2008tennessee, krueger1999experimental}. The experimental arms were (i) regular-sized classes (22-25 students), (ii) small classes (13-17 students), and (iii) regular-sized classes with a teacher aide. To define the sample for my analysis, I drop observations in the teacher-aide arm and focus on the effect of being assigned to a small class versus a regular-sized class. Following \citet{goldsmith2024contamination}, I focus on the effects of treatment in kindergarten, which mitigates attrition and other complications. This leaves $K=78$ schools and a total of $n=3783$ students. Following \citet{schanzenbach2006have}, I define the outcome as the average of a student's math and reading scores on the SAT, where the scores are demeaned and standardized relative to control-group scores. This puts the outcome into effect size units that have benchmarks in causal studies of preK-12 education interventions on student achievement: an effect size $\tau_{0}$ is small if $\tau_{0} < 0.05$, medium if $0.05 \leq \tau_{0} < 0.20$, and large if $\tau_{0} \geq 0.20$ \citep{kraft2020interpreting}. 

In my notation, $\theta_{k}$ is the ATE of being assigned to a small class on the test score outcome at school $k$, $\htheta_{k}$ is the estimated ATE at school $k$ from regressing the outcome on school-specific treatment indicators and fixed effects, $\tSigma$ is the heteroskedasticity-robust covariance matrix from this regression, and the baseline weights are equal weights (EW) $w = w_{\EW} = \1/K$ so that $\tau_{w}(\theta) = \tau_{\EW}(\theta)$ is the ATE in the empirical distribution of Project STAR schools. The baseline results establish medium-to-large positive effects of small class size on test scores. In particular, the EW estimate is $\htau_{\EW}=0.188$ with standard error $\tsigma_{\EW}=0.028$, yielding conventional EW CI $[0.133, 0.242]$. My baseline estimate is consistent with the one in \citet[Table 4]{schanzenbach2006have}. However, my standard error is about 0.01 smaller. In particular, \citet{schanzenbach2006have} clusters her standard errors at the classroom level. By contrast, I follow \citet{goldsmith2024contamination} and use heteroskedasticity-robust standard errors since the randomization of students to classrooms was at the individual level.\footnote{This is consistent with the when-to-cluster framework of \citet{abadie2023should}.}

\paragraph{Weighting Issues.} As reviewed by \citet{schanzenbach2006have}, Project STAR has been used in many policy discussions. However, \citet[page 207]{schanzenbach2006have} notes external validity concerns: 
\begin{itemize}
    \item ``A few aspects of the sample may limit the validity of generalizing the study to other settings. In order to be eligible to participate in the program, schools were required to have a minimum-size cohort of 57 students \ldots As a result, the schools that participated were about 25\% larger, on average, than other Tennessee schools'';
    \item ``Because of requirements imposed by the legislature for geographic diversity, schools in inner cities were overrepresented, and the students included were more economically disadvantaged and more likely to be African American'';
    \item ``Finally, average school spending in Tennessee was about three-fourths of the nationwide average, and teachers were less likely to have a master’s degree.''
\end{itemize}
Thus, the Project STAR results may not be representative of Tennessee (or U.S.) schools more broadly. Formally, the EW estimand averages over the empirical distribution of STAR schools, while alternative estimands of interest may place unequal weights across the STAR schools. This suggests a robustness question: are the baseline STAR results robust to alternative choices of weights/distributions? To answer this question, I compute robust estimators and CIs for classes of alternative weights that nest the EW weights, following implementations from Section \ref{arxiv2:sec:implementation.multisite.experiments}.

\paragraph{Measuring Heterogeneity.} The estimated heterogeneity UCB is $\teta = 16.214$, which is over twice as large as the baseline $t$-statistic $\htau_{\EW}/\tsigma_{\EW} = 6.714$. This suggests a meaningful degree of heterogeneity in the school-level ATE estimates, which previews the sensitivity of baseline results to the choice of weights.

\paragraph{Truncated Simplex Class.} I first consider the truncated simplex class 
\begin{align*}
    \L_{+}(\epsilon) = \curly{\l \in \W_{+}: \l \geq (1-\epsilon)/K} = \curly{(1-\epsilon)w_{\EW} + \epsilon \l: \l \in \W_{+}},
\end{align*}
which collects the set of $\epsilon$-deviations in weight distribution from the STAR empirical distribution. For each $\tau_{0} \in \{0, 0.05\}$, Table \ref{arxiv2:tab:star-truncated-simplex-breakdown} presents the breakdown value $\epsilon^{*}$, which is the smallest value of $\epsilon$ at which $\tau_{0}$ is included in the robust $CI^{*}(\epsilon)$ centered at robust estimate $\htau^{*}(\epsilon)$ with standard error $\tsigma^{*}(\epsilon)$. The breakdown value for $\tau_{0} = 0.05$ is $\epsilon^{*} = 0.0169$, while the breakdown value for $\tau_{0} = 0$ is $\epsilon^{*} = 0.0261$. These are both small departures from EW: even though the corresponding robust estimates are larger than the baseline EW estimate $\htau_{\EW} = 0.188$, the robust CIs indicate that baseline inferences are robust to only small departures from EW.

\begin{table}[!htbp]
\centering
\caption{Breakdown $\epsilon^{*}$ for $\tau_{0} \in CI^{*}(\epsilon)$ under $\L_{+}(\epsilon)$}

\label{arxiv2:tab:star-truncated-simplex-breakdown}
\vspace{0.5em}
\begin{threeparttable}
\begin{tabular}{lcccc}
\toprule
Target & $\epsilon^*$ & $\htau^{*}(\epsilon^{*})$ & $\tsigma^{*}(\epsilon^{*})$ & $CI^{*}(\epsilon^{*})$ \\
\midrule
$\tau_0=0.05$ & 0.0169 & 0.191 & 0.029 & [0.050, 0.332] \\
$\tau_0=0$ & 0.0261 & 0.193 & 0.029 & [0.000, 0.386] \\
\bottomrule
\end{tabular}
\end{threeparttable}
\end{table}

\paragraph{Truncated Covariate Balance Class.} I now intersect the truncated simplex $\L_{+}(\epsilon)$ with the covariate balance class $\hL_{X}(\epsilon\Bar{c}) = \{\l \in \W_{+}: c_{\l}(\hbX) \leq \epsilon\Bar{c}\}$, where $\epsilon > 0$, yielding
\begin{align*}
    \hL_{X}^{\epsilon}(\epsilon\Bar{c}) &= \curly{\l \in \L_{+}(\epsilon): c_{\l}(\hbX) \leq \epsilon\Bar{c}}, \quad c_{\l}(\hbX) = \max_{m \in \{1,\ldots,M\}}\frac{\abs{\l'\hbX_{m} - \mu(\hbX_{m})}}{\text{sd}(\hbX_{m})}, \quad \mu(\hbX_{m}) = w_{\EW}'\hbX_{m}, \\[5pt]
    &= \curly{(1-\epsilon)w_{\EW} + \epsilon \l: \l \in \hL_{X}(\Bar{c})}
\end{align*}
In particular, given school-level covariates $\hbX_{m} \in \R^{K}$ indexed by $m=1,\ldots,M$, the class considers $\epsilon$-deviations from the STAR baseline toward alternative populations with balance gaps no more than $\Bar{c}$. Table \ref{arxiv2:tab:star-covariate-summary} reports the $M=7$ Project STAR school-level covariates used in my analysis.

\begin{table}[!htbp]
\centering
\caption{Means and standard deviations of Project STAR covariates}
\label{arxiv2:tab:star-covariate-summary}
\vspace{0.5em}
\begin{threeparttable}
\begin{tabular}{lcc}
\toprule
Covariate & $\mu(\hbX_m)$ & $\operatorname{sd}(\hbX_m)$ \\
\midrule
Number of students & 80.65 & 25.67 \\
Fraction students female & 0.49 & 0.06 \\
Fraction students White or Asian & 0.71 & 0.39 \\
Average years of experience of students' teachers & 9.32 & 2.67 \\
Fraction students on FRPL & 0.47 & 0.27 \\
Fraction students with a White teacher & 0.87 & 0.25 \\
Fraction students with a teacher holding a master's degree or above & 0.36 & 0.27 \\
\bottomrule
\end{tabular}
\end{threeparttable}
\end{table}

I fix $\epsilon = \epsilon^{*}(0) = 0.0261$ at the truncated simplex breakdown value for $\tau_{0} = 0$ from the previous analysis and consider $\Bar{c} \in \{\Bar{c}_{d}\}_{d=1}^{9}$ equal to the deciles of the school-level balance gaps $c_{k}(\hbX)$. Intuitively, $c_{k}(\hbX)$ is the balance gap from a population that has covariate means $\hbX_{m,k}$ equal to those of school $k$. The deciles $\Bar{c}_{d}$ presented in Table \ref{arxiv2:tab:star-balance-gap-deciles} therefore represent typical values for the raw balance gaps $\Bar{c}$. For example, $\hL_{X}^{\epsilon}(\epsilon \Bar{c}_{5})$ consists of $\epsilon^{*}(0)$-deviations towards populations with covariate profiles that yield balance gaps no higher than the median school's gap $\Bar{c}_{5}$.

\begin{table}[!htbp]
\centering
\caption{Deciles of STAR school-level balance gaps}
\begin{threeparttable}
\label{arxiv2:tab:star-balance-gap-deciles}
\begin{tabular}{lccccccccc}
\toprule
Quantity & $d=1$ & $d=2$ & $d=3$ & $d=4$ & $d=5$ & $d=6$ & $d=7$ & $d=8$ & $d=9$ \\
\midrule
Gap decile $\Bar{c}_{d}$ & 0.805 & 1.040 & 1.310 & 1.419 & 1.635 & 1.792 & 1.880 & 2.222 & 2.497 \\
\bottomrule
\end{tabular}
\end{threeparttable}
\end{table}

Figure \ref{arxiv2:fig:truncated-balance-deciles} presents robust CIs centered at robust estimates under $\hL_{X}^{\epsilon}(\epsilon \Bar{c}_{d})$ for $\epsilon = \epsilon^{*}(0)$ and $d \in \{1,\ldots,9\}$. Across every balance gap decile $\Bar{c}_{d}$, the robust CIs include $\tau_{0} = 0.05$. Thus, while the covariate balance constraints allow the robust CIs to exclude zero relative to imposing only the truncated simplex constraint, this is not enough to allow one to robustly conclude that effects are at least medium-sized under $\epsilon^{*}(0)$-departures from EW. Overall, these results provide further support for the conclusion that baseline STAR inferences are robust to only small departures from EW. There is too much heterogeneity for the conclusions of a baseline CI to be (uniformly) informative for alternative weights that depart meaningfully from the baseline weights.

\begin{figure}[!htbp]
\centering
\caption{Robust CIs for $\hL_{X}^{\epsilon}(\epsilon\Bar{c}_{d})$ at $\epsilon = \epsilon^{*}(0)$}
\label{arxiv2:fig:truncated-balance-deciles}

\vspace{0.05em}

\includegraphics[width=0.85\textwidth]{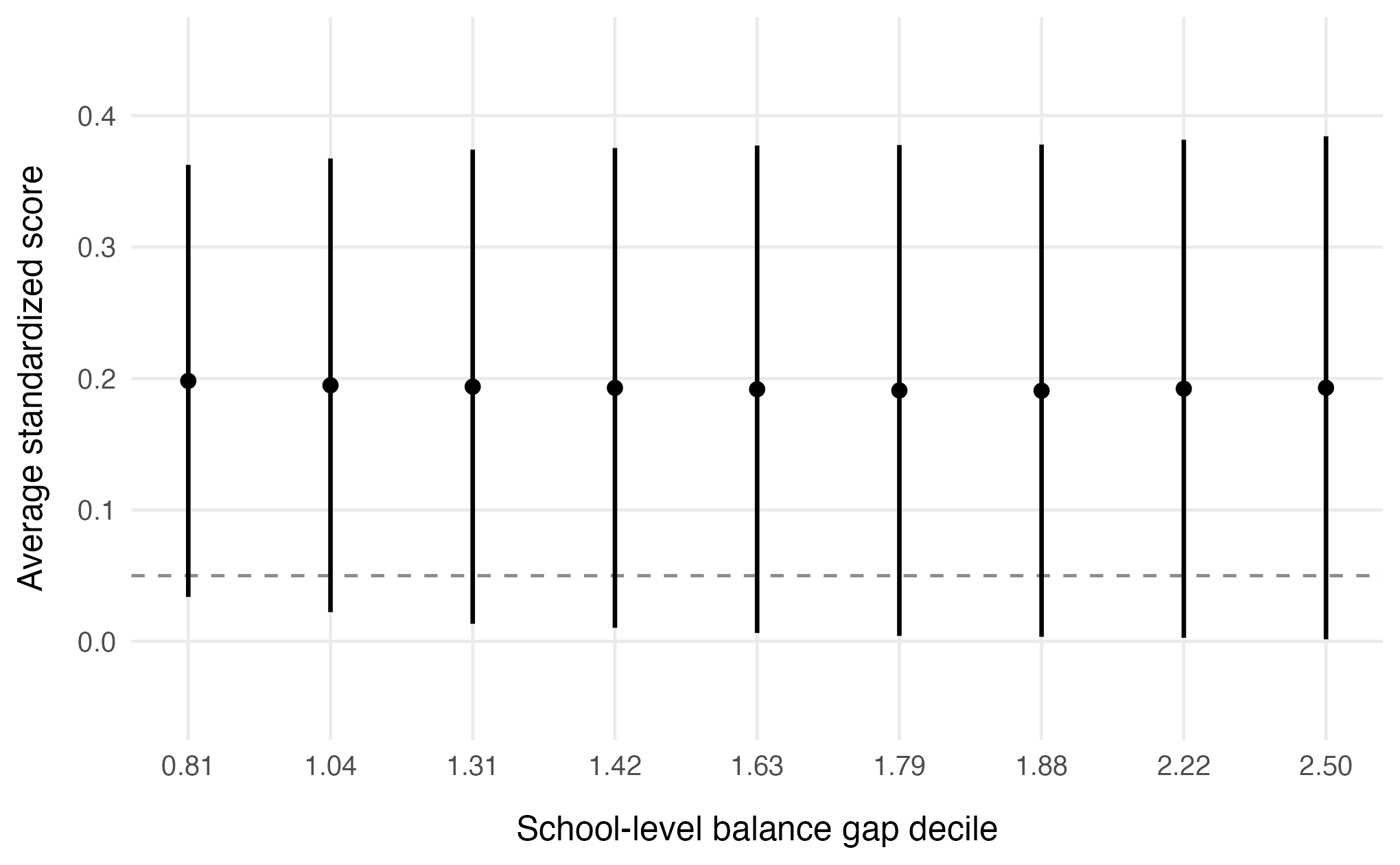}

\end{figure}

\section{Conclusion}\label{arxiv2:sec:conclusion}
Weighted estimands are ubiquitous in empirical research. The choice of weights is not merely a statistical nuisance: it shapes the empirical and policy relevance of a researcher's weighted estimand. Conventional inference procedures can therefore be limited when readers entertain alternative weights, leading to ambiguity and disagreement over the choice of weights.

This paper develops inference procedures that account for these issues. The key geometric observation is that the difference between two weighted estimands can be controlled by two objects: the heterogeneity of the underlying parameters and the distance between the corresponding weights. This decomposition leads to two practical tools. First, I derive minimax-bias weights that minimize the maximum distance to a class of alternative weights. The resulting robust estimator provides a natural default when researchers face ambiguity over the choice of weights. Second, I construct robust confidence intervals by combining a confidence bound for the heterogeneity in parameters with the maximum distance between the centering weights and the alternative weights. These intervals provide coverage guarantees for classes of weighted estimands rather than for a single baseline estimand. 

My framework accommodates many empirically relevant classes of alternatives. The bounded variance class captures concerns about comparing the baseline estimand to alternatives that can be estimated with reasonable precision. The truncated simplex class captures departures from a baseline target population and, when strictly truncated, ensures that groups represented under the baseline continue to receive positive weight. The covariate balance class captures external validity concerns by considering alternative populations whose covariate means remain close to the baseline. These classes can also be intersected, allowing researchers to combine restrictions on precision, positivity, and covariate-balance in a single robustness analysis.

The empirical applications illustrate that robustness to the choice of weights can be context-specific. For the event study in \citet{lakdawala2023dynamic}, conclusions are broadly robust to classes of weights that yield reasonable estimation precision. For the experiment in Project STAR, however, conclusions are sensitive even to small departures from the baseline weighting scheme. These findings highlight the value of reporting robust estimators and intervals alongside conventional ones: doing so makes explicit whether the conclusions reached under a researcher's baseline weights are informative for alternative weighted estimands that may be of interest to readers.

\clearpage
\bibliography{references.bib}

@article{de2020two,
  title={Two-way fixed effects estimators with heterogeneous treatment effects},
  author={de Chaisemartin, Cl{\'e}ment and D'Haultfoeuille, Xavier},
  journal={American economic review},
  volume={110},
  number={9},
  pages={2964--2996},
  year={2020},
  publisher={American Economic Association 2014 Broadway, Suite 305, Nashville, TN 37203}
}

@article{goodman2021difference,
  title={Difference-in-differences with variation in treatment timing},
  author={Goodman-Bacon, Andrew},
  journal={Journal of econometrics},
  volume={225},
  number={2},
  pages={254--277},
  year={2021},
  publisher={Elsevier}
}

@article{sun2021estimating,
  title={Estimating dynamic treatment effects in event studies with heterogeneous treatment effects},
  author={Sun, Liyang and Abraham, Sarah},
  journal={Journal of econometrics},
  volume={225},
  number={2},
  pages={175--199},
  year={2021},
  publisher={Elsevier}
}

@article{callaway2021difference,
  title={Difference-in-differences with multiple time periods},
  author={Callaway, Brantly and Sant’Anna, Pedro HC},
  journal={Journal of econometrics},
  volume={225},
  number={2},
  pages={200--230},
  year={2021},
  publisher={Elsevier}
}

@article{gardner2022two,
  title={Two-stage differences in differences},
  author={Gardner, John},
  journal={arXiv preprint arXiv:2207.05943},
  year={2022}
}

@article{athey2022design,
  title={Design-based analysis in difference-in-differences settings with staggered adoption},
  author={Athey, Susan and Imbens, Guido W},
  journal={Journal of econometrics},
  volume={226},
  number={1},
  pages={62--79},
  year={2022},
  publisher={Elsevier}
}

@article{de2022two,
  title={Two-way fixed effects and differences-in-differences with heterogeneous treatment effects: a survey},
  author={de Chaisemartin, Cl{\'e}ment and D'Haultf{\oe}uille, Xavier},
  journal={The Econometrics Journal},
  volume={26},
  number={3},
  pages={C1--C30},
  year={2022},
  publisher={Oxford University Press (OUP)}
}

@article{roth2023efficient,
  title={Efficient estimation for staggered rollout designs},
  author={Roth, Jonathan and Sant’Anna, Pedro HC},
  journal={Journal of Political Economy Microeconomics},
  volume={1},
  number={4},
  pages={669--709},
  year={2023},
  publisher={The University of Chicago Press Chicago, IL}
}

@article{roth2023s,
  title={What’s trending in difference-in-differences? A synthesis of the recent econometrics literature},
  author={Roth, Jonathan and Sant’Anna, Pedro HC and Bilinski, Alyssa and Poe, John},
  journal={Journal of Econometrics},
  volume={235},
  number={2},
  pages={2218--2244},
  year={2023},
  publisher={Elsevier}
}

@article{borusyak2024revisiting,
  title={Revisiting event-study designs: robust and efficient estimation},
  author={Borusyak, Kirill and Jaravel, Xavier and Spiess, Jann},
  journal={Review of Economic Studies},
  volume={91},
  number={6},
  pages={3253--3285},
  year={2024},
  publisher={Oxford University Press UK}
}

@techreport{wing2024stacked,
  title={Stacked difference-in-differences},
  author={Wing, Coady and Freedman, Seth M and Hollingsworth, Alex},
  year={2024},
  institution={National Bureau of Economic Research}
}

@techreport{abadie2025harvesting,
  title={Harvesting Differences-in-Differences and Event-Study Evidence},
  author={Abadie, Alberto and Angrist, Joshua and Frandsen, Brigham and Pischke, J{\"o}rn-Steffen},
  year={2025},
  institution={National Bureau of Economic Research}
}

@article{auer2025communication,
  title={Communication barriers and infant health: the intergenerational effect of randomly allocating refugees across language regions},
  author={Auer, Daniel and Kunz, Johannes S},
  journal={American Economic Journal: Economic Policy},
  volume={17},
  number={3},
  pages={71--106},
  year={2025},
  publisher={American Economic Association 2014 Broadway, Suite 305, Nashville, TN 37203-2425}
}

@article{chiu2026causal,
  title={Causal panel analysis under parallel trends: lessons from a large reanalysis study},
  author={Chiu, Albert and Lan, Xingchen and Liu, Ziyi and Xu, Yiqing},
  journal={American Political Science Review},
  volume={120},
  number={1},
  pages={245--266},
  year={2026},
  publisher={Cambridge University Press}
}

@article{lakdawala2023dynamic,
  title={Dynamic impacts of school-based internet access on student learning: Evidence from Peruvian public primary schools},
  author={Lakdawala, Leah K and Nakasone, Eduardo and Kho, Kevin},
  journal={American Economic Journal: Economic Policy},
  volume={15},
  number={4},
  pages={222--254},
  year={2023},
  publisher={American Economic Association 2014 Broadway, Suite 305, Nashville, TN 37203-2425}
}

@article{yitzhaki1996using,
  title={On using linear regressions in welfare economics},
  author={Yitzhaki, Shlomo},
  journal={Journal of Business \& Economic Statistics},
  volume={14},
  number={4},
  pages={478--486},
  year={1996},
  publisher={Taylor \& Francis}
}

@article{angrist1998estimating,
  title={Estimating the Labor Market Impact of Voluntary Military Service Using Social Security Data on Military Applicants},
  author={Angrist, Joshua},
  journal={Econometrica},
  volume={66},
  number={2},
  pages={249--288},
  year={1998}
}

@misc{crump2006moving,
  title={Moving the goalposts: Addressing limited overlap in the estimation of average treatment effects by changing the estimand},
  author={Crump, Richard K and Hotz, V Joseph and Imbens, Guido and Mitnik, Oscar},
  year={2006},
  publisher={National Bureau of Economic Research Cambridge, Mass., USA}
}

@article{achilles2008tennessee,
  title={Tennessee’s student teacher achievement ratio (STAR) project},
  author={Achilles, Charles M and Bain, Helen Pate and Bellott, Fred and Boyd-Zaharias, Jayne and Finn, Jeremy and Folger, John and Johnston, John and Word, Elizabeth},
  journal={Harvard Dataverse},
  volume={1},
  pages={2008},
  year={2008}
}

@techreport{angrist2010extrapolate,
  title={Extrapolate-ing: External validity and overidentification in the late framework},
  author={Angrist, Joshua and Fernandez-Val, Ivan},
  year={2010},
  institution={National Bureau of Economic Research}
}

@article{kennedy2019nonparametric,
  title={Nonparametric causal effects based on incremental propensity score interventions},
  author={Kennedy, Edward H},
  journal={Journal of the American Statistical Association},
  volume={114},
  number={526},
  pages={645--656},
  year={2019},
  publisher={Taylor \& Francis}
}

@article{de2021trading,
  title={Trading-off Bias and Variance in Stratified Experiments and in Matching Studies, Under a Boundedness Condition on the Magnitude of the Treatment Effect},
  author={de Chaisemartin, Cl{\'e}ment},
  journal={arXiv preprint arXiv:2105.08766},
  year={2021}
}

@article{zhou2022marginal,
  title={Marginal interventional effects},
  author={Zhou, Xiang and Opacic, Aleksei},
  journal={arXiv preprint arXiv:2206.10717},
  year={2022}
}

@article{goldsmith2024contamination,
  title={Contamination bias in linear regressions},
  author={Goldsmith-Pinkham, Paul and Hull, Peter and Koles{\'a}r, Michal},
  journal={American Economic Review},
  volume={114},
  number={12},
  pages={4015--4051},
  year={2024},
  publisher={American Economic Association 2014 Broadway, Suite 305, Nashville, TN 37203}
}

@article{kwon2025estimating,
  title={Estimating Treatment Effects Under Bounded Heterogeneity},
  author={Kwon, Soonwoo and Sun, Liyang},
  journal={arXiv preprint arXiv:2510.05454},
  year={2025}
}

@article{imbens1994identification,
  author  = {Imbens, Guido W. and Angrist, Joshua D.},
  title   = {Identification and Estimation of Local Average Treatment Effects},
  journal = {Econometrica},
  year    = {1994},
  volume  = {62},
  number  = {2},
  pages   = {467--475},
  doi     = {10.2307/2951620},
  url     = {https://www.jstor.org/stable/2951620}
}

@article{heckman2005structural,
  title={Structural equations, treatment effects, and econometric policy evaluation 1},
  author={Heckman, James J and Vytlacil, Edward},
  journal={Econometrica},
  volume={73},
  number={3},
  pages={669--738},
  year={2005},
  publisher={Wiley Online Library}
}

@techreport{kolesar2013ivheterogeneity,
  author       = {Koles{\'a}r, Michal},
  title        = {Estimation in an Instrumental Variables Model with Treatment Effect Heterogeneity},
  institution  = {Princeton University, Department of Economics},
  type         = {Working Paper},
  number       = {2013-2},
  year         = {2013},
  month        = {November},
  url          = {https://ideas.repec.org/p/pri/econom/2013-2.html}
}

@article{bhuller20242sls,
  title={2SLS with multiple treatments},
  author={Bhuller, Manudeep and Sigstad, Henrik},
  journal={Journal of Econometrics},
  volume={242},
  number={1},
  pages={105785},
  year={2024},
  publisher={Elsevier}
}

@article{abadie2024instrumental,
  title={Instrumental variable estimation with first-stage heterogeneity},
  author={Abadie, Alberto and Gu, Jiaying and Shen, Shu},
  journal={Journal of Econometrics},
  volume={240},
  number={2},
  pages={105425},
  year={2024},
  publisher={Elsevier}
}

@article{sloczynski2020should,
  title={When should we (not) interpret linear iv estimands as late?},
  author={Sloczy{\'n}ski, Tymon},
  journal={arXiv preprint arXiv:2011.06695},
  year={2020}
}

@article{huntington2020instruments,
  title={Instruments with heterogeneous effects: Bias, monotonicity, and localness},
  author={Huntington-Klein, Nick},
  journal={Journal of Causal Inference},
  volume={8},
  number={1},
  pages={182--208},
  year={2020},
  publisher={De Gruyter}
}

@article{coussens2021improving,
  title={Improving inference from simple instruments through compliance estimation},
  author={Coussens, Stephen and Spiess, Jann},
  journal={arXiv preprint arXiv:2108.03726},
  year={2021}
}

@article{mogstad2021causal,
  title={The causal interpretation of two-stage least squares with multiple instrumental variables},
  author={Mogstad, Magne and Torgovitsky, Alexander and Walters, Christopher R},
  journal={American Economic Review},
  volume={111},
  number={11},
  pages={3663--3698},
  year={2021},
  publisher={American Economic Association 2014 Broadway, Suite 305, Nashville, TN 37203}
}

@techreport{blandhol2022tsls,
  title={When is TSLS actually late?},
  author={Blandhol, Christine and Bonney, John and Mogstad, Magne and Torgovitsky, Alexander},
  year={2022},
  institution={National Bureau of Economic Research Cambridge, MA}
}

@article{sloczynski2022interpreting,
  title={Interpreting OLS estimands when treatment effects are heterogeneous: Smaller groups get larger weights},
  author={Sloczy{\'n}ski, Tymon},
  journal={Review of Economics and Statistics},
  volume={104},
  number={3},
  pages={501--509},
  year={2022},
  publisher={MIT Press One Rogers Street, Cambridge, MA 02142-1209, USA journals-info~…}
}

@article{krueger1999experimental,
  title={Experimental estimates of education production functions},
  author={Krueger, Alan B},
  journal={The quarterly journal of economics},
  volume={114},
  number={2},
  pages={497--532},
  year={1999},
  publisher={MIT Press}
}

@article{hotz2005predicting,
  title={Predicting the efficacy of future training programs using past experiences at other locations},
  author={Hotz, V Joseph and Imbens, Guido W and Mortimer, Julie H},
  journal={Journal of econometrics},
  volume={125},
  number={1-2},
  pages={241--270},
  year={2005},
  publisher={Elsevier}
}

@article{schanzenbach2006have,
  title={What have researchers learned from Project STAR?},
  author={Schanzenbach, Diane Whitmore},
  journal={Brookings papers on education policy},
  pages={205--228},
  year={2006},
  publisher={JSTOR}
}

@article{allcott2015site,
  title={Site selection bias in program evaluation},
  author={Allcott, Hunt},
  journal={The Quarterly journal of economics},
  volume={130},
  number={3},
  pages={1117--1165},
  year={2015},
  publisher={MIT Press}
}

@article{aronow2016does,
  title={Does regression produce representative estimates of causal effects?},
  author={Aronow, Peter M and Samii, Cyrus},
  journal={American Journal of Political Science},
  volume={60},
  number={1},
  pages={250--267},
  year={2016},
  publisher={Wiley Online Library}
}

@article{li2018balancing,
  title={Balancing covariates via propensity score weighting},
  author={Li, Fan and Morgan, Kari Lock and Zaslavsky, Alan M},
  journal={Journal of the American Statistical Association},
  volume={113},
  number={521},
  pages={390--400},
  year={2018},
  publisher={Taylor \& Francis}
}

@article{abadie2003economic,
  title={The economic costs of conflict: A case study of the Basque Country},
  author={Abadie, Alberto and Gardeazabal, Javier},
  journal={American economic review},
  volume={93},
  number={1},
  pages={113--132},
  year={2003},
  publisher={American Economic Association}
}

@article{abadie2010synthetic,
  title={Synthetic control methods for comparative case studies: Estimating the effect of California’s tobacco control program},
  author={Abadie, Alberto and Diamond, Alexis and Hainmueller, Jens},
  journal={Journal of the American statistical Association},
  volume={105},
  number={490},
  pages={493--505},
  year={2010},
  publisher={Taylor \& Francis}
}

@article{liu2025synthetic,
  title={Synthetic Parallel Trends},
  author={Liu, Yiqi},
  journal={arXiv preprint arXiv:2511.05870},
  year={2025}
}

@article{henry2012set,
  title={Set coverage and robust policy},
  author={Henry, Marc and Onatski, Alexei},
  journal={Economics Letters},
  volume={115},
  number={2},
  pages={256--257},
  year={2012},
  publisher={Elsevier}
}

@article{manski2021econometrics,
  title={Econometrics for decision making: Building foundations sketched by Haavelmo and Wald},
  author={Manski, Charles F},
  journal={Econometrica},
  volume={89},
  number={6},
  pages={2827--2853},
  year={2021},
  publisher={Wiley Online Library}
}

@article{armstrong2025adapting,
  title={Adapting to misspecification},
  author={Armstrong, Timothy B and Kline, Patrick and Sun, Liyang},
  journal={Econometrica},
  volume={93},
  number={6},
  pages={1981--2005},
  year={2025},
  publisher={Wiley Online Library}
}

@article{chernozhukov2025policy,
  title={Policy Learning with Confidence},
  author={Chernozhukov, Victor and Lee, Sokbae and Rosen, Adam M and Sun, Liyang},
  journal={arXiv preprint arXiv:2502.10653},
  year={2025}
}

@article{andrews2025certified,
  title={Certified Decisions},
  author={Andrews, Isaiah and Chen, Jiafeng},
  journal={arXiv preprint arXiv:2502.17830},
  year={2025}
}

@article{ben2025safe,
  title={Safe policy learning through extrapolation: Application to pre-trial risk assessment},
  author={Ben-Michael, Eli and Greiner, D James and Imai, Kosuke and Jiang, Zhichao},
  journal={Journal of the American Statistical Association},
  pages={forthcoming},
  year={2025},
  publisher={Taylor \& Francis}
}

@misc{lau2026aggregating,
  author = {Lau, Chun Pong},
  title  = {Aggregating Treatment Effects across Multiple Outcomes},
  year   = {2026},
  note   = {Job market paper},
  url    = {https://conroylau.github.io/conroy_lau_jmp.pdf}
}

@article{adusumilli2026you,
  title={You've Got to be Efficient: Ambiguity, Misspecification and Variational Preferences},
  author={Adusumilli, Karun},
  journal={arXiv preprint arXiv:2604.05327},
  year={2026}
}

@article{sarfati2026integrating,
  title={Integrating Diagnostic Checks into Estimation},
  author={Sarfati, Reca and Vilfort, Vod},
  journal={arXiv preprint arXiv:2604.16690},
  year={2026}
}

@article{scheffe1953method,
  title={A method for judging all contrasts in the analysis of variance},
  author={Scheff{\'e}, Henry},
  journal={Biometrika},
  volume={40},
  number={1-2},
  pages={87--110},
  year={1953},
  publisher={Oxford University Press}
}

@book{pfanzagl1994parametric,
  title={Parametric statistical theory},
  author={Pfanzagl, Johann},
  year={1994},
  publisher={Walter de Gruyter}
}

@article{imbens2004confidence,
  title={Confidence intervals for partially identified parameters},
  author={Imbens, Guido W and Manski, Charles F},
  journal={Econometrica},
  volume={72},
  number={6},
  pages={1845--1857},
  year={2004},
  publisher={Wiley Online Library}
}

@article{armstrong2018optimal,
  title={Optimal inference in a class of regression models},
  author={Armstrong, Timothy B and Koles{\'a}r, Michal},
  journal={Econometrica},
  volume={86},
  number={2},
  pages={655--683},
  year={2018},
  publisher={Wiley Online Library}
}

@article{armstrong2020simple,
  title={Simple and honest confidence intervals in nonparametric regression},
  author={Armstrong, Timothy B and Koles{\'a}r, Michal},
  journal={Quantitative Economics},
  volume={11},
  number={1},
  pages={1--39},
  year={2020},
  publisher={Wiley Online Library}
}

@article{kline2020leave,
  title={Leave-out estimation of variance components},
  author={Kline, Patrick and Saggio, Raffaele and S{\o}lvsten, Mikkel},
  journal={Econometrica},
  volume={88},
  number={5},
  pages={1859--1898},
  year={2020},
  publisher={Wiley Online Library}
}

@article{armstrong2021finite,
  title={Finite-Sample Optimal Estimation and Inference on Average Treatment Effects Under Unconfoundedness},
  author={Armstrong, Timothy B and Koles{\'a}r, Michal},
  journal={Econometrica},
  volume={89},
  number={3},
  pages={1141--1177},
  year={2021},
  publisher={Wiley Online Library}
}

@article{armstrong2021sensitivity,
  title={Sensitivity analysis using approximate moment condition models},
  author={Armstrong, Timothy B and Koles{\'a}r, Michal},
  journal={Quantitative Economics},
  volume={12},
  number={1},
  pages={77--108},
  year={2021},
  publisher={Wiley Online Library}
}

@incollection{van1996weak,
  title={Weak convergence},
  author={Van Der Vaart, Aad W and Wellner, Jon A},
  booktitle={Weak convergence and empirical processes: with applications to statistics},
  pages={16--28},
  year={1996},
  publisher={Springer}
}

@book{van2000asymptotic,
  title={Asymptotic statistics},
  author={Van der Vaart, Aad W},
  volume={3},
  year={2000},
  publisher={Cambridge university press}
}

@article{heckman2007econometric,
  title={Econometric evaluation of social programs, part II: Using the marginal treatment effect to organize alternative econometric estimators to evaluate social programs, and to forecast their effects in new environments},
  author={Heckman, James J and Vytlacil, Edward J},
  journal={Handbook of econometrics},
  volume={6},
  pages={4875--5143},
  year={2007},
  publisher={Elsevier}
}

@article{sun2010monotonicity,
  title={On the Monotonicity, Log-Concavity, and Tight Bounds of the Generalized Marcum and Nuttall $ Q $-Functions},
  author={Sun, Yin and Baricz, {\'A}rp{\'a}d and Zhou, Shidong},
  journal={IEEE Transactions on Information Theory},
  volume={56},
  number={3},
  pages={1166--1186},
  year={2010},
  publisher={IEEE}
}

@article{romano2014practical,
  title={A practical two-step method for testing moment inequalities},
  author={Romano, Joseph P and Shaikh, Azeem M and Wolf, Michael},
  journal={Econometrica},
  volume={82},
  number={5},
  pages={1979--2002},
  year={2014},
  publisher={Wiley Online Library}
}

@article{seri2015tight,
  title={A tight bound on the distance between a noncentral chi square and a normal distribution},
  author={Seri, Raffaello},
  journal={IEEE Communications Letters},
  volume={19},
  number={11},
  pages={1877--1880},
  year={2015},
  publisher={IEEE}
}

@article{belloni2015some,
  title={Some new asymptotic theory for least squares series: Pointwise and uniform results},
  author={Belloni, Alexandre and Chernozhukov, Victor and Chetverikov, Denis and Kato, Kengo},
  journal={Journal of Econometrics},
  volume={186},
  number={2},
  pages={345--366},
  year={2015},
  publisher={Elsevier}
}

@article{mccloskey2017bonferroni,
  title={Bonferroni-based size-correction for nonstandard testing problems},
  author={McCloskey, Adam},
  journal={Journal of Econometrics},
  volume={200},
  number={1},
  pages={17--35},
  year={2017},
  publisher={Elsevier}
}

@article{anatolyev2019many,
  title={Many instruments and/or regressors: A friendly guide},
  author={Anatolyev, Stanislav},
  journal={Journal of Economic Surveys},
  volume={33},
  number={2},
  pages={689--726},
  year={2019},
  publisher={Wiley Online Library}
}

@article{molinari2020microeconometrics,
  title={Microeconometrics with partial identification},
  author={Molinari, Francesca},
  journal={Handbook of econometrics},
  volume={7},
  pages={355--486},
  year={2020},
  publisher={Elsevier}
}

@article{kraft2020interpreting,
  title={Interpreting effect sizes of education interventions},
  author={Kraft, Matthew A},
  journal={Educational researcher},
  volume={49},
  number={4},
  pages={241--253},
  year={2020},
  publisher={Sage Publications Sage CA: Los Angeles, CA}
}

@misc{athey2023thirdnumber,
  author       = {Athey, Susan and Imbens, Guido W.},
  title        = {In Search of the Third Number: What to Report Beyond Point Estimates and Standard Errors},
  year         = {2023},
  howpublished = {\url{https://www.youtube.com/watch?v=i-K9a6UhtsE}},
  note         = {Keynote Lecture at the Berkeley Initiative for Transparency in the Social Sciences (BITSS)},
  publisher    = {YouTube},
  urldate = {2025-10-10}
}

@article{abadie2023should,
  title={When should you adjust standard errors for clustering?},
  author={Abadie, Alberto and Athey, Susan and Imbens, Guido W and Wooldridge, Jeffrey M},
  journal={The Quarterly Journal of Economics},
  volume={138},
  number={1},
  pages={1--35},
  year={2023},
  publisher={Oxford University Press}
}

@book{lehmann2024testing,
  title={Testing statistical hypotheses},
  author={Lehmann, E.L. and Romano, Joseph P.},
  year={2024},
  publisher={Springer}
}

@incollection{mogstad2024instrumental,
  title={Instrumental variables with unobserved heterogeneity in treatment effects},
  author={Mogstad, Magne and Torgovitsky, Alexander},
  booktitle={Handbook of Labor Economics},
  volume={5},
  pages={1--114},
  year={2024},
  publisher={Elsevier}
}

@article{poirier2024quantifying,
  title={Quantifying the internal validity of weighted estimands},
  author={Poirier, Alexandre and Sloczy{\'n}ski, Tymon},
  journal={arXiv preprint arXiv:2404.14603},
  year={2024}
}

@article{10.1257/pandp.115.743,
Author = {Luttmer, Erzo F.P.},
Title = {Report of the Editor, American Economic Review},
Journal = {AEA Papers and Proceedings},
Volume = {115},
Year = {2025},
Month = {May},
Pages = {743–60},
DOI = {10.1257/pandp.115.743},
URL = {https://www.aeaweb.org/articles?id=10.1257/pandp.115.743}}

@article{10.1257/pandp.116.742,
Author = {Luttmer, Erzo F.P.},
Title = {Report of the Editor, American Economic Review},
Journal = {AEA Papers and Proceedings},
Volume = {116},
Year = {2026},
Month = {May},
Pages = {742–760},
DOI = {10.1257/pandp.116.742},
URL = {https://www.aeaweb.org/articles?id=10.1257/pandp.116.742}}

@article{10.1257/pandp.115.779,
Author = {Olken, Benjamin},
Title = {Report of the Editor, American Economic Journal: Applied Economics},
Journal = {AEA Papers and Proceedings},
Volume = {115},
Year = {2025},
Month = {May},
Pages = {779–88},
DOI = {10.1257/pandp.115.779},
URL = {https://www.aeaweb.org/articles?id=10.1257/pandp.115.779}}

@article{10.1257/pandp.116.780,
Author = {Olken, Benjamin},
Title = {Report of the Editor, American Economic Journal: Applied Economics},
Journal = {AEA Papers and Proceedings},
Volume = {116},
Year = {2026},
Month = {May},
Pages = {780–789},
DOI = {10.1257/pandp.116.780},
URL = {https://www.aeaweb.org/articles?id=10.1257/pandp.116.780}}

@article{10.1257/pandp.115.789,
Author = {Davis, Lucas},
Title = {Report of the Editor, American Economic Journal: Economic Policy},
Journal = {AEA Papers and Proceedings},
Volume = {115},
Year = {2025},
Month = {May},
Pages = {789–99},
DOI = {10.1257/pandp.115.789},
URL = {https://www.aeaweb.org/articles?id=10.1257/pandp.115.789}}

@article{10.1257/pandp.116.790,
Author = {Jackson, C. Kirabo},
Title = {Report of the Editor, American Economic Journal: Economic Policy},
Journal = {AEA Papers and Proceedings},
Volume = {116},
Year = {2026},
Month = {May},
Pages = {790–800},
DOI = {10.1257/pandp.116.790},
URL = {https://www.aeaweb.org/articles?id=10.1257/pandp.116.790}}
\clearpage

\appendix


\setcounter{section}{0}
\renewcommand{\thesection}{\Alph{section}}

\begin{center}
\LARGE Supplemental Appendix to \\ ``Robust Inference for Weighted Estimands''
\end{center}

\begin{center}
\large Vod Vilfort
\end{center}

\vspace{1.5em}

\section{Event-Study Survey Construction and Coding}\label{arxiv2:app:event-study-survey}
This appendix describes the construction and coding of the survey summarized in Section \ref{arxiv2:sec:implementation.event.studies}. I use the economics literature search tool at \url{https://paulgp.com/econlit-pipeline/search.html} to search 2024-2025 records from AER, AEJ: Applied, and AEJ: Policy. On July 10, 2026, I downloaded the CSV files produced by each of the following exact-phrase search keys:
\begin{itemize}
    \item ``Two-Way Fixed Effects Estimators with Heterogeneous Treatment Effects''
    \item ``Difference-in-Differences with Multiple Time Periods''
    \item ``Difference-in-Differences with Variation in Treatment Timing''
    \item ``Estimating Dynamic Treatment Effects in Event Studies with Heterogeneous Treatment Effects''
    \item ``Revisiting Event-Study Designs: Robust and Efficient Estimation''
\end{itemize}
These are the titles of \citet{de2020two}, \citet{callaway2021difference}, \citet{goodman2021difference}, \citet{sun2021estimating}, and \citet{borusyak2024revisiting}, respectively. Thus, the title-citation screen identifies papers whose references contain at least one of these five canonical contributions to the event-study literature on weighting and heterogeneity issues. The papers defining the screen are not identical to the five core HR methods used in the classification coding. In particular, \citet{goodman2021difference} defines a title-citation category but is not one of the five core HR methods; when researchers use the \citet{goodman2021difference} decomposition, I code it as a diagnostic.

After combining the CSV files and deduplicating by DOI, the title-citation screen yields 66 papers. As described in Section \ref{arxiv2:sec:implementation.event.studies}, I manually apply two further screens: (i) whether a paper contains an event-study or dynamic DiD design, which 64 papers satisfy, and (ii) whether that design has staggered or otherwise varying treatment timing or exposure across groups, which 60 papers satisfy. The resulting 60 papers form the main survey denominator.

I performed the initial screening and coding manually, reviewing the published main-text and online appendix (OA) PDFs for each paper. As a second pass, I asked Codex to compare my coding spreadsheet with the published main-text PDFs and flag possible inconsistencies. The Codex review did not independently inspect the OAs. It could therefore corroborate main-text reporting and identify main-text references to OA analyses, but it could not independently verify OA-only coding. I reviewed the flagged cases manually and made all final inclusion and coding decisions.

For each paper, I code whether it reports a conventional TWFE event-study specification; any of five core HR methods---Sun-Abraham (SA), Callaway-Sant'Anna (CS), stacked DiD, Borusyak-Jaravel-Spiess imputation (BJS), and de Chaisemartin-D'Haultfoeuille (dCDH); any non-core alternative intended to address weighting or heterogeneity concerns; and any diagnostic method for assessing such concerns. The core-method count is the number of nonempty indicators among SA, CS, stacked DiD, BJS, and dCDH. I count a method as reported when its estimates or results appear in either the published main text or OA. 

I code the reporting location for each distinct method, giving priority to the main text: I record M if the method is reported in the main text, even if it also appears in the OA, and I record OA only when the method is reported exclusively in the OA. Thus, OA always means OA-only, while M does not rule out additional OA reporting. Table \ref{arxiv2:tab:event-study-survey-method-location} summarizes method reporting by category and location. Table \ref{arxiv2:tab:event-study-survey-paper-list} lists the 60 included papers, preserves their full titles, and displays the location code for each method category. A paper may report multiple distinct non-core alternatives or diagnostics, so the Other and Diagnostic columns retain both counts and locations. For example, M+OA denotes one distinct item in each location, while M(2)+OA denotes two distinct main-text items and one OA item. A blank cell means that the paper does not report that method or category.

\begin{table}[!htbp]
\centering
\caption{Method reporting by category and location}
\label{arxiv2:tab:event-study-survey-method-location}
\vspace{0.5em}
\begin{threeparttable}
\begin{tabular}{lccc}
\toprule
Method category & Any report & Main text & Online appendix \\
\midrule
TWFE & 47 (78.3\%) & 40 (66.7\%) & 7 (11.7\%) \\
Sun-Abraham & 14 (23.3\%) & 7 (11.7\%) & 7 (11.7\%) \\
Callaway-Sant'Anna & 15 (25.0\%) & 5 (8.3\%) & 10 (16.7\%) \\
Stacked DiD & 11 (18.3\%) & 10 (16.7\%) & 1 (1.7\%) \\
BJS imputation & 13 (21.7\%) & 6 (10.0\%) & 7 (11.7\%) \\
dCDH & 12 (20.0\%) & 6 (10.0\%) & 6 (10.0\%) \\
Other alternative & 16 (26.7\%) & 10 (16.7\%) & 6 (10.0\%) \\
Diagnostic & 12 (20.0\%) & 4 (6.7\%) & 9 (15.0\%) \\
\bottomrule
\end{tabular}
\vspace{0.35em}
\footnotesize
\textit{Notes}: Counts use the 60-paper included sample. The ``Any report'' column counts papers with at least one report in either location. The ``Main text'' and ``Online appendix'' columns count papers with at least one coded item in that location. The Other and Diagnostic categories can contain multiple distinct items and their location columns can therefore overlap in papers. ``Other alternative'' refers to non-core alternatives intended to address weighting or heterogeneity concerns.
\end{threeparttable}
\end{table}

\clearpage
\begin{landscape}
\begingroup
\scriptsize
\setlength{\tabcolsep}{2pt}
\renewcommand{\arraystretch}{1.08}
\setcounter{LTchunksize}{100}
\begin{longtable}{@{}>{\raggedright\arraybackslash}p{0.052\linewidth}>{\raggedright\arraybackslash}p{0.145\linewidth}>{\raggedright\arraybackslash}p{0.38\linewidth}cccccccc@{}}
\caption{Included papers and method-reporting locations}\label{arxiv2:tab:event-study-survey-paper-list}\\
\toprule
J./yr. & DOI & Paper & TWFE & SA & CS & Stk. & BJS & dCDH & Other & Diag. \\
\midrule
\endfirsthead
\multicolumn{11}{c}{\tablename~\thetable: Included papers and method-reporting locations (continued)}\\[0.35em]
\toprule
J./yr. & DOI & Paper & TWFE & SA & CS & Stk. & BJS & dCDH & Other & Diag. \\
\midrule
\endhead
\midrule
\multicolumn{11}{r}{Continued on next page}\\
\endfoot
\bottomrule
\endlastfoot
\rule{0pt}{2.4em}Pol-25 & \href{https://doi.org/10.1257/pol.20210618}{\nolinkurl{10.1257/pol.20210618}} & A Machine Learning Approach to Analyze and Support Anticorruption Policy & M &  & OA &  &  &  &  &  \\
\rule{0pt}{2.4em}Pol-25 & \href{https://doi.org/10.1257/pol.20230220}{\nolinkurl{10.1257/pol.20230220}} & Communication Barriers and Infant Health: The Intergenerational Effect of Randomly Allocating Refugees across Language Regions & M &  &  &  &  &  &  & OA \\
\rule{0pt}{2.4em}Pol-25 & \href{https://doi.org/10.1257/pol.20230613}{\nolinkurl{10.1257/pol.20230613}} & Do Consumers Distinguish Fixed Cost from Variable Cost? ``Schmeduling'' in Two-Part Tariffs in Energy & OA &  &  &  &  & M &  & OA \\
\rule{0pt}{2.4em}Pol-25 & \href{https://doi.org/10.1257/pol.20220009}{\nolinkurl{10.1257/pol.20220009}} & Does the ``Boost for Mathematics'' Boost Mathematics? A Large-Scale Evaluation of the ``Lesson Study'' Methodology on Student Performance &  & M &  &  &  &  &  &  \\
\rule{0pt}{2.4em}Pol-25 & \href{https://doi.org/10.1257/pol.20230036}{\nolinkurl{10.1257/pol.20230036}} & Effective Health Aid: Evidence from Gavi's Vaccine Program & M &  &  &  & M &  &  &  \\
\rule{0pt}{2.4em}Pol-25 & \href{https://doi.org/10.1257/pol.20230460}{\nolinkurl{10.1257/pol.20230460}} & Eliminating Fares to Expand Opportunities: Experimental Evidence on the Impacts of Free Public Transportation on Economic and Social Disparities & M &  &  &  & M &  & M &  \\
\rule{0pt}{2.4em}Pol-25 & \href{https://doi.org/10.1257/pol.20220760}{\nolinkurl{10.1257/pol.20220760}} & Increasing Organ Donor Registration as a Means to Increase Transplantation: An Experiment with Actual Organ Donor Registrations &  &  &  & M &  &  &  &  \\
\rule{0pt}{2.4em}Pol-25 & \href{https://doi.org/10.1257/pol.20220227}{\nolinkurl{10.1257/pol.20220227}} & Mental Health Consequences of Correctional Sentencing & M &  & OA &  & M &  &  &  \\
\rule{0pt}{2.4em}Pol-25 & \href{https://doi.org/10.1257/pol.20220555}{\nolinkurl{10.1257/pol.20220555}} & Physician Group Influences on Treatment Intensity and Health: Evidence from Physician Switchers & M &  &  &  &  &  & OA &  \\
\rule{0pt}{2.4em}Pol-25 & \href{https://doi.org/10.1257/pol.20220580}{\nolinkurl{10.1257/pol.20220580}} & Profit Taxation, R\&D Spending, and Innovation & M &  &  &  &  &  & OA &  \\
\rule{0pt}{2.4em}Pol-25 & \href{https://doi.org/10.1257/pol.20230125}{\nolinkurl{10.1257/pol.20230125}} & Reconstruction-Era Education and Long-Run Black-White Inequality & M &  &  &  &  &  & OA &  \\
\rule{0pt}{2.4em}Pol-25 & \href{https://doi.org/10.1257/pol.20220667}{\nolinkurl{10.1257/pol.20220667}} & Spillover, Efficiency, and Equity Effects of Regional Firm Subsidies & M & OA &  &  &  &  & OA &  \\
\rule{0pt}{2.4em}Pol-25 & \href{https://doi.org/10.1257/pol.20230615}{\nolinkurl{10.1257/pol.20230615}} & Terrorism and Voting: The Rise of Right-Wing Populism in Germany &  & M &  &  & OA &  &  &  \\
\addlinespace[0.35em]
\rule{0pt}{2.4em}Pol-24 & \href{https://doi.org/10.1257/pol.20220810}{\nolinkurl{10.1257/pol.20220810}} & Achieving Air Pollution Control Targets with Technology-Aided Monitoring: Better Enforcement or Localized Efforts? & M &  & OA &  &  &  &  & OA \\
\rule{0pt}{2.4em}Pol-24 & \href{https://doi.org/10.1257/pol.20220701}{\nolinkurl{10.1257/pol.20220701}} & Administrative Burden and Procedural Denials: Experimental Evidence from SNAP & M &  & OA &  &  &  &  &  \\
\rule{0pt}{2.4em}Pol-24 & \href{https://doi.org/10.1257/pol.20200843}{\nolinkurl{10.1257/pol.20200843}} & Attendance Boundary Policies and the Limits to Combating School Segregation & M &  &  &  &  &  & M &  \\
\rule{0pt}{2.4em}Pol-24 & \href{https://doi.org/10.1257/pol.20220181}{\nolinkurl{10.1257/pol.20220181}} & Did the Paycheck Protection Program Help Small Businesses? Evidence from Commercial Mortgage-Backed Securities & M & M &  &  &  &  &  &  \\
\rule{0pt}{2.4em}Pol-24 & \href{https://doi.org/10.1257/pol.20200734}{\nolinkurl{10.1257/pol.20200734}} & Does Health-Care Consolidation Harm Patients? Evidence from Maternity Ward Closures & M &  &  &  &  & OA &  & M \\
\rule{0pt}{2.4em}Pol-24 & \href{https://doi.org/10.1257/pol.20220530}{\nolinkurl{10.1257/pol.20220530}} & Downward Revision of Investment Decisions after Corporate Tax Hikes & M & OA &  &  & OA &  &  &  \\
\rule{0pt}{2.4em}Pol-24 & \href{https://doi.org/10.1257/pol.20200587}{\nolinkurl{10.1257/pol.20200587}} & Fighting Crime in Lawless Areas: Evidence from Slums in Rio de Janeiro & M &  &  &  &  & OA &  &  \\
\rule{0pt}{2.4em}Pol-24 & \href{https://doi.org/10.1257/pol.20200579}{\nolinkurl{10.1257/pol.20200579}} & Hassle Costs versus Information: How Do Prescription Drug Monitoring Programs Reduce Opioid Prescribing? & OA &  &  &  &  &  & M &  \\
\rule{0pt}{2.4em}Pol-24 & \href{https://doi.org/10.1257/pol.20220355}{\nolinkurl{10.1257/pol.20220355}} & How Do Copayment Coupons Affect Branded Drug Prices and Quantities Purchased? &  &  &  & M &  &  &  &  \\
\rule{0pt}{2.4em}Pol-24 & \href{https://doi.org/10.1257/pol.20220357}{\nolinkurl{10.1257/pol.20220357}} & Indirect Savings from Public Procurement Centralization & M &  &  &  &  & OA &  &  \\
\rule{0pt}{2.4em}Pol-24 & \href{https://doi.org/10.1257/pol.20210729}{\nolinkurl{10.1257/pol.20210729}} & Killing Prescriptions Softly: Low Emission Zones and Child Health from Birth to School & OA &  &  & M &  &  &  &  \\
\rule{0pt}{2.4em}Pol-24 & \href{https://doi.org/10.1257/pol.20220732}{\nolinkurl{10.1257/pol.20220732}} & Managers' Productivity and Recruitment in the Public Sector & M &  & M &  &  & M &  & M+OA \\
\rule{0pt}{2.4em}Pol-24 & \href{https://doi.org/10.1257/pol.20210667}{\nolinkurl{10.1257/pol.20210667}} & Mandatory Retirement for Judges Improved the Performance of US State Supreme Courts & M &  &  &  &  &  & OA & M \\
\rule{0pt}{2.4em}Pol-24 & \href{https://doi.org/10.1257/pol.20220246}{\nolinkurl{10.1257/pol.20220246}} & The Effects of a Large-Scale Mental Health Reform: Evidence from Brazil &  &  &  &  &  & M &  &  \\
\rule{0pt}{2.4em}Pol-24 & \href{https://doi.org/10.1257/pol.20220208}{\nolinkurl{10.1257/pol.20220208}} & The Impact of Legal Abortion on Maternal Mortality & M & M &  &  &  &  &  &  \\
\rule{0pt}{2.4em}Pol-24 & \href{https://doi.org/10.1257/pol.20210388}{\nolinkurl{10.1257/pol.20210388}} & The Introduction of the Income Tax, Fiscal Capacity, and Migration: Evidence from US States &  & OA &  & M &  &  &  &  \\
\rule{0pt}{2.4em}Pol-24 & \href{https://doi.org/10.1257/pol.20210300}{\nolinkurl{10.1257/pol.20210300}} & The Returns to Public Library Investment & OA &  & OA &  &  &  & M &  \\
\rule{0pt}{2.4em}Pol-24 & \href{https://doi.org/10.1257/pol.20200694}{\nolinkurl{10.1257/pol.20200694}} & The Returns to STEM Programs for Less-Prepared Students &  &  &  & M &  &  &  &  \\
\addlinespace[0.35em]
\rule{0pt}{2.4em}App-25 & \href{https://doi.org/10.1257/app.20230468}{\nolinkurl{10.1257/app.20230468}} & Down to the Wire: Leveraging Technology to Improve Electric Utility Cost Recovery & M & OA & OA &  &  &  &  &  \\
\rule{0pt}{2.4em}App-25 & \href{https://doi.org/10.1257/app.20230377}{\nolinkurl{10.1257/app.20230377}} & Labor Supply and Entertainment Innovations: Evidence from the US TV Rollout & M &  &  &  &  &  &  &  \\
\rule{0pt}{2.4em}App-25 & \href{https://doi.org/10.1257/app.20230195}{\nolinkurl{10.1257/app.20230195}} & Skill Depreciation during Unemployment: Evidence from Panel Data & M &  &  &  & OA &  &  &  \\
\addlinespace[0.35em]
\rule{0pt}{2.4em}App-24 & \href{https://doi.org/10.1257/app.20230344}{\nolinkurl{10.1257/app.20230344}} & Better Alone? Evidence on the Costs of Intermunicipal Cooperation & M &  &  &  &  & M &  &  \\
\rule{0pt}{2.4em}App-24 & \href{https://doi.org/10.1257/app.20210117}{\nolinkurl{10.1257/app.20210117}} & Changes in Family Structure and Welfare Participation since the 1960s: The Role of Legal Services &  &  & M &  &  &  & M &  \\
\rule{0pt}{2.4em}App-24 & \href{https://doi.org/10.1257/app.20220572}{\nolinkurl{10.1257/app.20220572}} & Customer Capital Spillovers: Evidence from Sales Managers in International Markets & OA &  &  &  & M &  &  &  \\
\rule{0pt}{2.4em}App-24 & \href{https://doi.org/10.1257/app.20220341}{\nolinkurl{10.1257/app.20220341}} & Health Care Centralization: The Health Impacts of Obstetric Unit Closures in the United States & M &  &  &  & OA & M &  & OA \\
\rule{0pt}{2.4em}App-24 & \href{https://doi.org/10.1257/app.20210264}{\nolinkurl{10.1257/app.20210264}} & Patient versus Provider Incentives in Long-Term Care & M &  &  &  &  &  & M &  \\
\rule{0pt}{2.4em}App-24 & \href{https://doi.org/10.1257/app.20230089}{\nolinkurl{10.1257/app.20230089}} & Staggered Difference-in-Differences in Gravity Settings: Revisiting the Effects of Trade Agreements & M &  &  &  & OA &  & M &  \\
\rule{0pt}{2.4em}App-24 & \href{https://doi.org/10.1257/app.20180487}{\nolinkurl{10.1257/app.20180487}} & The Effects of Roads on Trade and Migration: Evidence from a Planned Capital City & M &  &  &  & OA &  &  &  \\
\rule{0pt}{2.4em}App-24 & \href{https://doi.org/10.1257/app.20230046}{\nolinkurl{10.1257/app.20230046}} & The Paradox of Innovation Nondisclosure: Evidence from Licensing Contracts & M & M & M &  & M &  & M(2) &  \\
\rule{0pt}{2.4em}App-24 & \href{https://doi.org/10.1257/app.20230376}{\nolinkurl{10.1257/app.20230376}} & Working Remotely? Selection, Treatment, and the Market for Remote Work &  &  &  & M &  &  &  &  \\
\addlinespace[0.35em]
\rule{0pt}{2.4em}AER-25 & \href{https://doi.org/10.1257/aer.20220528}{\nolinkurl{10.1257/aer.20220528}} & Corruption as a Local Advantage: Evidence from the Indigenization of Nigerian Oil & M &  & OA & M &  &  &  & OA \\
\rule{0pt}{2.4em}AER-25 & \href{https://doi.org/10.1257/aer.20220276}{\nolinkurl{10.1257/aer.20220276}} & Does Entry Remedy Collusion? Evidence from the Generic Prescription Drug Cartel & M & OA &  &  &  &  &  &  \\
\rule{0pt}{2.4em}AER-25 & \href{https://doi.org/10.1257/aer.20230971}{\nolinkurl{10.1257/aer.20230971}} & From Retributive to Restorative: An Alternative Approach to Justice in Schools & OA &  &  &  &  & M &  &  \\
\rule{0pt}{2.4em}AER-25 & \href{https://doi.org/10.1257/aer.20220540}{\nolinkurl{10.1257/aer.20220540}} & Moonshot: Public R\&D and Growth & M &  &  &  &  & OA &  &  \\
\rule{0pt}{2.4em}AER-25 & \href{https://doi.org/10.1257/aer.115.5.1369}{\nolinkurl{10.1257/aer.115.5.1369}} & Presidential Address: Investing in Children to Address the Child Mental Health Crisis &  &  & M &  &  &  &  &  \\
\rule{0pt}{2.4em}AER-25 & \href{https://doi.org/10.1257/aer.20230075}{\nolinkurl{10.1257/aer.20230075}} & The Benefits of Revealing Race: Evidence from Minority-Owned Local Businesses & M &  & OA & OA &  &  & OA & M \\
\rule{0pt}{2.4em}AER-25 & \href{https://doi.org/10.1257/aer.20221561}{\nolinkurl{10.1257/aer.20221561}} & Undergraduate Gender Diversity and the Direction of Scientific Research & M & M &  &  & OA &  &  & OA \\
\addlinespace[0.35em]
\rule{0pt}{2.4em}AER-24 & \href{https://doi.org/10.1257/aer.20230328}{\nolinkurl{10.1257/aer.20230328}} & Dying or Lying? For-Profit Hospices and End-of-Life Care & M & OA &  &  &  &  &  &  \\
\rule{0pt}{2.4em}AER-24 & \href{https://doi.org/10.1257/aer.20230585}{\nolinkurl{10.1257/aer.20230585}} & Financial Access and Labor Market Outcomes: Evidence from Credit Lotteries & M & M &  &  &  &  &  &  \\
\rule{0pt}{2.4em}AER-24 & \href{https://doi.org/10.1257/aer.20201952}{\nolinkurl{10.1257/aer.20201952}} & Financial Technology Adoption: Network Externalities of Cashless Payments in Mexico & M &  &  &  &  & OA &  &  \\
\rule{0pt}{2.4em}AER-24 & \href{https://doi.org/10.1257/aer.20200956}{\nolinkurl{10.1257/aer.20200956}} & From Fog to Smog: The Value of Pollution Information & M &  &  &  &  &  &  & OA \\
\rule{0pt}{2.4em}AER-24 & \href{https://doi.org/10.1257/aer.20220822}{\nolinkurl{10.1257/aer.20220822}} & In-Kind Transfers as Insurance & OA &  &  &  & M &  &  & OA \\
\rule{0pt}{2.4em}AER-24 & \href{https://doi.org/10.1257/aer.20221705}{\nolinkurl{10.1257/aer.20221705}} & Merchants of Death: The Effect of Credit Supply Shocks on Hospital Outcomes & M &  & OA &  &  &  &  &  \\
\rule{0pt}{2.4em}AER-24 & \href{https://doi.org/10.1257/aer.20221687}{\nolinkurl{10.1257/aer.20221687}} & Organized Crime and Economic Growth: Evidence from Municipalities Infiltrated by the Mafia &  &  &  &  &  &  & M &  \\
\rule{0pt}{2.4em}AER-24 & \href{https://doi.org/10.1257/aer.20230019}{\nolinkurl{10.1257/aer.20230019}} & Strengthening State Capacity: Civil Service Reform and Public Sector Performance during the Gilded Age &  & OA & OA & M &  & OA &  &  \\
\rule{0pt}{2.4em}AER-24 & \href{https://doi.org/10.1257/aer.20211547}{\nolinkurl{10.1257/aer.20211547}} & The Economics of the Public Option: Evidence from Local Pharmaceutical Markets &  &  &  & M &  &  &  &  \\
\rule{0pt}{2.4em}AER-24 & \href{https://doi.org/10.1257/aer.20230008}{\nolinkurl{10.1257/aer.20230008}} & The Gift of a Lifetime: The Hospital, Modern Medicine, and Mortality & M &  & M & M &  &  & M &  \\
\end{longtable}
\vspace{-0.6em}
\noindent\begin{minipage}{\linewidth}
\scriptsize
\textit{Notes}: J./yr. abbreviates journal and publication year: Pol is AEJ: Policy and App is AEJ: Applied. M means that a method is reported in the published main text; it may also appear in the OA. OA means online-appendix-only reporting. A blank cell means that the method or category is not reported. For Other and Diag., a parenthetical number counts distinct items at that location: for example, M(2)+OA denotes two distinct main-text items and one OA item. Stk. denotes stacked DiD.
\end{minipage}
\endgroup
\end{landscape}

\section{Optimal Quantile-Unbiased Estimation}\label{arxiv2:app:sec:pfanzagl.optimality}
If one restricts attention to the class of quantile-unbiased estimators that depend on $\htheta$ through $\htheta'Q\htheta$, then the optimality statement for $\heta_{1-\beta}$ in Proposition \ref{arxiv2:prop:optimality} can be strengthened. To see this, first recall that since $\htheta \sim N(\theta, \Sigma)$,
\begin{align*}
    \htheta'Q\htheta = (\Sigma^{-1/2}\htheta)'A(\Sigma^{-1/2}\htheta) \sim \chi_{K-1}^{2}\paren{H(\theta)}, \quad \forall \theta.
\end{align*}
The family $\{\chi_{K-1}^{2}(\eta): \eta \geq 0\}$ of noncentral chi-squared distributions has monotone likelihood ratios in $\htheta'Q\htheta$. These distributions are continuous, so their CDFs satisfy (5.3.10) from \citet{pfanzagl1994parametric}. Consequently, the estimator $\heta_{1-\beta}$ satisfies properties 5.4.1(i)-(ii) of \citet{pfanzagl1994parametric} over the parameter space $\eta \in H = (0, \infty)$. By \citet[Theorem 5.4.3]{pfanzagl1994parametric}, $\heta_{1-\beta}$ is maximally concentrated around $H(\theta)$ when $\theta$ is heterogeneous. That is, for any other potentially randomized estimator $\Tilde{\eta}_{1-\beta}$ that satisfies \eqref{arxiv2:eq:exact.upperCI.heterogeneity} and depends on $\htheta$ through $\htheta'Q\htheta$,
\begin{align*}
    \P{\underline{\eta} \leq \heta_{1-\beta} \leq \Bar{\eta}} \geq \P{\underline{\eta} \leq \Tilde{\eta}_{1-\beta} \leq \Bar{\eta}}, \quad \forall (\underline{\eta}, \Bar{\eta}, \theta):  0 \leq \underline{\eta} \leq H(\theta) \leq \Bar{\eta} \leq \infty, H(\theta) > 0.
\end{align*}
Considering $\underline{\eta} = 0$ and $\Bar{\eta} = H(\theta) + \e$ yields the statement that $\heta_{1-\beta}$ is uniformly most accurate in the class of quantile-unbiased estimators that depend on $\htheta$ through $\htheta'Q\htheta$.

By \citet[Proposition 2.5.3]{pfanzagl1994parametric}, the above maximal concentration is equivalent to minimum-risk optimality for every loss function $L(a, \theta)$ that attains its minimum at $a = H(\theta)$ and that is quasiconvex in $a$ for all heterogeneous $\theta$:
\begin{align*}
    \E{L(\heta_{1-\beta}, \theta)} \leq \E{L(\Tilde{\eta}_{1-\beta}, \theta)}, \quad \forall \theta: H(\theta) > 0.
\end{align*}
The above loss functions formalize the notion that losses get larger as $a$ moves away from $H(\theta)$. For example, taking $L(a,\theta) = (H(\theta) - a)^{2}$ to be squared error loss, the above result shows that when $\theta$ is heterogeneous, $\heta_{1-\beta}$ is optimal under mean squared error for estimating $H(\theta)$ in the class of quantile-unbiased estimators that depend on $\htheta$ through $\htheta'Q\htheta$.

\section{Additional Coverage Properties}\label{arxiv2:app:sec:additional.coverage.properties}

\subsection{Simultaneous Coverage}\label{arxiv2:app:sec:simultaneous.coverage}
Proposition \ref{arxiv2:prop:robust.coverage} shows that $CI_{w}^{*}$ is robust in the sense of uniform coverage: for each $\l$-weighted estimand, $CI_{w}^{*}$ provides valid coverage. A more general notion of robust coverage is \textit{simultaneous coverage}, which entails coverage of the entire set of $\l$-weighted estimands: $\L(\theta) = \{\l'\theta: \l \in \L\}$. The following result shows that $CI_{w}^{*}$ also provides such coverage, but at a different confidence level for the two-sided case.

\begin{proposition}\label{arxiv2:app:prop:simultaneous.coverage}
$CI_{w}^{*}$ provides simultaneous coverage for $\L(\theta) = \{\l'\theta: \l \in \L\}$ at confidence level (i) $1-(2\alpha+\beta)$ in the two-sided case and (ii) $1-(\alpha+\beta)$ in the one-sided case. That is, 
\begin{align}\label{arxiv2:app:eq:simultaneous.coverage}
    \P{\L(\theta) \subseteq CI_{w}^{*}} = \P{\l'\theta \in CI_{w}^{*}, \forall \l \in \L} \geq 
    \begin{cases}
    \hfil 1-(2\alpha+\beta), & \textnormal{two-sided}, \\ 
    \hfil 1-(\alpha+\beta), & \textnormal{one-sided},
    \end{cases} \quad \forall \theta.
\end{align}
\end{proposition}

\begin{proof}
See Appendix \ref{arxiv2:app:proof:simultaneous.coverage}.
\end{proof}

Thus, $CI_{w}^{*}$ is robust in the sense of simultaneous coverage, with the caveat that inferences in the two-sided case come at a lower confidence level relative to uniform coverage. In settings where simultaneous inferences are more relevant than uniform ones, this additional cost may be worthwhile. To see how, I draw connections to the literature on partial identification by viewing $\L(\theta)$ as the identified set for an estimand with unknown weights in $\L$. From this perspective, uniform coverage \eqref{arxiv2:eq:robust.coverage} corresponds to coverage of points in the identified set, while simultaneous coverage \eqref{arxiv2:app:eq:simultaneous.coverage} corresponds to coverage of the identified set---for precise definitions in a general partial identification setup, see \citet[Section 4.3]{molinari2020microeconometrics}.

Following \citet{imbens2004confidence}, coverage of points in the identified set is relevant when there is a true estimand $\l_{0}'\theta$. For example, suppose that $\l_{0}'\theta$ represents the ATE for a target population that one has partial information about and $\L$ represents the class of populations consistent with that information. In such cases, the relevant error is the exclusion of $\l_{0}'\theta$ from $CI_{w}^{*}$, which occurs with probability $\P{\l_{0}'\theta \notin CI_{w}^{*}}$. Since $CI_{w}^{*}$ provides the uniform coverage in \eqref{arxiv2:eq:robust.coverage}, this error rate is uniformly bounded over the potential values of $\l_{0} \in \L$.

Following \citet{henry2012set}, coverage of the identified set is relevant when there is concern for robust decision-making. For instance, suppose $\L$ represents readers $\l$ who will use $CI_{w}^{*}$ to make downstream decisions and the researcher is concerned about the worst-case loss from those decisions across readers.\footnote{The role of confidence intervals for decision-making has been considered more recently in \citet{manski2021econometrics, andrews2025certified, ben2025safe, chernozhukov2025policy}.} In such cases, \citet{henry2012set} show that the relevant error is the exclusion of \textit{any} reader's estimand $\l'\theta$ from $CI_{w}^{*}$, which occurs with probability $\P{\L(\theta) \not \subseteq CI_{w}^{*}}$. The simultaneous coverage in \eqref{arxiv2:app:eq:simultaneous.coverage} bounds this error rate.

In summary, the relevant notion of coverage depends on the type of error that one wishes to mitigate. In my setting, uniform coverage appears most appropriate. For example, a setup with a partially known target population $\l_{0}$ from the researcher's perspective is consistent with the mapping of my framework to multisite experiments discussed in Section \ref{arxiv2:sec:alternative.weights}. Moreover, even in a decision-making setup where $\L$ is taken to represent readers, it is plausible that each reader is only concerned about the exclusion of their own estimand, in which case simultaneous coverage may be too strong. With these examples in mind, I therefore focus on uniform coverage as the relevant notion of robust coverage.

\subsection{Upper Bounding the Coverage Rate}\label{arxiv2:app:sec:coverage.upper.bound}
Proposition \ref{arxiv2:prop:robust.coverage} shows that the coverage rate of $CI_{w}^{*}$ for $\taul$ is bounded below:
\begin{align*}
    \P{\taul \in CI_{w}^{*}} \geq 1-(\alpha+\beta), \quad \forall \l \in \L, \quad \forall \theta.
\end{align*}
A follow-up question is whether the coverage rate of $CI_{w}^{*}$ is nontrivial in the sense that there exists $(\l, \theta) \in \L \times \R^{K}$ where the left-hand side is strictly less than one. Below I show that this is indeed the case. To this end, define the function
\begin{align*}
    \delta_{w,\L}^{\eta}(\alpha,\beta) = \Phi\paren{-\frac{z_{1-\alpha}\sigma_{w} + \eta\paren{\displaystyle\max_{\l \in \L}\norm{\l - w_{\GLS}}_{\Sigma} + \max_{\l \in \L}\norm{\l - w}_{\Sigma}} + \sqrt{F_{\chi^{2}}^{-1}(\beta;\eta)}\norm{w - w_{\GLS}}_{\Sigma}}{\sigma_{\GLS}}}\beta.
\end{align*}
Let $\Bar{\Theta}^{\eta} = \{\theta: H(\theta) = \eta\}$ denote the set of parameters where the heterogeneity is equal to $\eta$. 
\begin{proposition}\label{arxiv2:app:prop:nontrivial.coverage}
For each $\eta$, the coverage rate of $CI_{w}^{*}$ for $\taul$ is bounded above as
\begin{align*}
    \P{\taul \in CI_{w}^{*}} 
    \leq \begin{cases}
    \hfil 1-2\delta_{w,\L}^{\eta}(\alpha/2,\beta), & \textnormal{two-sided}, \\ 
    \hfil 1-\delta_{w,\L}^{\eta}(\alpha,\beta), & \textnormal{one-sided},
    \end{cases} \quad \forall \l \in \L, \quad \forall \theta \in \Bar{\Theta}^{\eta}.
\end{align*}
\end{proposition}

\begin{proof}
See Appendix \ref{arxiv2:app:proof:nontrivial.coverage}.
\end{proof}

Proposition \ref{arxiv2:app:prop:nontrivial.coverage} implies that for any $\theta$, the coverage rate is strictly less than one---of course, this assumes that $\alpha,\beta > 0$, which I maintain throughout. The function $\delta_{w,\L}^{\eta}(\alpha,\beta)$ depends on $\eta$, so the bound is not uniform: $\delta_{w,\L}^{\eta}(\alpha,\beta) \to 0$ as $\eta \to \infty$ unless $\L = \{w\} = \{w_{\GLS}\}$. Under $w=w_{\GLS}$ and $\L = \{\l \in \W: \sigma_{\l} \leq r\sigma_{w}\}$, the function simplifies to $\delta_{w,\L}^{\eta}(\alpha,\beta) = \Phi(-z_{1-\alpha}-2\eta\sqrt{r^{2}-1})\beta$.

\section{Proofs of Results for the Normal Model}\label{arxiv2:app:sec:proofs}

\subsection{Proof of Proposition \ref{arxiv2:prop:bound}}\label{arxiv2:app:proof:bound}
Note that $\Sigma^{1/2}(\l - w) = A\Sigma^{1/2}(\l - w)$ and $A^{2} = A$. Thus, the Cauchy-Schwarz inequality yields
\begin{align*}
   |\theta'(\l - w)| &\leq \norm{A\Sigma^{-1/2}\theta}\norm{\Sigma^{1/2}(\l - w)} = \sqrt{\theta'Q\theta}\sqrt{(\l - w)'\Sigma(\l - w)} = H(\theta)\norm{\l - w}_{\Sigma},
\end{align*}
with equality if and only if $A\Sigma^{-1/2}\theta$ and $\Sigma^{1/2}(\l - w)$ are linearly dependent. For example, when $\l \neq w$, this occurs for $\theta^{\dagger} = \eta\Sigma(\l - w)/\norm{\l - w}_{\Sigma}$, which satisfies $H(\theta^{\dagger}) = \eta$. 

\subsection{Proof of Proposition \ref{arxiv2:prop:eta.validity}}\label{arxiv2:app:proof:eta.validity}
For $\theta$ such that $H(\theta) > 0$, observe that (i) the definition of $\heta_{1-\beta}$ yields 
\begin{align*}
    \curly{H(\theta) \leq \heta_{1-\beta}} &\subseteq \curly{H(\theta) \leq \heta_{1-\beta}, F_{\chi^{2}}(\htheta'Q\htheta; 0) > \beta} = \curly{F_{\chi^{2}}(\htheta'Q\htheta;H(\theta)) \geq \beta, F_{\chi^{2}}(\htheta'Q\htheta; 0) > \beta},
\end{align*}
and (ii) the monotonicity of $\eta \mapsto F_{\chi^{2}}(\htheta'Q\htheta;\eta)$ yields 
\begin{align*}
    \curly{F_{\chi^{2}}(\htheta'Q\htheta;H(\theta)) \geq \beta} \subseteq \curly{F_{\chi^{2}}(\htheta'Q\htheta;H(\theta)) \geq \beta, F_{\chi^{2}}(\htheta'Q\htheta; 0) > \beta}.
\end{align*}
But since the reverse of the set inclusions in (i) and (ii) hold for any $\theta$, then
\begin{align*}
    \P{H(\theta) \leq \heta_{1-\beta}} 
    = \P{F_{\chi^{2}}(\htheta'Q\htheta;H(\theta)) \geq \beta} = 1-\beta, \quad \forall \theta: H(\theta) > 0,
\end{align*}
where the second equality holds for any $\theta$, by the probability integral transform.

\subsection{Proof of Proposition \ref{arxiv2:prop:optimality}}\label{arxiv2:app:prop:optimality}
For a given $(\eta, \e) > 0$, consider the problem of testing null hypothesis $\theta \in \Theta_{0} = \{\theta: H(\theta) = \eta + \e\}$ against alternative hypothesis $\theta \in \Theta_{1} = \{\theta: H(\theta) = \eta\}$ in the model $\{N(\theta, \Sigma): \theta \in \Theta_{0} \cup \Theta_{1}\}$. For the class of potentially randomized tests $\Tilde{T}$ with size 
\begin{align*}
    \sup_{\theta \in \Theta_{0}}\E{\Tilde{T}} = \beta,
\end{align*}
a maximin test $\hat{T}$ is one that satisfies, for any other test $\Tilde{T}$ with size $\beta$, the inequality
\begin{align*}
    \inf_{\theta \in \Theta_{1}}\E{\hat{T}} \geq \inf_{\theta \in \Theta_{1}}\E{\Tilde{T}}.
\end{align*}
This testing structure follows the setup in \citet[Section 8]{lehmann2024testing}.

By quantile-unbiasedness, the tests $\Tilde{T} = \I{\eta + \e > \Tilde{\eta}_{1-\beta}}$ and $\hat{T} = \I{\eta + \e > \hat{\eta}_{1-\beta}}$ have size $\beta$. Note further that $\Theta_{1} = \Bar{\Theta}^{\eta}$. Thus, if one can show that $\hat{T} = \I{\eta + \e > \hat{\eta}_{1-\beta}}$ is a maximin test in the above sense for any given $(\eta, \e) > 0$, it would then follow that
\begin{align*}
    \sup_{\theta \in \Bar{\Theta}^{\eta}}\P{\hat{\eta}_{1-\beta} \geq H(\theta) + \e} \leq \sup_{\theta \in \Bar{\Theta}^{\eta}}\P{\Tilde{\eta}_{1-\beta} \geq H(\theta) + \e}, \quad \forall (\eta, \e) > 0,
\end{align*}
for any quantile-unbiased $\Tilde{\eta}_{1-\beta}$, which would be the desired conclusion.

To solve the maximin problem, I now appeal to invariance arguments. Formally, the above testing problems are invariant to the group of transformations 
\begin{align*}
    G = \curly{\theta \mapsto M\theta + c\1: c \in \R, M\1=\1, M'\Sigma^{-1}M = \Sigma^{-1}}.
\end{align*}
Indeed, $M\htheta + c\1 \sim N(M\theta + c\1, \Sigma)$ with $M^{-1} = \Sigma M' \Sigma^{-1}$ and $H(M\theta + c\1) = H(\theta)$ for the above $(c,M)$. Thus, for maximizing worst-case power in the class of size $\beta$ tests, it follows from the Hunt-Stein theorem \citep[Theorem 8.5.1]{lehmann2024testing} and \citet[Lemma 8.4.1]{lehmann2024testing} that one can restrict attention to invariant tests. In this setup, an invariant test depends on $\htheta$ through the maximal invariant $\htheta'Q\htheta$. Thus, below I restrict attention to the class of size $\beta$ tests that depend on $\htheta$ through $\htheta'Q\htheta$.

Now considering tests that depend on $\htheta$ through $\htheta'Q\htheta$, recall that
\begin{align*}
    \htheta'Q\htheta = (\Sigma^{-1/2}\htheta)'A(\Sigma^{-1/2}\htheta) \sim \chi_{K-1}^{2}\paren{H(\theta)}, \quad \forall \theta.
\end{align*}
The family $\{\chi_{K-1}^{2}(\eta): \eta \geq 0\}$ of noncentral chi-squared distributions has monotone likelihood ratios in $\htheta'Q\htheta$. Combined with the observation that $\hat{T} = \I{\eta + \e > \heta_{1-\beta}}$ rejects for small values of $\htheta'Q\htheta$, the Neyman-Pearson lemma implies that $\hat{T}$ is most powerful for testing $\chi_{K-1}^{2}(\eta + \e)$ against $\chi_{K-1}^{2}(\eta)$. In particular, for any $\Tilde{T} = \I{\eta + \e > \Tilde{\eta}_{1-\beta}}$ based on quantile-unbiased $\Tilde{\eta}_{1-\beta}$,
\begin{align*}
    \P{\eta + \e > \heta_{1-\beta}} = \E{\hat{T}} &\geq \E{\Tilde{T}} = \P{\eta + \e > \Tilde{\eta}_{1-\beta}}, \quad \forall \theta: H(\theta) = \eta.
\end{align*}
This establishes the optimality of $\hat{T} = \I{\eta + \e > \heta_{1-\beta}}$ in the maximin problem, from which the optimality result for $\heta_{1-\beta}$ follows.

\subsection{Proof of Proposition \ref{arxiv2:prop:eta.to.UCB}}\label{arxiv2:app:proof:eta.to.UCB}
Observe that
\begin{align*}
    \curly{H(\theta) \leq \heta_{1-\beta}} \subseteq \curly{H(\theta)\max_{\l \in \L}\norm{\l - w}_{\Sigma} \leq \hB_{w}^{\beta}(\L)} \subseteq \curly{\max_{\l \in \L}|(\l - w)'\theta| \leq \hB_{w}^{\beta}(\L)},
\end{align*}
where the latter set inclusion follows from Proposition \ref{arxiv2:prop:bound} and the former inclusion holds with equality when $\L \neq \{w\}$. Proposition \ref{arxiv2:prop:eta.validity} shows $\P{H(\theta) \leq \heta_{1-\beta}} \geq 1-\beta$ for all $\theta$, with equality when $H(\theta) > 0$. Together with the above expressions, this yields the desired coverage bound. Now consider $\theta$ where $H(\theta) > 0$ and suppose that $\L \neq \{w\}$. Then the above implies
\begin{align*}
    \P{H(\theta)\max_{\l \in \L}\norm{\l - w}_{\Sigma} \leq \hB_{w}^{\beta}(\L)} = 1-\beta, \quad \forall \theta: H(\theta) > 0.
\end{align*}
Moreover, taking $\l^{*} \in \arg\max_{\l \in \L}\norm{\l - w}_{\Sigma}$ and $\theta^{\dagger} = \Sigma(\l^{*} - w)/\norm{\l^{*} - w}_{\Sigma}$ yields $H(\theta^{\dagger}) = 1$ and 
\begin{align*}
    H(\theta^{\dagger})\max_{\l \in \L}\norm{\l - w}_{\Sigma} = \norm{\l^{*} - w}_{\Sigma} = |(\l^{*} - w)'\theta^{\dagger}| \leq \max_{\l \in \L}|(\l - w)'\theta^{\dagger}| \leq H(\theta^{\dagger})\max_{\l \in \L}\norm{\l - w}_{\Sigma}.
\end{align*}
Thus, the desired coverage probability is equal to $1-\beta$ under $\theta^{\dagger} = \Sigma(\l^{*} - w)/\norm{\l^{*} - w}_{\Sigma}$.

\subsection{Proof of Proposition \ref{arxiv2:prop:optimal.weights}}\label{arxiv2:app:proof:optimal.weights}
To show existence and uniqueness, it suffices to show that
\begin{align*}
    w^{*} = \arg\min_{\Bar{w} \in \W^{*}} f(\Bar{w}), \quad f(\Bar{w}) = \max_{\l \in \L}f_{\l}(\Bar{w}), \quad f_{\l}(\Bar{w}) = (\l - \Bar{w})'\Sigma (\l - \Bar{w}).
\end{align*}
Each $f_{\l}$ has Hessian $2\Sigma$, and hence is strongly convex since $\Sigma$ is positive definite. Because this Hessian does not depend on $\l$, the pointwise maximum $f$ is also strongly convex on the convex set $\W^{*}$. Moreover, $f$ is (i) continuous by compactness of $\L$ and continuity of $(\l,\Bar{w}) \mapsto f_{\l}(\Bar{w})$ and (ii) coercive since for any fixed $\l_{0} \in \L$,
\begin{align*}
    f(\Bar{w}) \geq f_{\l_{0}}(\Bar{w}) \geq \norm{\l_{0} - \Bar{w}}_{\Sigma}^{2}
    \to \infty, \quad \norm{\Bar{w}} \to \infty.
\end{align*}
Hence, $f$ attains its minimum on the nonempty closed set $\W^{*}$.
Strong convexity then implies that this minimizer $w^{*}$ is unique.

I now establish the bias-optimality of $w^{*}$. For each $\Bar{w} \in \W^{*}$, Proposition \ref{arxiv2:prop:bound} implies
\begin{align*}
    \max_{\l \in \L} \max_{\theta \in \Theta^{\eta}} \abs{\E{\htau_{\Bar{w}}} - \taul} = \max_{\l \in \L}\max_{\theta \in \Theta^{\eta}}|(\l - \Bar{w})'\theta| \leq \max_{\l \in \L}\max_{\theta \in \Theta^{\eta}}H(\theta)\norm{\l - \Bar{w}}_{\Sigma} \leq \eta \max_{\l \in \L}\norm{\l - \Bar{w}}_{\Sigma}.
\end{align*}
If the above inequalities become equalities, then bias-optimality follows from the definition of $w^{*}$. It suffices to show that the inequalities become equalities in the case of $\max_{\l \in \L}\norm{\l - \Bar{w}}_{\Sigma} > 0$. For $\l^{*} \in \arg\max_{\l \in \L}\norm{\l - \Bar{w}}_{\Sigma}$ and $\theta^{\dagger} = \eta \Sigma (\l^{*} - \Bar{w})/\norm{\l^{*} - \Bar{w}}_{\Sigma}$, observe that 
\begin{align*}
     H(\theta^{\dagger}) = \eta, \quad \abs{(\l^{*} - \Bar{w})'\theta^{\dagger}} = \eta \norm{\l^{*} - \Bar{w}}_{\Sigma} = \eta \max_{\l \in \L}\norm{\l - \Bar{w}}_{\Sigma}.
\end{align*}
Thus, since $(\l^{*}, \theta^{\dagger}) \in \L \times \Theta^{\eta}$,
\begin{align*}
    \max_{\l \in \L}\max_{\theta \in \Theta^{\eta}} |(\l - \Bar{w})'\theta| \geq |(\l^{*} - \Bar{w})'\theta^{\dagger}| = \eta \max_{\l \in \L}\norm{\l - \Bar{w}}_{\Sigma}.
\end{align*}
Together with the previous inequalities, this establishes the bias-optimality of $w^{*}$.

\subsection{Proof of Proposition \ref{arxiv2:prop:robust.coverage}}\label{arxiv2:app:proof:robust.coverage}
On $\cE_{\theta} = \curly{\max_{\l \in \L}|(\l - w)'\theta| \leq \hB_{w}^{\beta}(\L)}$, the noncoverage event for $\l \in \L$ is $\curly{\l'\theta \notin CI_{w}^{*}, \cE_{\theta}}$. Note $\P{\l'\theta \notin CI_{w}^{*}, \cE_{\theta}^{c}} \leq \P{\cE_{\theta}^{c}} \leq \beta$ for all $(\l, \theta)$, where the last inequality follows from Proposition \ref{arxiv2:prop:eta.to.UCB}. Thus, it suffices to show that $\P{\l'\theta \notin CI_{w}^{*}, \cE_{\theta}} \leq \alpha$ for all $\l \in \L$ and $\theta$, since then 
\begin{align*}
    \P{\l'\theta \notin CI_{w}^{*}} = \P{\l'\theta \notin CI_{w}^{*}, \cE_{\theta}} + \P{\l'\theta \notin CI_{w}^{*}, \cE_{\theta}^{c}} \leq \alpha+\beta, \quad \forall \l \in \L, \quad \forall \theta,
\end{align*}
which would yield the desired conclusion. 

For the two-sided $CI_{w}^{*}$, the noncoverage event for $\l \in \L$ on $\cE_{\theta}$ is
\begin{align*}
    \curly{\abs{\l'\theta - w'\htheta} > \cv{\hB_{w}^{\beta}(\L)/\sigma_{w}}\sigma_{w}, \cE_{\theta}} 
    &\subseteq \curly{\abs{\l'\theta - w'\htheta} > \cv{\max\nolimits_{\l \in \L}|(\l - w)'\theta|/\sigma_{w}}\sigma_{w}} \\
    &\subseteq \curly{\abs{\l'\theta - w'\htheta} > \cv{|(\l - w)'\theta|/\sigma_{w}}\sigma_{w}} \\
    &= \curly{\abs{\l'\theta - w'\htheta} > \cv{(\l - w)'\theta/\sigma_{w}}\sigma_{w}},
\end{align*}
where the first line follows from the definition of $\cE_{\theta}$, the second from $\cv{|b|}$ increasing in $|b|$, and the third from $\cv{|b|} = \cv{b}$. Thus, for all $\l \in \L$ and $\theta$,
\begin{align*}
    \P{\l'\theta \notin CI_{w}^{*}, \cE_{\theta}} \leq \P{\abs{\l'\theta - w'\htheta}/\sigma_{w} > \cv{(\l - w)'\theta/\sigma_{w}}} = \alpha.
\end{align*}
Now consider the upper one-sided $CI_{w}^{*}$---the proof for the lower case follows analogously. Using similar steps as above, the noncoverage probability is
\begin{align*}
    \P{\l'\theta \notin CI_{w}^{*}, \cE_{\theta}} &\leq \P{\l'\theta - w'\htheta > z_{1-\alpha}\sigma_{w} + (\l - w)'\theta} = \P{w'(\theta - \htheta)/\sigma_{w} > z_{1-\alpha}} = \alpha.
\end{align*}
In conclusion, $\P{\l'\theta \notin CI_{w}^{*}, \cE_{\theta}} \leq \alpha$ for all $\l \in \L$ and $\theta$, as desired.

\subsection{Proof of Proposition \ref{arxiv2:app:prop:simultaneous.coverage}}\label{arxiv2:app:proof:simultaneous.coverage}
For each $\theta$, denote
\begin{align*}
    \Bar{\l}_{\theta} \in \arg\max_{\l \in \L}\l'\theta, \quad \underline{\l}_{\theta} \in \arg\min_{\l \in \L}\l'\theta, \quad \cE_{\theta} = \curly{\max_{\l \in \L}|(\l - w)'\theta| \leq \hB_{w}^{\beta}(\L)},
\end{align*}
where $\Bar{\l}_{\theta}$ and $\underline{\l}_{\theta}$ exist given Assumption \ref{arxiv2:ass:alternative.weights} and the linearity of $\l \mapsto \l'\theta$. The noncoverage event on $\cE_{\theta}$ is $\curly{\L(\theta) \not \subseteq CI_{w}^{*}, \cE_{\theta}}$. Note $\P{\L(\theta) \not \subseteq CI_{w}^{*}, \cE_{\theta}^{c}} \leq \P{\cE_{\theta}^{c}} \leq \beta$ for all $\theta$. Thus, in the two-sided case, it suffices to show $\P{\L(\theta) \not \subseteq CI_{w}^{*}, \cE_{\theta}} \leq 2\alpha$ for all $\theta$, since then 
\begin{align*}
    \P{\L(\theta) \not\subseteq CI_{w}^{*}} = \P{\L(\theta) \not\subseteq CI_{w}^{*}, \cE_{\theta}} + \P{\L(\theta) \not\subseteq CI_{w}^{*}, \cE_{\theta}^{c}} \leq 2\alpha+\beta, \quad \forall \theta.
\end{align*}
Likewise, in the one-sided case, it suffices to show $\P{\L(\theta) \not \subseteq CI_{w}^{*}, \cE_{\theta}} \leq \alpha$ for all $\theta$. 

For the two-sided $CI_{w}^{*}$, the noncoverage probability on $\cE_{\theta}$ at the upper endpoint is
\begin{align*}
    \P{\max_{\l \in \L}\l'\theta > w'\htheta + \cv{\frac{\hB_{w}^{\beta}(\L)}{\sigma_{w}}}\sigma_{w}, \cE_{\theta}} &\leq \P{\frac{\Bar{\l}_{\theta}'\theta - w'\htheta}{\sigma_{w}} > \cv{\max_{\l \in \L}\frac{|(\l - w)'\theta|}{\sigma_{w}}}} \\
    &\leq \P{\frac{\Bar{\l}_{\theta}'\theta - w'\htheta }{\sigma_{w}} > \cv{\frac{|(\Bar{\l}_{\theta} - w)'\theta|}{\sigma_{w}}}} \\
    &= \P{\frac{\Bar{\l}_{\theta}'\theta - w'\htheta }{\sigma_{w}} > \cv{\frac{(\Bar{\l}_{\theta} - w)'\theta}{\sigma_{w}}}} \\
    &\leq \P{\abs{\frac{\Bar{\l}_{\theta}'\theta - w'\htheta }{\sigma_{w}}} > \cv{\frac{(\Bar{\l}_{\theta} - w)'\theta}{\sigma_{w}}}} \\ 
    &= \alpha,
\end{align*}
where the first line follows from the definition of $\cE_{\theta}$, the second from $\cv{|b|}$ increasing in $|b|$, and the third from $\cv{|b|} = \cv{b}$. Likewise, at the lower endpoint, 
\begin{align*}
    \P{\min_{\l \in \L}\l'\theta < w'\htheta - \cv{\frac{\hB_{w}^{\beta}(\L)}{\sigma_{w}}}\sigma_{w}, \cE_{\theta}} \leq \P{\abs{\frac{\underline{\l}_{\theta}'\theta - w'\htheta}{\sigma_{w}}} > \cv{\frac{(\underline{\l}_{\theta} - w)'\theta}{\sigma_{w}}}} = \alpha.
\end{align*}
Since $\curly{\L(\theta) \not \subseteq CI_{w}^{*}, \cE_{\theta}}$ is the union of the above two LHS events, a union bound yields
\begin{align*}
    \P{\L(\theta) \not \subseteq CI_{w}^{*}, \cE_{\theta}} \leq \alpha + \alpha = 2\alpha.
\end{align*}
Thus, in the two-sided case, $\P{\L(\theta) \not \subseteq CI_{w}^{*}, \cE_{\theta}} \leq 2\alpha$ for all $\theta$, as desired. 

Now consider the upper one-sided $CI_{w}^{*}$---the proof for the lower case is symmetric. Using similar steps as above, the probability of noncoverage is
\begin{align*}
    \P{\L(\theta) \not \subseteq CI_{w}^{*}, \cE_{\theta}} \leq \P{\frac{\Bar{\l}_{\theta}'\theta - w'\htheta}{\sigma_{w}} > z_{1-\alpha} + \frac{(\Bar{\l}_{\theta} - w)'\theta}{\sigma_{w}}} = \P{\frac{w'\theta - w'\htheta}{\sigma_{w}} > z_{1-\alpha}} = \alpha.
\end{align*}
Thus, in the one-sided case, $\P{\L(\theta) \not \subseteq CI_{w}^{*}, \cE_{\theta}} \leq \alpha$ for all $\theta$, as desired.

\subsection{Proof of Proposition \ref{arxiv2:app:prop:nontrivial.coverage}}\label{arxiv2:app:proof:nontrivial.coverage}
Fix $\eta$, $\theta \in \Bar{\Theta}^{\eta}$, and $\l \in \L$. Denote
\begin{align*}
    B_{w}^{\theta}(\L) = H(\theta)\max_{\l \in \L}\norm{\l - w}_{\Sigma}, \quad B_{\GLS}^{\theta}(\L) = H(\theta)\max_{\l \in \L}\norm{\l - w_{\GLS}}_{\Sigma}, \quad H(\theta) = \eta.
\end{align*}
In this proof, I make use of the following statements.
\begin{enumerate}[label=(\roman*)]
    \item $\{H(\theta) \geq \heta_{1-\beta}\} \supseteq \curly{H^{2}(\htheta) \leq F_{\chi^{2}}^{-1}(\beta;\eta)}$, which follows from
    \begin{align*}
        \curly{H(\theta) < \heta_{1-\beta}} = \curly{\eta < \heta_{1-\beta}} \subseteq \curly{F_{\chi^{2}}(H^{2}(\htheta);\eta) > \beta} = \curly{H^{2}(\htheta) > F_{\chi^{2}}^{-1}(\beta;\eta)},
    \end{align*}
    where the set inclusion follows from the definition of $\heta_{1-\beta}$.
    \item $w_{\GLS}'\htheta \ind H^{2}(\htheta)$, since $H^{2}(\htheta) = \norm{(I-\1 w_{\GLS}')\htheta}_{\Sigma^{-1}}^{2}$ and $w_{\GLS}'\htheta \ind (I-\1 w_{\GLS}')\htheta$. Indeed,
    \begin{align*}
        \begin{bmatrix}
            w_{\GLS}'\htheta \\
            (I-\1 w_{\GLS}')\htheta
        \end{bmatrix}
        \sim 
        N\paren{
        \begin{bmatrix}
        w_{\GLS}' \\
        I - \1 w_{\GLS}'
        \end{bmatrix}
        \theta,
        \begin{bmatrix}
        \sigma_{\GLS}^{2} & w_{\GLS}'\Sigma (I - w_{\GLS}\1') \\
        (I - \1 w_{\GLS}') \Sigma w_{\GLS} & (I - \1 w_{\GLS}')\Sigma(I - w_{\GLS}\1'),
        \end{bmatrix}
        },
    \end{align*}
    where $w_{\GLS}'\Sigma (I - w_{\GLS}\1') = (\1'\Sigma^{-1}\1)^{-1}\1' - \sigma_{\GLS}^{2}\1' = 0$, which yields independence.
    \item $\P{H^{2}(\htheta) \leq F_{\chi^{2}}^{-1}(\beta;\eta)} = \beta$, which follows from 
    \begin{align*}
        \P{H^{2}(\htheta) \leq F_{\chi^{2}}^{-1}(\beta;\eta)} = \P{F_{\chi^{2}}(H^{2}(\htheta);\eta) \leq \beta} = \beta, \quad H^{2}(\htheta) \sim \chi_{K-1}^{2}(\eta),
    \end{align*}
    where $F_{\chi^{2}}(H^{2}(\htheta);\eta) \sim U(0,1)$ follows from the probability integral transform.
\end{enumerate}
I begin by considering the upper one-sided $CI_{w}^{*}$. Observe that
\begin{align*}
    &\P{\l'\theta \notin CI_{w}^{*}} \\
    \geq &\P{\l'\theta > w'\htheta + z_{1-\alpha}\sigma_{w} + \hB_{w}^{\beta}(\L), H(\theta) \geq \heta_{1-\beta}} \\ 
    \geq &\P{\l'\theta > w'\htheta + z_{1-\alpha}\sigma_{w} + B_{w}^{\theta}(\L), H(\theta) \geq \heta_{1-\beta}} \\
    \geq &\P{\l'\theta > w'\htheta + z_{1-\alpha}\sigma_{w} + B_{w}^{\theta}(\L), H^{2}(\htheta) \leq F_{\chi^{2}}^{-1}(\beta;\eta)} \\
    = &\P{\l'\theta > w_{\GLS}'\htheta + z_{1-\alpha}\sigma_{w} + (w-w_{\GLS})'\htheta + B_{w}^{\theta}(\L), H^{2}(\htheta) \leq F_{\chi^{2}}^{-1}(\beta;\eta)} \\
    \geq &\P{\l'\theta > w_{\GLS}'\htheta + z_{1-\alpha}\sigma_{w} + H(\htheta)\norm{w - w_{\GLS}}_{\Sigma} + B_{w}^{\theta}(\L), H^{2}(\htheta) \leq F_{\chi^{2}}^{-1}(\beta;\eta)} \\
    \geq &\P{\l'\theta > w_{\GLS}'\htheta + z_{1-\alpha}\sigma_{w} + \sqrt{F_{\chi^{2}}^{-1}(\beta;\eta)}\norm{w - w_{\GLS}}_{\Sigma} + B_{w}^{\theta}(\L), H^{2}(\htheta) \leq F_{\chi^{2}}^{-1}(\beta;\eta)} \\
    = &\P{\l'\theta > w_{\GLS}'\htheta + z_{1-\alpha}\sigma_{w} + \sqrt{F_{\chi^{2}}^{-1}(\beta;\eta)}\norm{w - w_{\GLS}}_{\Sigma} + B_{w}^{\theta}(\L)}\P{H^{2}(\htheta) \leq F_{\chi^{2}}^{-1}(\beta;\eta)} \\
    \geq &\P{\l'\theta > w_{\GLS}'\htheta + z_{1-\alpha}\sigma_{w} + \sqrt{F_{\chi^{2}}^{-1}(\beta;\eta)}\norm{w - w_{\GLS}}_{\Sigma} + B_{w}^{\theta}(\L)}\beta \\
    = &\P{(\l-w_{\GLS})'\theta > w_{\GLS}'(\htheta-\theta) + z_{1-\alpha}\sigma_{w} + \sqrt{F_{\chi^{2}}^{-1}(\beta;\eta)}\norm{w - w_{\GLS}}_{\Sigma} + B_{w}^{\theta}(\L)}\beta \\
    \geq &\P{-B_{\GLS}^{\theta}(\L) > w_{\GLS}'(\htheta-\theta) + z_{1-\alpha}\sigma_{w} + \sqrt{F_{\chi^{2}}^{-1}(\beta;\eta)}\norm{w - w_{\GLS}}_{\Sigma} + B_{w}^{\theta}(\L)}\beta \\[10pt]
    = &\Phi\paren{-\frac{B_{\GLS}^{\theta}(\L) + z_{1-\alpha}\sigma_{w} + \sqrt{F_{\chi^{2}}^{-1}(\beta;\eta)}\norm{w - w_{\GLS}}_{\Sigma} + B_{w}^{\theta}(\L)}{\sigma_{\GLS}}}\beta
\end{align*}
where the third line follows from statement (i), the fifth from Proposition \ref{arxiv2:prop:bound}, the seventh from statement (ii), the eighth from statement (iii), and the tenth from Proposition \ref{arxiv2:prop:bound}.

A similar argument for the lower one-sided $CI_{w}^{*}$ yields
\begin{align*}
    &\P{\l'\theta \notin CI_{w}^{*}} \\
    \geq &\P{\l'\theta < w_{\GLS}'\htheta - z_{1-\alpha}\sigma_{w} - \sqrt{F_{\chi^{2}}^{-1}(\beta;\eta)}\norm{w - w_{\GLS}}_{\Sigma} - B_{w}^{\theta}(\L)}\beta \\
    = &\P{(\l-w_{\GLS})'\theta < w_{\GLS}'(\htheta-\theta) - z_{1-\alpha}\sigma_{w} - \sqrt{F_{\chi^{2}}^{-1}(\beta;\eta)}\norm{w - w_{\GLS}}_{\Sigma} - B_{w}^{\theta}(\L)}\beta \\
    \geq &\P{B_{\GLS}^{\theta}(\L) < w_{\GLS}'(\htheta-\theta) - z_{1-\alpha}\sigma_{w} - \sqrt{F_{\chi^{2}}^{-1}(\beta;\eta)}\norm{w - w_{\GLS}}_{\Sigma} - B_{w}^{\theta}(\L)}\beta \\[10pt]
    = &\Phi\paren{-\frac{B_{\GLS}^{\theta}(\L) + z_{1-\alpha}\sigma_{w} + \sqrt{F_{\chi^{2}}^{-1}(\beta;\eta)}\norm{w - w_{\GLS}}_{\Sigma} + B_{w}^{\theta}(\L)}{\sigma_{\GLS}}}\beta,
\end{align*}
which is the same final bound as the upper one-sided case.

Finally, for the case of a two-sided $CI_{w}^{*}$, observe that
\begin{align*}
    \P{\l'\theta \notin CI_{w}^{*}} &\geq \P{|\l'\theta-w'\htheta| > z_{1-\alpha/2}\sigma_{w} + \hB_{w}^{\beta}(\L)} \\
    &= \P{\l'\theta > w'\htheta + z_{1-\alpha/2}\sigma_{w} + \hB_{w}^{\beta}(\L)} + \P{\l'\theta < w'\htheta -z_{1-\alpha/2}\sigma_{w} - \hB_{w}^{\beta}(\L)} \\[10pt]
    &\geq 2\Phi\paren{-\frac{B_{\GLS}^{\theta}(\L) + z_{1-\alpha/2}\sigma_{w} + \sqrt{F_{\chi^{2}}^{-1}(\beta;\eta)}\norm{w - w_{\GLS}}_{\Sigma} + B_{w}^{\theta}(\L)}{\sigma_{\GLS}}}\beta,
\end{align*}
where the first line follows from the bound $\cv{|b|} \leq z_{1-\alpha/2} + |b|$, and the second from the above derivations in the one-sided case. Thus, the desired coverage bounds hold.

\section{Proofs of Results for the Uniform Asymptotics}

\subsection{Proof of Proposition \ref{arxiv2:prop:uniform.hausdorff.convergence}}\label{arxiv2:app:proof:uniform.hausdorff.convergence}
Under Assumptions \ref{arxiv2:ass:consistent.S} and \ref{arxiv2:ass:Lambda.representation}, there is an event whose probability uniformly approaches one as $n \to \infty$ on which $\hSn \in \bbS$ and Lemma \ref{arxiv2:app:lemma:hausdorff.continuity} implies
\begin{align*}
    d_{H}(\hLn, \Ln) \leq \frac{\diam(\bbW)}{\delta_{g}} L_{g}d_{\cS}(\hSn, \Sn).
\end{align*}
Thus, since $d_{\cS}(\hSn, \Sn)$ uniformly converges in probability to zero, so does $d_{H}(\hLn, \Ln)$.

\subsection{Proof of Proposition \ref{arxiv2:prop:consistent.distance}}\label{arxiv2:app:proof:consistent.distance}
Under Assumptions \ref{arxiv2:ass:consistent.covariance}-\ref{arxiv2:ass:Lambda.representation}, there is an event whose probability uniformly approaches one as $n \to \infty$ on which $\LB{e}/2 \leq \emin(\hSigman) \leq \emax(\hSigman) \leq 2\UB{e}$, $\hSn \in \bbS$, and Lemma \ref{arxiv2:app:lemma:perturbation.bounds} implies
\begin{align*}
    \abs{\max_{\l \in \hLn}\norm{\l - \hwn}_{\hSigman} - \max_{\l \in \Ln}\norm{\l - \wn}_{\Sigman}} &\leq \abs{\max_{\l \in \hLn}\norm{\l - \hwn}_{\hSigman} - \max_{\l \in \hLn}\norm{\l - \wn}_{\hSigman}} \\[5pt]
    &+ \abs{\max_{\l \in \hLn}\norm{\l - \wn}_{\hSigman} - \max_{\l \in \Ln}\norm{\l - \wn}_{\hSigman}} \\[5pt]
    &+ \abs{\max_{\l \in \Ln}\norm{\l - \wn}_{\hSigman} - \max_{\l \in \Ln}\norm{\l - \wn}_{\Sigman}} \\[5pt]
    &\leq \sqrt{2\UB{e}}\norm{\hwn - \wn} \\[5pt]
    &+ \sqrt{2\UB{e}}d_{H}(\hLn, \Ln) \\[5pt]
    &+ \frac{1}{\sqrt{2}\LB{e}}\norm{\hSigman - \Sigman} \max_{\l \in \Ln}\norm{\l - \wn}_{\Sigman},
\end{align*}
where $\max_{\l \in \Ln}\norm{\l - \wn}_{\Sigman} \leq \sqrt{\UB{e}}(\max_{\l \in \bbW}\norm{\l} + \UB{C}_{w}) < \infty$. Thus, since the above right-hand side (RHS) uniformly converges in probability to zero, so does the left-hand side (LHS).

\subsection{Proof of Proposition \ref{arxiv2:prop:consistent.minimax.weights}}\label{arxiv2:app:proof:consistent.minimax.weights}
The boundedness claim follows because $\rwn \in \Ln \subseteq \bbW$ is contained in the compact set $\bbW$ for all $\Pn \in \cPn$ and $n$ under the maintained assumptions. I now prove the consistency claim. Under Assumptions \ref{arxiv2:ass:consistent.covariance} and \ref{arxiv2:ass:consistent.S}, there is an event whose probability uniformly approaches one as $n \to \infty$ on which $\LB{e}/2 \leq \emin(\hSigman) \leq \emax(\hSigman) \leq 2\UB{e}$ and $\hSn \in \bbS$. The following arguments are made on this event under Assumptions \ref{arxiv2:ass:consistent.covariance}, \ref{arxiv2:ass:consistent.S}, and \ref{arxiv2:ass:Lambda.representation}. Define the Euclidean projections
\begin{align*}
    \pi_{n}(\rhwn) = \arg\min_{\l \in \Ln} \norm{\rhwn - \l}, \quad \hat{\pi}_{n}(\rwn) = \arg\min_{\l \in \hLn} \norm{\rwn - \l}.
\end{align*}
Since $\rhwn \in \hLn$ and $\rwn \in \Ln$, the definition of Hausdorff distance yields the bounds
\begin{align*}
    \norm{\rhwn - \pi_{n}(\rhwn)} = \dist(\rhwn, \Ln) \leq d_{H}(\hLn, \Ln), \quad \norm{\rwn - \hat{\pi}_{n}(\rwn)} = \dist(\rwn, \hLn) \leq d_{H}(\hLn, \Ln).
\end{align*}
By triangle inequality and Lemma \ref{arxiv2:app:lemma:minimax.quadratic.bounds},
\begin{align*}
    \norm{\rhwn - \rwn} &\leq \norm{\rhwn - \pi_{n}(\rhwn)} + \norm{\pi_{n}(\rhwn) - \rwn} \\
    &\leq d_{H}(\hLn, \Ln) + \sqrt{\frac{2\sqrt{\UB{e}}}{\LB{e}}\diam(\bbW)\paren{\max_{\l \in \Ln}\norm{\l - \pi_{n}(\rhwn)}_{\Sigman} - \max_{\l \in \Ln}\norm{\l - \rwn}_{\Sigman}}}.
\end{align*}
I now look to bound the second term. To this end, first note that Lemma \ref{arxiv2:app:lemma:perturbation.bounds} implies
\begin{align*}
    \max_{\l \in \Ln}\norm{\l - \pi_{n}(\rhwn)}_{\Sigman} \leq \max_{\l \in \Ln}\norm{\l - \rhwn}_{\Sigman} + \sqrt{\UB{e}} d_{H}(\hLn, \Ln).
\end{align*}
By Lemma \ref{arxiv2:app:lemma:perturbation.bounds} again,
\begin{align*}
    \max_{\l \in \Ln}\norm{\l - \rhwn}_{\Sigman} \leq \max_{\l \in \hLn}\norm{\l - \rhwn}_{\hSigman} + \sqrt{\UB{e}}d_{H}(\hLn, \Ln) + \frac{1}{\dsqrt{2\LB{e}}}\norm{\hSigman - \Sigman} \diam(\bbW).
\end{align*}
By optimality of $\rhwn$ and Lemma \ref{arxiv2:app:lemma:perturbation.bounds} again,
\begin{align*}
    \max_{\l \in \hLn}\norm{\l - \rhwn}_{\hSigman} &\leq \max_{\l \in \hLn}\norm{\l - \hat{\pi}_{n}(\rwn)}_{\hSigman} \\
    &\leq \max_{\l \in \Ln}\norm{\l - \hat{\pi}_{n}(\rwn)}_{\Sigman} + \sqrt{\UB{e}}d_{H}(\hLn, \Ln) + \frac{1}{\dsqrt{2\LB{e}}}\norm{\hSigman - \Sigman} \diam(\bbW).
\end{align*}
Note that Lemma \ref{arxiv2:app:lemma:perturbation.bounds} implies
\begin{align*}
    \max_{\l \in \Ln}\norm{\l - \hat{\pi}_{n}(\rwn)}_{\Sigman} \leq \max_{\l \in \Ln}\norm{\l - \rwn}_{\Sigman} + \sqrt{\UB{e}} d_{H}(\hLn, \Ln).
\end{align*}
Combining the above steps yields the inequality
\begin{align*}
    \max_{\l \in \Ln}\norm{\l - \pi_{n}(\rhwn)}_{\Sigman} - \max_{\l \in \Ln}\norm{\l - \rwn}_{\Sigman} \leq 4\sqrt{\UB{e}}d_{H}(\hLn, \Ln) + \sqrt{\frac{2}{\LB{e}}}\norm{\hSigman - \Sigman} \diam(\bbW).
\end{align*}
In summary, under Assumptions \ref{arxiv2:ass:consistent.covariance}, \ref{arxiv2:ass:consistent.S}, and \ref{arxiv2:ass:Lambda.representation}, there is an event whose probability uniformly approaches one as $n \to \infty$ on which for $\eSigman = \norm{\Sigman^{-1/2}(\hSigman - \Sigman)\Sigman^{-1/2}}$,
\begin{align*}
    \norm{\rhwn - \rwn} \leq d_{H}(\hLn, \Ln) + \sqrt{\frac{2\UB{e}}{\LB{e}}\diam(\bbW)\paren{4 d_{H}(\hLn, \Ln) + \sqrt{\frac{2\UB{e}}{\LB{e}}}\eSigman \diam(\bbW)}}.
\end{align*}
Thus, since the RHS uniformly converges in probability to zero, so does the LHS.

\subsection{Proof of Proposition \ref{arxiv2:prop:asymptotic.bias}}\label{arxiv2:app:proof:asymptotic.bias}
Let $\cthetan = \hthetan - \thetan$ and $\cwn = \hwn - \wn$. By the triangle and Cauchy-Schwarz inequalities,
\begin{align*}
    \abs{\E[\Pn]{\htauwn} - \tauwn[\l]} &\leq \abs{\E[\Pn]{\cwn'\cthetan}} + \abs{\E[\Pn]{\wn'\cthetan}} + \abs{\E[\Pn]{\cwn'\thetan}} + \abs{(\l - \wn)'\thetan} \\[5pt]
    &\lesssim \sqrt{\E[\Pn]{\norm{\cwn}^{2}}\E[\Pn]{\norm{\cthetan}^{2}}} + \sqrt{\E[\Pn]{\norm{\cthetan}^{2}}} + \sqrt{\E[\Pn]{\norm{\cwn}^{2}}} + \abs{(\l - \wn)'\thetan}.
\end{align*}
I show at the end of this proof that the UI condition \eqref{arxiv2:eq:UI.condition} implies
\begin{align}\label{arxiv2:app:eq:UI.condition.implication}
    \limn \supPn \paren{\E[\Pn]{\norm{\cthetan}^{2}}, \E[\Pn]{\norm{\cwn}^{2}}} = 0.
\end{align}
Combined with Proposition \ref{arxiv2:prop:bound}, this yields
\begin{align*}
    \limsupn \paren{\supPn \max_{\l \in \Ln} \abs{\E[\Pn]{\htauwn} - \tauwn[\l]}} \leq \limsupn \paren{\UB{\eta}(\cPn) \supPn \max_{\l \in \Ln} \norm{\l - \wn}_{\Sigman}},
\end{align*}
where the RHS is finite under Assumptions \ref{arxiv2:ass:linear.CLT}-\ref{arxiv2:ass:Lambda.representation}.

For the equality statement, observe that the reverse triangle inequality implies
\begin{align*}
    \max_{\l \in \Ln}\abs{\E[\Pn]{\htauwn} - \tauwn[\l]} &\geq \max_{\l \in \Ln}\abs{(\l - \wn)'\thetan} - \abs{\E[\Pn]{\cwn'\cthetan} + \E[\Pn]{\wn'\cthetan} + \E[\Pn]{\cwn'\thetan}} \\
    &\geq \max_{\l \in \Ln}\abs{(\l - \wn)'\thetan} - \supPn\paren{\abs{\E[\Pn]{\cwn'\cthetan}} + \abs{\E[\Pn]{\wn'\cthetan}} + \abs{\E[\Pn]{\cwn'\thetan}}},
\end{align*}
where \eqref{arxiv2:app:eq:UI.condition.implication} and the previous Cauchy-Schwarz bound implies
\begin{align*}
    \limn \supPn \paren{\abs{\E[\Pn]{\cwn'\cthetan}} + \abs{\E[\Pn]{\wn'\cthetan}} + \abs{\E[\Pn]{\cwn'\thetan}}} = 0.
\end{align*}
If $P_{n}^{\dagger} \in \cPn$ for all $\Pn \in \cPn$ and $n$, then for all $\Pn \in \cPn$ and $n$,
\begin{align*}
    \supPn \max_{\l \in \Ln}\abs{(\l - \wn)'\thetan} \geq \max_{\l \in \Ln}\abs{(\l - \wn)'\theta(P_{n}^{\dagger})} = \abs{(\ln^{*} - \wn)'\theta(P_{n}^{\dagger})} &= \UB{\eta}(\cPn)\norm{\ln^{*} - \wn}_{\Sigman} \\
    &= \UB{\eta}(\cPn)\max_{\l \in \Ln}\norm{\l - \wn}_{\Sigman}.
\end{align*}
Taking the supremum over $\Pn \in \cPn$ and combining with the previous displays yields
\begin{align*}
    \limsupn \paren{\supPn \max_{\l \in \Ln} \abs{\E[\Pn]{\htauwn} - \tauwn[\l]}} \geq \limsupn \paren{\UB{\eta}(\cPn) \supPn \max_{\l \in \Ln} \norm{\l - \wn}_{\Sigman}},
\end{align*}
from which the equality statement follows.

I now verify \eqref{arxiv2:app:eq:UI.condition.implication}. Note that $\cthetan$ and $\cwn$ are uniformly consistent for zero, given Assumptions \ref{arxiv2:ass:linear.CLT} and \ref{arxiv2:ass:consistent.weights}. Focusing on $\cwn$ for the moment, observe that for any $C > 0$,
\begin{align*}
    \E[\Pn]{\norm{\cwn}^{2}} \leq \abs{\E[\Pn]{\norm{\cwn}^{2}} - \E[\Pn]{\min\paren{\norm{\cwn}^{2}, C}}} + \abs{\E[\Pn]{\min\paren{\norm{\cwn}^{2}, C}}}.
\end{align*}
Since $x \mapsto \min(\norm{x}^{2}, C)$ is bounded and continuous, uniform convergence in probability of $\cwn$ to zero implies
\begin{align*}
    \limn \supPn \abs{\E[\Pn]{\min\paren{\norm{\cwn}^{2}, C}}} = 0.
\end{align*}
Since $\abs{\E[\Pn]{\norm{\cwn}^{2}} - \E[\Pn]{\min\paren{\norm{\cwn}^{2}, C}}} \leq \E[\Pn]{\norm{\cwn}^{2} \I{\norm{\cwn}^{2} > C}}$, then
\begin{align*}
    \limsupn \supPn \E[\Pn]{\norm{\cwn}^{2}} \leq \limsupn \supPn \E[\Pn]{\norm{\cwn}^{2} \I{\norm{\cwn}^{2} > C}}.
\end{align*}
The above arguments also apply to $\cthetan$. Taking $C \to \infty$, the UI condition \eqref{arxiv2:eq:UI.condition} implies
\begin{align*}
    \limsupn \supPn \paren{\E[\Pn]{\norm{\cwn}^{2}}, \E[\Pn]{\norm{\cthetan}^{2}}} \leq 0.
\end{align*}
In particular, \eqref{arxiv2:app:eq:UI.condition.implication} holds.

\subsection{Proof of Proposition \ref{arxiv2:prop:asymptotic.optimality}}\label{arxiv2:app:proof:asymptotic.optimality}
Under these conditions, Proposition \ref{arxiv2:prop:asymptotic.bias} can be applied to both $\rhwn$ and $\hwn$. In particular,
\begin{align*}
    \limsupn \paren{\supPn \max_{\l \in \Ln} \abs{\E[\Pn]{\rhtaun} - \tauwn[\l]}}
    &\leq \limsupn \paren{\UB{\eta}(\cPn) \supPn \max_{\l \in \Ln} \norm{\l - \rwn}_{\Sigman}} \\
    &\leq \limsupn \paren{\UB{\eta}(\cPn) \supPn \max_{\l \in \Ln} \norm{\l - \wn}_{\Sigman}} \\
    &= \limsupn \paren{\supPn \max_{\l \in \Ln} \abs{\E[\Pn]{\htauwn} - \tauwn[\l]}},
\end{align*}
where the second inequality follows from $\max_{\l \in \Ln} \norm{\l - \rwn}_{\Sigman} \leq \max_{\l \in \Ln} \norm{\l - \wn}_{\Sigman}$.

\subsection{Proof of Proposition \ref{arxiv2:prop:consistent.heterogeneity}}\label{arxiv2:app:proof:consistent.heterogeneity}
Under the uniform consistency and eigenvalue bounds in Assumption \ref{arxiv2:ass:consistent.covariance}, there exists an event whose probability uniformly approaches one as $n \to \infty$ on which $\hSigman$ and $\Sigman$ belong to a fixed compact subset of positive definite matrices. The following arguments are made on this event under Assumption \ref{arxiv2:ass:consistent.covariance}. On the compact subset of positive definite $\Sigma$, the map $\Sigma \mapsto A(\Sigma)\Sigma^{-1/2}$ is Lipschitz: there exists a constant $C > 0$ such that
\begin{align*}
    \norm{\hAn\hSigman^{-1/2} - \An\Sigman^{-1/2}} = \norm{A(\hSigman)\hSigman^{-1/2} - A(\Sigman)\Sigman^{-1/2}} \leq C\norm{\hSigman - \Sigman}.
\end{align*}
Let $\hat{v}_{n}(\thetan) = \hAn\hSigman^{-1/2}\thetan$ and $v_{n}(\thetan) = \An\Sigman^{-1/2}\thetan$. Observe that
\begin{align*}
    \hat{v}_{n}(\1) = v_{n}(\1) = 0, \quad v_{n}(\thetan) = \Sigman^{-1/2}(\thetan - \1\gamma_{n}(\thetan)), \quad \gamma_{n}(\thetan) = \frac{\1'\Sigman^{-1}\thetan}{\1'\Sigman^{-1}\1}.
\end{align*}
Therefore $\hat{v}_{n}(\thetan) - v_{n}(\thetan) = (\hAn\hSigman^{-1/2} - \An\Sigman^{-1/2})(\thetan - \1\gamma_{n}(\thetan))$, and the eigenvalue bounds give $\norm{\thetan - \1\gamma_{n}(\thetan)} \leq \UB{e}^{1/2}\norm{\Sigman^{-1/2}(\thetan - \1\gamma_{n}(\thetan))} = \UB{e}^{1/2}\norm{v_{n}(\thetan)}$. Combined with the above displays,
\begin{align*}
    \norm{\hat{v}_{n}(\thetan) - v_{n}(\thetan)} \leq \norm{\hAn\hSigman^{-1/2} - \An\Sigman^{-1/2}}\norm{\thetan - \1\gamma_{n}(\thetan)} \leq \UB{e}^{1/2}C\norm{\hSigman - \Sigman} \norm{v_{n}(\thetan)}.
\end{align*}
Thus, for each $\e > 0$, a union bound yields
\begin{align*}
    \supPn \P[\Pn]{\norm{\hat{v}_{n}(\thetan) - v_{n}(\thetan)} > \e\norm{v_{n}(\thetan)}} \leq \sup_{P_{n} \in \mathcal{P}_{n}} \P[\Pn]{\norm{\hSigman - \Sigman} > \frac{\e}{\UB{e}^{1/2}C}} + o(1).
\end{align*}
Since the RHS converges to zero under Assumption \ref{arxiv2:ass:consistent.covariance}, so does the LHS, which establishes the first claim. Consequently, since the reverse triangle inequality under $\Hn(\thetan) > 0$ gives
\begin{align*}
    \abs{\hHn(\thetan) - \Hn(\thetan)} = \abs{\norm{\hAn\hSigman^{-1/2}\thetan} - \norm{\An\Sigman^{-1/2}\thetan}} \leq \norm{\hat{v}_{n}(\thetan) - v_{n}(\thetan)},
\end{align*}
then the second claim follows from the first.

\subsection{Proof of Proposition \ref{arxiv2:prop:asymptotic.UCB}}\label{arxiv2:app:proof:asymptotic.UCB}
If \eqref{arxiv2:eq:UCB.criteria} does not hold, then there exists a $p > 0$, a subsequence $\{n_{s}\} \subseteq \{n\}$, and a corresponding sequence of distributions $P_{n_{s}} \in \mathcal{P}_{n_{s}}$ such that
\begin{align*}
    \abs{\P[P_{n_{s}}]{\hat{H}_{n_{s}}(\theta_{n_{s}}) \leq \hat{\eta}_{1-\beta,n_{s}}} - (1-\beta)}\I{H_{n_{s}}(\theta_{n_{s}}) > 0} > p, \quad \forall s.
\end{align*}
Thus, to prove Proposition \ref{arxiv2:prop:asymptotic.UCB}, it suffices to show that the above yields a contradiction under Assumptions \ref{arxiv2:ass:linear.CLT} and \ref{arxiv2:ass:consistent.covariance}. To do so, I will iteratively pass to (sub)subsequences of $\{n_{s}\}$ along which variables of interest converge under Assumptions \ref{arxiv2:ass:linear.CLT} and \ref{arxiv2:ass:consistent.covariance}. 

For conciseness, I will keep using the index $s$ when passing to subsequences, and will also use $s$ to abbreviate the subscript $n_{s}$. In this notation, the above becomes
\begin{align*}
    \abs{\P[\Pns]{\sqrt{n_{s}}\hHns(\thetans) \leq \tetans} - (1-\beta)}\I{\sqrt{n_{s}}\Hns(\thetans) > 0} > p, \quad \forall s.
\end{align*}
This implies $\Hns(\thetans) > 0$ and hence $\hHns(\thetans) > 0$ for all $s$. Thus, using similar arguments to those in Appendix \ref{arxiv2:app:proof:eta.validity}, 
\begin{align}\label{arxiv2:app:eq:UCB.criteria.contradiction}
    \abs{\P[\Pns]{F_{\chi^{2}}(n_{s}\hHns^{2}(\hthetans); \sqrt{n_{s}}\hHns(\thetans)) < \beta} - \beta} > p, \quad \forall s.
\end{align}
It suffices to derive a contradiction from \eqref{arxiv2:app:eq:UCB.criteria.contradiction}. Let $\hZns = \sqrt{n_{s}}(\hthetans - \thetans)$. Under Assumptions \ref{arxiv2:ass:linear.CLT} and \ref{arxiv2:ass:consistent.covariance}, I pass to a subsequence where
\begin{align*}
    \Sigmans \to \limSigma, \quad \norm{\limSigma} \in (0,\infty), \quad \hSigmans \to[p] \limSigma, \quad \hZns \to[d] \limZ \sim N(0, \limSigma).
\end{align*}
By continuous mapping theorem (CMT), analogous convergences hold for $(\hAns, \Ans, \hQns, \Qns)$ relative to $(\limA, \limQ)$, etc. Let $\hvns = \hAns\hSigmans^{-1/2}\sqrt{n_{s}}\thetans$ and $\vns = \Ans\Sigmans^{-1/2}\sqrt{n_{s}}\thetans$ so that $\sqrt{n_{s}}\hHns(\thetans) = \norm{\hvns}$ and $\sqrt{n_{s}}\Hns(\thetans) = \norm{\vns}$. By Proposition \ref{arxiv2:prop:consistent.heterogeneity}, 
\begin{align*}
    \frac{\norm{\hvns - \vns}}{\norm{\vns}} = \frac{\norm{\hAns \hSigmans^{-1/2} \thetans - \Ans \Sigmans^{-1/2} \thetans}}{\norm{\Ans \Sigmans^{-1/2} \thetans}} \to[p] 0, \quad \frac{\norm{\hvns}}{\norm{\vns}} = \frac{\hHns(\thetans)}{\Hns(\thetans)} \to[p] 1.
\end{align*}
By compactness, I pass to a further subsequence where
\begin{align*}
    \vns \to \limv \in [-\infty, \infty]^{K}, \quad \norm{\vns} = \sqrt{n_{s}}\Hns(\thetans) \to \limeta = \norm{\limv} \in [0, \infty], \quad \frac{\vns}{\norm{\vns}} \to \limu, \quad \norm{\limu} = 1. 
\end{align*}
To establish a contradiction of \eqref{arxiv2:app:eq:UCB.criteria.contradiction}, it suffices to establish contradictions in both the case of $\limeta < \infty$ and the case of $\limeta = \infty$.

\subsubsection{Case 1: Bounded Normalized Heterogeneity} 
For the case of $\limeta < \infty$, I appeal to CMT. First, note that $\hvns \to[p] \limv$. Indeed, 
\begin{align*}
    \norm{\hvns - \limv} \leq \frac{\norm{\hvns - \vns}}{\norm{\vns}}\norm{\vns} + \norm{\vns - \limv} = o_{p}(1)O(1) + o(1) = o_{p}(1).
\end{align*}
Given that $\hQns = \hSigmans^{-1/2}\hAns\hSigmans^{-1/2}$ with $\hAns\hAns = \hAns$ and $\hAns\hvns = \hvns$, CMT yields
\begin{align*}
    n_{s}\hHns^{2}(\hthetans) = (\hZns + \sqrt{n_{s}}\thetans)'\hQns(\hZns + \sqrt{n_{s}}\thetans) &= \norm{\hAns(\hSigmans^{-1/2}\hZns + \hvns)}^{2} \\
    &\to[d] \norm{\limA(\limSigma^{-1/2}\limZ +  \limv)}^{2} \sim \chi_{K-1}^{2}(\norm{\limv}).
\end{align*}
Thus, since $(x, \eta) \mapsto F_{\chi^{2}}(x; \eta)$ is continuous, another application of CMT implies
\begin{align*}
    F_{\chi^{2}}\paren{n_{s}\hHns^{2}(\hthetans); \sqrt{n_{s}}\hHns(\thetans)} \to[d] F_{\chi^{2}}\paren{\chi_{K-1}^{2}(\norm{\limv}); \norm{\limv}} \sim U(0,1),
\end{align*}
which contradicts \eqref{arxiv2:app:eq:UCB.criteria.contradiction}. 

\subsubsection{Case 2: Unbounded Normalized Heterogeneity} 
For the case of $\limeta = \infty$, I appeal to noncentral chi-squared asymptotics. For setup, define the standardized noncentral chi-squared quantile
\begin{align*}
    R(\beta; \eta) = \frac{F_{\chi^{2}}^{-1}(\beta; \eta) - (\eta^{2} + (K-1))}{\sqrt{4\eta^{2} + 2(K-1)}},
\end{align*}
which satisfies
\begin{align*}
    R(\beta; \eta) = G^{-1}(\beta; \eta), \quad G(x; \eta) = F_{\chi^{2}}\paren{\eta^{2} + (K-1) + \sqrt{4\eta^{2} + 2(K-1)}x; \eta},
\end{align*}
where \citet[III.2]{seri2015tight} implies
\begin{align*}
    \kappa(\eta) = \sup_{x}\abs{G(x;\eta) - \Phi(x)} = \sup_{x}\abs{F_{\chi^{2}}(x;\eta) - \Phi\paren{\frac{x - (\eta^{2} + (K-1))}{\sqrt{4\eta^{2} + 2(K-1)}}}} = O\paren{\frac{1}{\sqrt{\eta^{2} + (K-1)}}}.
\end{align*}
Moreover, $R(\beta; \eta) \to z_{\beta}$ as $\eta \to \infty$. Indeed, $\abs{G(x;\eta) - \Phi(x)} \leq \kappa(\eta)$ at $x = \Phi^{-1}(\beta \pm \kappa(\eta))$ implies $G(\Phi^{-1}(\beta - \kappa(\eta)); \eta) \leq \beta \leq G(\Phi^{-1}(\beta + \kappa(\eta));\eta)$ for $\kappa(\eta) < \min\{\beta, 1-\beta\}$, and therefore $\Phi^{-1}(\beta - \kappa(\eta)) \leq R(\beta; \eta) \leq \Phi^{-1}(\beta + \kappa(\eta))$ since $x \mapsto G(x;\eta)$ is continuous and strictly increasing. Thus, since $\kappa(\eta) \to 0$ as $\eta \to \infty$ and $\beta \mapsto \Phi^{-1}(\beta) = z_{\beta}$ is continuous at $\beta \in (0,1)$, this implies $R(\beta; \eta) \to z_{\beta}$ as $\eta \to \infty$. 

Letting $\hat{c}_{1s} = \norm{\hvns}^{2} + (K-1)$ and $\hat{c}_{2s} = \sqrt{4\norm{\hvns}^{2} + 2(K-1)}$, observe that
\begin{align*}
    \P[\Pns]{F_{\chi^{2}}(n_{s}\hHns^{2}(\hthetans); \sqrt{n_{s}}\hHns(\thetans)) < \beta} &= \P[\Pns]{n_{s}\hHns^{2}(\hthetans) < F_{\chi^{2}}^{-1}(\beta;\sqrt{n_{s}}\hHns(\thetans))} \\[5pt]
    &= \P[\Pns]{\frac{\norm{\hAns\hSigmans^{-1/2}\hZns + \hvns}^{2} - \hat{c}_{1s}}{\hat{c}_{2s}} < R(\beta; \norm{\hvns})}.
\end{align*}
Note that $\norm{\hvns} \to[p] \infty$. Indeed, for any $\e \in (0,1)$ and $C > 0$, a union bound yields
\begin{align*}
    \limsup_{s}\P[\Pns]{\norm{\hvns} \leq C} \leq \limsup_{s}\P[\Pns]{\norm{\vns} \leq \frac{C}{1-\e}} + \limsup_{s}\P[\Pns]{\frac{\norm{\hvns}}{\norm{\vns}}-1 < -\e} = 0 + 0,
\end{align*}
where the first zero follows from $\norm{\vns} \to \infty$ and the second from $\norm{\hvns}/\norm{\vns} \to[p] 1$. Since this holds for any $C > 0$, then $\norm{\hvns} \to[p] \infty$. Moreover, note that $\hvns/\norm{\hvns} \to[p] \limu$. Indeed,
\begin{align*}
    \norm{\frac{\hvns}{\norm{\hvns}} - \frac{\vns}{\norm{\vns}}} = \norm{\frac{\hvns(\norm{\vns} - \norm{\hvns}) + (\hvns - \vns)\norm{\hvns}}{\norm{\hvns}\norm{\vns}}} \leq \abs{1-\frac{\norm{\hvns}}{\norm{\vns}}} + \frac{\norm{\hvns - \vns}}{\norm{\vns}} \to[p] 0.
\end{align*}
Thus, $\hvns/\norm{\hvns} \to[p] \limu$ follows from $\vns/\norm{\vns} \to \limu$. The above convergences yield
\begin{align*} 
    \frac{\norm{\hAns\hSigmans^{-1/2}\hZns + \hvns}^{2} - \norm{\hvns}^{2}}{2\norm{\hvns}} = \frac{\norm{\hAns\hSigmans^{-1/2}\hZns}^{2}}{2\norm{\hvns}} + \inner{\frac{\hvns}{\norm{\hvns}}, \hSigmans^{-1/2}\hZns} \to[d] \inner{\limu,\limSigma^{-1/2}\limZ} \sim N(0,1),
\end{align*}
and moreover
\begin{align*}
    \abs{\frac{\norm{\hAns\hSigmans^{-1/2}\hZns + \hvns}^{2} - \norm{\hvns}^{2}}{2\norm{\hvns}} - \frac{\norm{\hAns\hSigmans^{-1/2}\hZns + \hvns}^{2} - \hat{c}_{1s}}{\hat{c}_{2s}}} \to[p] 0,
\end{align*}
which then implies
\begin{align*}
    \frac{\norm{\hAns\hSigmans^{-1/2}\hZns + \hvns}^{2} - \hat{c}_{1s}}{\hat{c}_{2s}} \to[d] N(0,1).
\end{align*}
Furthermore, $R(\beta; \norm{\hvns}) \to[p] z_{\beta}$. Indeed, since $R(\beta; \eta) \to z_{\beta}$ as $\eta \to \infty$, then for each $\e > 0$, one can choose $C_{\e} > 0$ large such that $|R(\beta; \eta) - z_{\beta}| \leq \e$ for any $\eta \geq C_{\e}$. Therefore $\norm{\hvns} \to[p] \infty$ implies $\P[\Pns]{|R(\beta; \norm{\hvns}) - z_{\beta}| > \e} \leq \P[\Pns]{\norm{\hvns} < C_{\e}} \to 0$, and hence $R(\beta; \norm{\hvns}) \to[p] z_{\beta}$. Altogether,
\begin{align*}
    \P[\Pns]{\frac{\norm{\hAns\hSigmans^{-1/2}\hZns + \hvns}^{2} - \hat{c}_{1s}}{\hat{c}_{2s}} < R(\beta; \norm{\hvns})} \to \P[]{N(0,1) < z_{\beta}} = \beta,
\end{align*}
which is a contradiction of \eqref{arxiv2:app:eq:UCB.criteria.contradiction}.

\subsection{Proof of Proposition \ref{arxiv2:prop:asymptotic.robust.CI}}\label{arxiv2:app:proof:asymptotic.robust.CI}
If \eqref{arxiv2:eq:asymptotic.robust.coverage} does not hold, then there exist $p > 0$, subsequence $\{n_{s}\} \subseteq \{n\}$, sequence of distributions $P_{n_{s}} \in \mathcal{P}_{n_{s}}$, and sequence of alternative weights $\lns[n_{s}] \in \Lns[0,n_{s}]$ such that
\begin{align*}
    \P[P_{n_{s}}]{\lns[n_{s}]'\thetans[n_{s}] \in \rCIns[n_{s}]} - (1-\alpha-\beta) < -p, \quad \forall s.
\end{align*}
Thus, to prove Proposition \ref{arxiv2:prop:asymptotic.robust.CI}, it suffices to show that the above yields a contradiction under Assumptions \ref{arxiv2:ass:linear.CLT}-\ref{arxiv2:ass:Lambda.representation}. To do so, I will iteratively pass to (sub)subsequences of $\{n_{s}\}$ along which variables of interest converge under Assumptions \ref{arxiv2:ass:linear.CLT}-\ref{arxiv2:ass:Lambda.representation}.

For conciseness, I will keep using the index $s$ when passing to subsequences, and will also use $s$ to abbreviate the subscript $n_{s}$. In this notation, the above becomes
\begin{align}\label{arxiv2:app:eq:CI.criteria.contradiction}
    \P[\Pns]{\lns'\thetans \notin \rCIns} > p + (\alpha+\beta), \quad \forall s.
\end{align}
On event $\cE_{s} = \curly{\hHns(\thetans)\max_{\l \in \hLns}\norm{\l - \hwns}_{\hSigmans} \leq \hBns}$, observe that
\begin{align*}
    \P[\Pns]{\lns'\thetans \notin \rCIns} &\leq \P[\Pns]{\lns'\thetans \notin \rCIns, \cE_{s}} + \P[\Pns]{\cE_{s}^{c}} \\
    &\leq \P[\Pns]{\lns'\thetans \notin \rCIns, \lns \in \hLns, \cE_{s}} + \P[\Pns]{\lns \notin \hLns} + \P[\Pns]{\cE_{s}^{c}},
\end{align*}
where since $\hBns = \hetans \max_{\l \in \hLns}\norm{\l - \hwns}_{\hSigmans}$,
\begin{align*}
    \P[\Pns]{\cE_{s}^{c}} &= \P[\Pns]{\hHns(\thetans)\max_{\l \in \hLns}\norm{\l - \hwns}_{\hSigmans} > \hetans \max_{\l \in \hLns}\norm{\l - \hwns}_{\hSigmans}} \leq \P[\Pns]{\hHns(\thetans) > \hetans}.
\end{align*}
Proposition \ref{arxiv2:prop:asymptotic.UCB} and containment condition \eqref{arxiv2:eq:target.containment} yield, along any subsequence,
\begin{align*}
    \limsup_{s \to \infty}\P[\Pns]{\hHns(\thetans) > \hetans} \leq \beta, \quad \limsup_{s \to \infty}\P[\Pns]{\lns \notin \hLns} = 0.
\end{align*}
Thus, to establish a contradiction of \eqref{arxiv2:app:eq:CI.criteria.contradiction}, it suffices to find a subsequence where
\begin{align}\label{arxiv2:app:eq:noncoverage.bound.criteria}
    \limsup_{s \to \infty}\P[\Pns]{\lns'\thetans \notin \rCIns, \lns \in \hLns, \cE_{s}} \leq \alpha.
\end{align}

I first consider the upper one-sided CI---the lower one-sided case follows analogously. 
\begin{align*}
    \P[\Pns]{\lns'\thetans \notin \rCIns, \lns \in \hLns, \cE_{s}} &\leq \P[\Pns]{\frac{\lns'\thetans - \hwns'\hthetans}{\tsigmawns} > z_{1-\alpha} + \max_{\l \in \hLns}\frac{\hHns(\thetans)\norm{\l - \hwns}_{\hSigmans}}{\tsigmawns}, \lns \in \hLns} \\
    &\leq \P[\Pns]{\frac{\lns'\thetans - \hwns'\hthetans}{\tsigmawns} > z_{1-\alpha} + \max_{\l \in \hLns}\frac{\abs{(\l - \hwns)'\thetans}}{\tsigmawns}, \lns \in \hLns} \\
    &\leq \P[\Pns]{\frac{\lns'\thetans - \hwns'\hthetans}{\tsigmawns} > z_{1-\alpha} + \frac{\abs{(\lns - \hwns)'\thetans}}{\tsigmawns}} \\
    &\leq \P[\Pns]{\frac{\hwns'\hZns}{\hsigmawns} < z_{\alpha}}, \quad \hZns = \sqrt{n_{s}}(\hthetans - \thetans),
\end{align*}
where the second line follows from Proposition \ref{arxiv2:prop:bound}. Under Assumptions \ref{arxiv2:ass:linear.CLT}-\ref{arxiv2:ass:consistent.S}, I now pass to a subsequence where 
\begin{align*}
    \begin{pmatrix}
    \Sigmans \\
    \wns \\
    \Sns
    \end{pmatrix}
    \to 
    \begin{pmatrix}
    \limSigma \\
    \limw \\
    \limS
    \end{pmatrix},
    \quad \norm{\limSigma},\norm{\limw} \in (0, \infty), \quad \limS \in \bbS, \quad 
    \begin{pmatrix}
    \hSigmans \\
    \hwns \\
    \hSns
    \end{pmatrix}
    \to[p]
    \begin{pmatrix}
    \limSigma \\
    \limw \\
    \limS
    \end{pmatrix},
    \quad \hZns \to[d] \limZ \sim N(0, \limSigma).
\end{align*}
CMT implies $\hsigmawns \to[p] \limsigma = \sqrt{\limw'\limSigma \limw} \in (0,\infty)$ and $\hwns'\hZns/\hsigmawns \to[d] N(0,1)$, from which
\begin{align*}
    \limsup_{s \to \infty}\P[\Pns]{\lns'\thetans \notin \rCIns, \lns \in \hLns, \cE_{s}} \leq \limsup_{s \to \infty}\P[\Pns]{\frac{\hwns'\hZns}{\hsigmawns} < z_{\alpha}} = \P[]{N(0,1) < z_{\alpha}} = \alpha,
\end{align*}
which establishes \eqref{arxiv2:app:eq:noncoverage.bound.criteria}, and hence a contradiction of \eqref{arxiv2:app:eq:CI.criteria.contradiction}.

I now consider the two-sided CI, staying within the previous subsequence to maintain the above convergences relative to $(\limSigma, \limw, \limS, \limZ)$. As in Appendix \ref{arxiv2:app:proof:asymptotic.UCB}, let $\vns = \Ans\Sigmans^{-1/2}\sqrt{n_{s}}\thetans$ and $\hvns = \hAns\hSigmans^{-1/2}\sqrt{n_{s}}\thetans$ so that $\sqrt{n_{s}}\Hns(\thetans) = \norm{\vns}$ and $\sqrt{n_{s}}\hHns(\thetans) = \norm{\hvns}$. By compactness, I pass to a further subsequence where
\begin{align*}
    \vns \to \limv \in [-\infty, \infty]^{K}, \quad \norm{\vns} = \sqrt{n_{s}}\Hns(\thetans) \to \limeta = \norm{\limv} \in [0, \infty].
\end{align*}
From here it suffices to, separately for the case of $\limeta < \infty$ and the case of $\limeta = \infty$, find a further subsequence where \eqref{arxiv2:app:eq:noncoverage.bound.criteria} holds. In what follows, denote
\begin{align*}
    \hDeltans &= (\lns - \hwns)'\sqrt{n_{s}}\thetans = (\lns - \hwns)'\Sigmans^{1/2}\vns = (\lns - \hwns)'\hSigmans^{1/2}\hvns, & \hDns &= \max_{\l \in \hLns}\norm{\l - \hwns}_{\hSigmans}, \\
    \Deltans &= (\lns - \wns)'\sqrt{n_{s}}\thetans = (\lns - \wns)'\Sigmans^{1/2}\vns, & \Dns &= \max_{\l \in \Lns}\norm{\l - \wns}_{\Sigmans}.
\end{align*}
In this notation, 
\begin{align}
    \P[\Pns]{\lns'\thetans \notin \rCIns, \lns \in \hLns, \cE_{s}} &\leq \P[\Pns]{\frac{\abs{\lns'\thetans - \hwns'\hthetans}}{\tsigmawns} > \cv{\frac{\hHns(\thetans)\hDns}{\tsigmawns}},\lns \in \hLns} \notag \\
    &= \P[\Pns]{\abs{\frac{\hDeltans}{\hsigmawns} - \frac{\hwns'\hZns}{\hsigmawns}} > \cv{\frac{\norm{\hvns}\hDns}{\hsigmawns}},\lns \in \hLns}. \label{arxiv2:app:eq:two.sided.criteria}
\end{align}

\subsubsection{Case 1: Bounded Normalized Heterogeneity} 
For the case of $\limeta < \infty$, observe that
\begin{align*}
    \abs{\hDeltans - \Deltans} = \abs{(\hwns - \wns)'\Sigmans^{1/2}\vns} \lesssim \norm{\hwns - \wns}\norm{\vns} \to[p] 0, \quad \abs{\Deltans} \lesssim \norm{\vns} = O(1).
\end{align*}
By compactness, I pass to a subsequence where $\Deltans \to \limDelta \in \R$, from which the above implies $\hDeltans \to[p] \limDelta$. Using Proposition \ref{arxiv2:prop:bound} (under $\lns \in \hLns$) to bound \eqref{arxiv2:app:eq:two.sided.criteria}, it suffices to show that
\begin{align}\label{arxiv2:app:eq:two.sided.criteria.finite.eta}
    \limsup_{s \to \infty}\P[\Pns]{\abs{\frac{\hDeltans}{\hsigmawns} - \frac{\hwns'\hZns}{\hsigmawns}} > \cv{\frac{\abs{\hDeltans}}{\hsigmawns}}} \leq \alpha.
\end{align}
And indeed, for $\limsigma = \sqrt{\limw'\limSigma \limw}$, CMT implies
\begin{align*}
    \abs{\frac{\hDeltans}{\hsigmawns} - \frac{\hwns'\hZns}{\hsigmawns}} - \cv{\frac{\abs{\hDeltans}}{\hsigmawns}} \to[d] \abs{\frac{\limDelta}{\limsigma} - N(0,1)} - \cv{\frac{\abs{\limDelta}}{\limsigma}},
\end{align*}
By definition of $\cv{|b|}$, \eqref{arxiv2:app:eq:two.sided.criteria.finite.eta} follows with equality (and with $\lim_{s \to \infty}$). This establishes \eqref{arxiv2:app:eq:noncoverage.bound.criteria}, and hence a contradiction of \eqref{arxiv2:app:eq:CI.criteria.contradiction}.

\subsubsection{Case 2: Unbounded Normalized Heterogeneity}
Suppose $\limeta = \infty$. The argument in Appendix \ref{arxiv2:app:proof:asymptotic.UCB} gives $\norm{\hvns} \to[p] \infty$. For $\limL = \L(\limS)$, compactness and Proposition \ref{arxiv2:prop:consistent.distance} give
\begin{align*}
    \Dns \to \limD, \quad
    \hDns \to[p] \limD, \quad
    \limD = \max_{\l \in \limL}\norm{\l - \limw}_{\limSigma} \in (0,\infty),
\end{align*}
where $\limD > 0$ follows from Lemma \ref{arxiv2:app:lemma:nondegenerate.disagreement}. Define
\begin{align*}
    A_{s} = \frac{\hDeltans}{\hsigmawns}, \quad
    B_{s} = \frac{\norm{\hvns}\hDns}{\hsigmawns}, \quad
    T_{s} = \frac{\hwns'\hZns}{\hsigmawns}.
\end{align*}
The preceding arguments give $B_{s} \to[p] \infty$ and $T_{s} \to[d] N(0,1)$. Moreover,
\begin{align*}
    \abs{\frac{\Deltans}{\norm{\vns}\Dns}} \leq 1
\end{align*}
because $\lns \in \Lns$ (due to $\lns \in \Lns[0,s] \subseteq \Lns$). Pass to a further subsequence on which
\begin{align*}
    \frac{\Deltans}{\norm{\vns}\Dns} \to c_{\infty} \in [-1,1].
\end{align*}
The perturbation bounds used above imply
\begin{align*}
    \abs{\frac{A_{s}}{B_{s}} - \frac{\Deltans}{\norm{\vns}\Dns}}
    = \abs{\frac{\hDeltans}{\norm{\hvns}\hDns} - \frac{\Deltans}{\norm{\vns}\Dns}} \lesssim_{p}
    \frac{\norm{\hvns - \vns}}{\norm{\vns}}
    + \frac{\norm{\hvns}}{\norm{\vns}}\frac{\abs{\hDns - \Dns}}{\Dns}
    + \frac{\norm{\hwns - \wns}}{\Dns}
    \to[p] 0.
\end{align*}
Thus $A_{s}/B_{s} \to[p] c_{\infty}$. On $\curly{\lns \in \hLns}$, Proposition \ref{arxiv2:prop:bound} gives $\abs{A_{s}} \leq B_{s}$. Lemma \ref{arxiv2:app:lemma:diverging.folded.normal} therefore yields, for every $\alpha \in (0,1)$,
\begin{align*}
    \limsup_{s \to \infty}
    \P[\Pns]{\abs{A_{s} - T_{s}} > \cv{B_{s}},\ \lns \in \hLns}
    \leq \alpha.
\end{align*}
Together with \eqref{arxiv2:app:eq:two.sided.criteria}, this establishes \eqref{arxiv2:app:eq:noncoverage.bound.criteria} in the unbounded-normalized-heterogeneity case and completes the proof.

\section{Supplementary Lemmas}

\begin{lemma}\label{arxiv2:app:lemma:hausdorff.continuity}
Under Assumption \ref{arxiv2:ass:Lambda.representation}, $S \mapsto \L(S)$ is Lipschitz in the Hausdorff metric:
\begin{align*}
    d_{H}(\L(S_{1}),\L(S_{2})) \leq \frac{\diam(\bbW)}{\delta_{g}}L_{g}d_{\cS}(S_{1},S_{2}), \quad \forall S_{1},S_{2} \in \bbS.
\end{align*}
\end{lemma}

\begin{proof}
Under Assumption \ref{arxiv2:ass:Lambda.representation}, the distance from $\l$ to $\L(S)$ is bounded as
\begin{align*}
    \dist(\l,\L(S)) \leq \frac{\diam(\bbW)}{\delta_{g}}\max\curly{g(\l,S), 0}, \quad \forall \l \in \bbW, \quad \forall S \in \bbS.
\end{align*}
Indeed, take any $(\l,S) \in \bbW \times \bbS$ and consider
\begin{align*}
    t = \frac{\max\curly{g(\l,S), 0}}{\max\curly{g(\l,S), 0} + \delta_{g}}, \quad \l_{t} = (1-t)\l + t\l^{\circ}(S) \in \bbW.
\end{align*}
By convexity,
\begin{align*}
    g(\l_{t},S) \leq (1-t)g(\l,S) + tg(\l^{\circ}(S),S) \leq (1-t)\max\curly{g(\l,S), 0} - t\delta_{g} = 0.
\end{align*}
This means $\l_{t} \in \L(S)$, and hence
\begin{align*}
    \dist(\l, \L(S)) \leq \norm{\l - \l_{t}} \leq t\diam(\bbW) \leq \frac{\diam(\bbW)}{\delta_{g}}\max\curly{g(\l,S), 0}.
\end{align*}
Now take any $S_{1},S_{2} \in \bbS$. For $\l \in \L(S_{2})$, observe that
\begin{align*}
    g(\l,S_{1}) = g(\l,S_{1})-g(\l,S_{2})+g(\l,S_{2}) \leq g(\l,S_{1})-g(\l,S_{2}) \leq \max_{\l \in \bbW}\abs{g(\l,S_{1})-g(\l,S_{2})}.
\end{align*}
Thus, the distance bound and the Lipschitz condition yield
\begin{align*}
    \max_{\l \in \L(S_{2})}\dist(\l, \L(S_{1})) \leq \frac{\diam(\bbW)}{\delta_{g}}\max_{\l \in \L(S_{2})}\max\curly{g(\l,S_{1}),0} \leq \frac{\diam(\bbW)}{\delta_{g}}L_{g}d_{\cS}(S_{1},S_{2}).
\end{align*}
Reversing the roles of $S_{1}$ and $S_{2}$ gives the same bound on $\max_{\l \in \L(S_{1})}\dist(\l, \L(S_{2}))$, and hence the desired bound on $d_{H}(\L(S_{1}), \L(S_{2}))$. 
\end{proof}

\begin{lemma}\label{arxiv2:app:lemma:perturbation.bounds}
For positive definite $\Sigma_{1}, \Sigma_{2}$ with $\emin(\Sigma_{1}), \emin(\Sigma_{2}) \geq \LB{e} > 0$,
\begin{align*}
    \abs{\norm{x}_{\Sigma_{1}}-\norm{x}_{\Sigma_{2}}} \leq \frac{1}{2\sqrt{\LB{e}}}\norm{\Sigma_{1} - \Sigma_{2}}\norm{x}, \quad \forall x \in \R^{K}.
\end{align*}
For nonempty compact $\L_{1},\L_{2} \subseteq \W$, $w_{1},w_{2} \in \W$, and positive definite $\Sigma$ with $\emax(\Sigma) \leq \UB{e}$,
\begin{align*}
    \abs{\max_{\l \in \L_{1}}\norm{\l - w_{1}}_{\Sigma} - \max_{\l \in \L_{2}}\norm{\l - w_{1}}_{\Sigma}} &\leq \sqrt{\UB{e}} d_{H}(\L_{1}, \L_{2}), \\
    \abs{\max_{\l \in \L_{1}}\norm{\l - w_{1}}_{\Sigma} - \max_{\l \in \L_{1}}\norm{\l - w_{2}}_{\Sigma}} &\leq \sqrt{\UB{e}} \norm{w_{1}-w_{2}}.
\end{align*}
\end{lemma}

\begin{proof}
For the first claim, $x=0$ is immediate. For $x \neq 0$, observe that
\begin{align*}
    \abs{\norm{x}_{\Sigma_{1}} - \norm{x}_{\Sigma_{2}}} = \frac{\abs{x'(\Sigma_{1} - \Sigma_{2})x}}{\norm{x}_{\Sigma_{1}} + \norm{x}_{\Sigma_{2}}} \leq \frac{\norm{\Sigma_{1} - \Sigma_{2}}\norm{x}^{2}}{\sqrt{\emin(\Sigma_{1})}\norm{x} + \sqrt{\emin(\Sigma_{2})}\norm{x}} \leq \frac{1}{2\sqrt{\LB{e}}}\norm{\Sigma_{1} - \Sigma_{2}}\norm{x}.
\end{align*}

For the first inequality of the second claim, note that for each $\l_{1} \in \L_{1}$, compactness of $\L_{2}$ gives a nearest point $\l_{2} \in \L_{2}$ such that
\begin{align*}
    \norm{\l_{1} - \l_{2}} = \dist(\l_{1}, \L_{2}) \leq d_{H}(\L_{1}, \L_{2}).
\end{align*}
The reverse triangle inequality and max-eigenvalue bound imply
\begin{align*}
    |\norm{\l_{1} - w_{1}}_{\Sigma} - \norm{\l_{2} - w_{1}}_{\Sigma}| \leq \norm{\l_{1} - \l_{2}}_{\Sigma} \leq \sqrt{\UB{e}}\norm{\l_{1} - \l_{2}}.
\end{align*}
Combining with the previous inequality yields
\begin{align*}
    \norm{\l_{1} - w_{1}}_{\Sigma} \leq \norm{\l_{2} - w_{1}}_{\Sigma} + \sqrt{\UB{e}}d_{H}(\L_{1}, \L_{2}) \leq \max_{\l \in \L_{2}}\norm{\l - w_{1}}_{\Sigma} + \sqrt{\UB{e}}d_{H}(\L_{1}, \L_{2})
\end{align*}
Taking the maximum over $\l \in \L_{1}$ then yields
\begin{align*}
    \max_{\l \in \L_{1}}\norm{\l - w_{1}}_{\Sigma} \leq \max_{\l \in \L_{2}}\norm{\l - w_{1}}_{\Sigma} + \sqrt{\UB{e}}d_{H}(\L_{1}, \L_{2})
\end{align*}
Reversing the roles of $\L_{1}$ and $\L_{2}$ then yields the first inequality of the second claim. For the second inequality of the second claim, note that for each $\l \in \L_{1}$,
\begin{align*}
    \abs{\norm{\l - w_{1}}_{\Sigma} - \norm{\l - w_{2}}_{\Sigma}} \leq \norm{w_{1} - w_{2}}_{\Sigma} \leq \sqrt{\UB{e}}\norm{w_{1} - w_{2}},
\end{align*}
and hence
\begin{align*}
    \norm{\l - w_{1}}_{\Sigma} &\leq \norm{\l - w_{2}}_{\Sigma} + \sqrt{\UB{e}}\norm{w_{1} - w_{2}} \leq \max_{\l \in \L_{1}}\norm{\l - w_{2}}_{\Sigma} + \sqrt{\UB{e}}\norm{w_{1} - w_{2}}, \\
    \norm{\l - w_{2}}_{\Sigma} &\leq \norm{\l - w_{1}}_{\Sigma} + \sqrt{\UB{e}}\norm{w_{1} - w_{2}} \leq \max_{\l \in \L_{1}}\norm{\l - w_{1}}_{\Sigma} + \sqrt{\UB{e}}\norm{w_{1} - w_{2}}.
\end{align*}
Taking the maximum over $\l \in \L_{1}$ on the LHS of both lines then yields the second inequality of the second claim.
\end{proof}

\begin{lemma}\label{arxiv2:app:lemma:minimax.quadratic.bounds}
For nonempty, compact, and convex $\L \subseteq \W$ and positive definite $\Sigma$ with $\emin(\Sigma) \geq \LB{e} > 0$ and $\emax(\Sigma) \leq \UB{e}$, the corresponding minimax-bias weights $w^{*} = \arg\min_{\Bar{w} \in \L}\max_{\l \in \L}\norm{\l - \Bar{w}}_{\Sigma}$ satisfy
\begin{align*}
    \norm{\Bar{w} - w^{*}}^{2} \leq \frac{2\sqrt{\UB{e}}}{\LB{e}}\diam(\L)\paren{\max_{\l \in \L}\norm{\l - \Bar{w}}_{\Sigma} - \max_{\l \in \L}\norm{\l - w^{*}}_{\Sigma}}, \quad \forall \Bar{w} \in \L.
\end{align*}
\end{lemma}

\begin{proof}
By Proposition \ref{arxiv2:prop:optimal.weights}, $w^{*}$ exists uniquely. Below I will show the inequality
\begin{align*}
    \LB{e}\norm{\Bar{w} - w^{*}}^{2} \leq \max_{\l \in \L}\norm{\l - \Bar{w}}_{\Sigma}^{2} - \max_{\l \in \L}\norm{\l - w^{*}}_{\Sigma}^{2},
\end{align*}
from which the lemma follows by bounding the RHS as
\begin{align*}
    \max_{\l \in \L}\norm{\l - \Bar{w}}_{\Sigma}^{2} - \max_{\l \in \L}\norm{\l - w^{*}}_{\Sigma}^{2} &= \paren{\max_{\l \in \L}\norm{\l - \Bar{w}}_{\Sigma} - \max_{\l \in \L}\norm{\l - w^{*}}_{\Sigma}} \\
    &\times \paren{\max_{\l \in \L}\norm{\l - \Bar{w}}_{\Sigma} + \max_{\l \in \L}\norm{\l - w^{*}}_{\Sigma}} \\
    &\leq \paren{\max_{\l \in \L}\norm{\l - \Bar{w}}_{\Sigma} - \max_{\l \in \L}\norm{\l - w^{*}}_{\Sigma}} \\
    &\times \paren{\UB{e}^{1/2}\diam(\L) + \UB{e}^{1/2}\diam(\L)}.
\end{align*}
It now remains to show the inequality. First note that
\begin{align*}
    \norm{(1-t)a + tb}_{\Sigma}^{2} = (1-t)\norm{a}_{\Sigma}^{2} + t\norm{b}_{\Sigma}^{2} - t(1-t)\norm{a-b}_{\Sigma}^{2}, \quad \forall a,b \in \R^{K}, \quad \forall t \in (0,1).
\end{align*}
For each $\Bar{w} \in \L$, define $\Bar{w}_{t} = (1-t)w^{*} + t\Bar{w} \in \L$. For each $\l \in \L$, the above implies
\begin{align*}
    \max_{\l \in \L}\norm{\l - \Bar{w}_{t}}_{\Sigma}^{2} \leq (1-t)\max_{\l \in \L}\norm{\l - w^{*}}_{\Sigma}^{2} + t\max_{\l \in \L}\norm{\l - \Bar{w}}_{\Sigma}^{2} - t(1-t)\norm{\Bar{w} - w^{*}}_{\Sigma}^{2}.
\end{align*}
The LHS is bounded below by $\max_{\l \in \L}\norm{\l - w^{*}}_{\Sigma}^{2}$. Applying this bound and manipulating terms yields
\begin{align*}
    \max_{\l \in \L}\norm{\l - \Bar{w}}_{\Sigma}^{2} - \max_{\l \in \L}\norm{\l - w^{*}}_{\Sigma}^{2} \geq (1-t)\norm{\Bar{w} - w^{*}}_{\Sigma}^{2} \geq (1-t)\LB{e}\norm{\Bar{w} - w^{*}}^{2}.
\end{align*}
Taking $t \to 0$ yields the desired inequality.
\end{proof}

\begin{lemma}\label{arxiv2:app:lemma:nondegenerate.disagreement}
Under Assumption \ref{arxiv2:ass:Lambda.representation}, the diameter of $\L(S)$ is bounded away from zero:
\begin{align*}
    \diam(\L(S)) = \max_{\l,w \in \L(S)}\norm{\l - w} \geq \frac{\delta_{g}}{\delta_{g} + \UB{C}_{g}}\diam(\bbW) > 0, \quad \forall S \in \bbS,
\end{align*}
where $\UB{C}_{g} = \max_{S \in \bbS} \max_{\l \in \bbW} \max\curly{g(\l,S), 0} < \infty$. Consequently, the maximum distance under a positive definite $\Sigma$ with $\emin(\Sigma) \geq \LB{e} > 0$ is bounded below as
\begin{align*}
   \max_{\l \in \L(S)} \norm{\l - w}_{\Sigma} \geq \frac{\sqrt{\LB{e}}}{2}\frac{\delta_{g}}{\delta_{g} + \UB{C}_{g}}\diam(\bbW) > 0, \quad \forall S \in \bbS, \quad \forall w \in \W.
\end{align*}
\end{lemma}

\begin{proof}
The continuity, Lipschitz, and compactness properties in Assumption \ref{arxiv2:ass:Lambda.representation} yield
\begin{align*}
    \UB{C}_{g} = \max_{S \in \bbS} \max_{\l \in \bbW} \max\curly{g(\l,S), 0} < \infty.
\end{align*}
By compactness, there exist $w_{a},w_{b} \in \bbW$ such that
\begin{align*}
    \diam(\bbW) = \max_{\l,w \in \bbW}\norm{\l - w} = \norm{w_{a} - w_{b}}.
\end{align*}
Given the Slater point $\l^{\circ}(S) \in \bbW$, define 
\begin{align*}
    \l_{a}(S) = (1-t)\l^{\circ}(S) + t w_{a}, \quad \l_{b}(S) = (1-t)\l^{\circ}(S) + t w_{b}, \quad t = \frac{\delta_{g}}{\delta_{g} + \UB{C}_{g}} \in (0,1].
\end{align*}
Since $\bbW$ is convex, then $\l_{a}(S), \l_{b}(S) \in \bbW$. Moreover, convexity of $\l \mapsto g(\l,S)$ yields
\begin{align*}
    g(\l_{a}(S), S) \leq (1-t)g(\l^{\circ}(S), S) + t g(w_{a},S) \leq -(1-t)\delta_{g} + t\UB{C}_{g} = 0, 
\end{align*}
and hence $\l_{a}(S) \in \L(S)$. The analogous argument for $w_{b}$ yields $\l_{b}(S) \in \L(S)$. Thus,
\begin{align*}
    \diam(\L(S)) = \max_{\l,w \in \L(S)}\norm{\l - w} \geq \norm{\l_{a}(S) - \l_{b}(S)} = t \norm{w_{a}-w_{b}} = \frac{\delta_{g}}{\delta_{g} + \UB{C}_{g}} \diam(\bbW) > 0.
\end{align*}
Consequently, using $\norm{x}_{\Sigma} \geq \norm{x}\sqrt{\LB{e}}$; $\l_{a}(S),\l_{b}(S) \in \L(S)$; and triangle inequality,
\begin{align*}
    \frac{2\max_{\l \in \L(S)} \norm{\l - w}_{\Sigma}}{\sqrt{\LB{e}}} \geq \norm{\l_{a}(S) - w} + \norm{\l_{b}(S) - w} \geq \norm{\l_{a}(S)-\l_{b}(S)} = \frac{\delta_{g}}{\delta_{g} + \UB{C}_{g}}\diam(\bbW).
\end{align*}
\end{proof}

\begin{lemma}\label{arxiv2:app:lemma:sufficient.slack}
Let Assumptions \ref{arxiv2:ass:consistent.S} and \ref{arxiv2:ass:Lambda.representation} be satisfied and suppose that $\L_{0,n} \subseteq \Ln$ satisfies slack condition \eqref{arxiv2:eq:slack.condition}. Then containment condition \eqref{arxiv2:eq:target.containment} is satisfied:
\begin{align*}
    \limn \supPn \sup_{\l \in \L_{0,n}}\P[\Pn]{\l \notin \hLn} = 0.
\end{align*}
\end{lemma}

\begin{proof}
Observe that 
\begin{align*}
    \sup_{\l \in \L_{0,n}}\P[\Pn]{\l \notin \hLn} \leq \P[\Pn]{\sup_{\l \in \L_{0,n}}g(\l, \hSn) > 0, \hSn \in \bbS} + \P[\Pn]{\hSn \notin \bbS}.
\end{align*}
Under Assumption \ref{arxiv2:ass:consistent.S}, $\limn \supPn \P[\Pn]{\hSn \notin \bbS} = 0$. Thus, it suffices to show
\begin{align*}
    \limn \supPn \P[\Pn]{\sup_{\l \in \L_{0,n}}g(\l, \hSn) > 0, \hSn \in \bbS} = 0.
\end{align*}
For $\l \in \L_{0,n} \subseteq \Ln$ and on an event where $\hSn \in \bbS$, Assumption \ref{arxiv2:ass:Lambda.representation} implies
\begin{align*}
    g(\l , \hSn) \leq g(\l, \Sn) + L_{g}d_{\cS}(\hSn, \Sn),
\end{align*}
and hence slack condition \eqref{arxiv2:eq:slack.condition} implies
\begin{align*}
    \sup_{\l \in \L_{0,n}}g(\l , \hSn) \leq -\nu + L_{g}d_{\cS}(\hSn, \Sn).
\end{align*}
When $L_{g}d_{\cS}(\hSn, \Sn) \leq \nu$, the above implies $\sup_{\l \in \L_{0,n}}g(\l , \hSn) \leq 0$. By contraposition,
\begin{align*}
    \P[\Pn]{\sup_{\l \in \L_{0,n}}g(\l, \hSn) > 0, \hSn \in \bbS} \leq \P[\Pn]{L_{g}d_{\cS}(\hSn, \Sn) > \nu, \hSn \in \bbS} \leq \P[\Pn]{L_{g}d_{\cS}(\hSn, \Sn) > \nu}.
\end{align*}
But since $\hSn$ is uniformly consistent for $\Sn$ under Assumption \ref{arxiv2:ass:consistent.S}, the above implies
\begin{align*}
    \limn \supPn \P[\Pn]{\sup_{\l \in \L_{0,n}}g(\l, \hSn) > 0, \hSn \in \bbS} = 0.
\end{align*}
Altogether, this yields the desired conclusion.
\end{proof}

\begin{lemma}\label{arxiv2:app:lemma:diverging.folded.normal}
Fix $\alpha \in (0,1)$. Let $A_{n}$, $B_{n}$, and $T_{n}$ be random variables and let $\cE_{n}$ be events such that $B_{n} \geq 0$, $B_{n} \to[p] \infty$, $T_{n} \to[d] G \sim N(0,1)$, $A_{n}/B_{n} \to[p] c \in [-1,1]$, and $\abs{A_{n}} \leq B_{n}$ on $\cE_{n}$. Then
\begin{align*}
    \limsup_{n \to \infty}\P[]{\abs{A_{n} - T_{n}} > \cv{B_{n}},\ \cE_{n}} \leq \alpha.
\end{align*}
\end{lemma}

\begin{proof}
By the definition of the folded-normal critical value,
\begin{align*}
    \cv{b} - b \to z_{1-\alpha}, \quad b \to \infty.
\end{align*}
First suppose that $\abs{c} < 1$. On $\cE_{n}$,
\begin{align*}
    \cv{B_{n}} - \abs{A_{n}}
    = \curly{\cv{B_{n}} - B_{n}} + B_{n}\paren{1 - \frac{\abs{A_{n}}}{B_{n}}}
    \to[p] \infty.
\end{align*}
The triangle inequality and $T_{n} = O_{p}(1)$ therefore give
\begin{align*}
    \P[]{\abs{A_{n} - T_{n}} > \cv{B_{n}},\ \cE_{n}}
    \leq \P[]{\abs{T_{n}} > \cv{B_{n}} - \abs{A_{n}},\ \cE_{n}} \to 0.
\end{align*}
Now suppose that $c = 1$. With probability approaching one, $A_{n} \geq 0$ and $A_{n} + \cv{B_{n}} \to[p] \infty$. On $\cE_{n}$, $A_{n} \leq B_{n}$, so
\begin{align*}
    \curly{\abs{A_{n} - T_{n}} > \cv{B_{n}}}
    \subseteq
    \curly{T_{n} < B_{n} - \cv{B_{n}}}
    \cup
    \curly{T_{n} > A_{n} + \cv{B_{n}}}
\end{align*}
up to an event whose probability tends to zero. The preceding critical-value limit and Slutsky's theorem imply
\begin{align*}
    \P[]{T_{n} < B_{n} - \cv{B_{n}}} \to \P[]{G < z_{\alpha}} = \alpha,
\end{align*}
while the probability of the second event tends to zero. Hence the desired bound holds for $c = 1$. Finally, if $c = -1$, then with probability approaching one $A_{n} \leq 0$ and $A_{n} - \cv{B_{n}} \to[p] -\infty$. On $\cE_{n}$, $A_{n} \geq -B_{n}$ and the symmetric argument gives
\begin{align*}
    \limsup_{n \to \infty} \P[]{\abs{A_{n} - T_{n}} > \cv{B_{n}},\ \cE_{n}}
    \leq \lim_{n \to \infty} \P[]{T_{n} > \cv{B_{n}} - B_{n}}
    = \P[]{G > z_{1-\alpha}} = \alpha.
\end{align*}
\end{proof}

\section{Large-K Uniform Asymptotics}\label{arxiv2:app:sec:UK}
The uniform asymptotic results in Section \ref{arxiv2:sec:asymptotic.validity} keep the number of groups fixed. This appendix provides an extension in which $K = K_{n} \geq 2$ may grow with $n$. Under $P_{n} \in \mathcal{P}_{n}$, the researcher observes $\hthetan \in \R^{\Kn}$ and a positive definite covariance estimator $\tSigman \in \R^{\Kn \times \Kn}$. Let
\begin{align*}
    \hZn = \sqrt{n}(\hthetan - \thetan), \quad \thetan = \thetan(\Pn),  \quad \hSigman = n\tSigman, \quad \Wn = \curly{w \in \R^{\Kn}: \1'w = 1}.
\end{align*}
The population covariance of $\hZn$ is the positive definite matrix $\Sigman = \Sigman(\Pn)$. Note that the $n$ subscript in $(\thetan(\cdot), \Sigman(\cdot))$ makes the dimensional dependence on $\Kn$ explicit. For some objects below, I leave this dependence implicit. All stochastic orders below are uniform over $P_{n} \in \mathcal{P}_{n}$. Thus, $a_{n} = o_{p, \cPn}(b_{n})$ means that $a_{n}/b_{n}$ converges uniformly in probability to zero, and $a_{n} \lesssim_{p, \cPn} b_{n}$ means that $a_{n} \leq Cb_{n}$ with probability uniformly approaching one for a constant $C$ that does not depend on $n$ or $P_{n}$.

\subsection{Distributional and Covariance Conditions}

\begin{assumptionUK}[Linear CLT]\label{arxiv2:app:ass:UK.linear.CLT}
For $\cUn = \curly{u \in \R^{\Kn}: \norm{u} = 1}$,
\begin{align*}
    \limn \supPn \sup_{u \in \cUn}\sup_{x \in \R} \abs{\P[\Pn]{\frac{u'\hZn}{\sqrt{u'\Sigman u}} \leq x} - \Phi(x)} = 0, \quad \hZn = \sqrt{n}(\hthetan - \thetan).
\end{align*}
Moreover, there exists a constant $\UB{C}_{\theta} > 0$ such that $\supPn \norm{\thetan} \leq \dsqrt{\Kn}\UB{C}_{\theta}$ for all $n$.
\end{assumptionUK}

For example, when $\hZn$ consists of normalized regression coefficients or sample means admitting an observation- or cluster-level asymptotically linear representation, the linear CLT can be obtained from uniform moment and dependence bounds, a well-conditioned low-leverage design, and restrictions on the growth of $\Kn$ relative to the number of independent observations or clusters; see \citet[Theorem 4.2]{belloni2015some} for illustrative regression conditions.

Define the relative covariance error
\begin{align*}
    \eSigman = \norm{\Sigman^{-1/2} (\hSigman - \Sigman) \Sigman^{-1/2}}.
\end{align*}

\begin{assumptionUK}[Relative Covariance Consistency]\label{arxiv2:app:ass:UK.consistent.covariance}
The relative covariance error satisfies
\begin{align*}
    \eSigman = o_{p, \cPn}(1).
\end{align*}
Moreover, there exist deterministic sequences $\LB{e}_{n}, \UB{e}_{n} > 0$ and a constant $\UB{\rho} > 0$ such that
\begin{align*}
    \LB{e}_{n} \leq \infPn \emin(\Sigman) \leq \supPn \emax(\Sigman) \leq \UB{e}_{n}, \quad \rhon = \frac{\UB{e}_{n}}{\LB{e}_{n}} \leq \UB{\rho}, \quad \forall n.
\end{align*}
\end{assumptionUK}

The common scale of the eigenvalues may therefore vary with $n$; only the condition number is uniformly bounded. This formulation includes, for example, balanced independent group means, for which the eigenvalues of the covariance of $\sqrt{n}(\hthetan - \thetan)$ are typically of order $\Kn$. Uniform relative covariance consistency follows from matrix laws of large numbers for covariance-standardized score outer products. \citet[Theorem 4.6]{belloni2015some} gives illustrative rates when a high-dimensional regression covariance matrix is estimated using fitted residuals.

On the event $\curly{\eSigman < 1}$, the definition of $\eSigman$ gives the scale-free Loewner comparisons
\begin{align}\label{arxiv2:app:eq:UK.loewner}
    (1 - \eSigman)\Sigman \preceq \hSigman \preceq (1 + \eSigman)\Sigman,
\end{align}
and, after inversion,
\begin{align}\label{arxiv2:app:eq:UK.inverse.loewner}
    (1 + \eSigman)^{-1}\Sigman^{-1} \preceq \hSigman^{-1} \preceq (1 - \eSigman)^{-1}\Sigman^{-1}.
\end{align}
These inequalities replace the fixed-$K$ arguments that separately bound every occurrence of $\Sigman$ and $\Sigman^{-1}$ by fixed upper and lower eigenvalue constants. For a positive definite matrix $V$, define
\begin{align*}
    A(V) = I - \frac{V^{-1/2}\1\1'V^{-1/2}}{\1'V^{-1}\1}, \quad Q(V) = V^{-1/2}A(V)V^{-1/2}.
\end{align*}
Write $(\An, \Qn) = (A(\Sigman), Q(\Sigman))$ and $(\hAn, \hQn) = (A(\hSigman), Q(\hSigman))$, and let
\begin{align*}
    \Hn^{2}(a) = a'\Qn a, \quad \hHn^{2}(a) = a'\hQn a, \quad \nu_{n} = \Kn - 1, \quad \eta_{n} = \sqrt{n}\Hn(\thetan), \quad T_{n}^{0} = n\Hn^{2}(\hthetan).
\end{align*}
The population-covariance statistic $T_{n}^{0}$ has the decomposition
\begin{align}\label{arxiv2:app:eq:UK.oracle.decomposition}
    T_{n}^{0} - (\nu_{n} + \eta_{n}^{2}) = 2v_{n}'\xi_{n} + \curly{\xi_{n}'\An\xi_{n} - \nu_{n}}, \quad \xi_{n} = \Sigman^{-1/2}\hZn, \quad v_{n} = \An\Sigman^{-1/2}\sqrt{n}\thetan.
\end{align}
The first term is linear in $\xi_{n}$ while the second term is an increasing-rank quadratic form in $\xi_{n}$. This decomposition identifies the two sources of randomness. Assumption \ref{arxiv2:app:ass:UK.linear.CLT} controls fixed deterministic linear combinations. Assumption \ref{arxiv2:app:ass:UK.quadratic.CLT} below is a separate high-level assumption on the standardized \textit{full oracle statistic}; primitive conditions for Assumption \ref{arxiv2:app:ass:UK.quadratic.CLT} must control the linear and quadratic pieces jointly.

\begin{assumptionUK}[Oracle Quadratic CLT]\label{arxiv2:app:ass:UK.quadratic.CLT}
For every subsequence $\{n_{s}\} \subseteq \{n\}$ and every sequence $P_{n_{s}} \in \mathcal{P}_{n_{s}}$ for which $K_{n_{s}} \to \infty$,
\begin{align*}
    \frac{T_{n_{s}}^{0} - \curly{\nu_{n_{s}} + \eta_{n_{s}}^{2}}}{\sqrt{2\nu_{n_{s}} + 4\eta_{n_{s}}^{2}}} \to[d] N(0, 1).
\end{align*}
\end{assumptionUK}

When $\xi_{n} = \Sigman^{-1/2}\hZn$ admits an observation- or cluster-level asymptotically linear representation, Assumption \ref{arxiv2:app:ass:UK.quadratic.CLT} can be verified from uniform higher-moment and dependence bounds together with two dispersion requirements: the random fluctuation of the quadratic form cannot be driven by a small number of sampling units, and the quadratic form cannot place most of its weight on the squares of only a few covariance-standardized linear combinations of $\hZn$; see \citet{anatolyev2019many} and \citet{kline2020leave}. The second requirement is automatic for the oracle statistic in \eqref{arxiv2:app:eq:UK.oracle.decomposition}, because $\An$ is a nonrandom rank-$\nu_{n}$ projection whose nonzero eigenvalues are all equal to one. This observation does not by itself justify an analogous quadratic CLT for the feasible statistic, which is based on the random matrix $\hAn$. The next result handles that separate covariance-estimation step.

\begin{propositionUK}[Covariance and Quadratic-Statistic Transfer]\label{arxiv2:app:prop:UK.quadratic.transfer}
Under Assumption \ref{arxiv2:app:ass:UK.consistent.covariance}, on the event $\curly{\eSigman \leq 1/2}$,
\begin{align}\label{arxiv2:app:eq:UK.vector.heterogeneity}
    \norm{\hAn\hSigman^{-1/2}\thetan - \An\Sigman^{-1/2}\thetan} &\lesssim \eSigman \norm{\An\Sigman^{-1/2}\thetan}, \\[5pt]
    \abs{\hHn(\thetan) - \Hn(\thetan)} &\lesssim \eSigman \Hn(\thetan). \label{arxiv2:app:eq:UK.relative.heterogeneity}
\end{align}
Let $\hat{\eta}_{0,n} = \sqrt{n}\hHn(\thetan)$ and $\hat{T}_{n} = n\hHn^{2}(\hthetan)$. If, in addition, the following two rate conditions hold directly,
\begin{align}\label{arxiv2:app:eq:UK.transfer.rate}
    \sqrt{\Kn}\eSigman = o_{p, \cPn}(1), \quad \norm{\xi_{n}} = O_{p, \cPn}(\sqrt{\Kn}),
\end{align}
then Assumption \ref{arxiv2:app:ass:UK.quadratic.CLT} implies, along every sequence $P_{n} \in \mathcal{P}_{n}$ with $\Kn \to \infty$,
\begin{align}\label{arxiv2:app:eq:UK.feasible.CLT}
    \frac{\hat{T}_{n} - \curly{\nu_{n} + \hat{\eta}_{0,n}^{2}}}{\sqrt{2\nu_{n} + 4\hat{\eta}_{0,n}^{2}}} \to[d] N(0, 1).
\end{align}
\end{propositionUK}

\begin{proof}
Since $a'Q(V)a = \min_{\gamma \in \R}\norm{a - \1\gamma}_{V^{-1}}^{2}$, equation \eqref{arxiv2:app:eq:UK.inverse.loewner} yields
\begin{align*}
    (1 + \eSigman)^{-1/2}\Hn(\thetan) \leq \hHn(\thetan) \leq (1 - \eSigman)^{-1/2}\Hn(\thetan),
\end{align*}
which proves \eqref{arxiv2:app:eq:UK.relative.heterogeneity}. For \eqref{arxiv2:app:eq:UK.vector.heterogeneity}, both $A(V)V^{-1/2}$ maps annihilate $\1$, so their difference acts only on population GLS residual $\thetan - \1\gamma_{n}$. Relative covariance perturbation and $\rhon \leq \UB{\rho}$ then give \eqref{arxiv2:app:eq:UK.vector.heterogeneity}.

Because $\Qn$ and $\hQn$ also annihilate $\1$, the same residualization yields
\begin{align*}
    \abs{\curly{\hat{T}_{n} - \hat{\eta}_{0,n}^{2}} - \curly{T_{n}^{0} - \eta_{n}^{2}}} \lesssim \eSigman \curly{\norm{\xi_{n}}^{2} + \eta_{n}\norm{\xi_{n}}}.
\end{align*}
After division by $\sqrt{2\nu_{n} + 4\eta_{n}^{2}}$, the right side is $O_{p, \cPn}(\dsqrt{\Kn}\eSigman) = o_{p, \cPn}(1)$. If $\eta_{n} = 0$, then $\hat{\eta}_{0,n} = 0$ exactly. Otherwise, \eqref{arxiv2:app:eq:UK.relative.heterogeneity} gives
\begin{align*}
    \frac{2\nu_{n} + 4\hat{\eta}_{0,n}^{2}} {2\nu_{n} + 4\eta_{n}^{2}} \to[p, \cPn] 1.
\end{align*}
Slutsky's theorem proves \eqref{arxiv2:app:eq:UK.feasible.CLT}.
\end{proof}

The rate in \eqref{arxiv2:app:eq:UK.transfer.rate} is a convenient worst-case condition; structured covariance estimators can permit weaker rates.

\subsection{Classes, Weights, and Maximum Distance}
Let $\bbWn \subseteq \Wn$ be compact and convex, let $d_{n} = \diam(\bbWn)$, and suppose
\begin{align*}
    \Ln = \curly{\l \in \bbWn: \gn(\l, \Sn) \leq 0}, \quad \hLn = \curly{\l \in \bbWn: \gn(\l, \hSn) \leq 0}
\end{align*}
are nonempty, compact, and convex. Suppose $\gn(\cdot, S)$ is continuous and convex, is $\Lgn$-Lipschitz in $S$ under a metric $d_{\cS,n}$, and admits a Slater point with margin $\deltagn > 0$. The fixed-$K$ error-bound proof applies pointwise for every $n$ and gives
\begin{align}\label{arxiv2:app:eq:UK.hausdorff}
    h_{n} = d_{H}(\hLn, \Ln) \leq \frac{d_{n}\Lgn}{\deltagn} d_{\cS,n}(\hSn, \Sn).
\end{align}
Thus, unlike in the fixed-$K$ analysis, the Lipschitz constants, Slater margins, diameters, and first-stage rates may all depend on $n$. Their product in \eqref{arxiv2:app:eq:UK.hausdorff} is the key.

For generic population and feasible estimated centering weights $\wn, \hwn \in \bbWn$, define
\begin{align*}
    D_{n} = \max_{\l \in \Ln}\norm{\l - \wn}_{\Sigman}, \quad \sigma_{w,n} = \norm{\wn}_{\Sigman}, \quad \hat{D}_{n} = \max_{\l \in \hLn}\norm{\l - \hwn}_{\hSigman}, \quad \hat{\sigma}_{w,n} = \norm{\hwn}_{\hSigman}.
\end{align*}

\begin{propositionUK}[Weight and Criterion Consistency]\label{arxiv2:app:prop:UK.weights}
Under Assumption \ref{arxiv2:app:ass:UK.consistent.covariance}, on the event $\curly{\eSigman \leq 1/2}$,
\begin{align}
    \abs{\hat{D}_{n} - D_{n}} &\lesssim \sqrt{\UB{e}_{n}} \curly{h_{n} + \norm{\hwn - \wn}} + \eSigman D_{n}, \label{arxiv2:app:eq:UK.distance.bound} \\
    \abs{\hat{\sigma}_{w,n} - \sigma_{w,n}} &\lesssim \sqrt{\UB{e}_{n}}\norm{\hwn - \wn} + \eSigman\sigma_{w,n}. \label{arxiv2:app:eq:UK.se.bound}
\end{align}
Consequently,
\begin{align}
    \frac{\abs{\hat{D}_{n} - D_{n}}}{D_{n}} &\lesssim_{p, \cPn} \eSigman + \frac{\sqrt{\UB{e}_{n}} \curly{h_{n} + \norm{\hwn - \wn}}}{D_{n}}, \label{arxiv2:app:eq:UK.relative.distance} \\
    \abs{\frac{\hat{\sigma}_{w,n}}{\sigma_{w,n}} - 1} &\lesssim_{p, \cPn} \eSigman + \frac{\sqrt{\UB{e}_{n}}\norm{\hwn - \wn}}{\sigma_{w,n}}, \label{arxiv2:app:eq:UK.relative.se}
\end{align}
whenever the denominators are positive. Let $\rwn$ minimize $\max_{\l \in \Ln}\norm{\l - w}_{\Sigman}$ over $w \in \Ln$, and let $\rhwn$ be its feasible analogue over $\hLn$ using $\hSigman$. Then
\begin{align}\label{arxiv2:app:eq:UK.minimax.bound}
    \norm{\rhwn - \rwn} \lesssim_{p, \cPn} h_{n} + \sqrt{\rhon d_{n} \curly{h_{n} + \eSigman d_{n}}}.
\end{align}
Hence the right side of \eqref{arxiv2:app:eq:UK.minimax.bound} converging to zero is sufficient for minimax-weight consistency. For the bounded variance class, the minimax center is GLS and the sharper bound
\begin{align}\label{arxiv2:app:eq:UK.GLS.bound}
    \frac{\norm{\rhwn - \rwn}_{\Sigman}}{\norm{\rwn}_{\Sigman}} \lesssim_{p, \cPn} \eSigman
\end{align}
holds directly from the GLS formula.
\end{propositionUK}

\begin{proof}
By \eqref{arxiv2:app:eq:UK.loewner}, changing $\Sigman$ to $\hSigman$ changes a quadratic norm by at most a constant times $\eSigman$ times its population value. Holding the covariance fixed, Hausdorff matching changes the radius by at most $\dsqrt{\UB{e}_{n}}h_{n}$, while changing the center contributes at most $\dsqrt{\UB{e}_{n}}\norm{\hwn - \wn}$. This proves \eqref{arxiv2:app:eq:UK.distance.bound}. The same argument applied to the distance from the singleton $\curly{0}$ proves \eqref{arxiv2:app:eq:UK.se.bound}.

For the minimax weights, project $\rhwn$ onto $\Ln$ and denote the projection by $\pi_{n}(\rhwn)$. The projection error is at most $h_{n}$. The preceding criterion comparisons and optimality of $\rhwn$ imply
\begin{align*}
    R_{n}(\pi_{n}(\rhwn)) - R_{n}(\rwn) \lesssim_{p, \cPn} \sqrt{\UB{e}_{n}}\curly{h_{n} + \eSigman d_{n}}, \quad R_{n}(w) = \max_{\l \in \Ln}\norm{\l - w}_{\Sigman}.
\end{align*}
The strong-convexity identity used in the fixed-$K$ proof gives
\begin{align*}
    \norm{\pi_{n}(\rhwn) - \rwn}_{\Sigman}^{2} \leq R_{n}^{2}(\pi_{n}(\rhwn)) - R_{n}^{2}(\rwn).
\end{align*}
Both radii are at most $\dsqrt{\UB{e}_{n}}d_{n}$. Using $\norm{x} \leq \norm{x}_{\Sigman}/\dsqrt{\LB{e}_{n}}$ and adding the projection error gives \eqref{arxiv2:app:eq:UK.minimax.bound}. Finally, the inverse identity in relative covariance coordinates gives \eqref{arxiv2:app:eq:UK.GLS.bound}.
\end{proof}

For the benchmark classes, \eqref{arxiv2:app:eq:UK.hausdorff} can be replaced by sharper class-specific identities when available. For example, the truncated simplex satisfies
\begin{align*}
    d_{H}\curly{(1 - \epsilon)w_{1} + \epsilon\W_{+,n}, (1 - \epsilon)w_{2} + \epsilon\W_{+,n}} = (1 - \epsilon)\norm{w_{1} - w_{2}},
\end{align*}
and the bounded variance class is homogeneous in the common covariance scale. The relevant large-$K$ conditions are therefore the rate statements in \eqref{arxiv2:app:eq:UK.relative.distance}-\eqref{arxiv2:app:eq:UK.minimax.bound}, not separate absolute bounds on $\UB{e}_{n}$ and $\LB{e}_{n}$.

\subsection{Asymptotic Minimax-Bias Optimality}
The population optimization comparison is unchanged when $K_{n}$ grows:
\begin{align*}
    \max_{\l \in \Ln} \norm{\l - \rwn}_{\Sigman}
    \leq
    \max_{\l \in \Ln} \norm{\l - \wn}_{\Sigman}
\end{align*}
for every $(n,P_{n})$. To translate this geometric inequality into an expected-bias comparison for feasible estimators, it is enough that, for each center under comparison, the three remainders
\begin{align*}
    \supPn \abs{\E[\Pn]{(\hwn - \wn)'(\hthetan - \thetan)}}, \quad
    \supPn \abs{\E[\Pn]{\wn'(\hthetan - \thetan)}}, \quad
    \supPn \abs{\E[\Pn]{(\hwn - \wn)'\thetan}}
\end{align*}
are $o(1)$. Under these remainder conditions and the same sharpness condition as in the fixed-$K$ analysis, the fixed-$K$ maximum-bias expansion and minimax comparison continue to hold. The only additional issue created by growing dimension is that the displayed remainders require dimension-appropriate moment and rate conditions; because they are not needed for the UCB or robust-CI results below, I do not develop primitive conditions for them here.

\subsection{Heterogeneity UCB}
For each $n$, let $F_{n}(x; \eta)$ denote the CDF of a noncentral chi-squared variable with $\nu_{n} = \Kn - 1$ degrees of freedom and noncentrality parameter $\eta^{2}$. For fixed $x > 0$, $F_{n}(x; \eta)$ is strictly decreasing in $\eta$. Let $F_{n}^{-1}(x; \beta)$ denote its inverse in the second argument whenever $F_{n}(x; 0) > \beta$, so that
\begin{align*}
    F_{n}\paren{x; F_{n}^{-1}(x; \beta)} = \beta.
\end{align*}
For $\hat{T}_{n} = n\hHn^{2}(\hthetan)$, define the large-$K$ heterogeneity UCB by
\begin{align*}
    \hetan = \frac{\tetan}{\sqrt{n}}, \quad 
    \tetan = 
    \begin{cases} 
    \hfil 0, & F_{n}(\hat{T}_{n}; 0) \leq \beta, \\ 
    \hfil F_{n}^{-1}(\hat{T}_{n}; \beta), & F_{n}(\hat{T}_{n}; 0) > \beta. 
    \end{cases}
\end{align*}

\begin{propositionUK}[Large-$K$ Heterogeneity UCB]\label{arxiv2:app:prop:UK.asymptotic.UCB}
Suppose Assumptions \ref{arxiv2:app:ass:UK.linear.CLT}-\ref{arxiv2:app:ass:UK.quadratic.CLT} hold and the transfer conditions in Proposition \ref{arxiv2:app:prop:UK.quadratic.transfer} are satisfied. Then
\begin{align}\label{arxiv2:app:eq:UK.UCB.coverage}
    \limn \supPn \abs{\P[\Pn]{\hHn(\thetan) \leq \hetan} - (1 - \beta)} \I{\Hn(\thetan) > 0} = 0.
\end{align}
In particular,
\begin{align*}
    \liminfn \infPn \P[\Pn]{\hHn(\thetan) \leq \hetan} \geq 1 - \beta.
\end{align*}
\end{propositionUK}

\begin{proof}
The proof follows the same inversion and subsequence structure as the fixed-$K$ UCB proof. For $\Hn(\thetan) > 0$, monotonicity of $\eta \mapsto F_{n}(\hat{T}_{n}; \eta)$ gives the exact event identity
\begin{align*}
    \curly{\hHn(\thetan) \leq \hetan} = \curly{F_{n}(\hat{T}_{n}; \hat{\eta}_{0,n}) \geq \beta}, \quad \hat{\eta}_{0,n} = \sqrt{n}\hHn(\thetan).
\end{align*}
It is therefore enough to show
\begin{align}\label{arxiv2:app:eq:UK.uniform.pivot}
    F_{n}(\hat{T}_{n}; \hat{\eta}_{0,n}) \to[d] U(0, 1)
\end{align}
along every sequence $P_{n} \in \mathcal{P}_{n}$ with nonzero heterogeneity $\Hn(\thetan) > 0$ for all $n$. Below I restrict to such sequences.

Consider an arbitrary subsequence. If $K_{n}$ is bounded, pass to a further subsequence on which it is constant. The fixed-$K$ proof then applies after covariance standardization. Its event inversion and probability-integral-transform argument are unchanged: bounded normalized heterogeneity is handled by the continuous mapping theorem, while diverging normalized heterogeneity is handled by the one-dimensional linear CLT and the normal approximation to a noncentral chi-squared distribution. The covariance-estimation inputs used in those steps are the feasible-population norm equivalence and relative stability of the heterogeneity vector; equations \eqref{arxiv2:app:eq:UK.loewner} and \eqref{arxiv2:app:eq:UK.vector.heterogeneity} provide exactly those inputs even though the common eigenvalue scale may drift.

Now consider a subsequence with $K_{n} \to \infty$. Define the standardized noncentral chi-squared CDF
\begin{align*}
    G_{n}(z; \eta) = F_{n}\paren{\nu_{n} + \eta^{2} + \sqrt{2\nu_{n} + 4\eta^{2}}z; \eta}.
\end{align*}
Uniform normal-approximation bounds for noncentral chi-squared distributions imply
\begin{align}\label{arxiv2:app:eq:UK.chisq.normal}
    \sup_{\eta \geq 0}\sup_{z \in \R} \abs{G_{n}(z; \eta) - \Phi(z)} \to 0
\end{align}
when $\nu_{n} \to \infty$; see \citet{seri2015tight}. By Proposition \ref{arxiv2:app:prop:UK.quadratic.transfer},
\begin{align*}
    Z_{n}^{Q} = \frac{\hat{T}_{n} - (\nu_{n} + \hat{\eta}_{0,n}^{2})}{\sqrt{2\nu_{n} + 4\hat{\eta}_{0,n}^{2}}} \to[d] N(0, 1).
\end{align*}
Because $F_{n}(\hat{T}_{n}; \hat{\eta}_{0,n}) = G_{n}(Z_{n}^{Q}; \hat{\eta}_{0,n})$, equation \eqref{arxiv2:app:eq:UK.chisq.normal} gives
\begin{align*}
    F_{n}(\hat{T}_{n}; \hat{\eta}_{0,n}) = \Phi(Z_{n}^{Q}) + o_{p}(1) \to[d] U(0, 1).
\end{align*}
This proves \eqref{arxiv2:app:eq:UK.uniform.pivot} and hence \eqref{arxiv2:app:eq:UK.UCB.coverage}. 
\end{proof}

Thus, the large-$K$ proof changes only the distributional step of the fixed-$K$ argument. The event inversion and reduction of UCB validity to a probability-integral-transform statement are identical. Along diverging-$K$ subsequences, the increasing degrees of freedom replace the fixed-$K$ proof's separate appeal to diverging noncentrality.

\subsection{Robust CI}
Maintain $\alpha \in (0,1)$. Let $\L_{0,n} \subseteq \Ln$ be a target class. As in the fixed-$K$ analysis, estimation of the working class need not imply containment of its population boundary. I therefore impose the direct target-containment condition
\begin{align}\label{arxiv2:app:eq:UK.target.containment}
    \limn \supPn \sup_{\l \in \L_{0,n}} \P[\Pn]{\l \notin \hLn} = 0.
\end{align}
A sufficient condition is the existence of a possibly shrinking slack $\nu_{0,n} > 0$ such that
\begin{align*}
    \sup_{\l \in \L_{0,n}}\gn(\l, \Sn) \leq - \nu_{0,n}, \quad \Lgn d_{\cS,n}(\hSn, \Sn) = o_{p, \cPn}(\nu_{0,n}).
\end{align*}

For proof purposes, define the population and feasible signal-based normalized biases
\begin{align*}
    b_{n} = \frac{\eta_{n}D_{n}}{\sigma_{w,n}}, \quad d_{\l,n} = \frac{\sqrt{n}(\l - \wn)'\thetan}{\sigma_{w,n}}, \quad \bar{b}_{n} = \frac{\hat{\eta}_{0,n}\hat{D}_{n}}{\hat{\sigma}_{w,n}}, \quad \bar{d}_{\l,n} = \frac{\sqrt{n}(\l - \hwn)'\thetan}{\hat{\sigma}_{w,n}}.
\end{align*}
The population bias bound gives $\abs{d_{\l,n}} \leq b_{n}$ for $\l \in \Ln$, while on $\curly{\l \in \hLn}$ its feasible analogue gives $\abs{\bar{d}_{\l,n}} \leq \bar{b}_{n}$. Assume the feasible score $T_{w,n} = \hwn'\hZn/\hat{\sigma}_{w,n}$ satisfies
\begin{align}\label{arxiv2:app:eq:UK.feasible.score}
    \limn \supPn \sup_{x \in \R} \abs{\P[\Pn]{T_{w,n} \leq x} - \Phi(x)} = 0, \quad T_{w,n} = \frac{\hwn'\hZn}{\hat{\sigma}_{w,n}}.
\end{align}
For a two-sided CI, additionally assume the effective-bias stability conditions
\begin{align}\label{arxiv2:app:eq:UK.bias.stability}
    \frac{\abs{\bar{b}_{n} - b_{n}}}{1 + b_{n}} = o_{p, \cPn}(1), \quad \sup_{\l \in \L_{0,n}} \frac{\abs{\bar{d}_{\l,n} - d_{\l,n}}}{1 + b_{n}} = o_{p, \cPn}(1).
\end{align}

\begin{propositionUK}[Large-$K$ Robust CI]\label{arxiv2:app:prop:UK.asymptotic.robust.CI}
Suppose the conditions of Proposition \ref{arxiv2:app:prop:UK.asymptotic.UCB}, target containment \eqref{arxiv2:app:eq:UK.target.containment}, and the feasible-score approximation \eqref{arxiv2:app:eq:UK.feasible.score} hold. Then every one-sided robust CI constructed from $(\hLn, \hwn, \hSigman, \hetan)$ satisfies
\begin{align}\label{arxiv2:app:eq:UK.robust.coverage.one}
    \liminfn \infPn \inf_{\l \in \L_{0,n}} \P[\Pn]{\tauwn[\l] \in \rCIn} \geq 1 - (\alpha + \beta).
\end{align}
If, in addition, \eqref{arxiv2:app:eq:UK.bias.stability} holds, then the same conclusion holds for the two-sided robust CI.
\end{propositionUK}

\begin{proof}
Take arbitrary sequences $P_{n} \in \mathcal{P}_{n}$ and $\l_{n} \in \L_{0,n}$. Let
\begin{align*}
    \cE_{H,n} = \curly{\hHn(\thetan) \leq \hetan}, \quad \cE_{\L,n} = \curly{\l_{n} \in \hLn}.
\end{align*}
Proposition \ref{arxiv2:app:prop:UK.asymptotic.UCB} and \eqref{arxiv2:app:eq:UK.target.containment} imply
\begin{align}\label{arxiv2:app:eq:UK.good.event}
    \liminf_{n \to \infty} \P[\Pn]{\cE_{H,n}\cap\cE_{\L,n}} \geq 1 - \beta.
\end{align}
On this event, the bias UCB used by the CI is at least the true feasible bias:
\begin{align*}
    \sqrt{n}\hetan \frac{\hat{D}_{n}}{\hat{\sigma}_{w,n}} \geq \bar{b}_{n} \geq \abs{\bar{d}_{\l_{n},n}}.
\end{align*}

For an upper one-sided interval, the good-event noncoverage probability satisfies
\begin{align*}
    \P[\Pn]{\tauwn[\l_{n}] \notin \rCIn,\ \cE_{H,n} \cap \cE_{\L,n}} &\leq
    \P[\Pn]{\bar{d}_{\l_{n},n} - T_{w,n} > z_{1-\alpha} + \abs{\bar{d}_{\l_{n},n}},\ \cE_{H,n} \cap \cE_{\L,n}} \\
    &\leq \P[\Pn]{T_{w,n} < z_{\alpha}}.
\end{align*}
Indeed, $\abs{d} - d \geq 0$ for every real $d$, so the event in the second line implies $T_{w,n} < z_{\alpha}$. The last probability converges to $\alpha$ by \eqref{arxiv2:app:eq:UK.feasible.score}; the lower one-sided case is identical. Combining this bound with \eqref{arxiv2:app:eq:UK.good.event} proves \eqref{arxiv2:app:eq:UK.robust.coverage.one} for one-sided intervals.

For the two-sided interval, monotonicity of $\cv{b}$ and the bound in the good event give
\begin{align}\label{arxiv2:app:eq:UK.two.sided.event}
    \P[\Pn]{\tauwn[\l_{n}] \notin \rCIn,\ \cE_{H,n} \cap \cE_{\L,n}}
    \leq
    \P[\Pn]{\abs{\bar{d}_{\l_{n},n} - T_{w,n}} > \cv{\bar{b}_{n}},\ \cE_{\L,n}}.
\end{align}
The fixed-$K$ proof obtains deterministic subsequential limits from compactness of the finite-dimensional primitives. Since that route is unavailable when $K_{n}$ grows, condition \eqref{arxiv2:app:eq:UK.bias.stability} supplies directly the two scalar comparisons that the folded-normal argument needs. Pass to a subsequence along which $b_{n}$ is bounded or diverges. 

If $b_{n}$ is bounded, pass further so that
\begin{align*}
    b_{n} \to b, \quad d_{\l_{n},n} \to d, \quad \abs{d} \leq b.
\end{align*}
Then \eqref{arxiv2:app:eq:UK.bias.stability} gives $\bar{b}_{n} \to[p] b$ and $\bar{d}_{\l_{n},n} \to[p] d$. Together with \eqref{arxiv2:app:eq:UK.feasible.score} and continuity of the folded-normal critical value,
\begin{align*}
    \limsup_{n \to \infty} \P[\Pn]{\abs{\bar{d}_{\l_{n},n} - T_{w,n}} > \cv{\bar{b}_{n}}, \ \cE_{\L,n}} \leq \P[]{\abs{d - G} > \cv{b}} \leq \alpha,
\end{align*}
where $G \sim N(0, 1)$ and the last inequality uses $\abs{d} \leq b$ and monotonicity of folded-normal tail probabilities in the magnitude of the mean.

If $b_{n} \to \infty$, pass further so that $d_{\l_{n},n}/b_{n} \to c \in [-1,1]$. Condition \eqref{arxiv2:app:eq:UK.bias.stability} then gives
\begin{align*}
    \bar{b}_{n} \to[p] \infty, \qquad
    \frac{\bar{d}_{\l_{n},n}}{\bar{b}_{n}} \to[p] c.
\end{align*}
Moreover, $\abs{\bar{d}_{\l_{n},n}} \leq \bar{b}_{n}$ on $\cE_{\L,n}$, and \eqref{arxiv2:app:eq:UK.feasible.score} gives $T_{w,n} \to[d] N(0,1)$. Lemma \ref{arxiv2:app:lemma:diverging.folded.normal}, applied with
\begin{align*}
    A_{n} = \bar{d}_{\l_{n},n}, \quad
    B_{n} = \bar{b}_{n}, \quad
    T_{n} = T_{w,n}, \quad
    \cE_{n} = \cE_{\L,n},
\end{align*}
therefore shows that the right-hand side of \eqref{arxiv2:app:eq:UK.two.sided.event} has limit superior at most $\alpha$. The same conclusion was established above when $b_{n}$ is bounded. Combining the good-event noncoverage bound with \eqref{arxiv2:app:eq:UK.good.event} proves the two-sided result for every $\alpha \in (0,1)$.
\end{proof}

For a prespecified center $\hwn = \wn$, Assumption \ref{arxiv2:app:ass:UK.linear.CLT} gives the oracle-score normal approximation, while Assumption \ref{arxiv2:app:ass:UK.consistent.covariance} gives $\hat{\sigma}_{w,n}/\sigma_{w,n} \to[p,\cPn] 1$; together they imply \eqref{arxiv2:app:eq:UK.feasible.score}. For estimated weights, the same conclusion follows if
\begin{align*}
    \frac{(\hwn-\wn)'\hZn}{\sigma_{w,n}} = o_{p,\cPn}(1),
    \qquad
    \frac{\hat{\sigma}_{w,n}}{\sigma_{w,n}} \to[p,\cPn] 1.
\end{align*}
A convenient sufficient condition for the first display is
\begin{align*}
    \sqrt{\Kn}\frac{\norm{\hwn-\wn}_{\Sigman}}{\sigma_{w,n}} = o_{p,\cPn}(1),
    \quad
    \norm{\Sigman^{-1/2}\hZn} = O_{p,\cPn}(\sqrt{\Kn}).
\end{align*}
For feasible GLS, \eqref{arxiv2:app:eq:UK.GLS.bound} and $\dsqrt{\Kn}\eSigman = o_{p,\cPn}(1)$ verify this route. Formulations based on conditioning or sample splitting can avoid having to appeal to such conditions when using estimated weights; see \citet{coussens2021improving} for an illustration of sample splitting.

For effective-bias stability, relative consistency of $\hat{\eta}_{0,n}$, $\hat{D}_{n}$, and $\hat{\sigma}_{w,n}$ gives the first condition in \eqref{arxiv2:app:eq:UK.bias.stability} when $\eta_{n}$, $D_{n}$, and $\sigma_{w,n}$ are positive. If $\eta_{n} = 0$, then $\hat{\eta}_{0,n} = 0$ exactly. If $D_{n} = 0$, it suffices to verify $\bar{b}_{n} = o_{p,\cPn}(1)$ directly. For the target-bias term, the identity
\begin{align*}
    \Bar{d}_{\l,n} - d_{\l,n}
    = \Bar{d}_{\l,n}\paren{1-\frac{\hat\sigma_{w,n}}{\sigma_{w,n}}}
    - \frac{\sqrt n(\hwn-\wn)'\thetan}{\sigma_{w,n}}
\end{align*}
and $\1'(\hwn-\wn) = 0$ imply
\begin{align*}
    \frac{\sqrt n\abs{(\hwn-\wn)'\thetan}}{\sigma_{w,n}}
    \leq
    \eta_{n}\frac{\norm{\hwn-\wn}_{\Sigman}}{\sigma_{w,n}}.
\end{align*}
Thus, relative standard-error consistency and
\begin{align}\label{arxiv2:app:eq:UK.center.bias.rate}
    \eta_{n}\frac{\norm{\hwn-\wn}_{\Sigman}}{\sigma_{w,n}} = o_{p,\cPn}(1+b_{n})
\end{align}
first imply $\sup_{\l \in \L_{0,n}}\abs{\Bar{d}_{\l,n}} = O_{p,\cPn}(1+b_{n})$ because $\abs{d_{\l,n}} \leq b_{n}$. Substitution into the preceding identity then gives the second condition in \eqref{arxiv2:app:eq:UK.bias.stability}. When $D_{n} > 0$, condition \eqref{arxiv2:app:eq:UK.center.bias.rate} follows from
\begin{align*}
    \frac{\norm{\hwn-\wn}_{\Sigman}}{D_{n}} = o_{p,\cPn}(1),
\end{align*}
since $b_{n} = \eta_{n}D_{n}/\sigma_{w,n}$. The relative distance, standard-error, and heterogeneity bounds derived above supply the remaining relative-consistency components. 

\end{document}